\documentclass{article}

\usepackage[margin=1in]{geometry}
\usepackage[utf8]{inputenc}
\usepackage{bbm, bm}
\usepackage{graphicx}
\usepackage{color, xcolor}
\usepackage{amsmath, amsfonts, amssymb, amsthm, thmtools, mathtools}
\usepackage{algorithm, algpseudocode}
\usepackage{enumitem}

\PassOptionsToPackage{hyphens}{url}
\usepackage[colorlinks=true, allcolors=blue]{hyperref}
\usepackage[capitalise,nameinlink]{cleveref}

\usepackage[style=alphabetic, backend=biber, minalphanames=3, maxalphanames=4, maxbibnames=99]{biblatex}

\crefformat{section}{#2\S#1#3}
\crefformat{subsection}{#2\S#1#3}
\crefformat{subsubsection}{#2\S#1#3}
\Crefformat{section}{#2\S#1#3}
\Crefformat{subsection}{#2\S#1#3}
\Crefformat{subsubsection}{#2\S#1#3}

\Crefrangeformat{section}{#3\S\S#1–#2#4}
\Crefmultiformat{section}{#2\S\S#1#3}{ and #2#1#3}{, #2#1#3}{, and #2#1#3}

\declaretheoremstyle[bodyfont=\it,qed=\qedsymbol]{noproofstyle}

\numberwithin{equation}{section}

\declaretheorem[name=Observation,numbered=no]{observation*}

\declaretheorem[numberlike=equation]{fact}

\declaretheorem[numberlike=equation]{theorem}

\declaretheorem[name=Theorem,numbered=no]{theorem*}

\declaretheorem[numberlike=equation]{lemma}
\declaretheorem[name=Lemma,numbered=no]{lemma*}

\declaretheorem[numberlike=equation]{corollary}
\declaretheorem[name=Corollary,numbered=no]{corollary*}

\declaretheorem[numberlike=equation]{proposition}
\declaretheorem[name=Proposition,numbered=no]{proposition*}

\declaretheorem[numberlike=equation]{claim}
\declaretheorem[name=Claim,numbered=no]{claim*}

\declaretheorem[name=Conjecture,numbered=no]{conjecture*}

\declaretheorem[name=Question,numbered=no]{question*}

\declaretheoremstyle[
  headfont=\normalfont\bfseries,
  bodyfont=\normalfont,
  postheadspace=1em,
  qed=$\lozenge$
]{rmkstyle}

\declaretheorem[
  style=rmkstyle,
  numberlike=equation
]{remark}

\declaretheorem[numberlike=equation,style=rmkstyle]{definition}
\declaretheorem[unnumbered,name=Definition,style=rmkstyle]{definition*}

\declaretheorem[numberlike=equation,style=rmkstyle]{example}
\declaretheorem[unnumbered,name=Example,style=rmkstyle]{example*}

\declaretheorem[unnumbered,name=Notation=rmkstyle]{notation*}

\declaretheorem[unnumbered,name=Construction,style=rmkstyle]{construction*}

\crefname{fact}{Fact}{Facts}
\Crefname{fact}{Fact}{Facts}
\crefname{claim}{Claim}{Claims}
\Crefname{fact}{Claim}{Claims}
\usepackage{singer-macros}
\usepackage{bbm}
\usepackage{subcaption}
\usepackage{tikz}
\usetikzlibrary{shapes.geometric, arrows.meta, matrix, calc, positioning}

\usetikzlibrary{backgrounds}
\usetikzlibrary{patterns}
\usetikzlibrary{decorations.pathmorphing}

\newcommand{\addwire}[3]{ 
  \draw (0,-#2) -- (#1,-#2);
  \node at (-0.5,-#2) {#3};
}
\newcommand{\addwiregroup}[3]{ 
  \draw (0,-#2) -- (#1,-#2);
  \draw (0,-#2-.05) -- (#1,-#2-.05);
  \draw (0,-#2+.05) -- (#1,-#2+.05);
  \node at (-0.5,-#2) {#3};
}
\newcommand{\addmoat}[3]{ 
  \fill[pattern=north east lines]  (.5,-#2-.4) rectangle (#1-.5,-#2-.6);
  \node at (0.25,-#2-.5) {#3};
}

\newcommand{\addgate}[4]{ 
  \draw (#1,-#2) -- (#1,-#3-.15); 
  \node[draw=black, fill=white] at (#1,-#2) {#4}; 
  \draw (#1,-#3) circle (.15); 
}
\newcommand{\addselfgate}[3]{ 
  \draw[fill=white, draw=black]
    (#1-.3,-#2+.3) -- (#1+.3,-#2+.3) -- (#1+.3,-#2-.3) -- cycle;
  \node[inner sep=0, fill=white] at (#1,-#2) {#3}; 
}
\newcommand{\addgategroup}[4]{ 
  \draw (#1,-#2) -- (#1,-#3); 
  \draw (#1-.05,-#2) -- (#1-.05,-#3-.15); 
  \draw (#1+.05,-#2) -- (#1+.05,-#3-.15); 
  \node[draw=black, fill=white] at (#1,-#2) {#4}; 
  \draw[fill=white] (#1,-#3) circle (.2); 
  \draw (#1,-#3-.2) -- (#1,-#3+.2);
  \draw (#1-.2,-#3) -- (#1+.2,-#3);
}

\allowdisplaybreaks

\title{Low-soundness direct-product testers and PCPs \\
from Kaufman--Oppenheim complexes}
\author{Ryan O'Donnell\thanks{Computer Science Department, Carnegie Mellon University. \texttt{\{odonnell,ngsinger\}@cs.cmu.edu}.\\${~~}^{\text{\tiny \textregistered}}$\!\,Author order randomized.}${~}^{\text{\small \textregistered}}$ \and Noah G. Singer\footnotemark[1]${~}^{\text{\small \textregistered}}$}
\date{\today}

\usepackage{xspace}
\usepackage{mathtools}
\usepackage{stmaryrd}

\newcommand{\CplxA}[3]{\mathfrak{A}_{#1}^{\kappa,#2}(#3)}
\newcommand{\LinkCplxA}[2]{\mathfrak{LA}_{#2}^{\kappa}(#1)}

\newcommand{\GrF}[2]{#2^{(#1)}}
\newcommand{\SL}[2]{\mathsf{SL}_{#1}(#2)}
\newcommand{\EL}[2]{\mathsf{EL}_{#1}(#2)}
\newcommand{\GrSL}[3]{\SL{#1}{\GrF{#2}{#3}}}
\newcommand{\CGrUnip}[4]{\mathsf{Unip}_{#1,#2}^{\kappa,(#3)}(#4)}

\newcommand{\GrUnip}[2]{\mathrm{Unip}_{#2}^{\kappa}(#1)}
\newcommand{\GrStair}[3]{\mathrm{Stair}_{#2}^\kappa(#1; #3)}
\newcommand{\GrTri}[3]{\mathrm{Triangle}_{#2}^\kappa(#1; #3)}
\newcommand{\GrTriII}[4]{\GrTri{#1}{#2}{\rangeII{#3}{#4}}}
\newcommand{\GrTriIE}[4]{\GrTri{#1}{#2}{\rangeIE{#3}{#4}}}
\newcommand{\GrTriEE}[4]{\GrTri{#1}{#2}{\rangeEE{#3}{#4}}}

\newcommand{\basic}[1]{\mathrm{BasicMat}^\kappa(#1)}

\newcommand{\retract}[3]{\pi^{(#1;#3)}_#2}
\newcommand{\retractII}[4]{\pi^{(#1;#3,#4)}_#2}
\newcommand{\isom}[3]{\Phi^{(#1;#3)}_#2}
\newcommand{\isomII}[4]{\Phi^{(#1;#3,#4)}_#2}

\newcommand{\Comm}[2]{\bracks*{ #1,#2 } }

\newcommand{\El}[1]{\braces*{ #1 }}
\newcommand{\Let}[1]{\angles*{ #1 }}

\newcommand{\RSt}{\calR_{\mathrm{Steinberg}}}

\newcommand{\rangeSep}{\mathinner{.\,.}}
\newcommand{\rangeIE}[2]{[#1\rangeSep#2)}
\newcommand{\rangeII}[2]{[#1\rangeSep#2]}
\newcommand{\rangeEE}[2]{(#1\rangeSep#2)}
\newcommand{\rangeEI}[2]{(#1\rangeSep#2]}
\newcommand{\rangeOne}[1]{\rangeII{1}{#1}}

\newcommand{\typA}{{\mathtt{A}}}
\newcommand{\typB}{{\mathtt{B}}}
\newcommand{\typC}{{\mathtt{C}}}
\newcommand{\typD}{{\mathtt{D}}}
\newcommand{\typF}{{\mathtt{F}}}
\newcommand{\typG}{{\mathtt{G}}}
\newcommand{\typH}{{\mathtt{H}}}
\newcommand{\typP}{{\mathtt{P}}}
\newcommand{\typQ}{{\mathtt{Q}}}
\newcommand{\typZ}{{\mathtt{Z}}}

\newcommand{\typL}{{\lambda}}
\newcommand{\Types}{\mathbb{L}}

\newcommand{\alias}[1]{\mathrm{alias} \parens*{ #1 }}
\newcommand{\ElBlockZ}[4]{\zeta^{#1}_{#2;#3} \parens*{ #4 } }

\newcommand{\Succ}{\operatorname{succ}}

\newcommand{\zero}{e}

\newcommand{\Mats}[1]{\mathrm{Mat}_{#1}}

\newcommand{\diamA}{\Mats\typA}
\newcommand{\diamB}{\Mats\typB}
\newcommand{\diamC}{\Mats\typC}
\newcommand{\diamF}{\Mats\typF}
\newcommand{\diamG}{\Mats\typG}
\newcommand{\diamP}{\Mats\typP}

\newcommand{\cbdyA}{\Mats\typA}
\newcommand{\cbdyB}{\Mats\typB}
\newcommand{\cbdyC}{\Mats\typC}
\newcommand{\cbdyD}{\Mats\typD}
\newcommand{\cbdyF}{\Mats\typF}
\newcommand{\cbdyG}{\Mats\typG}
\newcommand{\cbdyH}{\Mats\typH}
\newcommand{\cbdyP}{\Mats\typP}
\newcommand{\cbdyQ}{\Mats\typQ}
\newcommand{\cbdyZ}{\Mats\typZ}

\newcommand{\rhoP}{\rho_{\mathrm{P}}}
\newcommand{\rhoH}{\rho_{\mathrm{H}}}

\newcommand{\partA}[1]{a(#1)}
\newcommand{\partB}[1]{b(#1)}
\newcommand{\partC}[1]{c(#1)}
\newcommand{\partD}[1]{d(#1)}

\newcommand{\blockF}[4]{F^{#1}_{#2;#3} \parens*{ #4 }}
\newcommand{\blockG}[4]{G^{#1}_{#2;#3} \parens*{ #4 }}
\newcommand{\blockH}[4]{H^{#1}_{#2;#3} \parens*{ #4 }}
\newcommand{\blockP}[4]{P^{#1}_{#2;#3} \parens*{ #4 }}
\newcommand{\blockQ}[4]{Q^{#1}_{#2;#3} \parens*{ #4 }}
\newcommand{\blockZ}[4]{Z^{#1}_{#2;#3} \parens*{ #4 }}

\newcommand{\Mbody}{M_{\mathrm{body}}}
\newcommand{\Mtail}{M_{\mathrm{tail}}}
\newcommand{\Mgap}{M_{\mathrm{gap}}}
\newcommand{\Nbody}{N_{\mathrm{body}}}
\newcommand{\Ntail}{N_{\mathrm{tail}}}
\newcommand{\Ngap}{N_{\mathrm{gap}}}

\newcommand{\initA}{a}
\newcommand{\initB}{b}
\newcommand{\initC}{c}
\newcommand{\initD}{d}

\newcommand{\WiresA}{{\mathcal{I}_{\mathrm{A}}}}
\newcommand{\WiresB}{{\mathcal{I}_{\mathrm{B}}}}
\newcommand{\WiresC}{{\mathcal{I}_{\mathrm{C}}}}
\newcommand{\WiresD}{{\mathcal{I}_{\mathrm{D}}}}
\newcommand{\Wires}{{\mathcal{I}_n}}
\newcommand{\InitsA}{{\mathcal{S}_{\mathrm{A}}}}
\newcommand{\InitsB}{{\mathcal{S}_{\mathrm{B}}}}
\newcommand{\InitsC}{{\mathcal{S}_{\mathrm{C}}}}
\newcommand{\InitsD}{{\mathcal{S}_{\mathrm{D}}}}

\newcommand{\I}[1]{{\mathcal{I}_{#1}^m}}

\newcommand{\IntA}[1]{{\mathcal{I}_\mathrm{A}^{#1}}}
\newcommand{\IntB}[1]{{\mathcal{I}_\mathrm{B}^{#1}}}
\newcommand{\IntC}[1]{{\mathcal{I}_\mathrm{C}^{#1}}}
\newcommand{\IntD}[1]{{\mathcal{I}_\mathrm{D}^{#1}}}

\newcommand{\IntPartA}[1]{\IntA{\partA{#1}}}
\newcommand{\IntPartB}[1]{\IntB{\partB{#1}}}
\newcommand{\IntPartC}[1]{\IntC{\partC{#1}}}
\newcommand{\IntPartD}[1]{\IntD{\partD{#1}}}

\newcommand{\blockPartH}[4]{H^{#1}_{\partC{#2};\partD{#3}} \parens*{ #4 }}
\newcommand{\blockPartP}[4]{P^{#1}_{\partA{#2};\partC{#3}} \parens*{ #4 }}

\newcommand{\blockPartZ}[4]{Z^{#1}_{\partA{#2};\partD{#3}} \parens*{ #4 }}

\newcommand{\nA}{{n_{\mathrm{A}}}}
\newcommand{\nB}{{n_{\mathrm{B}}}}
\newcommand{\nC}{{n_{\mathrm{C}}}}
\newcommand{\nD}{{n_{\mathrm{D}}}}

\newcommand{\powF}{t_{\mathrm{F}}}
\newcommand{\powG}{t_{\mathrm{G}}}
\newcommand{\powH}{t_{\mathrm{H}}}
\newcommand{\powP}{t_{\mathrm{P}}}
\newcommand{\powQ}{t_{\mathrm{Q}}}
\newcommand{\powZ}{t_{\mathrm{Z}}}

\newcommand{\mF}{m_{\mathrm{F}}}
\newcommand{\mG}{m_{\mathrm{G}}}
\newcommand{\mH}{m_{\mathrm{H}}}

\newcommand{\GU}{\mathsf{GU}}
\newcommand{\GUa}{\mathsf{GU}_{\mathrm{A}}}
\newcommand{\GUb}{\mathsf{GU}_{\mathrm{B}}}
\newcommand{\GUc}{\mathsf{GU}_{\mathrm{C}}}

\newcommand{\ia}{\bar{a}}
\newcommand{\ib}{\bar{b}}
\newcommand{\ic}{\bar{c}}

\newcommand{\X}{\mathfrak{X}}

\newcommand{\res}[2]{#1[#2]}
\newcommand{\resLink}[3]{#1_{#3}[#2]}
\newcommand{\link}[2]{#1_{#2}}

\newcommand{\Avg}{\mathfrak{A}}

\newcommand{\ori}[1]{\vec{#1}}

\newcommand{\CoCo}[2]{\mathfrak{CC}(#1;#2)}

\newcommand{\Sym}[1]{\mathfrak{S}(#1)}

\newlist{conditions}{enumerate}{1}
\setlist[conditions]{label=\theconditionsi.,ref=\theconditionsi}

\crefname{conditionsi}{condition}{conditions}
\Crefname{conditionsi}{Condition}{Conditions}

\newcommand{\conditem}[1]{%
  \item\emph{#1:}\\
}

\newcommand{\ptiz}[3]{#2 \perp_{#1} #3}

\newcommand{\pack}[3]{\mathrm{pack} \parens*{ #1; #2, #3 }}
\newcommand{\Tri}{\mathfrak{T}}

\newcommand{\gdiam}[2]{\operatorname{diam}(#1;#2)}
\newcommand{\garea}[3]{\operatorname{area}^{#1}(#2;#3)}

\newcommand{\instAdd}[2]{\mathtt{add}~#1~\mathtt{to}~#2\mathtt{;}}
\newcommand{\instSub}[2]{\mathtt{subtract}~#1~\mathtt{from}~#2\mathtt{;}}

\newcommand{\partOf}[2]{\Omega^{#2}}
\newcommand{\tildepartOf}[2]{\tilde{\Omega}^{#2}}
\newcommand{\universe}{universe\xspace}
\newcommand{\universes}{universes\xspace}
\newcommand{\strng}{string\xspace}
\newcommand{\strngs}{strings\xspace}
\newcommand{\substring}{substring\xspace}
\newcommand{\Substrings}{Substrings\xspace}
\newcommand{\substrings}{substrings\xspace}

\newcommand{\symbls}{symbols\xspace}
\newcommand{\face}
{word\xspace}
\newcommand{\faces}{words\xspace}

\newcommand{\Faces}{Words\xspace}
\newcommand{\triword}{triword\xspace}
\newcommand{\triwords}{triwords\xspace}
\newcommand{\biword}{biword\xspace}
\newcommand{\biwords}{biwords\xspace}

\newcommand{\Kpin}{K_{\mathrm{pin}}}
\newcommand{\epin}{\epsilon_{\mathrm{pin}}}

\newcommand{\area}[2]{\mathrm{area}_{#2}\parens*{ #1 }}
\newcommand{\SubRels}[3]{\calR_{\mathrm{common}}^{#3}(#1;#2)}

\newcommand{\eval}{\operatorname{eval}}
\newcommand{\Eq}[2]{\text{``}#1\equiv#2\text{''}}
\newcommand{\Embed}[1]{\langle #1 \rangle}

\newcommand{\Fle}[1]{\F^{\le #1}[X]}

\newcommand{\sorl}{s}

\date{\today}

\begin{document}

\maketitle


\begin{abstract}
     We study the Kaufman--Oppenheim coset complexes~\cite{KO23}, which have an elementary and strongly explicit description.  Answering an open question of~\cite{KOW25}, we show that they support sparse direct-product testers in the low soundness regime.
     Our proof relies on the HDX characterization of agreement testing by \cite{BM24,DD24-covers}, the recent result of \cite{KOW25}, and follows  techniques from~\cite{BLM24,DDL24}.     
     Ultimately, the task reduces to showing dimension-independent coboundary expansion of certain $2$-dimensional subcomplexes of the KO complex; following the ``Dehn method'' of~\cite{KO21},  we do this by establishing efficient presentation bounds for  certain matrix groups over polynomial rings.
     
     As shown in~\cite{BMVY25}, a consequence of our direct-product testing result is that the Kaufman--Oppenheim complexes can also be used to obtain PCPs with arbitrarily small constant soundness and quasilinear length.  Thus the use of sophisticated number theory and algebraic group-theoretic tools in the construction of these PCPs can be avoided.
\end{abstract}

\section{Introduction} \label{sec:intro}

\renewcommand{\floatpagefraction}{.8}
\renewcommand{\textfraction}{.1}
\begin{figure}
\begin{tikzpicture}[
  node distance=0.5cm,
  ourthm/.style={rectangle, draw, fill=green!20!white, minimum width=2.5cm, minimum height=1cm, align=center},
  theirthm/.style={rectangle, draw, minimum width=2.5cm, minimum height=1cm, align=center},
  reduction/.style={rectangle, rounded corners=8pt, draw, minimum width=2.5cm, minimum height=1cm, align=center},
  arrow/.style={-{Latex[length=2mm]}, thick}]
\node[ourthm] (pcp) {KO complex can be used within low-soundness quasilinear PCPs};
\node[reduction, above=of pcp] (bmvy) {\cite{BMVY25}};
\node[ourthm, above=of bmvy] (dp) {\Cref{thm:main-dp}: KO complex is a direct-product tester};
\node[reduction, above=of dp] (bm) {\cite{BM24,DD24-covers}};
\node[theirthm, left=of bm] (kow) {KO complex has no nontrivial $m$-covers~\cite{KOW25}};
\node[ourthm, above=of bm] (us) {\Cref{thm:main-coboundary}: KO complex links have triword expansion};
\node[reduction, above=of us] (blm) {Induction/restriction methods for \\ showing coboundary expansion from~\cite{BLM24,DDL24}};
\node[ourthm, above=of blm] (us2) {\Cref{thm:main-separated}: Separated restrictions of KO link-complex are coboundary expanders};
\node[reduction, above=of us2] (ko) {``Absolute Dehn method'' from~\cite{KO21,OS25}};
\node[ourthm, above=of ko] (us3) {\Cref{thm:main-unip}: Derivation bounds in unipotent group};
\draw[arrow] (bmvy) -- (pcp);
\draw[arrow] (dp) -- (bmvy);
\draw[arrow] (bm) -- (dp);
\draw[arrow] (kow) -- (bm);
\draw[arrow] (us) -- (bm);
\draw[arrow] (blm) -- (us);
\draw[arrow] (us2) -- (blm);
\draw[arrow] (ko) -- (us2);
\draw[arrow] (us3) -- (ko);
\end{tikzpicture}
\caption{A schematic diagram of our proof; the contributions of our work are shaded.
Rounded rectangles are ``reduction'' theorems.}
\end{figure}
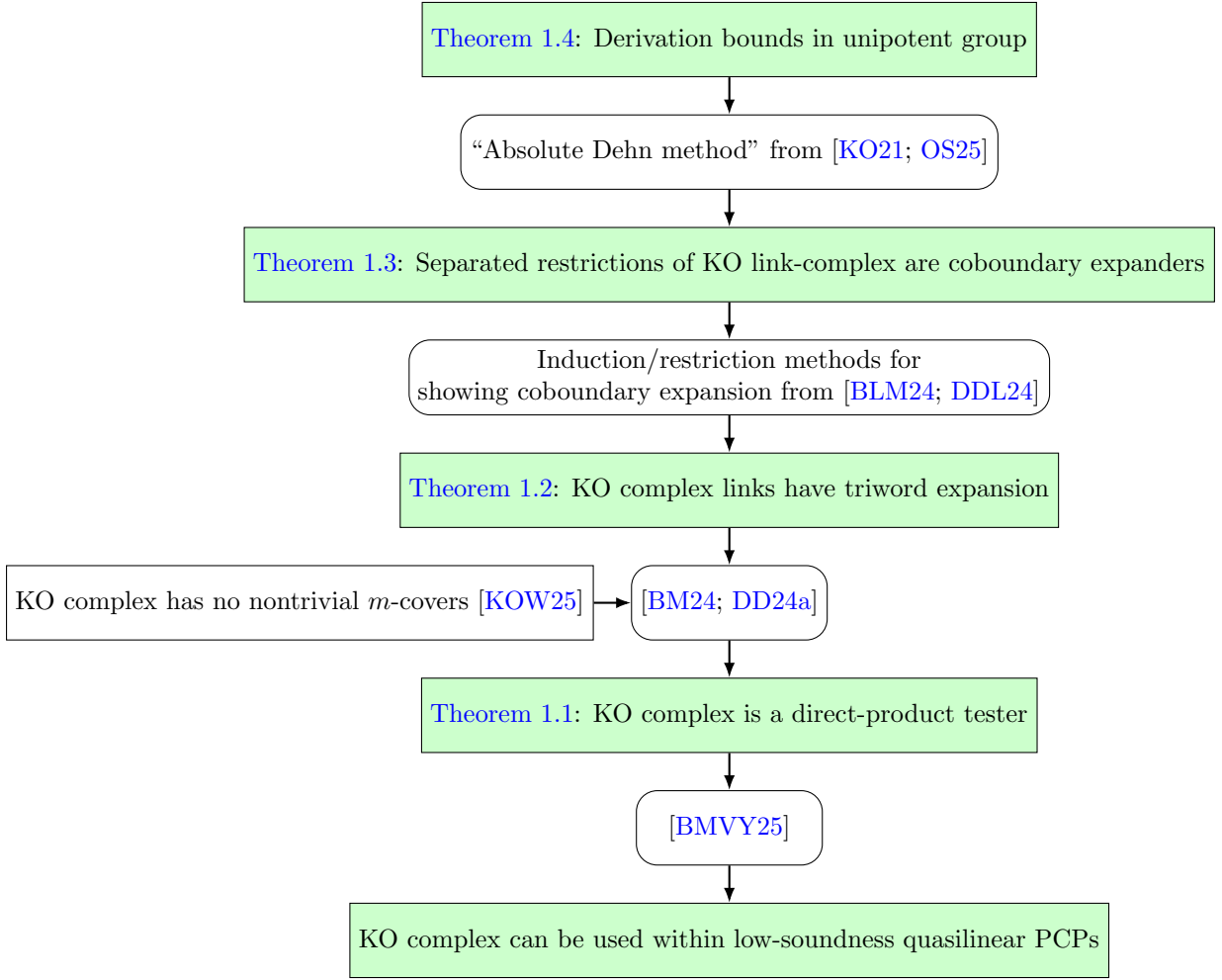

The main result of our work is an elementary, strongly explicit, sparse, two-query direct-product tester, based on the high-dimensional expanders of Kaufman, Oppenheim, and Weinberger~\cite{KO23,KOW25}: 
\begin{theorem}[Direct-product testing application theorem]\label{thm:main-dp}
    Given $\delta > 0$, provided $k \geq \exp(\poly(1/\delta))$, $n \geq \exp(\poly(k))$, and $p \geq \poly(n)$ is prime, the \cite{KOW25} complexes  $\{ \mathfrak{A}_n(\F_{Q}) : Q = p^\sorl,\ \sorl > 3n\}$ yield a strongly explicit family of $k$-uniform hypergraphs with $N \approx Q^{(n+1)^2 - 1}$ vertices and $p^{O(n^2)} \cdot N$ hyperedges, such that performing the standard ``V-test'' with them gives a $\delta$-sound direct-product tester.
\end{theorem}
\noindent Previously, Bafna--Lifshitz--Minzer~\cite{BLM24} and Dikstein--Dinur--Lubotzky~\cite{DDL24} independently showed a similar result for complexes constructed as quotients of affine Bruhat--Tits buildings associated to the symplectic group over the $p$-adics.
These complexes require deep number theory and algebraic group theory to construct and analyze, and at a technical level are not known to be strongly explicit.  By contrast, in \Cref{sec:high-dp-concrete} we will give a completely concrete description of the direct-product tester described in \Cref{thm:main-dp}; moreover, its analysis, though lengthy, can arguably be described as elementary.

\medskip

As shown by Bafna, Minzer, and Vyas~\cite{BMVY25}, direct-product testers as in \cite{BLM24,DDL24} and \Cref{thm:main-dp} can be successfully put to use for their original purpose: highly efficient PCPs. More precisely, \cite{BMVY25} constructs from them an efficient PCP system for $\mathsf{NP}$ with the shortest known proof lengths, $n \mapsto n \cdot \polylog n$, for any small constant soundness $\delta > 0$.
Our \Cref{thm:main-dp} provides an alternative ingredient that allows their entire construction to be elementary.  It is our hope that the improved simplicity of this construction will help lead to quantitative improvements in the future.\footnote{For example, we are hopeful that \Cref{thm:main-dp}'s strong explicitness property will lead to improved prover/verifier \emph{time-efficiency} in the PCP, but we leave this for future work.}

\medskip

To establish \Cref{thm:main-dp}, we use a characterization of the precise high-dimensional expansion (HDX) properties needed to obtain a (two-query, low-soundness) direct-product tester from a simplicial complex.  
This characterization was determined by Bafna--Minzer~\cite{BM24} and Dikstein--Dinur~\cite{DD24-covers} following a sequence of important works including~\cite{DM22,GK23-list,DD24-local}.
Assuming the $k$-uniform hypergraph tester is embedded inside an $n$-dimensional complex~$\X$ (for $n \gg k$) with excellent spectral expansion, roughly speaking the following two properties of~$\X$ are needed (see \Cref{sec:high-coboundary} for an explanation of terms):
\begin{enumerate}
    \item \label{itm:global} The global hypergraph $\X$ should have a qualitative connectedness property: no nontrivial $m$-covers for $m \leq \poly(1/\delta)$ (aka ``trivial $1$-cohomology over the symmetric group $\Sym{m}$'').
    \item \label{itm:local} The local neighborhood (``link'') of each vertex of~$\X$ should have a quantitatively strong form of topological expansion vis-\`{a}-vis its size-$r$ sets; in this paper's terminology, it should have ``$r$-triword expansion at least $\exp(-O(r^{.99}))$''.\footnote{Two remarks: First, as in Property~(\ref{itm:global}), one only needs this expansion over~$\Sym{m}$ for $m = \poly(1/\delta)$; however, we get it for all~$m$.  Second, ``triword expansion'' is our own terminology. In \cite{BM24}, the condition would be called ``$(m,r,\xi, \exp(-O(r^{.99}))\cdot \xi)$-UG-coboundary expansion'', and in   
    \cite{DD24-covers} it would be called ``$(\exp(-O(r^{.99})), r)$-swap-coboundary expansion over~$\Sym{m}$''.}
\end{enumerate}
Property~(\ref{itm:global}) was recently shown for the Kaufman--Oppenheim complexes by \textcite{KOW25}. Thus to establish \Cref{thm:main-dp}, it suffices for us to establish Property~(\ref{itm:local}); i.e., to prove the following:  
\begin{theorem}[KO link coboundary expansion]\label{thm:main-coboundary}
    Given $r \in \N^+$, provided $n \gg r^2$ and $p \gg r^2 n^2$ is prime, the common vertex-link of the \cite{KOW25} complexes, $\mathfrak{LA}_n(\F_p)$, has $r$-triword expansion at least $\exp(-O(r^{.67}))$ (over every group~$\Gamma$).
\end{theorem}
Before discussing how we prove \Cref{thm:main-coboundary}, we recap the prior analogous work. 
Dikstein and Dinur~\cite[Theorem 1.4]{DD24-swap} showed that the links of the ``usual'' (non-symplectic) \cite{LSV05-explicit} building quotient complexes satisfy Property~(\ref{itm:local}) above, but no form of them is known to satisfy Property~(\ref{itm:global}).
However Chapman and Lubotzky~\cite{CL25-applications} subsequently showed that  more exotic, symplectic forms of these HDXs \emph{do} satisfy Property~(\ref{itm:global}) over~$\Sym{2}$.  This paved the way for Bafna--Lifshitz--Minzer~\cite{BLM24} and Dikstein--Dinur--Lubotzky~\cite{DDL24} to show the same result over~$\Sym{m}$, thereby establishing the symplectic building quotients satisfy Property~(\ref{itm:global}). Thus to derive low soundness direct-product testers, it remained for them to establish Property~(\ref{itm:local}) for the links of the symplectic building quotients, which \cite[Theorem 5.3]{BLM24}, \cite[Theorem 1.3]{DDL24} did by building on the work of~\cite{DD24-swap}.
(We remark that it is the construction and analysis for Property~(\ref{itm:global}) that uses deep number theory and algebraic group theory;  \cite{BLM24,DDL24}'s establishment of Property~(\ref{itm:local}), though not easy, uses only elementary tools.)

\medskip

The KO vertex-link complex $\mathfrak{LA}_n(\F_p)$ is $n$-partite. Moreover, the parts $P = \{1, \dots, n\}$ have a natural linear ordering; i.e., it is relevant to think of them as lying on an $n$-path graph and to consider the ``distance'' $|a - b|$ between parts $a$ and~$b$. 
Because of this, if $w$ is a set of vertices in $\mathfrak{LA}_n(\F_p)$ from different parts $W \subseteq P$, and we study the ``$w$-pinned subcomplex'' (i.e., ``$w$-link'') $\mathfrak{LA}_n(\F_p)_w$, the expansion of this subcomplex is affected by the ``geometry'' of how $W$ sits within the path~$P$.  Similarly, if we take, say, $3$ parts $a < b < c$ from $P$  and consider the subcomplex $\mathfrak{LA}_n(\F_p)[\{a\}, \{b\}, \{c\}]$ restricted to those parts, then the distances $b-a$ and $c-b$ play a role in the expansion of the subcomplex.
(Similar considerations arise in the link-complexes studied in \cite{DD24-swap,BLM24,DDL24}.)

To prove \Cref{thm:main-coboundary}, we follow an induction/restriction strategy on parts.
Specifically, we follow the strategy from \cite{BLM24} (with some streamlining); \cite{DD24-swap,DDL24} uses similar-looking ideas, though the technical details seem to diverge.
This ultimately lets us reduce \Cref{thm:main-coboundary} to the below theorem about $1$-coboundary expansion within pinnings/restrictions of $\mathfrak{LA}_n(\F_p)$, where the three parts to which we restrict  are at path-distance~$\Omega(n)$.

\begin{theorem}[Separated restrictions have dimension-independent coboundary expansion] \label{thm:main-separated}
    In $\mathfrak{LA}_n(\F)$, let $a < b < c$ be in $\rangeOne{n}$, and assume the ``well-separation'' condition $b - a, c - b \geq \xi n$. Let $W \subseteq \rangeOne{n} \setminus \rangeII{a}{c}$, and let $w$ consist of one vertex per part in~$W$.  Then the pinned-and-restricted subcomplex   $\mathfrak{LA}_n(\F_p)_w[\{a\}, \{b\}, \{c\}]$ has $1$-coboundary expansion at least~$\poly(\xi)$ (over any group~$\Gamma$).
\end{theorem}
\noindent We remark that \textcite{KO21} (essentially) proved this theorem in the particular case of $n = 3$ (and $\mathrm{char}\;\F \neq 2$).  The crucial aspect of our \Cref{thm:main-separated} is that $\xi \geq \Omega_{n\to\infty}(1)$ for ``most'' triples $a < b < c$ in $\rangeOne{n}$, and hence the $1$-coboundary expansion is \emph{independent of~$n$}. (Even though $n$ is ultimately a ``constant'' in our construction, to get \Cref{thm:main-coboundary} we require this quantitative independence from~$n$.)
\medskip

As in \textcite{KO21} (and \cite{DD24-swap}; see also~\cite{Gro10,LMM16}), we establish the coboundary-expansion result  \Cref{thm:main-separated} using the ``cones method''.  As shown in \cite{KO21}, the strong symmetry of the KO complex simplifies what one needs to prove: thinking of $\mathfrak{T}  \coloneqq \mathfrak{LA}_n(\F_p)_w[\{a\}, \{b\}, \{c\}]$ as a ``triangle-complex'' ($3$-uniform hypergraph), we only need to show:
\begin{itemize}
    \item its \emph{diameter} is small, $\diam(\mathfrak{T}) \leq C_0 = C_0(\xi)$;
    \item every loop of length at most~$2C_0+1$ can be triangulated using at most $C = C(\xi)$ triangles ($3$-uniform hyperedges).
\end{itemize}    
This is a quantitative form of saying that $\mathfrak{T}$ is ``simply connected''.  It is known~\cite{Lan50,Bun52} that for \emph{coset complexes} like~$\mathfrak{T}$, there is a \emph{group-theoretic} characterization of simple connectivity, and \textcite{KO21} (see also~\cite{OS25}) upgraded this to a quantitative form.  Thus to show \Cref{thm:main-separated}, it suffices to show certain quantitative statements about presentations within the group of unipotent\footnote{Upper-triangular, with $1$'s on the diagonal.} $n \times n$ matrices over~$\F_p[t]$.  To put it very briefly, given a sequence of generators that multiply to~$\Id$ in the unipotent group, we need to bound how many steps it takes to \emph{prove} this, using a limited set of relations.  
Such bounds are often expressed using the notion of the ``Dehn function'', or ``area bounds''. 
Ultimately, the main technical theorems we prove in our paper, constituting the bulk of the work, are the following:
\begin{restatable}[Unipotent group  presentation bounds]{theorem}{unip}\label{thm:main-unip}
For every $\xi > 0$ there exist $C_0, C \in \N$ such that for all sufficiently large~$n$, the following hold: For every field $\F$, parameter $\kappa \in \N^+$ (defining generators), and $\ell_1 < \ell_2 < \ell_3 \in \rangeOne{n}$ with $\ell_2 - \ell_1, \ell_3 - \ell_2 \ge \xi n$, 
\begin{align}
\gdiam{\GrUnip{n}{\F}}{(\GrStair{n}{\F}{\ell_i})_{i \in \{1,2\}}} &\le C_0, \label{eq:diam} \\
\garea{2C_0+1}{\GrUnip{n}{\F}}{(\GrStair{n}{\F}{\ell_i})_{i\in\{1,2,3\}}} &\le C \label{eq:cbdy}.
\end{align}
\end{restatable}

\noindent We explain the meaning of this problem in more detail, using the computer science language of \emph{circuits}, in \Cref{sec:high-groups}.

\subsection*{Organization of the rest of the paper}
\Cref{sec:high} is devoted to further high-level descriptions of the components of our work:
\begin{itemize}
    \item In \Cref{sec:high-dp}, we review the definition of low-soundness direct-product testers and concretely describe the one we obtain from the Kaufman--Oppenheim complex.  We also describe the subcomplexes $\mathfrak{LA}_n(\F_p)$ and $\mathfrak{LA}_n(\F_p)_w[\{a\},\{b\},\{c\}]$ arising in \Cref{thm:main-coboundary,thm:main-separated}.
    \item In \Cref{sec:high-dp}, we review the definitions of vanishing $1$-cohomology and $r$-triword expansion.
    \item In \Cref{sec:high-groups}, we give a high-level overview of our group-theoretic proof of \Cref{thm:main-unip} and some of the considerations therein.
\end{itemize}
The remainder of the work is devoted to technical details.
In \Cref{sec:complex-prelims}, we setup additional necessary definitions regarding simplicial complexes, coset complexes, unipotent groups, and the (local) Kaufman--Oppenheim complex.
In \Cref{sec:struct}, we review some useful known facts and prove some (simple) additional ones about unipotent groups and their subgroups.
In \Cref{sec:proof}, we prove \Cref{thm:main-coboundary} modulo \Cref{thm:main-unip}.
In \Cref{sec:setup}, we describe a few useful additional technical tools which we will use to prove \Cref{thm:main-unip}.
\Cref{thm:main-unip} is then proven in \Cref{sec:diam} (for \Cref{eq:diam}) and \Cref{sec:cbdy} (for \Cref{eq:cbdy}).
Finally, in \Cref{sec:megatheorem}, we collect together all known theorems about the KO complexes, prove our direct-product testing theorem \Cref{thm:main-dp}, and verify that the KO complexes can be used in the \cite{BMVY25} PCP construction.

\subsection*{Acknowledgments}
We would like to thank Mitali Bafna and Dor Minzer for their advice on applying~\cite{BLM24} to other complexes.  The second author in particular thanks Mitali for discussions about details in~\cite{BLM24}.

\section{High-level description of components} \label{sec:high}

\subsection{Direct-product testing} \label{sec:high-dp}

``Direct-product testers'' were introduced by Goldreich and Safra~\cite{GS00} as a kind of ``derandomized parallel repetition'' for use in PCPs. At a high level, one imagines that a ``prover'' claims to have a function $f : \calU \to \Sigma$, mapping a large universe~$\calU$ of size~$N$ to some alphabet~$\Sigma$.  
The tester wishes to simultaneously get $k$ values $f(u_1), \dots, f(u_k)$ with one ``question'' $\{u_1, \dots, u_k\}$ to the prover --- but is concerned that the prover might not honestly give the ``direct-product'' answer $f(u_1), \dots, f(u_k)$.  

A \emph{two-query agreement test (or ``V-test'')}~\cite{DR06} attempts to solve this by introducing a second prover, who cooperates but cannot communicate with the first prover.
The agreement tester is then defined by a  $k$-uniform hypergraph $\X$ on $\calU$.\footnote{One can allow for a weighted $k$-uniform hypergraph, but unlike in~\cite{BLM24,DDL24} we will not need to.}
The tester randomly chooses two hyperedges $\bw, \bw'$ from $\X$  conditioned on $|\bw \cap \bw'| = \sqrt{k}$, asks the first prover for $f$'s values on $\bw$, and asks the second prover for $f$'s values on $\bw'$.  
Finally, the tester verifies that the provers gave the same answers for the $\sqrt{k}$ elements in common.  
One says the agreement tester defined by $\X$ and $k$ achieves \emph{$\delta$-soundness} if, whenever the test accepts with probability at least $\delta$, there is an $f : \calU \to \Sigma$ that is (nearly) consistent with the provers' answers for at least a $\delta^{O(1)}$ fraction of $\X$'s hyperedges.  Here is a precise definition (which merges the two provers into one function~$F$):
\begin{definition}[Agreement tester soundness] \label{def:agreement-tester}
    Let $\X$ be a $k$-uniform hypergraph on~$\calU$.
    We say the agreement tester based on~$\X$ achieves \emph{$\delta$-soundness} if the following holds:  Whenever it  accepts $F : \X \ni w \mapsto \Sigma^w$ with probability least~$\delta$, there exists $f : \calU \to \Sigma$ such that 
    \begin{equation}    \label{eqn:delta-soundness}
        \Pr_{\bw \sim \X}[\text{dist}(F(\bw), (f(u))_{u \in \bw}) \leq \exp(-1/\delta^{\Omega(1)})] \geq \delta^{O(1)},
    \end{equation}
    where $\dist$ denotes fractional Hamming distance.\footnote{Strictly speaking, it is not sensible to put these $\exp(-1/\delta^{\Omega(1)})$ and $\delta^{O(1)}$ expressions into \Cref{eqn:delta-soundness}: they only make sense in the context of a growing family of direct-product testers, and it's also not particularly natural to insist on certain fixed functions of~$\delta$ in such a definition.  However, we \emph{will} eventually work in the context of a growing family, and the theorems will achieve some unspecified $\exp(-1/\delta^{\Omega(1)})$ and $\delta^{O(1)}$ bounds, so we write the definition this way to minimize the number of parameters that need to be introduced.}
\end{definition}

It is not easy even to show that the \emph{complete} $k$-uniform hypergraph on $\calU$ achieves $\delta$-soundness, but this was done by Dinur and Goldenberg~\cite{DG08}  for $k = \poly(1/\delta)$ (cf.~\cite{IKW12}).  
However, the fact that the complete hypergraph has size $\Theta(N^k)$ for $|\calU| = N$ renders it quantitatively undesirable for applications to PCPs.
Consequently, Dinur and Kaufman~\cite{DK17} sought \emph{sparse} agreement testers, meaning ones with $O_{\delta}(N)$ hyperedges.  They showed that the high-dimensional expanders (HDXs) of \cite{LSV05-explicit} yield sparse agreement testers in the \emph{high soundness} regime when~$\delta$ is close to~$1$, and conjectured that HDXs could also be used to get sparse agreement testers in the \emph{low soundness} regime of~$\delta$ close to~$0$.  As described in \Cref{sec:intro}, this conjecture was ultimately resolved positively by Bafna--Lifshitz--Mizer and Dikstein--Dinur--Lubotzky~\cite{BLM24,DDL24}, for certain
(nontrivial to construct and analyze) quotients of affine Bruhat–Tits buildings associated to the symplectic group over the p-adics.  

Subsequently, Bafna, Minzer, and Vyas~\cite{BMVY25} used these agreement testers as a key ingredient in their $(n \cdot \polylog n)$-length PCPs with arbitrarily small constant soundness $\delta > 0$.\footnote{This, however, required a nontrivial appendix by Yun to verify that one could allow the ``$p$'' parameter to be superconstant; namely $p = \polylog N$.  Note that no extra effort is required to verify this for the KO complex agreement tester we presently describe.}  Their work also essentially showed that \emph{if} our main \Cref{thm:main-dp} is true, then the \cite{KO23,KOW25} complexes can be substituted into their construction in place of the symplectic building quotients.  (There are a few minor details to check, which we do in \Cref{sec:megatheorem}.)  Thus our work shows that the elementary agreement tester which are described below can be used in their PCP construction.

\subsubsection{An elementary, sparse, low-soundness agreement tester} \label{sec:high-dp-concrete}

In this section, we give a concrete description of the $\delta$-sound agreement tester in our \Cref{thm:main-dp}, arising from the \cite{KO23,KOW25} complexes $\mathfrak{A}_n(\F_Q)$.  We describe it as a strongly explicit family of sparse, $k$-uniform hypergraphs ($k = \poly(1/\delta)$) over universes $\calU$ of growing size~$N$.  

Let $\delta$, $k$, $n$, $p$ be as in \Cref{thm:main-dp},\footnote{In this section, we will treat these parameters as ``constants'', but in actuality the construction remains strongly explicitly even for $p$ as large as $\polylog N$, and this is the parameter setting needed for the PCP application. We also mention that $p$ could be a prime power rather than a prime.} let $\sorl > 3n$ be a growing parameter, and let $Q = p^s$.
The universe~$\calU$ will be closely related to the set of $(n+1) \times (n+1)$ matrices\footnote{The choice of $n+1$ rather than $n$ here is for notational convenience in the main body of the paper.} over the finite field~$\F_Q$.; thus very roughly, $N \approx Q^{n^2} = p^{\sorl n^2}$, so $\sorl = \Theta(\log N)$.

Indexing the dimensions of such matrices as $\rangeII{0}{n}$, we define for $0 \leq i \neq j \leq n$ and $\alpha \in \F_Q$ the \emph{transvection} matrix
\begin{equation}
\phantom{\qquad \text{($\alpha$ in the $(i,j)$-entry)}}
    e_{i,j}(\alpha) = 
   \begin{bmatrix}
1      &        &        &        &        \\
       & 1      &        & \alpha &        \\
       &        & 1      &        &        \\
       &        &        & \ddots &        \\
       &        &        &        & 1
\end{bmatrix}
\qquad \text{($\alpha$ in the $(i,j)$-entry);}
\end{equation}
when multiplying from the right,
this is the elementary operation of adding $\alpha$-times-the-$i$th-column to the $j$th column, familiar from Gaussian elimination.  It is well known that if one starts with the identity matrix and repeatedly applies transvections, the set of all matrices that can be achieved is precisely $G = \mathrm{SL}_n(\F_Q)$, the matrices in $\F_Q^{n \times n}$ of determinant~$1$.  The size of $G$ is roughly $Q^{n^2 - 1}$.\footnote{Precisely, its size is $Q^{\binom{n+1}{2}} \cdot Q^{n} \cdot \sum_{\pi \in \Sym{n+1}} Q^{\text{inv}(\pi)}$, where $\text{inv}(\pi)$ denotes the number of inversions in~$\pi$.  This is because there is a $1$-$1$ correspondence between matrices in $\mathrm{SL}_{n+1}(\F_Q)$ and their ``Bruhat decompositions'' of the form $U D P_\pi L$, where $U$ is upper-triangular with $1$'s on the diagonal, $D$ is diagonal with entries multiplying to~$1$, $P_\pi$ is the permutation matrix for permutation $\pi$, and $L$ is lower-triangular with $1$'s on the diagonal and $L_{i,j} = 0$ whenever $\pi(i) > \pi(j)$. This $1$-$1$ correspondence gives an easy way to draw a matrix from~$\mathrm{SL}_{n+1}(\F_Q)$ exactly uniformly at random in $\Theta(\log Q)$ time.}

Let us fix a concrete representation of $\F_Q$ as $\F_p[X]/(\text{irred}(X))$, where $\text{irred}(X)$ denotes some irreducible polynomial in $\F_p[X]$ of degree~$\sorl$.\footnote{\cite{KO23} mod out by $X^\sorl$, but following \cite{OP22} we find it more natural to mod out by an irreducible; it makes no difference to the result in~\cite{KOW25}. Shoup~\cite{Sho90} shows how to find such an irreducible in deterministic $\poly(\sorl) = \polylog N$ time.}  
We define the following ``tiny'' set $T$ of $p^4$ field elements:\footnote{The distinction between this construction from \cite{KOW25} and the original one in~\cite{KO23} is that the latter has $\kappa = 1$ in the definition of~$T$.  \cite{KOW25} needs $\kappa = 3$ for a minor technical reason.}
\begin{equation}
    T \coloneqq \{f(X) \in \F_p[X] : \deg(f) \leq \kappa \coloneqq 3\} \subseteq \F_Q;
\end{equation}
then we define a transvection $e_{i,j}(\alpha)$ to be ``tiny'' if $\alpha \in T$ and $j = i+1$ mod~$n+1$.

Next, let us define an equivalence relation $\equiv_0$ on matrices $g \in G$:
\begin{multline} \label{eqn:sim}
    g \equiv_0 g' \iff g' \text{ can be obtained from $g$ by repeatedly applying tiny transvections $e_{i,i+1}(\alpha)$,} \\
    \text{$i+1 \neq 0$ mod $n+1$.}
\end{multline}
(In other words, we disallow use of the ``wraparound'' tiny transvections.)
It is not hard to show that the size of every associated equivalence class $[g]_0$ is some $c \coloneqq p^{O(n^2)}$, a (very large) constant.  
More generally, for $0 \leq a \leq n$, we 
define the equivalence relation $\equiv_a$ by allowing in \Cref{eqn:sim} only the tiny transvections $e_{i,i+1}$ with $i+1 \neq a$; this yields equivalence classes $[g]_a$ also of size~$c$.

Now for each matrix $g \in G$ we get $n+1$ equivalence classes $[g]_0, \dots, [g]_{n}$ of size~$c$.  The universe~$\calU$ for the assignment tester will be $(n+1)$-partite, with the $a$th part consisting of all equivalence classes~$[g]_a$; thus $N \coloneqq |\calU| = (n+1) |G|/c = \Theta(|G|)$.  We can naturally form an $(n+1)$-uniform hypergraph (i.e., $n$-dimensional simplicial complex) on~$\calU$ by taking all $|G|$ hyperedges of the form $\{[g]_a : a \in \rangeII{0}{n}\}$.
For the actual assignment tester, we simply take all the size-$k$ sub-edges $\{[g]_a : a \in A, |A| = k\}$; this yields a sparse agreement tester, because there are $\binom{n}{k} |G| = \Theta(|G|)$ such sets.

Explicitly, the agreement tester for $f : \calU \to \Sigma$, when described for two provers, is the following.  The tester picks two random size-$k$ subsets $\bA, \bA'$ of $\{0, 1, \dots, n\}$ conditioned on $|\bA \cap \bA'| = \sqrt{k}$, picks $\bg \in \mathrm{SL}_{n+1}(\F_Q)$ uniformly at random, and asks the two provers for $f$'s values on $([\bg]_{a})_{a \in \bA}$ and on $([\bg]_{a'})_{a' \in \bA'}$.

\paragraph{The \emph{local KO complex}.} Before ending this section, let us also describe the subhypergraphs appearing in \Cref{thm:main-coboundary,thm:main-separated}.  We first have what we call the \emph{local KO complex} $\mathfrak{LA}_n(\F_p)$, to which all vertex-neighborhoods (links) are isomorphic. It may be defined by restricting attention to the $n$-partite hypergraph with just the hyperedges
\begin{equation}
    \{[g]_1, \dots, [g]_n\}, \text{ where } g \equiv_0 \Id.
\end{equation}
Note that this hypergraph is of ``constant'' size, only being defined by the $c$ elements of~$[\Id]_0$.  These are all the matrices that can be obtained by multiplying together non-wraparound tiny transvections $e_{i,i+1}(\alpha)$.  It is not hard to show (using $\sorl > \kappa n$) that these are all upper-triangular matrices in $\F_Q^{(n+1)\times(n+1)}$ with $1$'s on the diagonal and with $(i,j)$-entry $\beta$ satisfying $\deg(\beta) \leq \kappa(j-i)$.  In particular, this set does not depend on $\sorl$ or ``$Q$'' per se, which is why we use the notation $\mathfrak{LA}_n(\F_p)$.

As for the further subcomplexes $\mathfrak{LA}_n(\F_p)_w[\{a\},\{b\},\{c\}]$ from \Cref{thm:main-separated}, these are all isomorphic to the tripartite hypergraph with just the hyperedges
\begin{equation}
    \{[g]_a, [g]_b, [g]_c\}, \text{ where } g \equiv_i \Id \text{ for all } i \in \{0\} \cup W.
\end{equation}

\subsection{Triword expansion}
\label{sec:high-coboundary}

In this section, we define the ``topological'' HDX properties that go into the Bafna--Minzer/Dikstein--Dinur~\cite{BM24,DD24-covers} characterization of when a hypergraph can serve as a low-soundness agreement tester.  (Roughly speaking, Properties~(\ref{itm:global}),~(\ref{itm:local}) from \Cref{sec:intro}.)  For the most precise definitions, see \Cref{sec:complex-prelims}.

Let $X = (V,E,T)$ be a connected ``triangle complex'' (pure $2$-dimensional simplicial complex); the reader may also think of $X$ as a graph $(V,E)$ formed by the union of a set $T$ of $3$-cliques. Let $f : \vec{E} \to \Sym{m}$ be a labeling of $X$'s directed edges by permutations of~$[m]$ (satisfying $f(u,v) = f(v,u)^{-1}$).  The associated \emph{$m$-lift} $X^f$ of $X$ is the graph on vertex set $V \times [m]$ which has an edge from $(u,i)$ to $(v, f(u,v)(i))$ whenever $(u,v) \in \vec{E}$.

We say the lift is an \emph{$m$-cover} if every triangle $(u,v,w)$ in~$T$ gets lifted to $m$ disjoint triangles in $X^f$ --- equivalently, $f(u,v)f(v,w)f(w,v) = \Id$.  There are certain ``trivial'' ways to get $m$-covers; take any vertex-labeling $g : V \to \Sym{m}$ and define $f(u,v) = g(u)g(v)^{-1}$.  For these ``trivial'' (sometimes called ``balanced'') lifts, the lifted graph~$X^f$ consists of $m$ disjoint copies of~$X$, and the product of $f$ around \emph{every} cycle is~$\Id$ (not just around every $3$-cycle).
When the \emph{only} $m$-covers of $X$ are these trivial ones, $X$ is said to have \emph{trivial $1$-cohomology over~$\Sym{m}$}.  More generally, this can be defined over any group~$\Gamma$:
\begin{definition}[Trivial $1$-cohomology over~$\Gamma$]  
    A triangle complex $X = (V,E,T)$ is said to have \emph{trivial $1$-cohomology over $\Gamma$} if, whenever $f : \vec{E} \to \Gamma$ (with $f(u,v) = f(v,u)^{-1}$) has $f(u,v)f(v,w)f(w,u) = \Id$ for all $(u,v,w) \in T$, there exists $g : V \to \Gamma$ such that $f(u,v) =  g(u)g(v)^{-1}$ for all $(u,v) \in E$.  
\end{definition}
The main theorem of \textcite{KOW25} is that the $k = 3$ hypergraph described in the preceding section \Cref{sec:high-dp-concrete} indeed has no nontrivial $m$-covers for $m < p$:
\begin{theorem} \label{thm:KOWabunga} (\cite{KOW25}.)  
    Provided $\Gamma$ has no nontrivial elements of order~$p$, the \cite{KOW25} complexes from \Cref{thm:main-dp} have vanishing $1$-cohomology over~$\Gamma$.
\end{theorem}
\noindent As mentioned in \Cref{sec:intro}, this result verifies that the complexes have Property~(\ref{itm:global}) of the Bafna--Minzer/Dikstein--Dinur HDX characterization of agreement testers.  

\medskip

Let us now define our term ``$r$-triword expansion'' from Property~(\ref{itm:local}).
First, when the $1$-cohomology of $X$ over $\Sym{m}$ vanishes, one can ask for a quantitative strengthening:
\begin{definition}[$1$-coboundary expansion]
    $X$ is said to have \emph{$1$-coboundary expansion at least~$\epsilon$ (over~$\Gamma$)} if, whenever $f : \vec{E} \to \Gamma$ (with $f(u,v) = f(v,u)^{-1}$) has $f(u,v)f(v,w)f(w,u) = \Id$ for a $1-\zeta$ fraction of triangles $(u,v,w) \in T$, there exists $g : V \to \Gamma$ such that $f(u,v) = g(u) g(v)^{-1} $ for at least a $1-\zeta/\epsilon$ fraction of edges $(u,v) \in E$. 
\end{definition}
At a (very) intuitive level (which is more accurate if $X$ has bounded diameter), to have good $1$-coboundary expansion over $\Sym{m}$ is to say that whenever $X^f$ is an $m$-lift in which most triangles get lifted to $m$ disjoint triangles, then  most longer loops \emph{also}  get lifted to $m$ disjoint copies.  One can then intuit that this  may be true if loops in~$X$ can be triangulated with ``few'' triangles.  This is the essence of the ``cones method'' described at the end of \Cref{sec:intro}.

\medskip

This defines $1$-coboundary expansion, but to finish explaining terms in \Cref{thm:main-coboundary}, we make a generalized notion suitable for $n$-uniform hypergraphs $\X$ such as~$\mathfrak{LA}_n(\F_p)$.
For any such~$\X$ and any parameter $1 \leq r \leq n/3$, we may turn $\X$ into a triangle complex~$(V_r(\X), E_r(\X), T_r(\X))$ (called the ``$r$-faces complex'' in~\cite{DD24-covers}) by specifying the following: the elements of $V_r(\X)$ are all sets of $r$ vertices in~$\X$ that appear together in some hyperedge; and, the elements of $T_r(\X)$ are all triples $(S_1, S_2, S_3)$ from $V_r(\X)$ where $S_1 \cup S_2 \cup S_3$ is a size-$3r$ set of vertices in~$\X$ appearing together in some hyperedge.  Finally, we say that $\X$ has \emph{$r$-triword expansion at least~$\epsilon$ over~$\Gamma$} if $(V_r(\X), E_r(\X), T_r(\X))$ has $1$-coboundary expansion at least~$\epsilon$ over~$\Gamma$.
With this definition in place, we have explained all terms appearing in \Cref{thm:main-coboundary,thm:main-separated}.

\subsection{Proofs of identity in the unipotent group}
\label{sec:high-groups}

In this subsection, before we finally dive into the main body of the paper, we preview some of the group-theoretic setup and proof strategies needed for our core technical theorem, \Cref{thm:main-unip}.
The discussion is substantially simplified, but we hope it gives a flavor of the considerations which do come up in the actual proof.

\subsubsection{Circuits}
Our key coboundary expansion result, \Cref{thm:main-separated}, is concerned with loops and triangles inside the ``local KO complex'' described at the end of \Cref{sec:high-dp-concrete}.
In turn, this means we need to understand well which matrices over can be obtained from others through multiplication by tiny transvections.  
For this, we found it extremely helpful to use the computer science language of \emph{circuits}.

Fix any field $\F$.
Let $\Fle{d}$ denote the set of polynomials in $X$ of degree at most $d$ over $\F$.
We will be interested in circuits that model $(n+1) \times (n+1)$ matrices over $\Fle{d}$ that are unipotent (upper-triangular with $1$'s on the diagonal).

Consider circuits with $n+1$ horizontal ``wires'', numbered $0$ through $n$ from top to bottom.\footnote{We follow the quantum computing convention of circuits with horizontal wires computing linear transformations; but unlike in that field, matrices will get multiplied in the same left-to-right order as gates, hence we consider action by right-multiplication on row vectors.}  We use the notation $\rangeII{0}{n} = \{0,\ldots,n\}$.
A circuit will then simply be an ordered sequence of \emph{gates}, each of which represents a(n upper-triangular) transvection.  Generalizing the set of ``tiny transvections'', each gate is specified by a pair $i < j \in \rangeII{0}{n}$ and a polynomial $f \in \Fle{3(j-i)}$; this gate is written $e_{i,j}(f)$.
We think of $i$ as the ``head'' of the gate, $j$ as the ``tail'', and $j-i$ as the ``height''.
The head wire is always (strictly) above the tail, and the height is therefore always (strictly) positive.
The degree is bounded by the constant $3$ times the height.
(The particular choice of the constant $3$ is not important for the purposes of this subsection.)
A circuit $C$ is a sequence of gates, denoted $\Embed{e_{i_1,j_1}(f_1)} \cdots \Embed{e_{i_T,j_T}(f_T)}$.
Here, we use the notation $\Embed{\cdot}$ around each individual gate to emphasize that a circuit $\Embed{e_{i_1,j_1}(f_1)} \cdots \Embed{e_{i_T,j_T}(f_T)}$
is formally distinct from the (unipotent) matrix product $e_{i_1,j_1}(f_1) \cdots e_{i_T,j_T}(f_T)$.
In particular, many distinct circuits may represent the same unipotent matrix.

Although we don't formally \emph{need} to think of these circuits as acting upon anything, it will psychologically helpful to do so.  
We think of the wires as ``registers'' that store polynomials $v_0, \dots, v_n  \in \F[X]$; collectively, we think of them as storing a row-vector $v = (v_0,\ldots,v_n) \in \F[X]^{n+1}$.
The action of the gate $e_{i,j}(f)$ is to update the contents of the $j$th register  by setting $v_j \gets v_j + fv_i$.
We may also think of a circuit as a ``program'', with a list of ``instructions'' (gates), 
where $e_{i,j}(f)$ is the instruction ``$\instAdd{fv_i}{v_j}$''.

Linear-algebraically, a circuit $C$ corresponds to an $(n+1) \times (n+1)$ (unipotent) matrix over $\F[X]$, and it acts on a row vector $v \in \F[X]^{n+1}$ to produce $v \cdot C$.  
We will write $C \equiv C'$ and say that $C$ and $C'$ are \emph{equivalent} if $v \cdot C = v \cdot C'$ for every $v \in \F^{n+1}$; in other words, if $C$ and $C'$ are equal as matrices.  We will be particularly interested in the case of $C \equiv \Id$, meaning $C$ is equivalent to the ``empty'' (trivial) circuit with no gates.
Note that $C \equiv C'$ iff $C^{-1} C' \equiv \Id$, where $C^{-1}$ is the ``inverse'' circuit which reverses the orders of gates and negates their entries (note that $\Embed{e_{i,j}(f)} \Embed{e_{i,j}(-f)} \equiv \Id$).

\subsubsection{``Relations'' among circuits}

Consider the following important examples:

\begin{example}[``Linearity'' relation]\label{ex:steinberg:linear}
Let $n=1$ and $C \coloneqq \Embed{e_{0,1}(f)} \Embed{e_{0,1}(f')}$.
If the initial state is $(v_0,v_1)$, then the states after the gates operate are: \[
\begin{pmatrix}
    v_0 \\ v_1
\end{pmatrix}
\stackrel{e_{0,1}(f)}{\longmapsto}
\begin{pmatrix}
    v_0 \\ v_1+fv_0
\end{pmatrix}
\stackrel{e_{0,1}(f')}{\longmapsto}
\begin{pmatrix}
    v_0 \\ v_1+(f+f')v_0
\end{pmatrix}. \]
Hence, $C \equiv \Embed{e_{0,1}(f+f')}$.
\end{example}

\begin{example}[``Nontrivial commutator'' relation]\label{ex:steinberg:nontrivial}
Let $n=2$ and $C \coloneqq \Embed{e_{0,1}(f)} \Embed{e_{1,2}(g)} \Embed{e_{0,1}(-f)} \Embed{e_{0,1}(-g)}$.
If the initial state is $(v_0,v_1,v_2)$, then the states after the gates operate are: \[
\begin{pmatrix}
    v_0 \\ v_1 \\ v_2
\end{pmatrix}
\stackrel{e_{0,1}(f)}{\longmapsto}
\begin{pmatrix}
    v_0 \\ v_1+fv_0 \\ v_2
\end{pmatrix}
\stackrel{e_{1,2}(g)}{\longmapsto}
\begin{pmatrix}
    v_0 \\ v_1+fv_0 \\ v_2+gv_1+fgv_0
\end{pmatrix}
\stackrel{e_{0,1}(-f)}{\longmapsto}
\begin{pmatrix}
    v_0 \\ v_1 \\ v_2+gv_1+fgv_0
\end{pmatrix}
\stackrel{e_{1,2}(-g)}{\longmapsto}
\begin{pmatrix}
    v_0 \\ v_1 \\ v_2+fgv_0
\end{pmatrix}.
\]
Hence, $C \equiv \Embed{e_{0,2}(fg)}$.
\end{example}

\begin{example}[``Trivial commutator'' relation]\label{ex:steinberg:trivial}
Let $n=3$ and $C \coloneqq \Embed{e_{0,2}(p)} \Embed{e_{1,3}(q)} \Embed{e_{0,2}(-p)} \Embed{e_{1,3}(-q)}$.
If the initial state is $(v_0,v_1,v_2,v_3)$, then the states after the gates operate are: \[
\begin{pmatrix}
    v_0 \\ v_1 \\ v_2 \\ v_3
\end{pmatrix}
\stackrel{e_{0,2}(p)}{\longmapsto}
\begin{pmatrix}
    v_0 \\ v_1 \\ v_2+pv_0 \\ v_3
\end{pmatrix}
\stackrel{e_{1,3}(q)}{\longmapsto}
\begin{pmatrix}
    v_0 \\ v_1 \\ v_2+pv_0 \\ v_3+qv_1
\end{pmatrix}
\stackrel{e_{0,2}(-p)}{\longmapsto}
\begin{pmatrix}
    v_0 \\ v_1 \\ v_2 \\ v_3+qv_1
\end{pmatrix}
\stackrel{e_{1,3}(-q)}{\longmapsto}
\begin{pmatrix}
    v_0 \\ v_1 \\ v_2 \\ v_3
\end{pmatrix}.
\]
Hence, $C \equiv \Id$.
\end{example}
(Here, the particular choice of the pairs $(0,2)$ and $(1,3)$ was not important.
The only important feature is disjointness; $(0,3)$ and $(1,2)$, or $(0,1)$ and $(2,3)$, would have worked as well.)

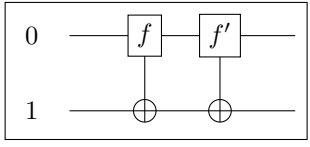
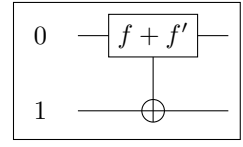
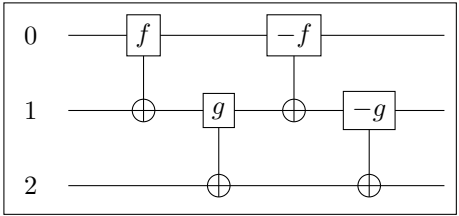
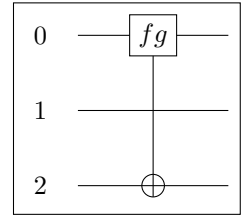
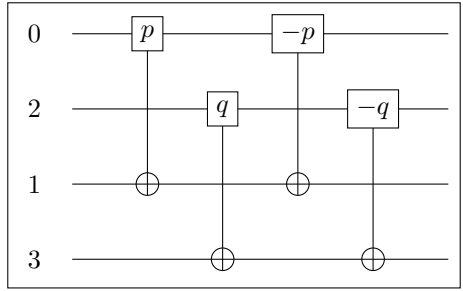
\begin{figure}[t!]
    \centering

    \begin{subfigure}[t]{0.45\textwidth}
    \centering
    \begin{tikzpicture}[framed]
    \addwire{3}{0}{$0$};
    \addwire{3}{1}{$1$};
    \addgate{1}{0}{1}{$f$};
    \addgate{2}{0}{1}{$f'$};
    \end{tikzpicture}
    \caption{A two-gate circuit ($n=1$).
    Written as a program: ``$\instAdd{fv_0}{v_1}$ $\instAdd{f'v_0}{v_1}$''.}
    \label{fig:steinberg:a}
    \end{subfigure}%
    \hspace{0.05\textwidth}
    \begin{subfigure}[t]{0.45\textwidth}
    \centering
    \begin{tikzpicture}[framed]
    \addwire{2}{0}{$0$};
    \addwire{2}{1}{$1$};
    \addgate{1}{0}{1}{$f + f'$};
    \end{tikzpicture}
    \caption{A one-gate circuit ($n=1$).
    Written as a program: ``$\instAdd{(f+f')v_0}{v_1}$''.}
    \label{fig:steinberg:b}
    \end{subfigure}
    
    \vspace{0.25in}
    
    \begin{subfigure}[t]{0.45\textwidth}
    \centering
    \begin{tikzpicture}[framed]
    \addwire{5}{0}{$0$};
    \addwire{5}{1}{$1$};
    \addwire{5}{2}{$2$};
    \addgate{1}{0}{1}{$f$};
    \addgate{2}{1}{2}{$g$};
    \addgate{3}{0}{1}{$-f$};
    \addgate{4}{1}{2}{$-g$};
    \end{tikzpicture}
    \caption{A four-gate circuit ($n=2$).
    Written as a program: ``$\instAdd{fv_0}{v_1}$ $\instAdd{gv_1}{v_2}$ $\instSub{fv_0}{v_1}$ $\instSub{gv_1}{v_2}$''.}
    \label{fig:steinberg:c}
    \end{subfigure}%
    \hspace{0.05\textwidth}
    \begin{subfigure}[t]{0.45\textwidth}
    \centering
    \begin{tikzpicture}[framed]
    \addwire{2}{0}{$0$};
    \addwire{2}{1}{$1$};
    \addwire{2}{2}{$2$};
    \addgate{1}{0}{2}{$fg$};
    \end{tikzpicture}
    \caption{A one-gate circuit ($n=2$).
    Written as a program: ``$\instAdd{fgv_0}{v_2}$''.}
    \label{fig:steinberg:d}
    \end{subfigure}
    
    \vspace{0.25in}
    
    \begin{subfigure}[t]{0.45\textwidth}
    \centering
    \begin{tikzpicture}[framed]
    \addwire{5}{0}{$0$};
    \addwire{5}{1}{$2$};
    \addwire{5}{2}{$1$};
    \addwire{5}{3}{$3$};
    \addgate{1}{0}{2}{$p$};
    \addgate{2}{1}{3}{$q$};
    \addgate{3}{0}{2}{$-p$};
    \addgate{4}{1}{3}{$-q$};
    \end{tikzpicture}
    \caption{A four-gate circuit ($n=3$).
    Written as a program: ``$\instAdd{pv_0}{v_2}$ $\instAdd{qv_1}{v_3}$ $\instSub{pv_0}{v_2}$ $\instSub{qv_1}{v_3}$''.}
    \label{fig:steinberg:e}
    \end{subfigure}
    \caption{Some circuits.
    The circuits in \Cref{fig:steinberg:a,fig:steinberg:b} are equivalent (the ``linearity'' relation, \Cref{ex:steinberg:linear}).
    The circuits in \Cref{fig:steinberg:c,fig:steinberg:d} are equivalent (the ``nontrivial commutator'' relation, \Cref{ex:steinberg:nontrivial}).
    The circuit in \Cref{fig:steinberg:e} is equivalent to the trivial (empty) circuit (the ``trivial commutator'' relation, \Cref{ex:steinberg:trivial}).}\label{fig:steinberg}
\end{figure}

These examples generalize easily to the following fact:

\begin{fact}[``Steinberg relations'']\label{fact:steinberg}
    We have:
    \begin{enumerate}
    \item For every $i < j \in \rangeII{0}{n}$ and $f, f' \in \Fle{3(j-i)}$, $\Embed{e_{i,j}(f)} \Embed{e_{i,j}(f')} \equiv \Embed{e_{i,j}(f+f')}$.
    \item For every $i < j < k \in \rangeII{0}{n}$ and $f \in \Fle{3(j-i)}, g \in \Fle{3(k-j)}$, $\Embed{e_{i,j}(f)} \Embed{e_{j,k}(g)} \Embed{e_{i,j}(-f)} \Embed{e_{j,k}(-g)} = \Embed{e_{i,k}(fg)}$.
    \item For every $i < j \in \rangeII{0}{n}$ and $k < \ell \in \rangeII{0}{n}$ such that $j \ne k$ and $i \ne \ell$, 
    and $f \in \Fle{3(j-i)}$ and $g \in \Fle{3(k-j)}$,
    $\Embed{e_{i,j}(f)} \Embed{e_{k,\ell}(h)} \Embed{e_{i,j}(-f)} \Embed{e_{k,\ell}(-h)} = \Id$.
    \end{enumerate}
\end{fact}

\noindent Compare to \Cref{def:graded unipotent group} below.

Given a circuit $C$, it is not hard to check whether $C \equiv \Id$; you can simply compute the matrix corresponding to~$C$ and check if it is the identity matrix.
But we are concerned with the problem of \emph{certifying} that $C \equiv \Id$ in particular, limited proof models.
We are interested both in the \emph{qualitative} question of ``Given $C$, is it possible to certify that $C \equiv \Id$?''
and the \emph{quantitative} question of whether this certification can be \emph{efficient}.
We first focus on the qualitative question, and then circle back to the quantitative version (which is slightly subtler to define) at the end of this section.

\subsubsection{Avoidant transformations}

We now turn to defining the particular proof model we will be interested in, which we call the ``avoidant transformations'' model.

Recall that we have $n+1$ wires labeled by $\rangeII{0}{n} = \{0,\ldots,n+1\}$.
We now also consider the \emph{gaps} between each pair of adjacent wires.
We label the gap between wires $0$ and $1$ by $1$, the gap between wires $1$ and $2$ by $2$, and so on;
hence there are $n$ gaps, labeled $1,\ldots,n$.
Now, let us fix a small set $\calA \subseteq \rangeOne{n}$, which we call the \emph{moats}.
We say a gate from wire $i$ to wire $j$ \emph{crosses} moat $a \in \calA$ if $a \in \rangeEI{i}{j} = \{i+1,\ldots,j\}$,
or equivalently $i+1 \le a \le j$;
we say a gate is \emph{$a$-avoiding} if it does \emph{not} cross $a$.
For $a \in \calA$, we say a circuit is \emph{$a$-avoiding} if every gate in it is $a$-avoiding,
and \emph{avoidant} if every gate in it avoids \emph{some} moat.

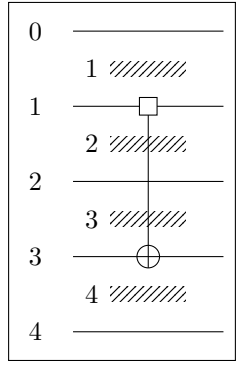
\begin{figure}[t!]
    \centering
    \begin{tikzpicture}[framed]
    \addwire{2}{0}{$0$};
    \addwire{2}{1}{$1$};
    \addwire{2}{2}{$2$};
    \addwire{2}{3}{$3$};
    \addwire{2}{4}{$4$};
    \addmoat{2}{0}{$1$};
    \addmoat{2}{1}{$2$};
    \addmoat{2}{2}{$3$};
    \addmoat{2}{3}{$4$};
    \addgate{1}{1}{3}{~};
    \end{tikzpicture}
    
    \caption{An illustration of the situation when $n = 4$.
    There are $n+1=5$ wires, labeled $0,\ldots,4$, and $n$ moats, labeled $1, \ldots,4$.
    The ``gate'' from wire $1$ to wire $3$ crosses moats $2$ and $3$ but not moats $1$ or $4$.
    This corresponds to the fact that $\rangeEI{1}{3} = \{2,3\}$.}
    \label{fig:moats0}
\end{figure}

Now, we define the concept of an \emph{avoidant transformation} on a circuit $C$.
Suppose $C$ is decomposed into three subcircuits $C = C_1 C_2 C_3$ (any of which may be empty), and $C_2$ is $a$-avoiding for some moat $a \in \calA$.
Then an \emph{avoidant transformation} replaces $C_2$ with another $a$-avoiding circuit $C'_2$ such that $C_2 \equiv C_2'$.

Given this setup, we can now formally define our goal: 
We want to show that for every set of three moats $\calA$ and circuit $C$ such that $C \equiv \Id$,\footnote{
    Technically, we need to restrict here to only those in which every gate avoids at least one moat,
    since those gates cannot possibly be changed via avoidant transformations.
    We elide this distinction in this overview.}
$C$ can be transformed to $\Id$ using only avoidant transformations.
(This is the ``qualitative'' goal; we address quantitative aspects, which are the main actual contribution of our work, at the end of this section.)

Fix three moats $a < b < c$.
We develop some notations for how these moats break up the set of wires $\rangeII{0}{n}$ into intervals
and how these in turn break up the gates into different types.
In particular, we use $\WiresA$, $\WiresB$, $\WiresC$, and $\WiresD$ to denote, respectively,
wires above moat $a$, wires between moats $a$ and $b$, wires between moats $b$ and $c$, and wires below moat $d$.
These four sets of wires, which we call \emph{bands}, partition the full set $\rangeII{0}{n}$.
We then divide the gates of an arbitrary circuit into ten \emph{types}, corresponding to letters, according to the following table:

\begin{table}[H]
\centering
\begin{tabular}{c||c|c|c|c|c|c|c|c|c|c}
    Gate type     & $\typA$   & $\typB$   & $\typC$   & $\typD$   & $\typF$   & $\typG$   & $\typH$   & $\typP$   & $\typQ$   & $\typZ$    \\ \hline\hline
    Head interval & $\WiresA$ & $\WiresB$ & $\WiresC$ & $\WiresD$ & $\WiresA$ & $\WiresB$ & $\WiresC$ & $\WiresA$ & $\WiresB$ & $\WiresA$ \\ \hline
    Tail interval & $\WiresA$ & $\WiresB$ & $\WiresC$ & $\WiresD$ & $\WiresB$ & $\WiresC$ & $\WiresD$ & $\WiresC$ & $\WiresD$ & $\WiresD$ \\ \hline
    Avoids?     & $a,b,c$   & $a,b,c$   & $a,b,c$   & $a,b,c$   & $b,c$     & $a,c$     & $a,b$     & $c$       & $a$       & none
\end{tabular}
\caption{The ten possible types of gates with respect to the three moats.}\label{tab:types}
\end{table}

See also \Cref{fig:cbdy} below for a visual depiction of the situation.

\subsubsection{The power of working space}

We now consider two concrete examples of circuits which show some important differences under the lens of avoidant transformations.

\begin{example}\label{ex:bd}
Consider the case $n = 3$, $a = 1$, $b = 2$, $c = 3$, and the circuit \[
C \coloneqq \Embed{e_{0,2}(p)} \Embed{e_{1,3}(q)} \Embed{e_{0,2}(-p)} \Embed{e_{1,3}(-q)} \]
(which is also depicted in \Cref{fig:bd} below).
Here, $p, q \in \Fle{6}$ are arbitrary sextic polynomials.
Since $0 \in \WiresA$, $1 \in \WiresB$, $2 \in \WiresC$, and $3 \in \WiresD$,
the gates $\Embed{e_{0,2}(\pm p)}$ are ``type-$\typP$'' and the gates $\Embed{e_{1,3}(\pm q)}$ are ``type-$\typQ$''.
By the ``trivial commutator'' Steinberg relation (i.e., \Cref{ex:steinberg:trivial}),
this circuit $C \equiv \Id$.
However, it is not completely obvious how to avoidantly transform this circuit into the empty circuit.
Doing so turns out to be possible, but it requires a few tricks
--- firstly, a clever ``symmetry-breaking'' expansion to solve the problem in the case where $p$ and $q$ are degree-$0$ polynomials~\cite[\S4]{BD01}
--- and secondly, a ``lifting'' argument to generalize to arbitrary $p,q \in \Fle{3}$ \cite[\S8]{KO21}.
(Actually, \textcite{KO21} only consider the case $p,q \in \Fle{1}$, but it is not hard to show that their argument generalizes.
See \cite{OS25} for more on lifting arguments.)
\end{example}

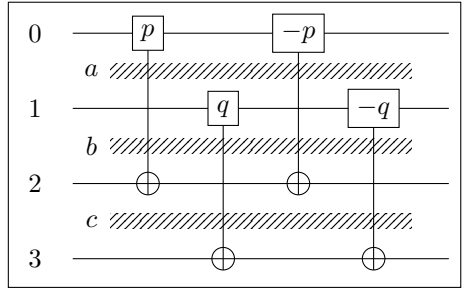
\begin{figure}[H]
    \centering
    \begin{tikzpicture}[framed]
    \addwire{5}{0}{$0$};
    \addwire{5}{1}{$1$};
    \addwire{5}{2}{$2$};
    \addwire{5}{3}{$3$};
    \addmoat{5}{0}{$a$};
    \addmoat{5}{1}{$b$};
    \addmoat{5}{2}{$c$};
    \addgate{1}{0}{2}{$p$};
    \addgate{2}{1}{3}{$q$};
    \addgate{3}{0}{2}{$-p$};
    \addgate{4}{1}{3}{$-q$};
    \end{tikzpicture}
    \hspace{0.05\textwidth}
    \caption{The circuit in \Cref{ex:bd} with $n=3$, $a=1$, $b=2$, and $c=3$.
    Every gate avoids at least one moat; in particular, the $\pm p$ gates avoid $c$ and the $\pm q$ gates avoid $a$.}\label{fig:bd}
\end{figure}

A simple but crucial insight of our group-theoretic work is that spreading out the moats $a$, $b$, and $c$ can facilitate the use of avoidant transformations.
For instance:

\begin{example}\label{ex:slide}
Consider the case $n = 5$, $a = 1$, $b = 3$, $c = 5$, and the circuit \[
C \coloneqq \Embed{e_{0,3}(p)} \Embed{e_{2,5}(q)} \Embed{e_{0,3}(-p)} \Embed{e_{2,5}(-q)} \]
(which is also depicted in \Cref{fig:slide} below).
$p, q \in \Fle{9}$ are arbitrary nonic polynomials.
Since $0 \in \WiresA$, $2 \in \WiresB$, $3 \in \WiresC$, and $5 \in \WiresD$,
the gates $e_{0,3}(\pm p)$ are again ``type-$\typP$'' and $e_{2,5}(\pm q)$ ``type-$\typQ$'', and the circuit $C \equiv \Id$.
This time, we claim it is (relatively) straightforward to certify that $C \equiv \Id$,
or equivalently that $e_{0,3}(p) e_{2,5}(q) \equiv e_{2,5}(q) e_{0,3}(p)$, avoidantly.

First, we claim that we can reduce to the case where we have factorizations $p = fg^{(1)}$ and $q = g^{(2)}h$,
for cubic polynomials $f, h \in \Fle{3}$ and sextic polynomials $g^{(1)},g^{(2)} \in \F[X]$.
For instance, write the degree decompositions $p = p_0 + p_1 X + p_2X^2 + \cdots + p_9 X^9$ and similarly for $q$.
By the linearity relations, to commute $e_{0,3}(p)$ and $e_{2,5}(q)$ using avoidant transformations,
it suffices to commute each pair $e_{0,3}(p_i X^i)$ and $e_{2,5}(q_j X^j)$ using avoidant transformations,
and for each such pair we have \[
p_i X^i = \underbrace{(p_i X^{\max\{i-3,0\}})}_{\eqqcolon f_i \in \Fle{3}}
\cdot \underbrace{(X^{i-\max\{i-3,0\}})}_{\eqqcolon g^{(1)}_i \in \Fle{6}} \text{ and }
q_j X^j = \underbrace{(X^{j-\max\{j-3,0\}})}_{\eqqcolon g^{(2)}_j \in \Fle{6}}
\cdot \underbrace{(q_j X^{\max\{j-3,0\}})}_{\eqqcolon h_j \in \Fle{3}}. \]

Now, consider the problem of commuting $e_{0,3}(fg^{(1)})$ with $e_{2,5}(g^{(2)}h)$.
We have, using (avoidant) nontrivial commutator relations,
\begin{equation}\label{eq:slide}
\Embed{e_{0,3}(fg^{(1)})} \Embed{e_{2,5}(g^{(2)}h)} = \Embed{e_{0,1}(f)} \Embed{e_{1,2}(g^{(1)})} \Embed{e_{0,1}(-f)} \Embed{e_{1,2}(-g^{(1)})}
\cdot \Embed{e_{2,4}(g^{(2)})} \Embed{e_{4,5}(h)} \Embed{e_{2,4}(-g^{(2)})} \Embed{e_{4,5}(-h)}.
\end{equation}
Now, one can check that using (avoidant) trivial commutator relations,
each of the first four gates on the right-hand side of \Cref{eq:slide} commutes with each of the last four.
See \Cref{fig:slide} for a visual depiction.
\end{example}

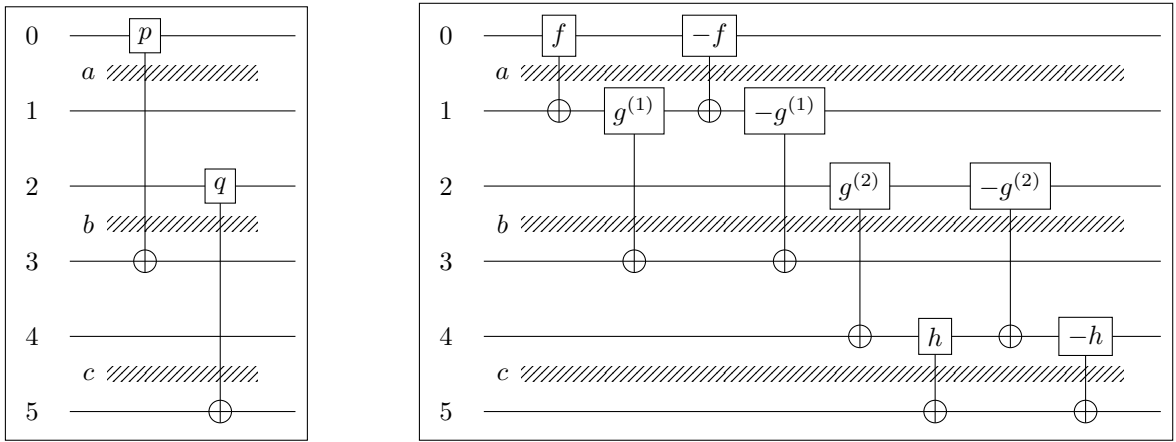
\begin{figure}[H]
    \centering
    \begin{subfigure}[t]{0.22\textwidth}
    \centering
    \begin{tikzpicture}[framed]
    \addwire{3}{0}{$0$};
    \addwire{3}{1}{$1$};
    \addwire{3}{2}{$2$};
    \addwire{3}{3}{$3$};
    \addwire{3}{4}{$4$};
    \addwire{3}{5}{$5$};
    \addmoat{3}{0}{$a$};
    \addmoat{3}{2}{$b$};
    \addmoat{3}{4}{$c$};
    \addgate{1}{0}{3}{$p$};
    \addgate{2}{2}{5}{$q$};
    \end{tikzpicture}
    \caption{A two-gate circuit $C$.}
    \end{subfigure}%
    \hspace{0.05\textwidth}
    \begin{subfigure}[t]{0.72\textwidth}
    \centering
    \begin{tikzpicture}[framed]
    \addwire{9}{0}{$0$};
    \addwire{9}{1}{$1$};
    \addwire{9}{2}{$2$};
    \addwire{9}{3}{$3$};
    \addwire{9}{4}{$4$};
    \addwire{9}{5}{$5$};
    \addmoat{9}{0}{$a$};
    \addmoat{9}{2}{$b$};
    \addmoat{9}{4}{$c$};
    \addgate{1}{0}{1}{$f$};
    \addgate{2}{1}{3}{$g^{(1)}$};
    \addgate{3}{0}{1}{$-f$};
    \addgate{4}{1}{3}{$-g^{(1)}$};
    \addgate{5}{2}{4}{$g^{(2)}$};
    \addgate{6}{4}{5}{$h$};
    \addgate{7}{2}{4}{$-g^{(2)}$};
    \addgate{8}{4}{5}{$-h$};
    \end{tikzpicture}
    \caption{Suppose $fg^{(1)} = p$ and $g^{(2)}h = q$; then this is a circuit $C' \equiv C$ with eight gates.
    $C'$ can be obtained from $C$ avoidantly (by expanding each gate using the Steinberg commutator relation).
    The $\pm g^{(1)}$ and $\pm g^{(2)}$ gates avoid $a$ and $c$; $\pm f$ and $\pm g^{(2)}$ avoid $c$; $\pm g^{(1)}$ and $\pm h$ avoid $a$; and $\pm f$ and $\pm h$ avoid $b$.
    Hence, the first half of the gates can be commuted past the second half using avoidant transformations.}
    \end{subfigure}
    \caption{The utility of additional working space.}\label{fig:slide}
\end{figure}

Intuitively, having two wires in the bands $\WiresB$ and $\WiresC$ is helpful when certifying this circuit
because we have enough ``room'' for gates to ``slide past'' one another without interference.
We view this extra room as a kind of \emph{working space}, allowing us to, e.g.,
simultaneously ``write'' values onto the band $\WiresB$ (the output of the $e_{0,1}(\pm f)$ gates)
\emph{and} ``read'' values from the band $\WiresB$ (the input of the $e_{2,4}(\pm g^{(2)})$ gates).

\subsubsection{The Steinberg relations}

In the previous subsection, \Cref{ex:slide} showed how having more than one wire in the bands $\WiresB$ and $\WiresC$
let us easily prove that type-$\typP$ and type-$\typQ$ gates commute using avoidant transformations.
This is only the initial step towards our full goal: Certifying triviality of an \emph{arbitrary} circuit (wherein all gates are avoidant).

For this full goal, we need to employ a useful reduction due to~\cite{KO21}, which we now explain.
It is a well-known fact in group theory that the Steinberg relations (\Cref{fact:steinberg}) suffice to ``present'' the group of unipotent matrices over a finite field~\cite{Ste62}.
What this means is that --- temporarily ignoring the issue of transformations needing to be ``avoiding'' ---
given any circuit which is equivalent to the identity, it is possible to transform it into the identity
simply by repeated applications of the equations in \Cref{fact:steinberg} to adjacent pairs of gates.
(This can be viewed essentially as an implementation of Gaussian elimination on a unipotent matrix.)

But how can we implement this via avoidant transformations?
Indeed, many Steinberg relations will involve type-$\typZ$ gates (that is, gates $e_{i,j}(z)$ where $i \in \WiresA$ and $j \in \WiresD$);
these gates are not avoidant, since they cross all three of the moats $a$, $b$, and $c$.
Therefore, these relations cannot be used (or even stated) in the world of avoidant transformations.

The idea of~\textcite{KO21} (as interpreted also by~\cite{OS25}) was to \emph{simulate} these ``missing'' type-$\typZ$ gates by an ``alias'' circuit $\zeta_{i,j}(z)$
to ``stand in'' for the gate $e_{i,j}(f)$.
We should have that $\zeta_{i,j}(f) \equiv \Embed{e_{i,j}(f)}$ and yet $\zeta_{i,j}(f)$ is composed entirely of avoidant gates.
When $f$ admits a nice factorization, the existence of such $\zeta_{i,j}(z)$ is immediate  ---
e.g., set $\zeta_{i,j}(z) \coloneqq \Embed{e_{i,k}(f)} \Embed{e_{i,k}(q)} \Embed{e_{i,k}(-f)} \Embed{e_{i,k}(-q)}$ when $z = fq$, $f \in \Fle{3(k-i)}$, and $q \in \Fle{3(j-k)}$, and $k \in \WiresB$ is of our choosing.
Then, we again extend to arbitrary $f$ by linearity.
It is enough to then implement the remaining Steinberg relations as avoidant transformations while replacing every occurrence of $e_{i,j}(z)$ with $\zeta_{i,j}(z)$.
These relations can be partitioned into three classes:
\begin{itemize}
    \item The ``nontrivial commutator relations'' producing type-$\typZ$ elements:
    \[ \zeta_{i,j}(z) \equiv \Embed{e_{i,k}(f)} \Embed{e_{i,k}(-q)} \Embed{e_{i,k}(-f)} \Embed{e_{i,k}(-q)} \] (for \emph{every} $k \in \WiresB$)
    and \[\zeta_{i,j}(z) \equiv \Embed{e_{i,k}(p)} \Embed{e_{i,k}(h)} \Embed{e_{i,k}(-p)} \Embed{e_{i,k}(-h)}\] (for \emph{every} $k \in \WiresC$).
    We emphasize that only one such relation, for one fixed value of $k$, 
    can be imposed by fiat in the definition of $\zeta_{i,j}(z)$.
    (There are also nontrivial commutators with type-$\typA$ and type-$\typD$ elements, though we defer these until the end.)
    \item The ``trivial commutator'' relations involving type-$\typZ$ elements:
    Every gate of type $\typB$, $\typC$, $\typF$, $\typG$, $\typF$, $\typG$, $\typH$, $\typP$, and $\typQ$ commutes the alias circuits $\zeta_{i,j}(z)$ for $\typZ$.
    \item The ``linearity'' relation for type-$\typZ$ elements: $\zeta_{i,j}(z_1) \zeta_{i,j}(z_2) \equiv \zeta_{i,j}(z_1+z_2)$.
\end{itemize}

\subsubsection{Quantitative aspects of the proof}

So far, we have only discussed the qualitative aspects of whether it is \emph{possible} to certify triviality of a circuit,
and not the quantitative question of how \emph{efficiently} this can be done.
In particular, we are interested in using a number of operations which is independent of the number of wires $n$.

An example of quantitative considerations already came up in \Cref{ex:slide}, when we used the following sort of reduction:
For a type-$\typP$ gate $e_{i,j}(p)$, we have \[
\Embed{e_{i,j}(p)} = \Embed{e_{i,j}(p_0 X)} \Embed{e_{i,j}(p_1 X)} \Embed{e_{i,j}(p_2 X^2)} \cdots \Embed{e_{i,j}(p_{3(j-i)} X^{3(j-i)})}. \]
However, $j-i$ could already be as large as $n$, so we cannot afford to implement this full decomposition.
Instead, we will have to break up the terms in $p$ by degree only into $O(1)$ blocks (and similarly for $q$) and carry the proof through carefully;
these considerations create a number of annoying pitfalls later in the proof.

Furthermore, to prove our full theorem, we will actually need to work in a slightly different model.
In this model, our circuits are now made up of \emph{supergates} of the form $\Embed{e_{i_1,j_1}(f_1) \cdots e_{i_T,j_T}(f_T)}$,
where $e_{i_1,j_1}(f_1),\ldots,e_{i_T,j_T}(f_T)$ all avoid the same moat.

Two supergates are treated as identical if the underlying products are the same when considered as matrices:
$e_{i_1,j_1}(f_1) \cdots e_{i_T,j_T}(f_T) = e_{i'_1,j'_1}(f'_1) \cdots e_{i_{T'},j_{T'}}(f'_{T'})$.
Importantly, we will have to deal with the case where $T$, the number of gates making up a supergate, is allowed to be $\Theta(n^2)$.

In this new model, an avoidant transformation is now allowed to do the following:
Suppose we have a decomposition $C = C_1 C_2 C_3$, where $C_1,C_2,C_3$ are subcircuits made up of supergates,
$C_2$ contains $O(1)$ supergates and every supergate in $C_2$ avoids the same moat.
Then, we can replace $C_2$ with an equivalent subcircuit, also of length $O(1)$, avoiding this same moat.

Qualitatively, the old notion of avoidantly certifying the triviality a circuit made up of gates
and the new notion of certifying the triviality of a circuit made up of supergates are equivalent;
this is because we can transform a supergate $\Embed{e_{i_1,j_1}(f_1) \cdots e_{i_T,j_T}(f_T)}$ avoidantly
into a sequence of gates $\Embed{e_{i_1,j_1}(f_1)} \cdots \Embed{e_{i_T,j_T}(f_T)}$ and vice versa.
However, quantitatively, this transformation is very expensive: A single supergate can become $\Theta(n^2)$ gates!
Hence, we will instead need to operate on the supergates themselves directly.

Our ultimate strategy is very roughly to divide up every supergate into $O(1)$ sub-supergates,
based on the wires and monomial degrees which appear.
We then systematically re-implement the Steinberg relations in a ``block-wise'' manner.
There are some technical details --- the most annoying of which disappear in the case where $c-b$ and $b-c$ have greater common divisor $\Omega(n)$ ---
but morally, the proof is quite systematic and elementary.
We made no (significant) attempt to optimize the final bounds we prove on the cost of avoidant transformations,
and view this as an interesting problem for future work.

\section{Preliminaries} \label{sec:complex-prelims}

\subsection*{Notations}

For $a \le b \in \N$, we let $\rangeII{a}{b} \coloneqq \{ c \in \N : a \le c \le b\}$, $\rangeIE{a}{b} \coloneqq \{ c \in \N : a \le c < b\}$, and $\rangeEE{a}{b} \coloneqq \{ c \in \N : a < c < b \}$.

For $r \in \N$, we will write $X \subseteq_r Y$ to denote that $X \subseteq Y$ and $|X| = r$.  
We will also very frequently need to refer to the situation where we have multiple pairwise disjoint subsets of a given set, so we invent some notation for this.
We use $X_1,\ldots,X_k \sqsubset Y$ to denote that $X_1,\ldots,X_k$ are pairwise disjoint subsets of $Y$.
Further, we use $X_1,\ldots,X_k \sqsubset_r Y$ to denote that each $|X_i| = r$,
and $X_1,\ldots,X_k \sqsubset_{\le r} Y$ to denote that each $|X_i| \le r$.
We also use $X \sqcup Y$ to denote the union of two sets $X$ and $Y$ when we want to emphasize that $X$ and~$Y$ are disjoint.

\subsection{Indexed simplicial complexes}

In the body of this paper (i.e., outside the appendices), we only consider the case of $n$-partite, pure rank-$n$, weighted simplicial complexes.
We term these \emph{indexed simplicial complexes}.\footnote{
    These are also called \emph{numbered} complexes by \textcite{Lan50} and \emph{completely balanced} complexes by \textcite{Sta79}.}

\begin{definition}[Indexed simplicial complex] \label{def:icomplex}
    Let $\calI$ be a set of finite size~$n \in \N$.
    An \emph{$\calI$-indexed simplicial complex} $\X$ is a probability distribution over a product set 
    \[
        \partOf{\X}{\calI} \coloneqq \prod_{i \in \calI} \partOf{\X}{i},
    \]
    where each $\partOf{\X}{i}$ is a nonempty finite set that we call a \emph{\universe} of~$\X$. We call the elements of $\partOf{\X}{i}$ \emph{\symbls}, and the elements of $\partOf{\X}{\calI}$ \emph{\strngs}. 
\end{definition}

\begin{definition}[\Substrings]
    Let $\X$ be an $\calI$-indexed simplicial complex.
    For $x \in \partOf{\X}{\calI}$ and $S \subseteq \calI$, we define the \emph{\substring}
    $x_S \coloneqq (x_i)_{i \in S} \in \partOf{\X}{S}$.  We sometimes call $x_S$ an \emph{$S$-\strng}.
\end{definition}

\begin{definition}[Restricted complexes]
    Let $\X$ be an $\calI$-indexed simplicial complex.
    For a set $\calJ$ and a 
    tuple of pairwise disjoint subsets $\calS = (S_j \subset \calI)_{j \in \calJ}$, 
    we define the \emph{$\calS$-restricted} $\calJ$-indexed simplicial complex $\res{\X}{\calS}$.
    We preliminarily define the \universes of this complex to be $\partOf{\X}{S_j}$ (for $j \in \calJ$). Then, the distribution $\res{\X}{\calS}$ is given by drawing $\bx \sim \X$ and outputting the tuple of \substrings $(\bx_{S_j})_{j \in \calJ}$.  
    Finally, in case the resulting indexed complex has symbols that never occur in $\supp(\res{\X}{\calS})$, we delete these symbols from their universes.
\end{definition}

\begin{remark}
    \textcite{DD24-swap} call the restricted complex $\res{\X}{\calS}$ the \emph{colored faces complex},
    and the special case of $\res{\X}{S_1,S_2}$ the \emph{colored swap walk}.
    This walk was also studied in~\cite{DD19,GLL22}.
\end{remark}

\begin{definition}[\Faces/faces]   
    In the setting of the previous definitions, if $x \in \supp(\X)$, we call $x_S$ a \emph{\face}\footnote{We apologize for renaming ``faces'' in this way, particularly since the term could be confused with the group-theoretic notion of ``word''.  Nevertheless, we find it has the right mnemonic value; as in English, every word is a string, but not every string is a word.} or a \emph{face} of~$\X$.  We also call $x_S$ an $S$-\face (to emphasize $S$) or an $|S|$-\face (to emphasize the size of~$S$).
    Hence, the set of $S$-faces is precisely $\supp(\res{\X}{S})$.
\end{definition}

Two indexed simplicial complexes are isomorphic if one can be obtained from the other by re-indexing and renaming \symbls.
Formally:
\begin{definition}[Isomorphism of indexed simplicial complexes]\label{def:isomorphism}
    Let $\X$ (respectively, $\tilde{\X}$) be an $\calI$-indexed (respectively, $\tilde{\calI}$-indexed) simplicial complex with universes $(\partOf{\X}{i})_{i \in \calI}$ (respectively, $(\tildepartOf{\X}{j})_{j \in \tilde{\calI}}$).
    We say they are \emph{isomorphic} (notated $\X \cong \tilde{\X}$) if there is a bijection $\iota : \calI \to \tilde{\calI}$ and, for each $i \in \calI$ a bijection $\phi_i : \partOf{\X}{i} \to \tildepartOf{\X}{\iota(i)}$
    (from $\{i\}$-words to $\{\iota(i)\}$-words)
    such that for $\bx \sim \X$, the \strng-valued random variable $(\phi_i(\bx_i))_{\iota(i) \in \tilde{\calI}}$ is distributed as $\tilde{\X}$.
\end{definition}

\begin{definition}[Clique complex]\label{def:clique complex}
    We say the $\calI$-indexed simplicial complex $\X$ is a \emph{clique complex} provided the following holds for all \substrings~$x_S$:  if $x_E$ is a \face for all $E \subseteq_2 S$, then $x_S$ is a \face.
\end{definition}
\begin{remark}
    An $\calI$-indexed clique complex is equivalent to a (weighted, finite) \emph{incidence geometry}; see, e.g.,~\cite{BC13}.
    Almost all indexed complexes we deal with will be clique complexes.
\end{remark}

\begin{definition}[Links/pinnings]
    Let $\X$ be an $\calI$-indexed simplicial complex, $S \subseteq \calI$, and $w$ an $S$-\face. 
    Write $\overline{S} = \calI \setminus S$.
    We define an associated 
    \emph{pinned} (or \emph{link}) complex $\link{\X}{w}$. This is the $\overline{S}$-indexed simplicial complex whose \universes are (preliminarily) defined to be $(\partOf{\link{\X}{s}}{i})_{i \in \overline{S}}$, and whose distribution $\link{\X}{w}$ is defined by drawing $\bx \sim \X$ conditioned on $\bx_{\calI \setminus S} = w$ and outputting $\bx_{\overline{S}}$.
    As in the definition of restricted complexes, in case this indexed complex has any symbols that occur with zero probability, they are deleted from their universe.
\end{definition}

We now effectively give the definitions of ``$r$-faces complex'' from~\cite{DD24-covers} (see also~\cite{BLM24}):
\begin{definition}[$r$-\triwords (and \biwords)]
    Let $\X$ be an $\calI$-indexed simplicial complex, let $1 \leq r \leq |\calI|/3$ and let $S_1,S_2,S_3 \sqsubset_r \calI$.  
    We then call any element of $\supp(\res{\X}{S_1,S_2,S_3})$ an \emph{$r$-\triword} of~$\X$.
    We can alternatively think of an $r$-\triword as an unordered triple of $r$-\faces (with disjoint index sets).
    We write $T_r(\X)$ for the probability distribution on $r$-\triwords obtained by drawing $\bx \sim \X$, independently drawing $\bS_1, \bS_2, \bS_3 \sqsubset_r \calI$ uniformly at random, and outputting $\{\bx_{\bS_1}, \bx_{\bS_2}, \bx_{\bS_3}\}$.
    We may analogously define \biwords (unordered \emph{pairs} of \faces) and the \biword distribution~$E_r(\X)$.
    And finally, we write $V_r(\X)$ for the analogous distribution on \faces (``mono-\faces'').
\end{definition}
\begin{remark}
    Identifying $V_r(\X)$, $E_r(\X)$, $T_r(\X)$ with their support, we may regard $\{\emptyset, V_r(\X), E_r(\X), T_r(\X)\}$ as a \emph{non-indexed} simplicial complex.  More precisely, it is a weighted, pure, rank-$3$ (dimension-$2$)  simplicial complex.
    We may call it the \emph{$r$-\triword complex} of~$\X$.
\end{remark}

\subsection{Spectral expansion}

\begin{definition}
    Let $\X$ be an $\{i_1,i_2\}$-indexed edge complex (a.k.a. a joint distribution on $\partOf{\X}{i_1} \times \partOf{\X}{i_2}$).
    We define the \emph{averaging operator} $A : \R^{\partOf{\X}{i_1}} \to \R^{\partOf{\X}{i_2}}$ via \[
    (Af)(v_2) \coloneqq \Exp_{v_1 \sim \resLink{\X}{\ell_1}{v_2}}[f(v_1)]. \]
    We say $\X$ is an \emph{$\epsilon$-bipartite (spectral) expander} if the second largest singular value of $A$ is at most $\epsilon$.
    (The ordering on $i_1,i_2$ does not matter; the converse averaging operator is the transpose $A^\intercal$, which has the same (nonzero) singular values.)
\end{definition}

\begin{fact}
    In the setup of the preceding definition, $\X$ is the \emph{uniform} distribution on the product set $\partOf{\X}{i_1} \times \partOf{\X}{i_2}$
    if and only if it is a $0$-bipartite expander.
\end{fact}

The following notion was studied in~\cite{DD19} and formally defined in~\cite[\S3]{GLL22}:

\begin{definition}
    Let $\X$ be an $\calI$-indexed simplicial complex and $\epsilon > 0$.
    $\X$ is \emph{$\epsilon$-product} if
    for every $\ell_1 \ne \ell_2 \in \calI$,
    $W \subseteq \calI \setminus \{\ell_1,\ell_2\}$,
    and $W$-word $w$,
    $\resLink{\X}{\{\ell_1\},\{\ell_2\}}{w}$ is an $\epsilon$-bipartite expander.
\end{definition}

\begin{remark}\label{rmk:heritable}
    $\epsilon$-productness, along with cliqueness, are \emph{heritable} properties of simplicial complexes:
    If a simplicial complex satisfies one of these properties, then so do all its links.
\end{remark}

We also define the following weaker notion:

\begin{definition}
    Let $\X$ be an $\calI$-indexed simplicial complex ($|\calI| \ge 2$) and $\epsilon > 0$.
    $\X$ is an \emph{$\epsilon$-local bipartite expander} if
    for every $\ell_1 \ne \ell_2 \in \calI$,
    and $(\calI \setminus \{\ell_1,\ell_2\})$-word $w$,
    $\resLink{\X}{\{\ell_1\},\{\ell_2\}}{w}$ is an $\epsilon$-bipartite expander.
\end{definition}

We then have the following trickle-down type theorem for establishing $\epsilon$-productness:

\begin{theorem}[{\cite[Corollary 7.6]{DD19}}]\label{thm:local bipartite to product}
    Let $\X$ be an $\calI$-indexed simplicial complex, $|\calI| = n$.
    Let $\epsilon < \frac12$ and suppose that:
    \begin{enumerate}
        \item $\X$ is an $(\frac{\epsilon}{(n-1)\epsilon+1})$-local bipartite expander.
        \item For every $W \subseteq_{\le n-2} \calI$ and $W$-word $w$, $\link{\X}{w}$ is connected.
    \end{enumerate}
    Then $\X$ is $\epsilon$-product.
\end{theorem}

\subsection{Triword expansion}

$1$-cohomology of a triangle complex over a group~$\Gamma$, and its quantitative form $1$-coboundary expansion, are defined in \Cref{sec:high-coboundary}.  
In the remainder of the paper, we will only need this definition in the context of the $r$-triword complex.  For clarity, we explicitly give the definition here:

\begin{definition}[Antisymmetric biword labelings]
    Let $\Gamma$ be a group, $\X$ an $\calI$-indexed simplicial complex, and $r \le \tfrac12 |\calI|$.
    A \emph{biword labeling} is a function $f : \ori{E}_r(\X) \to \Gamma$ satisfying the antisymmetry condition
    that for every $(s_1,s_2) \in \ori{E}_r(\X)$, $f(s_1,s_2) = f(s_2,s_1)^{-1}$.
\end{definition}

\begin{definition}[Triword expansion]
    Let $\Gamma$ be a group, $\X$ an $\calI$-indexed simplicial complex, $r \le \frac13 |\calI|$, and $\epsilon > 0$.
    $\X$ has \emph{$r$-triword expansion (at least) $\epsilon$ over $\Gamma$}
    if for every biword-labeling $f : \ori{E}_r(\X) \to \Gamma$,
    there exists a word-labeling $g : V_r(\X) \to \Gamma$ such that \[
    \Pr_{(\bs_1,\bs_2) \sim \ori{E}_r(\X)} \bracks*{ g(\bs_1) g(\bs_2)^{-1} \ne f(\bs_1,\bs_2) }
    \le \frac1{\epsilon} \cdot \Pr_{(\bs_1,\bs_2,\bs_3) \sim \ori{T}_{r}(\X)} \bracks*{ f(\bs_1,\bs_2) f(\bs_2,\bs_3) f(\bs_3,\bs_1) \ne \Id }. \]
    When $r$ is omitted, we assume it is $1$.
\end{definition}

\begin{remark}
    The property that we call ``triword expansion'' is called ``swap coboundary expansion'' in~\cite{DD24-swap,DDL24} and ``unique-games coboundary expansion'' in~\cite{BM24,BLM24}.
\end{remark}

\subsection{Direct-product testing from triword expansion}

\textcite{BM24} and \textcite{DD24-covers} determined ``high-dimensional expansion conditions'' on~$\calD$ that sufficed for the associated agreement tester to achieve good soundness.
We present the following version of their results, which can be viewed as a version of \cite[Theorem 3.1]{DD24-covers}, together with background results; see also \cite[Theorem A.1]{BMVY25}.  We comment that conditions~4,~5 below correspond to Properties~(\ref{itm:local}),~(\ref{itm:global}) in \Cref{sec:intro}, respectively.

\begin{restatable}[Sufficient conditions for direct-product testability, combining {\cite{DD24-covers,DD24-local}}]{theorem}{sufficient}\label{thm:sufficient conditions}
    Given $\delta > 0$, there is $m = \poly(1/\delta)$, such that if $k \geq \exp(\poly(1/\delta))$, $r \geq \poly(k)$, $n \geq \exp(\poly(r))$, and $\gamma < 1/n^2$, then the following holds:
    
    Suppose $\X$ is a rank-$n$ complex (not necessarily indexed) satisfying the following:
    \begin{conditions}
        \conditem{Spectral expansion}
        For any pinning~$w$ of size $n-2$, the graph $\X_w$ has normalized $2$nd eigenvalue at most~$\gamma$. \label{itm:1}
        \conditem{Connectivity}
        For every $S$-word $w$ with $|S| \leq r$, the graph $(V_r(\X_w), E_r(\X_w))$ is connected. \label{itm:2}
        \conditem{Clique complex}
        $\X$ is a clique complex.
        \conditem{Local triword expansion}
        For every $1$-word $w$ (``vertex''), $\X_w$ has $r$-triword expansion at least $\exp(-r^{.99})$ over~$\Sym{m_0}$ for all $1 < m_0 \leq m$.\label{item:sufficient conditions:typical triword expander}
        \conditem{Trivial global cohomology}
        $\X$ has nonzero $1$-triword expansion (i.e., trivial $1$-cohomology) over~$\Sym{m_0}$ for all $1 < m_0 \leq m$.\label{itm:4}
    \end{conditions}    
    Then the agreement tester with distribution given by $\X$'s rank-$k$ faces achieves $\delta$-soundness for any alphabet~$\Sigma$, in the sense of \Cref{def:agreement-tester}.
\end{restatable}

In \Cref{sec:sufficient conditions for direct-product testing} below, we discuss why this theorem follows from the works of~\cite{DD24-covers,DD24-local}.

\subsection{Strongly symmetric complexes}

We define a notion of strong symmetry in indexed simplicial complexes, which, in mathematical terms,
posits that the automorphism group of a complex acts transitively on words:

\begin{definition}\label{def:strong symm}
    Let $\X$ be an $\calI$-indexed simplicial complex.
    We say $\X$ is \emph{strongly symmetric} if for every pair of words $x  = (x_i)_{i \in \calI}, x'  = (x'_i)_{i \in \calI} \in \supp(\X)$,
    there exist bijections $\phi_i : \partOf{\X}{i} \to \partOf{\X}{i}$
    such that:
    \begin{enumerate}
    \item $\phi_i(x_i) = x'_i$ for every $i \in \calI$.
    \item $(\phi_i)_{i \in \calI}$ form an isomorphism from $\X$ to itself (together with the identity map $\iota : \calI \to \calI$), in the sense of \Cref{def:isomorphism}.
    That is, for $\bx = (\bx_i)_{i \in \calI} \sim \X$, $(\phi_i(\bx_i))_{i \in \calI}$ is also distributed as $\X$. \qedhere
    \end{enumerate}
\end{definition}

\begin{remark}
    This notion of strong symmetry in an indexed simplicial complex is stronger than the standard definition of strong symmetry in simplicial complexes.
    (The standard definition allows a single bijection $\Phi : V(\X) \to V(\X)$ that need not preserve the individual parts $\partOf{\X}{i}$.) 
    We use this condition because coset complexes do satisfy it and it is easier to state.
\end{remark}

\begin{definition}[Diameter]
    Let $\X$ be an $\calI$-indexed simplicial complex.
    For $s, s' \in V(\X)$, the \emph{distance} between $s$ and $s'$ is the minimum $k \in \N$ such that
    there exist $(s^{(i)} \in V(\X))_{i \in \rangeII{0}{k}}$ such that $s^{(0)} = s$, $s^{(k)} = s'$, and for every $i \in [k]$, $(s^{(i-1)},s^{(i)}) \in \ori{E}(\X)$.
    The \emph{diameter} $\diam(\X)$ of $\X$ is the maximum distance between any $s, s' \in V(\X)$.
\end{definition}

\subsection{Triword expansion of linearly ordered complexes}

We next turn to the question of proving that a simplicial complex is a triword expander.
We will need to do this for the local Kaufman--Oppenheim complex in order to fulfill
\Cref{item:sufficient conditions:typical triword expander} of \Cref{thm:sufficient conditions}.
We present a sufficient condition for triword expansion in complexes with ``linearly ordered'' geometry in \Cref{thm:linear condition} below.
This again follows the proof of \cite[\S3-5]{BLM24} which is specialized to the complex of \textcite{CL25-applications},
but our statement is generic and can be applied to the complex of \textcite{KO18} as well.

\begin{definition}[Productization]
    Let $\X$ an $\calI$-indexed simplicial complex, and $S_1, S_2 \sqsubset \calI$.
    We say $\X$ \emph{productizes over $S_1$ and $S_2$}, denoted $\ptiz{\X}{S_1}{S_2}$,
    if for $(\bs_1,\bs_2) \sim \res{\X}{S_1,S_2}$, $\bs_1$ and $\bs_2$ are independent.
    (Equivalently, the joint distribution factors into a product distribution, $\res{\X}{S_1,S_2} = \res{\X}{S_1} \times \res{\X}{S_2}$.)
\end{definition}

\begin{definition}[Ordering on sets]
    Let $n \in \N$.
    For two sets $S_1,S_2 \subseteq \rangeOne{n}$, we say $S_1$ \emph{precedes} $S_2$, denoted $S_1 \prec S_2$, if $\max S_1 < \min S_2$.
    In particular, $S_1 \prec S_2 \implies S_1 \cap S_2 = \emptyset$ (so that $S_1,S_2 \sqsubset \rangeOne{n}$).
    We write $S_1 \prec \cdots \prec S_k \sqsubset \rangeOne{n}\setminus W$ if each $S_i \subset \rangeOne{n} \setminus W$ and $S_i \prec S_{i+1}$ for every $1 \le i \le k-1$.
\end{definition}

\begin{definition}[Separated set]
    Let $\sigma \in \N$. We say $R \subseteq \rangeOne{n}$ is \emph{$\sigma$-separated}
    if for every $\ell_1 \ne \ell_2 \in R$, we have $|\ell_2 - \ell_1| \ge \sigma$.
\end{definition}

Note that if $S \subseteq R$ and $R$ is $\sigma$-separated, then $S$ is also $\sigma$-separated.

\begin{definition}[Productization on the line]\label{def:pol}
    Let $\X$ be an $\rangeOne{n}$-indexed simplicial complex.
    We say $\X$ is \emph{productizing on the line (PoL)} if the following is true:
    For every $W \subseteq \rangeOne{n}$ and $S_1 \prec S_2 \sqsubset \rangeOne{n} \setminus W$,
    if $W \cap \rangeEE{\max S_1}{\min S_2} \ne \emptyset$, then for every $W$-word $w$,
    $\ptiz{\link{\X}{w}}{S_1}{S_2}$.
\end{definition}

\begin{remark}
    The foregoing three definitions are the main reason why we need to consider simplicial complexes indexed by $\rangeOne{n}$, which has ``the geometry of the line''.
\end{remark}

For completeness, we prove the following theorem in \Cref{sec:triword expansion in linear complexes}, though it is essentially wholly due to~\cite{BLM24}:

\begin{restatable}[Proving triword expansion for ``linearly ordered'' complexes, abstracted from~\cite{BLM24}]{theorem}{line}{\label{thm:linear condition}}
    There exist absolute constants $K > 1$ and $\epsilon_0 > 0$ such that
    for every $r \in \N$, there exist $\xi > 0$ and $n_0 \in \N$ such that for every $n \ge n_0 \in \N$, the following holds.
    Let $\X$ be a $\rangeOne{n}$-indexed simplicial complex and $\Gamma$ an arbitrary group.
    Suppose that $\X$ satisfies the following conditions:
    \begin{conditions}
    \conditem{Spectral expansion}\label{item:lin:spectral}
    $\X$ is $\epsilon_0 r^{-2}$-product.
    \conditem{Productization}\label{item:lin:pol}
    $\X$ productizes on the line, in the sense of \Cref{def:pol}.
    \conditem{Symmetry}\label{item:lin:symmetry}
    $\X$ is strongly symmetric, in the sense of \Cref{def:strong symm}.
    \conditem{Coboundary expansion of separated singleton restrictions}\label{item:lin:cbdy}
    For every $W \subseteq \rangeOne{n}$ and $\ell_1 < \ell_2 < \ell_3 \in \rangeOne{n}$
    such that $W \cap \rangeII{\ell_1}{\ell_3} = \emptyset$ and $\{\ell_1,\ell_2,\ell_3\}$ is $\xi n$-separated,
    for every $W$-word $w$,
    the complex $\resLink{\X}{\{\ell_1\},\{\ell_2\},\{\ell_3\}}{w}$ has $1$-triword expansion at least $\beta$ over $\Gamma$.
    \conditem{Diameter of separated singleton restrictions}\label{item:lin:diam}
    For every $W \subseteq \rangeOne{n}$ and $\ell_1 < \ell_2 \in \rangeOne{n}$
    such that $W \cap \rangeII{\ell_1}{\ell_2} = \emptyset$ and $\{\ell_1,\ell_2\}$ is $\xi n$-separated,
    for every $W$-word $w$,
    the complex $\resLink{\X}{\{\ell_1\},\{\ell_2\}}{w}$ has diameter at most $C_0$.
    \end{conditions}
    Then for every $W \subset_{\le r} \rangeOne{n}$ and $W$-word $w$,
    $\link{\X}{w}$ has $r$-triword expansion at least $(K \beta/C_0)^{O(r^{0.67})}$ over $\Gamma$.
\end{restatable}

We note that there are a few minor differences in our theorem statement from that of~\cite{BLM24}.
In particular: we have a quantitatively spectral expansion hypothesis ($O(r^{-2})$ vs. $\exp(-\poly(r))$),
have a quantitatively stronger conclusion (exponent of $r^{0.67}$ vs. $r^{0.99}$),
and do not need to account for additive error.
However, we emphasize that these improvements arise only from some slack in their proof and we do not claim them as contributions.

\subsection{Subgroups and cosets}

For two groups $H$ and $G$, $H \le G$ denotes that $H$ is a subgroup of $G$.
For $g \in G$, $gH \coloneqq \{g h : h \in H\}$ is the \emph{left coset} corresponding to $g$.
$G/H \coloneqq \{ g H : g \in G \}$ denotes the set of all left cosets.
More generally, for two subgroups $H_1,H_2 \le G$, we write $H_1 H_2 \coloneqq \{ h_1 h_2 : h_1 \in H_1, h_2 \in H_2\}$.

\begin{definition}\label{def:indexed subgroup family}
    Let $G$ be a finite group and $\calI$ a finite set.
    An \emph{$\calI$-indexed subgroup family} for $G$ is a collection $\calH = (H_i \le G)_{i \in \calI}$ of subgroups indexed by $\calI$.
    For $S \subseteq \calI$, we write $H_S \coloneqq \bigcap_{i \in S} H_i$, with $H_\emptyset \coloneqq G$.
    An $\calI$-indexed subgroup family $\calH = (H_i \le G)_{i\in\calI}$ for $G$ and $\calI'$-indexed subgroup family $\calH' = (H'_i \le G)_{i \in \calI'}$ are \emph{isomorphic}
    if there exists a re-indexing bijection $\iota : \calI \to \calI'$ and a group isomorphism $\phi : G \stackrel{\sim}{\to} G'$
    such that $\phi(H_i) = H'_{\iota(i)}$ for every $i \in \calI$.
\end{definition}

We collect here some trivial but useful facts about cosets.
Mutually intersecting cosets of different subgroups are cosets of the intersection subgroup:

\begin{proposition}[{\cite[Ex.~2.26]{Rot95}}]\label{prop:prelim:coset intersection}
    For a group $G$, an $\calI$-indexed subgroup family $\calH = (H_i \le G)_{i \in \calI}$,
    and a tuple of cosets $(C_i \in G/H_i)_{i \in \calI}$,
    either $\bigcap_i C_i = \emptyset$ or $\bigcap_i C_i$ is a coset of the intersection subgroup $\bigcap_i H_i$.
\end{proposition}

Two cosets intersect if their ratio is a product:

\begin{proposition}[{\cite[Observation 3.2]{HS24}}]\label{prop:prelim:condition for coset intersection}
    For all groups $G$ and subgroups $H_1, H_2 \le G$ and $g_1, g_2 \in G$, $g_1 H_1 = g_2 H_2$ iff $g_1^{-1} g_2 \in H_1 H_2$.
\end{proposition}

Taking $H_1 = H_2$, we get:

\begin{corollary}[also {\cite[Thm.~2.8]{Rot95}}]
    For all groups $H \le G$ and $g_1, g_2 \in G$, $g_1H = g_2H$ iff $g_1^{-1} g_2 \in H$.
\end{corollary}

and therefore:

\begin{corollary}\label{fact:prelim:coset containment}
    For all groups $H \le G$, $g \in G$, and $C \in G/H$, $g \in C$ if and only if $C = gH$.
\end{corollary}

\begin{proof}
    Obviously, if $C = gH$ then $g = g \Id \in C$.
    Conversely, if $C = g' H$ and $g \in C$, there exists $h \in H$ such that $g' h = g$ and therefore $(g')^{-1} g \in H$;
    then apply the previous proposition.
\end{proof}

Also, the cosets $G/H$ partition $G$:

\begin{proposition}[{\cite[Thm.~2.9]{Rot95}}]\label{prop:prelim:coset overlap}
    For all groups $H \le G$, any two cosets $C, C' \in G / H$ are either identical or disjoint.
\end{proposition}

We have the following product formula:

\begin{proposition}[{\cite[Thm.~2.20]{Rot95}}]\label{prop:prelim:product formula}
    For all groups $H_1,H_2 \le G$, $|H_1 H_2| \cdot |H_1 \cap H_2| = |H_1| \cdot |H_2|$.
\end{proposition}

$G$ acts on $G/H$ by left multiplication:

\begin{fact}\label{fact:prelim:coset action}
    For all groups $H \le G$, cosets $C \in G/H$ of $H$, and elements $g \in G$, $g C \in G/H$,
    i.e., $g C$ is again a coset of $H$.
    Further, this is a group action, i.e., $\Id C = C$ and $g_1 (g_2 C) = (g_1g_2) C$.
\end{fact}

\begin{proof}
    We have $g(g' H) = (gg') H$, $\Id (g' H) = g' H$, and $(g_1g_2)(g' H) = (g_1g_2g') H g_1 (g_2 (g' H))$.
\end{proof}

\subsection{Coset complexes}

We now define a way to construct an indexed simplicial complex from an indexed subgroup family,
dating back at least as far as the work of~\textcite{Lan50}:

\begin{definition}\label{def:coset complex}
    Let $G$ be a finite group and $\calH = (H_i \le G)_{i \in \calI}$ an $\calI$-indexed subgroup family.
    The \emph{$G$-coset complex} $\CoCo{G}{\calH}$ is the $\calI$-indexed simplicial complex with \universes $\partOf{\CoCo{G}{\calH}}{i} \coloneqq G/H_i$ and the following distribution $\CoCo{G}{\calH}$:
    Draw $\bg \sim G$ uniformly and output the string of cosets $(\bg H_i)_{i \in \calI}$.
\end{definition}

\begin{remark}
    This definition goes under different names in the literature; for instance, if one ignores the probabilistic weighting, it is called a finite \emph{coset incidence system} in~\cite{BC13}.
    The definition in~\cite{KO18,KO23} of coset complexes is slightly different than ours (which is consistent with~\cite{Lan50,Gar79,Tit86,OP22,HS24,OS25}).
    The~\cite{KO18,KO23} definition takes the result of \Cref{def:coset complex}, which is not necessarily a clique complex, and then adds extra faces corresponding to all cliques.
    The definitions coincide when the complex defined in \Cref{def:coset complex} is already a clique complex.
\end{remark}

We now collect several basic properties of coset complexes which we shall use repeatedly:

\begin{proposition}\label{prop:coset complex:faces}
    For $S \subseteq \calI$, an $S$-\strng $w = (C_i \in G/H_i)_{i \in S}$ is an $S$-\face
    (that is, $w \in \supp(\res{(\CoCo{G}{\calH})}{S}$)
    iff the cosets intersect mutually; i.e., $\bigcap_{i \in S} C_i \ne \emptyset$.
\end{proposition}
\begin{proof}
    $w$ is an $S$-word if and only if there exists $g \in G$ such that $C_i = gH_i$ for every $i \in \calI$;
    then use \Cref{fact:prelim:coset containment}.
\end{proof}

\begin{proposition}\label{prop:coset complex:strongly symmetric}
    $\CoCo{G}{\calH}$ is always strongly symmetric.
\end{proposition}
\begin{proof}
    Let $x = (C_i)_{i \in \calI}, x' = (C'_i)_{i \in \calI} \in \supp(\CoCo{G}{\calH})$ be words.
    By \Cref{prop:coset complex:faces}, $\bigcap_{i \in \calI} C_i \ne \emptyset$ and $\bigcap_{i \in \calI} C'_i \ne \emptyset$.
    By \Cref{prop:prelim:coset intersection}, there exist $g, g' \in G$ such that
    $\bigcap_{i \in \calI} C_i = g H_\calI$ and $\bigcap_{i \in \calI} C'_i = g' H_\calI$.
    (Here, we used the notation $H_\calI = \bigcap_{i \in \calI} H_i$.)
    Hence $g \in C_i$ and $g' \in C'_i$ for every $i \in \calI$. 
    So by \Cref{fact:prelim:coset containment}, $C_i = gH_i$ and $C'_i = g'H_i$ for every $i \in \calI$.

    We use the bijections $\phi_i : G/H_i \to G/H_i$ defined by $C \mapsto g' g^{-1} C$.
    (If $C \in G/H_i$, then $g' g^{-1} C \in G/H_i$ by \Cref{fact:prelim:coset action}.
    Further, $\phi_i$ is a bijection on $G/H_i$ because its inverse is $C \mapsto g (g')^{-1} C$.)
    By definition, $\phi_i (C_i) = g' g^{-1} (g H_i) = g' H_i = C'_i$.
    Finally, these maps preserve the distribution $\CoCo{G}{\calH}$
    because for $\bg \sim G$, $\bg$ and $g'g^{-1} \bg$ have the same distribution;
    therefore, $(\bg H_i)_{i \in \calI}$ and $(g'g^{-1} \bg H_i)_{i \in \calI}$ also have the same distribution.
\end{proof}

\begin{remark}
    A coset clique complex is equivalent to a (finite) \emph{coset geometry}; see, e.g.,~\cite{BC13}.
\end{remark}

A useful property of $G$-coset complexes is that any (nontrivial) pinning is naturally a $G'$-coset complex for $G' \leq G$:

\begin{proposition}[e.g., {\cite[Lemma 3.4]{HS24}}]\label{prop:coset complex:link}
    For any $W \subseteq \calI$ and $W$-\face $w$, we have the isomorphism \[
    \link{(\CoCo{G}{\calH})}{w} \cong \CoCo{H_S}{(H_{W \cup \{i\}})_{i \in \calI\setminus W}}. \]
    (The left-hand side is a $\calI\setminus W$-indexed simplicial complex
    and the index mapping $\iota : \calI\setminus W \to \calI\setminus W$ is the identity function.)
\end{proposition}

Also, if $N \trianglelefteq G$ is a normal subgroup of $G$,
$\CoCo{G}{\calH}$ is isomorphic to a $G/N$-quotient complex,
assuming that $N \le H_i$ for every $i \in \calI$.
We present this proposition in the following equivalent form:

\begin{proposition}\label{prop:coset complex:pass to quotient}
    Let $G$ be a group, $L \le G$, $\pi : G \twoheadrightarrow L$ a surjective map onto $L$,
    and $\calH = (H_i \le G)_{i \in \calI}$ an indexed subgroup family satisfying $\ker(\pi) \le H_{\calI}$.
    Then $\CoCo{G}{\calH} \cong \CoCo{L}{(\pi(H_i))_{i \in \calI}}$.
\end{proposition}

See \Cref{sec:coset complex} below for a (standard) proof.
One useful special case of \Cref{prop:coset complex:pass to quotient} is when the surjective map is actually an isomorphism:

\begin{corollary}
    Suppose we have isomorphic indexed subgroup families, in the sense of \Cref{def:indexed subgroup family}:
    $G$ and $G'$ are groups, $\phi : G \stackrel{\sim}{\to} G'$ is a group isomorphism,
    $\calH = (H_i \le G)_{i \in \calI}$ an $\calI$-indexed subgroup family,
    and $\iota : \calI \to \calI'$ a reindexing function.
    Then $\CoCo{G}{\calH} \cong \CoCo{G'}{(\phi(H_i))_{i \in \calI}}$.
\end{corollary}

\begin{proof}
    Use the prior \Cref{prop:coset complex:pass to quotient} and the fact that $\ker(\phi)$ is trivial.
\end{proof}

This is the only possible application of \Cref{prop:coset complex:pass to quotient}
when the intersection of the subgroups $H_\calI$ is trivial, which is the typical situation
(for instance, this is the case in the local Kaufman--Oppenheim complex).
However, cases where the intersection of subgroups is nontrivial do arise.
One typical case is when we restrict the subgroup family $(H_i)_{i \in \calI}$ to a subfamily $(H_i)_{i \in \calJ}$ for some $\calJ \subsetneq \calI$.
We give a corollary which is useful for this case in \Cref{cor:coset complex:reduce} below.

Finally, we prove the following simple condition for productization in \Cref{sec:coset complex}:

\begin{proposition}\label{prop:coset complex:productization}
    Let $G$ be a group and $\calH = (H_i \le G)_{i \in \calI}$ an $\calI$-indexed subgroup family.
    Suppose $W,S_1,S_2 \sqsubseteq \calI$.
    Then for every $W$-word $w$, $\ptiz{\link{\CoCo{G}{\calH}}{w}}{S_1}{S_2}$
    if and only if $H_W = H_{W \cup S_1} H_{W \cup S_2}$.
    In particular, taking $W = \emptyset$, we get that $\ptiz{\CoCo{G}{\calH}}{S_1}{S_2}$ if and only if $G = H_{S_1} H_{S_2}$.
\end{proposition}

All rings are commutative rings with identity.
For a field $\F$, $\F[X]$ denotes the ring of polynomials over $\F$
in a single indeterminate $X$.
$\deg(f)$ is the degree of a polynomial $f \in \F[X]$, with the convention $\deg(0) \coloneqq -\infty$.

For a group $G$ and two elements $g,h \in G$, the \emph{commutator} of $g$ and $h$ is
\begin{equation}\label{eq:comm}
\Comm{g}{h} \coloneqq g^{-1} h^{-1} g h
\end{equation}
so that
\begin{equation}\label{eq:reorder}
gh = \Comm{g}{h} hg.
\end{equation}
(In particular, $g$ and $h$ commute in $G$ iff $\Comm{g}{h} = 1$.)

\subsection{The graded unipotent group and the local Kaufman--Oppenheim complex}

For a field $\F$, an indeterminate $X$, and $d \in \N$, let $\Fle{n}$ denote the subset of polynomials over $X$ of degree at most $d$.

\begin{definition}[Graded unipotent group]\label{def:graded unipotent group}
    Let $n,\kappa \in \N$, and $\F$ be a field.
    The group $\GrUnip{n}{\F}$ is defined by the following presentation:
    \begin{itemize}
        \item The generators are symbols $e_{i,j}(f)$, where $i < j \in \rangeII{0}{n}$
        and $f \in \Fle{\kappa(j-i)}$.
        \item The relations are:
        \begin{itemize}
            \item The \emph{linearity} relations: For every $i < j \in \rangeII{0}{n}$
            and $f,g \in \Fle{\kappa(j-i)}$,
            \[ e_{i,j}(f) \cdot e_{i,j}(g) = e_{i,j}(f+g). \]
            \item The \emph{nontrivial commutator} relations: For every $i < j < k \in \rangeII{0}{n}$
            and $f \in \Fle{\kappa(j-i)}, g \in \Fle{\kappa(k-j)}$,
            \[ [e_{i,j}(f), e_{j,k}(g)] = e_{i,k}(fg). \]
            \item The \emph{trivial commutator} relations: For every $i < j \in \rangeII{0}{n}$ and $k < \ell \in \rangeII{0}{n}$
            with $j \ne k$ and $i \ne \ell$ and $f \in \Fle{\kappa(j-i)}, g \in \Fle{\kappa(\ell-k)}$,
            \[ [e_{i,j}(f), e_{k,\ell}(g)] = 1. \]
        \end{itemize}
    \end{itemize}
    These relations are called the \emph{Steinberg relations} for $\GrUnip{n}{\F}$.
    (See also \Cref{ex:steinberg:linear,ex:steinberg:nontrivial,ex:steinberg:trivial} and \Cref{fig:steinberg}.)
\end{definition}
Note that the nontrivial commutator relations are well-defined
because $\deg(fg) \le \deg(f) + \deg(g) \le \kappa(j-i) + \kappa(k - j) = \kappa(k - i)$.

\begin{remark}
The group $\GrUnip{n}{\F}$ is denoted $\mathrm{Unip}_n(\F[t];\kappa)$ in \cite[\S2.7]{KOW25}.
When we omit $\kappa$, we assume it is $3$; this is the case for which \cite{KOW25} prove trivial cohomology of the global complex,
and therefore the only case where we currently get our downstream application of direct-product testing and PCPs.
Our results on local expansion work for all values of $\kappa$, including the case $\kappa = 1$, which corresponds to the complex originally defined by~\cite{KO21,KO23}.
We do not currently have direct-product testers in the $\kappa = 1$ case because trivial cohomology of the global complex is not (yet) known.
\end{remark}

We now define a collection of subgroups of the group $\GrUnip{n}{\F}$:
\begin{definition}[Staircase subgroups]\label{def:unip:staircase}
    For $n, \kappa \in \N$, $\F$ a field, and $\ell \in \rangeOne{n} = \rangeII{1}{n}$,
    we define the ``staircase subgroup'' $\GrStair{n}{\F}{\ell} \le \GrUnip{n}{\F}$
    as the subgroup generated by the elements
    $e_{i,j}(f)$ ($i < j \in \rangeII{0}{n}$ and $f \in \Fle{\kappa(j-i)}$) where $\rangeEI{i}{j} \not\ni \ell$.
    The staircase subgroups form an $\rangeOne{n}$-indexed subgroup family \[
    (\GrStair{n}{\F}{\ell} \le \GrUnip{n}{\F})_{\ell \in \rangeOne{n}}. \]
    Thus, (as in \Cref{def:indexed subgroup family},) we can also define the intersection subgroups, for every $S \subseteq \rangeOne{n}$, \[
    \GrStair{n}{\F}{S} \coloneqq \bigcap_{\ell \in S} \GrStair{n}{\F}{\ell}. \]
    (By a slight abuse of notation, we also refer to all these subgroups as ``staircase subgroups''.)
\end{definition}

To interpret the condition $\rangeEI{i}{j} \ni \ell$, it is helpful to consider the following trivial fact:
\begin{equation}
    \forall i < j \in \rangeII{0}{n}, \ell \in \rangeOne{n}, \quad
    \rangeEI{i}{j} \ni \ell \iff (j < \ell \vee  i \ge \ell).
\end{equation}
If we view the elements of $\rangeII{0}{n}$ as $n+1$ ``wires'' and the elements of $\rangeOne{n}$ as $n$ interspersed ``moats''
(so that moat $i$ is between wires $i-1$ and $i$),
this condition equivalently captures whether a ``gate'' stretching from wire $i$ to wire $j$ crosses moat $\ell$.
We call $j-i$ the ``height'' of the gate $(i,j)$.

See \Cref{sec:struct} below for intuition for and structure characterizations of the group $\GrUnip{n}{\F}$ and its subgroups $\GrStair{n}{\F}{a}$.

\subsection{The Kaufman--Oppenheim complex}

Having the stair allows us to then define the coset complex:
\begin{definition}\label{def:ko:local complex}
    For $n, \kappa \in \N$, we define the coset complex
    $\LinkCplxA{n}{\F} \coloneqq \CoCo{\GrUnip{n}{\F}}{(\GrStair{n}{\F}{a})_{a \in \rangeOne{n}}}$.
    We call this the \emph{local Kaufman--Oppenheim complex} (on $n+1$ ``wires'').
\end{definition}

For $n, \kappa, s \in \N$ and finite field $\F$, \textcite{KO23} defined a certain coset complex $\CplxA{n}{s}{\F}$, which we call a \emph{global Kaufman--Oppenheim complex}.\footnote{Actually, \cite{KO23} only had $\kappa = 1$; the very minor generalization to $\kappa = 3$ is from \cite{KOW25}.}
This complex is a coset complex with rank $n+1$ / dimension $n$.
In the interest of notational simplicity, we postpone the explicit definition of the global complex $\CplxA{n}{s}{\F}$ to \Cref{sec:kaufman-oppenheim};
in the body of our paper, we use only the following four statements about this complex, all of which are proven explicitly or implicitly in~\cite{KO23}:

\begin{theorem}[Vertex links of global Kaufman--Oppenheim complex]\label{prop:ko:vertex links}
    For every field $\F$, $n, \kappa \in \N$, $s \ge \kappa n$, $1$-word (vertex) $w \in V(\CplxA{n}{s}{\F})$,
    $v \in \supp([i])$,
    the link of $v$ in the global Kaufman--Oppenheim complex is isomorphic to the local Kaufman--Oppenheim complex:
    $\link{\CplxA{n}{s}{\F}}{v} \cong \LinkCplxA{n}{\F}$.
\end{theorem}

\begin{theorem}[Global Kaufman--Oppenheim complex is clique]\label{prop:KO:clique}
    For every field $\F$, $n, \kappa \in \N$, and $s \ge \kappa n$, the global Kaufman--Oppenheim complex $\CplxA{n}{s}{\F}$ is a clique complex.
\end{theorem}

\begin{theorem}[Local bipartite expansion of global Kaufman--Oppenheim complex]\label{prop:ko:local bipartite}
    For every prime $p$, $n, \kappa \in \N$, and $s \ge \kappa n$,
    the global Kaufman--Oppenheim complex $\CplxA{n}{s}{\F_p}$ is a $1/\sqrt{p}$-local bipartite expander.
\end{theorem}

\begin{theorem}[``Strong gallery connectivity'' of global Kaufman--Oppenheim complex]\label{prop:ko:connected}
    For every prime $p$, $n, \kappa \in \N$, and $s \ge \kappa n$,
    and for every $r \le n-1$ and $w \in V_r(\X)$, the link $\link{(\CplxA{n}{s}{\F_p})}{w}$ is connected.
\end{theorem}

We also use the following property, which was proven by~\textcite{KOW25}:

\begin{theorem}[Trivial cohomology of Kaufman--Oppenheim complex, \cite{KOW25}]
    Let $p$ be prime and $\Gamma$ a group which has no nontrivial elements of order $p$.
    (For instance, this is the case if $|\Gamma| < p$.)
    Let $\kappa \coloneqq 3$, $n \in \N$, and $s \ge \kappa n$.
    Then complex $\CplxA{n}{s}{\F_p}$ has trivial $1$-cohomology over $\Gamma$.
\end{theorem}

Finally, \Cref{thm:local bipartite to product} together with \Cref{prop:ko:local bipartite,prop:ko:connected} gives:

\begin{corollary}[Global Kaufman--Oppenheim complex is product]\label{cor:ko:product}
    For every prime $p$, $n, \kappa \in \N$, and $s \ge \kappa n$,
    if $\frac1{\sqrt{p}-n+1} < \frac12$, then $\CplxA{n}{s}{\F}$ is $(\frac1{\sqrt{p}-n+1})$-product.
\end{corollary}

\subsection{The absolute Dehn method}

\begin{definition}[Subgroup words]
    Let $G$ be a group and $\calH = (H_i < G)_{i \in \calI}$ an $\calI$-indexed family of subgroups.
    A \emph{word over $\calH$} is a string of group elements written $\Embed{g_1} \cdots \Embed{g_T}$ such that each $g_t \in \bigcup \calH$, i.e., $g_t \in H_{i_t}$ for some $i_t \in \calI$.
    (The notation $\Embed{\cdot}$ serves to highlight that a subgroup element $g \in \bigcup \calH$
    is formally different from the corresponding length-$1$ subgroup word $\Embed{g}$.)
    The \emph{length} of a word $w = \Embed{g_1} \cdots \Embed{g_T}$ is $|w| \coloneqq T$.
    The inverse of a word $w = \Embed{g_1} \cdots \Embed{g_T}$ over $\calH$ is the word $w^{-1} \coloneqq \Embed{g_T^{-1}} \cdots \Embed{g_1^{-1}}$, also over $\calH$.
    The \emph{concatenation} of two words $w = \Embed{g_1} \cdots \Embed{g_T}$ and $w' = \Embed{g'_1} \cdots \Embed{g'_{T'}}$
    is $w w' \coloneqq \Embed{g_1} \cdots \Embed{g_T} \Embed{g'_1} \cdots \Embed{g'_{T'}}$.
    The empty (length-$0$) word is denoted $\zero$.
\end{definition}

\begin{definition}[Evaluations]\label{def:eval}
    Let $w = \Embed{g_1} \cdots \Embed{g_T}$ be a word over $\calH$.
    The \emph{evaluation} of $w$, denoted $\eval(w) \in G$, is the ordered product $g_1 \cdots g_T \in G$.
    For two words $w, w'$ over $\calH$, we write $w \equiv w'$ iff $\eval(w) = \eval(w')$.
\end{definition}

\begin{definition}
    Let $G$ be a group and $\calH = (H_i < G)_{i \in \calI}$ an $\calI$-indexed subgroup family.
    Define: \[
    \gdiam{G}{\calH} \coloneqq \max_{g \in G} \min \{
        |w| : w\text{ is a word over }\calH\text{ and }\eval(w) = g \}. \qedhere \]
\end{definition}

Note that if $H_1H_2 = G$ then $\gdiam{G}{(H_1,H_2)} = 2$, while $\gdiam{G}{\calH} < \infty$ iff the subgroups $\calH$ generate $G$.
The following proposition justifies the use of the notation $\gdiam{G}{\calH}$:

\begin{proposition}[{\cite{KO21}}]\label{prop:dehn:diam}
    Let $G$ be a group and $\calH = (H_i < G)_{i \in \calI}$ an indexed subgroup family.
    Then \[
    \diam(\CoCo{G}{\calH}) \le \gdiam{G}{\calH}. \]
\end{proposition}

\begin{definition}[Relators]\label{def:relator}
    A word $w$ over $\calH$ is a \emph{relator} if its evaluation is $\Id$ in $G$.    
    (Hence, $w$ is a relator iff $w \equiv \zero$.)
\end{definition}

\begin{definition}[Equations]\label{def:eq}
    We define an equivalence relation~$\sim$ on words over $\calH$ via:
    \begin{enumerate}
        \item $x y \sim y x$,
        \item $x \sim x^{-1}$,
    \end{enumerate}
    If $u \sim v$, then $u$ is a relator iff $v$ is a relator.
    An \emph{equation} over $\calH$ is a string $\Eq{y}{z}$, where $y,z$ are words over $\calH$ and $yz^{-1}$ is a relator; we identify the equation with this relator.
\end{definition}

\begin{definition}[Derivations]\label{def:derivation}
    Let $r, x, y$ be relators over $\bigcup \calH$.
    We say $x$ is \emph{derived from} $y$ \emph{via} $r$ if (for some words $p,q,u,v$ over $\calH$), \[
        x \sim puq, \quad 
        y \sim pvq, \quad 
        \Eq{u}{v} \sim r. \qedhere \]
\end{definition}

\begin{definition}\label{def:in-subgroup rels}
    $\SubRels{G}{\calH}{\ell}$ is the set of relators $\Embed{g_1} \cdots \Embed{g_\ell}$
    such that there exists a single $i \in \calI$ with $g_1,\ldots,g_\ell \in H_i$.
\end{definition}

\begin{definition}[Derivation length]
    Let $w$ be a relator over~$\calH$.
    We write $\area{w}{\calH}$ for the least number of derivation steps, via relators from $\SubRels{G}{\calH}{3}$, that it takes to reduce $w$ to the word~$\Id$.
    (Or, we write $\area{w}{\calH} = \infty$ if this is not possible.)
    For $\calR$ a set of relators over $\calH$, we write \[
    \area{\calR}{\calH} \coloneqq \max_{w \in \calR} \area{w}{\calH} \]
    for the maximum area of any relation in $\calR$.
    Finally, for $T \in \N$, we write \[
    \garea{T}{G}{\calH} \coloneqq \area{\{w : w\text{ a relator and }|w| \le T\}}{\calH}, \]
    for the maximum area of any relator of length $T$ in $G$.
\end{definition}

\begin{remark}
    It is not hard to see that for every $\ell \in \N$, $\area{\SubRels{G}{\calH}{\ell}}{\calH} \le \ell - 3$.
\end{remark}

A theorem relating area and coboundary expansion in coset complexes was originally proven in~\textcite{KO21};
we give here a version with improved parameters due to~\textcite{OS25}.

\begin{theorem}[{\cite[\S C]{OS25}}]\label{prop:dehn:cbdy}
    For every group $G$ and subgroup family $\calH$, for $R_0 \coloneqq \gdiam{G}{\calH}$ and $R_1 \coloneqq \garea{2R_0+1}{G}{\calH}$,
    the coset complex $\CoCo{G}{\calH}$ has $1$-triword expansion at least $12/R_1$ over every group $\Gamma$.
\end{theorem}

\section{Structure of the local Kaufman--Oppenheim complex}\label{sec:struct}

In this section, we prove a number of useful properties of the graded unipotent group, the staircase subgroups, and local Kaufman--Oppenheim complex.
These proofs give good intuition for calculating with the Steinberg relations.
Many of the properties are well-known (or are obvious when working with matrix realizations of the group)
and have appeared (at least for the case $\kappa = 1$) in many previous works.

\subsection{Group theory review: Direct products and retractions}

\begin{theorem}[Von Dyck's theorem, {\cite[p. 346]{Rot95}}]\label{prop:groups:von dyck}
    Let $G = \langle S \mid R\rangle$ be a presented group ($S$ are the generators and $R$ the relations) and $H$ another group.
    Suppose $\phi : S \to H$ is a mapping such that the image under $\phi$ of every relation in $R$ is trivial.
    Then $\phi$ extends to a (unique) group homomorphism $\phi : G \to H$.
\end{theorem}

\begin{definition}
    Let $G$ be a finite group and $\calH = (H_i \le G)_{i \in \calI}$ an $\calI$-indexed subgroup family ($\calI$ finite).
    We say $\calH$ \emph{forms an internal direct product} in $G$ if:
    \begin{enumerate}
        \item For every $i \ne i' \in \calI$, the subgroups $H_i$ and $H_{i'}$ intersect trivially, i.e., $H_i \cap H_{i'} = \{\Id\}$.
        \item For every $i \ne i' \in \calI$, the subgroups $H_i$ and $H_{i'}$ commute, i.e., $H_i H_{i'} = H_{i'} H_i$.
    \end{enumerate}
    In this case, we use $\prod \calH$ to denote the subgroup of $G$ generated by the subgroups in $\calH$.
    We also say $\prod \calH$ \emph{splits as an internal direct product} of the family $\calH$.
\end{definition}

\begin{fact}
    When $G$ is a finite group, $\calH$ an $\calI$-indexed subgroup family, and $G$ splits as an internal direct product $G = \prod \calH$:
    \begin{enumerate}
        \item Every element of $G$ can be written uniquely as a product $\prod_{i \in \calI} h_i$ where $h_i \in H_i$ (order does not matter),
        \item There are homomorphisms $\pi_i : G \twoheadrightarrow H_i$ for every $i \in \calI$
        satisfying $\pi_i(h) = h$ for every $h \in H_i$, $\pi_i(\pi_i(g)) = \pi_i(g)$ for every $g \in G$, and
        $H_{i'} \in \ker(\pi_i)$ for every $i\ne i' \in \calI$.
    \end{enumerate}
\end{fact}

We now define the group-theoretic notion of a ``retraction'', which generalizes the homomorphisms $\pi_i$ in the preceding fact:

\begin{definition}[e.g., {\cite[Lemma 7.20]{Rot95}}]
    Let $G$ be a group and $L \le G$ a subgroup.
    A \emph{retraction} onto $L$ is a homomorphism $\pi : G \twoheadrightarrow L$ such that
    for every $\ell \in L$, $\pi(\ell) = \ell$.
\end{definition}

Recall that by the first isomorphism theorem, if $\pi : G \twoheadrightarrow L$ is a surjective homomorphism,
then $\ker(\pi) \trianglelefteq G$ is a normal subgroup and $L \cong G/\ker(\pi)$ is isomorphic to its quotient group.
A retraction onto $L$ is a special kind of a surjective map onto $L$,
wherein this quotient group is itself isomorphic to a subgroup of $G$.

When a retraction onto $L$ exists, $G$ is called a ``semidirect product'', written $\ker(\pi) \rtimes L$.
The subgroups $\ker(\pi)$ and $L$ need not commute (in which case $G$ is not a direct product of $\ker(\pi)$ and $L$),
but a semidirect product still exhibits product-like properties:

\begin{proposition}\label{prop:groups:retraction}
    Let $G$ be a group, $L \le G$ a subgroup, and $\pi : G \twoheadrightarrow L$ a retraction onto $L$.
    \begin{enumerate}
    \item The subgroups $\ker \pi, L \le G$ intersect trivially.
    \item $\pi$ is a projection, i.e., for every $g \in G$, $\pi(\pi(g)) = \pi(g)$.
    \item Every element $g\in G$ can be written uniquely as $g = k\ell$ for $k \in \ker \pi$ and $\ell \in L$ (and as $g = \ell' k'$ for $\ell' \in L$ and $k' \in \ker \pi$).
    \item Suppose $\ker (\pi) \le H$.
    Then $\pi(H) = H \cap L$.
    \item Suppose $M \le \ker(\pi)$.
    Then $\ker (\pi) \cap LM = \ker(\pi) \cap ML = M$.
    \end{enumerate}
\end{proposition}

\begin{proof}
    We have:

    \begin{enumerate}
    \item Every $g \in \ker \pi$ has $\pi(g) = \Id$ and every $g \in L$ has $\pi(g) = g$.

    \item Follows from the definition of retraction and that $\pi(g) \in L$.

    \item For existence, set $\ell \coloneqq \pi(g)$ and $k \coloneqq g \ell^{-1}$,
    so that $\pi(k\ell) = \pi(g) (\pi(\ell))^{-1} = \pi(g) (\pi(\pi(g)))^{-1} = \pi(g) (\pi(g))^{-1} = \Id$.
    For uniqueness, note that if $g = k\ell$ then $\pi(g) = \pi(k) \pi(\ell) = \pi(\ell) = \ell$.
    Hence $\ell$ is unique, and so so is $k$.
    The proof for the other case is symmetric.
    
    \item Let $h \in H$.
    On one hand, if $h = k\ell$ with $k \in \ker(\pi)$ and $\ell \in L$,
    then $\pi(h) = \ell$.
    But $k^{-1} h = \ell \in H$ since $\ker(\pi) \le H$.
    On the other hand, if $h \in H \cap L$, then by the retraction property, $h = \pi(h) \in \pi(H)$.

    \item Obviously, $M \le LM$ and $M \le \ker(\pi)$.
    Conversely, suppose that $h \in \ker(\pi) \cap LM$.
    Then on one hand, since $h \in LM$, we have $h = \ell m$ with $\ell \in L$ and $m \in M$,
    and so $\pi(h) = \pi(\ell) = \ell$ (since $\pi$ is a retraction onto $L$ and $m \in M \le \ker(\pi)$).
    At the same time, since $h \in \ker(\pi)$, we have $\pi(h) = \Id$.
    Hence $\ell = 1$ and so $h = m \in M$.
    The proof for the other case is symmetric. \qedhere
    \end{enumerate}
\end{proof}

\subsection{``Reducing'' a coset complex}

The definitions in the preceding section yield the following corollary:

\begin{proposition}\label{cor:coset complex:reduce}
    Let $G$ be a group and $\calH = (H_i \le G)_{i \in \calI}$ an $\calI$-indexed subgroup family.
    Suppose $W, \calJ \sqsubseteq \calI$.
    Further, suppose $\pi : G \twoheadrightarrow L$ is a retraction onto $L$
    and $H_W = LM$ with $M \le \ker(\pi) \cap H_{\calJ}$.
    Then $\CoCo{H_W}{(H_{W \cup \{a\}})_{a \in \calJ}} \cong \CoCo{L}{(L \cap H_a)_{a \in \calJ}}$
    (and the index map $\iota : \calJ \to \calJ$ is the identity function).
\end{proposition}

\begin{proof}
    We apply \Cref{prop:coset complex:pass to quotient} with $G' \coloneqq H_W$, $\pi' \coloneqq \pi|_{H_W}$,
    and $\calH' \coloneqq (H_{W\cup\{a\}})_{a \in \calJ}$.
    Note that $\tilde{\pi} : H_W \twoheadrightarrow L$ is still a retraction because $L \le H_W$.
    It remains to check that $\ker(\pi') \le H_{W \cup \{a\}}$ for every $a \in \calJ$.
    We have $\ker(\pi') = \ker(\pi) \cap H_W \le H_W$,
    so we need only verify that $\ker(\pi) \cap H_W \le H_a$.
    By the final item of \Cref{prop:groups:retraction}, we have $\ker(\pi) \cap H_W = M \le H_\calJ \le H_a$.
\end{proof}

\subsection{Triangular subgroups of the graded unipotent group}

We now define another family of subgroups of $\GrUnip{n}{\F}$ which will be of recurring interest to us in the paper.
We say a set $\calI \subseteq \N$ is an \emph{interval} if $\calI = \rangeII{\min(\calI)}{\max(\calI)}$.
Note that in this case, $|\calI| = \max(\calI) - \min(\calI) + 1$.

\begin{definition}[Triangular subgroups]\label{def:unip:triangle}
    For $n, \kappa \in \N$, $\F$ a field, and an interval $\calI \subseteq \rangeII{0}{n}$,
    we define the ``triangular subgroup'' $\GrTri{n}{\F}{\calI} \le \GrUnip{n}{\F}$
    as the subgroup generated by the elements $e_{i,j}(f)$, where $i < j \in \calI$ and $f \in \Fle{\kappa(j-i)}$.
    Note that by definition, $\GrTriII{n}{\F}{0}{n} = \GrUnip{n}{\F}$,
    while $\GrTriII{n}{\F}{j-1}{j} = \{\Id\}$ for every $j \in \rangeOne{n}$.
\end{definition}

We now prove several useful properties of these subgroups:

\begin{proposition}\label{prop:unip:triangle contain}
    Let $n, \kappa \in \N$, $\F$ be a field, and $\calI_1 \subseteq \calI_2 \subseteq \rangeII{0}{n}$ be intervals.
    We have the containment of subgroups: \[
    \GrTri{n}{\F}{\calI_1} \le \GrTri{n}{\F}{\calI_2}. \]
\end{proposition}

\begin{proof}
    Every generator of the left-hand group is also in the right-hand group.
\end{proof}

\begin{proposition}\label{prop:unip:iso}
    Let $n, \kappa \in \N$, $\F$ be a field, and $\calI \subseteq \rangeII{0}{n}$ be an interval.
    Consider the groups $\GrTri{n}{\F}{\calI}$ and $\GrUnip{|\calI|-1}{\F}$;
    write the generators of the former as $e_{i,j}(f)$
    ($i < j \in \calI$, $f \in \Fle{\kappa(j-i)}$)
    and of the latter as $\tilde{e}_{i,j}(f)$
    ($i < j \in \rangeII{0}{|\calI|-1}$, $f \in \Fle{\kappa(j-i)}$).
    Then there is an isomorphism:
    \begin{alignat*}{4}
    \isom{n}{\F}{\calI} \;:\;
    &&\GrTri{n}{\F}{\calI}
    &&\;\stackrel{\sim}{\longrightarrow}\;
    &\GrUnip{|\calI|-1}{\F}, \\
    && e_{i,j}(f)
    &&\;\longmapsto\;
    &\tilde{e}_{\,i-\min(\calI),\,j-\min(\calI)}(f).
    \end{alignat*}
    (Note that $\tilde{e}_{i-\min(\calI),j-\min(\calI)}(f)$ is a valid generator in the codomain $\GrUnip{|\calI|-1}{\F}$
    because $i \ge \min(\calI) \implies i-\min(\calI) \ge 0$,
    $j \le \max(\calI) \implies j - \min(\calI) \le \max(\calI) - \min(\calI)$,
    and $(j-\min(\calI))-(i-\min(\calI)) = j-i$.)
    Further, for every interval $\calI' \subseteq \calI$, \[
    \isom{n}{\F}{\calI} (\GrTri{n}{\F}{\calI'})
    = \GrTriII{|\calI|-1}{\F}{\min(\calI')-\min(\calI)}{\max(\calI')-\min(\calI)}. \]
\end{proposition}

\begin{proof}
    It is easy to check that this map, defined on the generators, is a homomorphism
    by applying von Dyck's theorem (\Cref{prop:groups:von dyck}) and the definition of the unipotent group (\Cref{def:graded unipotent group}).
    One can similarly verify that there is a homomorphism $\GrUnip{|\calI|-1}{\F} \to \GrTri{n}{\F}{\calI}
    : \tilde{e}_{i,j}(f) \mapsto e_{\min(\calI)+i,\min(\calI)+j}(f)$,
    and finally that this homomorphism is the inverse of $\isom{n}{\F}{\calI}$.
    The ``further'' statement follows since $\isom{n}{\F}{\calI}$, restricted to $\GrTri{n}{\F}{\calI'}$, agrees with $\isom{n}{\F}{\calI'}$.
\end{proof}

\begin{proposition}\label{prop:unip:disjoint triangles commute}
    Let $n, \kappa \in \N$, $\F$ be a field, and intervals $\calI_1, \calI_2 \subseteq \rangeII{0}{n}$ satisfy $\calI_1 \cap \calI_2 = \emptyset$.
    The subgroups $\GrTri{n}{\F}{\calI_1}$ and $\GrTri{n}{\F}{\calI_2}$ of $\GrUnip{n}{\F}$ commute.
\end{proposition}

\begin{proof}
    $\GrTri{n}{\F}{\calI_1}$ is generated by elements $e_{i,j}(f)$ with $i < j \in \calI_1$,
    and $\GrTri{n}{\F}{\calI_2}$ by elements $e_{i,j}(f)$ with $i < j \in \calI_2$.
    The disjointness assumption and the trivial commutator relations imply that
    all generators of the former subgroup commute with all generators of the latter subgroup.
\end{proof}

In order to prove richer properties of these subgroups, we now define a retraction of the ``big triangle'' group $\GrUnip{n}{\F} = \GrTriII{n}{\F}{0}{n}$
onto ``small triangle'' subgroups $\GrTri{n}{\F}{\calI}$.
Together with the isomorphism of $\GrTri{n}{\F}{\calI}$ and $\GrUnip{|\calI|-1}{\F}$ (\Cref{prop:unip:iso}),
the retraction will be useful to us several times in the sequel to analyze the ``self-reducibility'' of the unipotent group construction.

\begin{proposition}\label{prop:unip:retract}
    Let $n, \kappa \in \N$, $\calI \subseteq \rangeII{0}{n}$ be an interval, and $\F$ be a field.
    Consider the groups $\GrUnip{n}{\F} \ge \GrTri{n}{\F}{\calI}$;
    write the generators of the former as $e_{i,j}(f)$
    ($i < j \in \rangeII{0}{n}$, $f \in \Fle{\kappa(j-i)}$)
    and of the latter as $\tilde{e}_{i,j}(f)$
    ($i < j \in \calI$, $f \in \Fle{\kappa(j-i)}$).
    Then there is a retraction:
    \begin{alignat*}{4}
    \retract{n}{\F}{\calI} \;:\;
    &&\GrUnip{n}{\F} 
    &&\;\relbar\joinrel\twoheadrightarrow\;
    &\GrTri{n}{\F}{\calI}, \\
    && e_{i,j}(f)
    &&\;\longmapsto\;
    & \begin{cases}
        \tilde{e}_{i,j}(f) & \text{if }i < j \in \calI, \\ \Id & \text{otherwise}. \end{cases}
    \end{alignat*}
    (Thus, for every $i < j \in \rangeII{0}{n}$ with $\{i, j\} \not\subset \calI$,
    $e_{i,j}(f) \in \ker(\retract{n}{\F}{\calI})$
    for every $f \in \Fle{\kappa(j-i)}$.)
\end{proposition}

\begin{proof}
    This map, if it is a valid homomorphism, is obviously a retraction.
    To check that is is a homomorphism, we again apply von Dyck's theorem (\Cref{prop:groups:von dyck}),
    and verify that the image of every relation defining $\GrUnip{n}{\F}$ vanishes.
    This follows from some simple casework and the fact that commutators with trivial elements are trivial:
    \begin{itemize}
        \item For linearity relations $e_{i,j}(f) \cdot e_{i,j}(g) = e_{i,j}(f+g)$
        ($i < j \in \rangeII{0}{n}$, $f, g \in \F[X]$ with $\deg(f) , \deg(g) \le \kappa(j-i)$),
        either $i < j \in \calI$, in which case the image is the (true) linearity relation
        $\tilde{e}_{i,j}(f) \cdot \tilde{e}_{i,j}(g) = \tilde{e}_{i,j}(f+g)$,
        or the image is the trivial relation $\Id \cdot \Id = \Id$.
        \item For nontrivial commutator relations $\Comm{e_{i,j}(f)}{e_{j,k}(g)} = e_{i,k}(fg)$
        ($i < j < k \in \rangeII{0}{n}$, $f \in \Fle{\kappa(j-i)}$, $g \in \Fle{\kappa(k-j)}$),
        either $i, k \in \calI$, in which case the image is the (true) nontrivial commutator relation
        $\Comm{\tilde{e}_{i,j}(f)}{\tilde{e}_{j,k}(g)} = \tilde{e}_{i,k}(fg)$,
        or (at least) one of $e_{i,j}(f)$ and $e_{j,k}(g)$ has trivial image \emph{and}
        $e_{i,k}(fg)$ has trivial image.
        \item Finally, for trivial commutator relations $\Comm{e_{i,j}(f)}{e_{k,\ell}(g)} = \Id$
        ($i < j \in \rangeII{0}{n}$, $k < \ell \in \rangeII{0}{n}$,
        $j \ne k$, $i \ne \ell$, $f \in \Fle{\kappa(j-i)}$, $g \in \Fle{\kappa(\ell-k)}$),
        either $i, j, k, \ell \in \calI$,
        in which case the image is the (true) trivial commutator relation
        $\Comm{\tilde{e}_{i,j}(f)}{\tilde{e}_{k,\ell}(g)} = \Id$,
        or (at least) one of $e_{i,j}(f)$ and $e_{k,\ell}(g)$ has trivial image. \qedhere
    \end{itemize}
\end{proof}

This retraction has the following useful properties:

\begin{proposition}\label{prop:unip:retract triangle}
    Let $n, \kappa \in \N$, $\F$ be a field, and $\calI_1, \calI_2 \subseteq \rangeII{0}{n}$ be intervals.
    Then \[
        \retract{n}{\F}{\calI_1}(\GrTri{n}{\F}{\calI_2}) = \GrTri{n}{\F}{\calI_1 \cap \calI_2}. \]
\end{proposition}
\begin{proof}
    We claim that $\GrTri{n}{\F}{\calI_2}$ is generated by two subgroups, $\GrTri{n}{\F}{\calI_1\cap\calI_2}$ and another subgroup $L \le \ker(\retract{n}{\F}{\calI_1})$.
    This suffices because $\retract{n}{\F}{\calI_1}$ preserves $\GrTri{n}{\F}{\calI_1 \cap \calI_2}$ and kills $L$.

    Indeed, the generators of $\GrTri{n}{\F}{\calI_2}$ are of the form $e_{i,j}(f)$ where $i < j \in \calI_2$ and $f \in \Fle{\kappa(j-i)}$.
    We set $L$ to be the subgroup generated by the generators $e_{i,j}(f)$ where $\{i,j\} \not\subseteq \calI_1$.
    Such generators are killed by the retraction $\retract{n}{\F}{\calI_1}$,
    and the remaining generators of $\GrTri{n}{\F}{\calI_1\cap\calI_2}$ are those where $\{i,j\}\subseteq \calI_1$
    and are therefore contained in $\GrTri{n}{\F}{\calI_1\cap\calI_2}$.
\end{proof}

\begin{proposition}\label{prop:unip:overlapping triangles}
    Let $n, \kappa \in \N$, $\F$ be a field, and $\calI_1,\calI_2 \subseteq \rangeII{0}{n}$ be intervals.
    Then: \[
    \GrTri{n}{\F}{\calI_1} \cap \GrTri{n}{\F}{\calI_2} = \GrTri{n}{\F}{\calI_1\cap\calI_2}. \]
    In particular, if $|\calI_1 \cap \calI_2| \le 1$, then $\GrTri{n}{\F}{\calI_1}$ and $\GrTri{n}{\F}{\calI_2}$ intersect trivially.
\end{proposition}

\begin{proof}
    The $\ge$ direction follows immediately from \Cref{prop:unip:triangle contain}.
    For the $\le$ direction, consider any element $g \in \GrTri{n}{\F}{\calI_1} \cap \GrTri{n}{\F}{\calI_2}$.
    Since $g \in \GrTri{n}{\F}{\calI_1}$, $g$ is fixed by the retraction $\retract{n}{\F}{\calI_1}$,
    while by \Cref{prop:unip:retract triangle}, the same retraction maps $g$ into $\GrTri{n}{\F}{\calI_1 \cap \calI_2}$.
\end{proof}

Note that \Cref{prop:unip:disjoint triangles commute,prop:unip:overlapping triangles} together imply that
for any family of pairwise disjoint intervals $(\calI_t \subseteq \rangeII{0}{n})_{t \in \calT}$,
the triangle subgroups $\GrTri{n}{\F}{\calI_t}$ form an internal direct product within $\GrUnip{n}{\F}$.

\subsection{Staircase subgroups of the unipotent group}

We next turn to describing the staircase subgroups (\Cref{def:unip:staircase}) in terms of the triangle subgroups (\Cref{def:unip:triangle}).

\begin{proposition}\label{prop:unip:one step stair}
    In the setup of the preceding definition, $\GrStair{n}{\F}{\ell}$ splits as an internal direct product: \[
    \GrStair{n}{\F}{\ell} = \GrTriIE{n}{\F}{0}{\ell} \times \GrTriII{n}{\F}{\ell}{n}. \]
\end{proposition}

\begin{proof}
    The subgroups $\GrTriIE{n}{\F}{0}{\ell}$ and $\GrTriII{n}{\F}{\ell}{n}$ commute by \Cref{prop:unip:disjoint triangles commute}.
    Their intersection is trivial by \Cref{prop:unip:overlapping triangles}.
    They are contained in and generate $\GrStair{n}{\F}{\ell}$ since this group is generated by elements $e_{i,j}(f)$ with either $j < \ell$ or $i \ge \ell$;
    generators of the former kind are present in and generate the subgroup $\GrTriIE{n}{\F}{0}{\ell}$,
    while generators of the latter kind are present in and generate the subgroup $\GrTriII{n}{\F}{\ell}{n}$.
\end{proof}

Now, we prove the following simple fact in \Cref{sec:coset complex}:

\begin{proposition}\label{prop:group:intersection of products}
    Let $G$ be a group and $H, K \le G$ subgroups,
    with $H = H_1 \times H_2$ and $K = K_1 \times K_2$ splitting as internal direct products.
    Further, suppose that $K_1 \le H_1$ and $H_2 \le K_2$.
    Then $H \cap K$ also splits as an internal direct product: \[
    H \cap K = K_1 \times (H_1 \cap K_2) \times H_2. \]
\end{proposition}

Applying this inductively together with \Cref{prop:unip:one step stair,prop:unip:disjoint triangles commute} gives the following, which we also prove in \Cref{sec:coset complex}:

\begin{proposition}\label{prop:unip:stair form}
    For every $n, \kappa \in \N$ and field $\F$, and $W \subseteq \rangeOne{n}$, the following holds.
    Suppose $W = \{\ell_1,\ldots,\ell_T\}$ with $\ell_1 < \cdots < \ell_T$,
    and define $\ell_0 \coloneqq 0$ and $\ell_{T+1} \coloneqq n+1$.
    Define the family of intervals $(\calI_t \coloneqq \rangeIE{\ell_{t-1}}{\ell_t} \subseteq \rangeII{0}{n})_{t \in \rangeOne{T+1}}$.
    Then $\GrStair{n}{\F}{W}$ splits as an internal direct product: \[
    \GrStair{n}{\F}{W} = \prod_{t \in \rangeOne{T+1}} \GrTri{n}{\F}{\calI_t}. \]
    (Recall that $\GrTri{n}{\F}{\calI}$ is trivial when $|\calI| = 1$, so these intervals can be ignored.)
\end{proposition}

This yields several nice corollaries:

\begin{corollary}
    For every $n, \kappa \in \N$, field $\F$, and $W \subseteq \rangeOne{n}$,
    $\GrStair{n}{\F}{W}$ is the subgroup of $\GrUnip{n}{\F}$ generated by elements of the form $e_{i,j}(f)$
    for $i < j \in \rangeII{0}{n}$ satisfying $\rangeEI{i}{j} \cap W = \emptyset$
    and $f \in \Fle{\kappa(j-i)}$.
\end{corollary}

\begin{proof}
    Use the prior proposition and the definition of triangle subgroups by generators (\Cref{def:unip:triangle}).
\end{proof}

\begin{corollary}
    For every $n, \kappa \in \N$, field $\F$, and $W, W_1, W_2 \subseteq \rangeOne{n}$, the following holds.
    Suppose that $W_1$ and $W_2$ satisfy that $W_1 \supseteq W$, $W_2 \supseteq W$, and
    for every $i \in W_1 \setminus W$ and $j \in W_2 \setminus W$, $\rangeEI{\min\{i,j\}}{\max\{i,j\}} \cap W \ne \emptyset$.
    Then: \[
    \GrStair{n}{\F}{W} \subseteq \GrStair{n}{\F}{W_1} \cdot \GrStair{n}{\F}{W_2}. \]
\end{corollary}

\begin{corollary}\label{cor:unip:tri as stair}
    For every $n, \kappa \in \N$, field $\F$, and interval $\calI \subseteq \rangeII{0}{n}$, we have: \[
    \GrTri{n}{\F}{\calI} = \GrStair{n}{\F}{\rangeII{1}{\min(\calI)} \cup \rangeEI{\max(\calI)}{n}}. \]
    (Alternatively, $\rangeII{1}{\min(\calI)} \cup \rangeEI{\max(\calI)}{n} = \rangeOne{n} \setminus (\calI \setminus \{\min(\calI)\})$.)
\end{corollary}

\begin{corollary}\label{cor:unip:max stair empty}
    For every $n, \kappa \in \N$ and field $\F$, the staircase subgroup $\GrStair{n}{\F}{\rangeOne{n}} = \bigcap_{\ell=1}^n \GrStair{n}{\F}{\ell}$ is trivial.
\end{corollary}

\subsection{Some self-reducibility results}

The prior calculations let us easily calculate the retraction of a staircase subgroup:

\begin{proposition}\label{prop:unip:retract stair}
    For every $n, \kappa \in \N$, field $\F$, interval $\calI \subseteq \rangeII{0}{n}$, $a \in \calI$, we have:
    \[
    \retract{n}{\F}{\calI}(\GrStair{n}{\F}{a})
    = \GrStair{n}{\F}{a} \cap \GrTri{n}{\F}{\calI}
    = \GrTriIE{n}{\F}{\min(\calI)}{a} \times \GrTriII{n}{\F}{a}{\max(\calI)}. \]
\end{proposition}

\begin{proof}
Abbreviate $\pi \coloneqq\retract{n}{\F}{\calI}$.
Using \Cref{prop:unip:one step stair,prop:unip:retract triangle}:
\begin{multline*}
\GrTriIE{n}{\F}{\min(\calI)}{a} \times \GrTriII{n}{\F}{a}{\max(\calI)}
= \pi(\GrTriIE{n}{\F}{0}{a}) \times \pi(\GrTriII{n}{\F}{a}{n}) \\
= \pi(\GrTriIE{n}{\F}{0}{a} \times \GrTriII{n}{\F}{a}{n})
= \pi(\GrStair{n}{\F}{a}).
\end{multline*}
At the same time, using \Cref{cor:unip:tri as stair,prop:unip:stair form}:
\begin{multline*}
\GrStair{n}{\F}{a} \cap \GrTri{n}{\F}{\calI}
= \GrStair{n}{\F}{a} \cap \GrStair{n}{\F}{\rangeII{1}{\min(\calI)} \cup \rangeEI{\max(\calI)}{n}} \\
= \GrStair{n}{\F}{\rangeII{1}{\min(\calI)} \cup \{a\} \cup \rangeEI{\max(\calI)}{n}}
= \GrTriIE{n}{\F}{\min(\calI)}{a}) \times \GrTriII{n}{\F}{a}{\max(\calI)},
\end{multline*}
as desired.
\end{proof}

\begin{corollary}\label{cor:unip:isom stair}
    For every $n, \kappa \in \N$, field $\F$, interval $\calI \subseteq \rangeII{0}{n}$, and $a \in \calI$, we have: \[
    \isom{n}{\F}{\calI} \parens*{ \retract{n}{\F}{\calI}(\GrStair{n}{\F}{a}) }
    = \GrStair{|\calI|-1}{\F}{a-\min(\calI)}. \]
\end{corollary}

\begin{proof}
    Abbreviate $\pi \coloneqq\retract{n}{\F}{\calI}$ and $\Phi \coloneqq \isom{n}{\F}{\calI}$.
    We calculate using \Cref{prop:unip:one step stair,prop:unip:iso}, and the prior proposition:
    \begin{multline*}
    \GrStair{|\calI|-1}{\F}{a-\min(\calI)}
    = \GrTriIE{|\calI|-1}{\F}{0}{a-\min(\calI)} \times \GrTriII{|\calI|-1}{\F}{a-\min(\calI)}{|\calI|-1} \\
    = \Phi(\GrTriIE{n}{\F}{\min(\calI)}{a})
    \times \Phi(\GrTriII{n}{\F}{a}{\max(\calI)}) \\
    = \Phi(\GrTriIE{n}{\F}{\min(\calI)}{a} \times \GrTriII{n}{\F}{a}{\max(\calI)})
    = \Phi\parens*{ \pi(\GrStair{n}{\F}{a})},
    \end{multline*}
    as desired.
\end{proof}

\begin{proposition}\label{prop:unip:blank}
    Let $W \subseteq \rangeOne{n}$.
    Let $\ell^{(1)} < \ell^{(2)} \in W \cup \{0,n+1\}$.
    Let $W_1 \coloneqq W \cup \rangeEE{\ell}{\ell'}$
    and $W_2 \coloneqq W \cup \rangeIE{1}{\ell} \cup \rangeEI{\ell'}{n}$.
    Then $\GrStair{n}{\F}{W}$ splits as an internal direct product: \[
    \GrStair{n}{\F}{W} = \GrStair{n}{\F}{W_1} \times \GrStair{n}{\F}{W_2}. \]
\end{proposition}

\begin{proof}
    As in \Cref{prop:unip:stair form}, let $W = \{\ell_1,\ldots,\ell_T\}$ with $\ell_1 < \cdots < \ell_T$ and $\ell_0 \coloneqq 0$ and $\ell_{T+1} \coloneqq n+1$.
    Suppose $\ell^{(1)} = \ell_{t^{(1)}}$ and $\ell^{(2)} = \ell_{t^{(2)}}$, for some $t^{(1)} < t^{(2)} \in \rangeII{0}{T+1}$.
    By \Cref{prop:unip:stair form}, \[
    \GrStair{n}{\F}{W} = \prod_{t \in \rangeOne{T+1}} \calI_t =  \prod_{t\in\rangeIE{t^{(1)}}{t^{(2)}}} \calI_t \times \prod_{t\in\rangeOne{T+1}\setminus \rangeIE{t^{(1)}}{t^{(2)}}} \calI_t \]
    for the intervals $(\calI_t \coloneqq \rangeIE{\ell_{t-1}}{\ell_t} \subseteq \rangeII{0}{n})_{t \in \rangeOne{T+1}}$.
    We claim that the first factor equals $\GrStair{n}{\F}{W_2}$ and the second $\GrStair{n}{\F}{W_1}$.
    Indeed, note that:
    \begin{align*}
    W &= \{\ell_t : t \in \rangeIE{1}{T+1}\}, \\
    W_2 &= \rangeII{1}{\ell_{t^{(1)}}-1} \cup \{\ell_t : t \in \rangeIE{t^{(1)}}{t^{(2)}}\} \cup \rangeIE{\ell_{t^{(2)}}}{n+1}, \\
    W_1 &= \{\ell_t : t \in \rangeIE{1}{t^{(1)}-1}\} \cup \rangeIE{\ell_{t^{(1)}}}{\ell_{t^{(2)}}} \cup \{\ell_t : t \in \rangeIE{t^{(2)}}{T+1}\}.
    \end{align*}
    Thus, the nonempty intervals in the decomposition of $\GrStair{n}{\F}{W_2}$ are $(\calI_t)_{t \in \rangeIE{t^{(1)}}{t^{(2)}}}$,
    and for $\GrStair{n}{\F}{W_1}$ are $(\calI_t)_{t\in\rangeOne{T+1}\setminus \rangeIE{t^{(1)}}{t^{(2)}}}$.
    Applying \Cref{prop:unip:stair form} again finishes the proof.
\end{proof}

We also have the following consequence, which we will use to establish productization:

\begin{proposition}\label{prop:unip:group product}
    Let $\kappa, n \in \N$, field $\F$, and $W \subseteq \rangeOne{n}$.
    If $S_1 \prec S_2 \sqsubset \rangeOne{n} \setminus W$ satisfy that $W \cap \rangeEE{\max S_1}{\min S_2} \ne \emptyset$,
    then $\GrStair{n}{\F}{W}$ equals a product of subgroups: \[
    \GrStair{n}{\F}{W} = \GrStair{n}{\F}{S_1 \cup W} \cdot \GrStair{n}{\F}{S_2 \cup W}. \]
\end{proposition}

\begin{proof}
    Pick $\ell \in W \cap \rangeEE{\max S_1}{\min S_2}$.
    Note that $S_1 \subseteq \rangeIE{1}{\ell}$ and $S_2 \subseteq \rangeEI{\ell}{n}$.
    
    By the prior \Cref{prop:unip:blank}, $\GrStair{n}{\F}{W}$ splits as an internal direct product $\GrStair{n}{\F}{W \cup \rangeEE{\ell}{n+1}} \times \GrStair{n}{\F}{W \cup \rangeIE{1}{\ell} \cup \rangeEI{n+1}{n}}$.
    The first factor is contained in $\GrStair{n}{\F}{W \cup S_2}$ since $W \cup S_2 \subseteq W \cup \rangeEI{\ell}{n}$
    and the second factor is in $\GrStair{n}{\F}{W \cup S_1}$ since $W \cup S_1 \subseteq W \cup \rangeIE{1}{\ell}$.
\end{proof}

\Cref{cor:coset complex:reduce} has the following useful corollary which characterizes induced complexes in links of the local Kaufman--Oppenheim complex.
Roughly, it says that the induced complex on a set of wires between two moats is itself a (smaller) local Kaufman--Oppenheim complex.

\begin{proposition}[``Self-reducibility'' of the local Kaufman--Oppenheim complex]\label{cor:ko:self reducible}
    For every $n,\kappa \in \N$ and field $\F$, and every $W \subseteq \rangeOne{n}$, the following holds.
    Suppose $\ell_1 < \ell_2 \in W \cup \{0,n+1\}$,
    and $W \cap \rangeEE{\ell_1}{\ell_2} = \emptyset$ and $w$ is a $W$-word.
    Then the induced complex on $\rangeEE{\ell_1}{\ell_2}$ is isomorphic to a smaller Kaufman--Oppenheim complex: \[
    \resLink{(\LinkCplxA{n}{\F})}{\rangeEE{\ell_1}{\ell_2}}{w} \cong \LinkCplxA{\ell_2-\ell_1 - 1}{\F}. \]
    (The bijection on indices is $\iota : \rangeEE{\ell_1}{\ell_2} \to \rangeOne{\ell_2 - \ell_1 - 1} : \ell \mapsto \ell - \ell_1$.)
\end{proposition}

\begin{proof}
    We set $G \coloneqq \GrUnip{n}{\F}$, $\calJ \coloneqq \rangeEE{\ell_1}{\ell_2}$,
    $H_a \coloneqq \GrStair{n}{\F}{a}$,
    $L \coloneqq \GrTriIE{n}{\F}{\ell_1}{\ell_2}$,
    and $M \coloneqq \GrStair{n}{\F}{W \cup \rangeEE{\ell_1}{\ell_2}}$.
    By \Cref{prop:unip:retract}, $\pi \coloneqq \retract{n}{\F}{{\rangeIE{\ell_1}{\ell_2}}}$ is a retraction $G \twoheadrightarrow L$.
    By \Cref{prop:groups:retraction,prop:unip:retract stair}, $H \cap L_a = \pi(L_a).$
    
    Now, we apply \Cref{cor:coset complex:reduce}.
    If we can check that (1) $H_W = LM$ and (2) $M \le \ker(\pi) \cap H_\calJ$,
    we then have that $\resLink{(\LinkCplxA{n}{\F})}{\rangeEE{\ell_1}{\ell_2}}{w} \cong \CoCo{L}{H_a \cap L}$.
    Then, we can apply the isomorphism $\isom{n}{\F}{\rangeIE{\ell_1}{\ell_2}}$,
    which carries $L$ to $\GrUnip{\ell_2-\ell_1-1}{\F}$ and, by \Cref{cor:unip:isom stair}, $L \cap H_a$ to $\GrUnip{a-\ell_1}{\F}$.

    First, note that by \Cref{cor:unip:tri as stair}, $L = \GrTriIE{n}{\F}{\ell_1}{\ell_2}
    = \GrStair{n}{\F}{\rangeII{1}{\ell_1} \cup \rangeII{\ell_2}{n}}
    = \GrStair{n}{\F}{W \cup \rangeIE{1}{\ell_1} \cup \rangeEI{\ell_2}{n}}$
    (where the last equality used that $\ell_1, \ell_2 \in W$ and $\rangeEE{\ell_1}{\ell_2} \cap W = \emptyset$).
    Hence, (1) follows immediately from the prior \Cref{prop:unip:blank}.
    
    For (2), we have $H_\calJ = \GrStair{n}{\F}{\rangeEE{\ell_1}{\ell_2}}$.
    Hence $M \le H_\calJ$ since $\rangeEE{\ell_1}{\ell_2} \subseteq W \cup \rangeEE{\ell_1}{\ell_2}$.
    On the other hand,
    \begin{multline*}
    \pi(M) \le \bigcap_{a \in \rangeEE{\ell_1}{\ell_2}} \pi(\GrStair{n}{\F}{a})
    = \bigcap_{a \in \rangeEE{\ell_1}{\ell_2}} (\GrStair{n}{\F}{a} \cap \GrTriEE{n}{\F}{\ell_1}{\ell_2}) \\
    = \GrStair{n}{\F}{\rangeEE{\ell_1}{\ell_2}} \cap \GrTriEE{n}{\F}{\ell_1}{\ell_2}
    = \GrStair{n}{\F}{\rangeEE{\ell_1}{\ell_2}} \cap L
    = \GrStair{n}{\F}{\rangeOne{n}}
    = \{\Id\},
    \end{multline*}
    and so $M \in \ker(\pi)$, as desired.
\end{proof}

\begin{proposition}[Productization in the local Kaufman--Oppenheim complex]\label{prop:ko:pol}
    For every $n,\kappa \in \N$ and field $\F$, the local Kaufman--Oppenheim complex $\LinkCplxA{n}{\F}$ productizes on the line (in the sense of \Cref{def:pol}).
\end{proposition}

\begin{proof}
    We want to show that for every $W \subseteq \rangeOne{n}$ and $S_1 \prec S_2 \sqsubset \rangeOne{n} \setminus W$,
    if $W \cap \rangeEE{\max S_1}{\min S_2} \ne \emptyset$, then for every $W$-word $w$,
    $\ptiz{\link{\LinkCplxA{n}{\F}}{w}}{S_1}{S_2}$.
    By definition, $\LinkCplxA{n}{\F} = \CoCo{\GrUnip{n}{\F}}{(\GrStair{n}{\F}{\ell})_{\ell \in \rangeOne{n}})}$.
    By \Cref{prop:coset complex:productization}, we therefore need to check whether $\GrStair{n}{\F}{W} = \GrStair{n}{\F}{W \cup S_1} \cdot \GrStair{n}{\F}{W \cup S_2}$.
    But this is precisely \Cref{prop:unip:group product}.
\end{proof}

\subsection{Representing \texorpdfstring{$\GrUnip{n}{\F}$}{the unipotent group}}

The next important structural result lets us, in a certain sense, write unique expressions for every element in $\GrUnip{n}{\F}$:

\begin{proposition}
    Let $\Phi \coloneqq \{(i,j) : i < j \in \rangeII{0}{n}\}$ (so that $|\Phi| = \binom{n+1}{2}$).
    For every ordering $\pi : \rangeOne{\binom{n+1}2} \to \Phi$ and $g \in \GrUnip{n}{\F}$,
    there exist unique polynomials $(f_{i,j} \in \Fle{\kappa(j-i)})_{(i,j) \in \Phi}$
    such that \[
    g = \prod_{p=1}^{\binom{n+1}{2}} e_{i,j}(f_{i,j}). \]
\end{proposition}

By the direct product decomposition in \Cref{prop:unip:stair form}, we immediately get:

\begin{corollary}[Essentially {\cite[Lemma~3.14]{OP22}}]\label{prop:expressing group elements}
    Let $W \subseteq \rangeII{1}{n}$ and $\Phi_W \coloneqq \{(i, j) : i < j \in \rangeII{0}{n}, W \cap \rangeEI{i}{j} = \emptyset\}$.
    For every group element $x \in \GrStair{n}{\F}{S}$ and bijection $\pi : [|\Phi_W|] \to \Phi_W$,
    there exist unique polynomials $(f_{i,j})_{(i,j) \in \Phi_W}$ with $\deg(f_{i,j}) \le \kappa(j-i)$ such that \[
    x = \prod_{p=1}^{|\Phi_W|} e_{\pi(p)}(f_{\pi(p)}). \]
\end{corollary}

The group $\GrUnip{n}{\F}$ has a natural realization as a matrix group.
Consider the set of matrices with entries in $\F[X]$:
\begin{equation}\label{eq:matrices}
S^\kappa_n(\F) \coloneqq \braces*{ M \in (\F[X])^{\rangeII{0}{n} \times \rangeII{0}{n}} :
\begin{aligned}
    &\forall j < i \in \rangeII{0}{n}, && M_{i, j} = 0 \\
    &\forall i \in \rangeII{0}{n}, && M_{i,i} = 1 \\
    &\forall i < j \in \rangeII{0}{n}, && \deg(M_{i,j}) \le \kappa(j - i)
\end{aligned}}
\end{equation}
(note that these are $n+1 \times n+1$ matrices).

\begin{proposition}\label{prop:graded unipotent:representation}
    Let $n \in \N$ and $\F$ be a field.
    The set of matrices $S_n(\F)$ in \Cref{eq:matrices} forms a group
    under matrix multiplication, and moreover, there is an isomorphism
    $\phi : \GrUnip{n}{\F} \to S^\kappa_n(\F)$
    defined on the generators by \[
    \phi(e_{i,j}(f)) \coloneqq \Id + f \cdot 1_{i,j}. \]
    Hence, there is a faithful representation of $\GrUnip{n}{\F}$ on the vector space $(\F[X])^{\rangeII{0}{n}}$, defined on generators as \[
    (v_0,\ldots,v_n) \xmapsto{e_{i,j}(f)}(v_0,\ldots,v_{j-1},v_j + fv_i,v_{j+1},\ldots,v_n). \]
\end{proposition}

This gives some useful alternative perspectives on the group $\GrUnip{n}{\F}$:
\begin{itemize}
\item Given a ``state vector'' $v = (v_0,\ldots,v_n)$,
$e_{i,j}(f)$ can be viewed as an ``instruction'' of the form ``$\instAdd{fv_i}{v_j}$'',
and a product of generators can be viewed as a ``program'' consisting of a sequence of instructions.
\item Visually, we can view the coordinates $\rangeII{0}{n}$ as wires storing the current state,
the generators $e_{i,j}(f)$ as ``gates'' adding a higher wire (closer to $0$) to a lower wire (closer to $1$) with a certain (polynomial) weight.
\end{itemize}
Each program/circuit \emph{evaluates} to a linear operator in $\GrUnip{n}{\F}$ (the evaluation map is many-to-one).

In this view, for $W \subseteq \rangeOne{n}$, elements in the staircase subgroup $\GrStair{n}{\F}{W}$ are precisely
the evaluations of circuits where, when we place a ``moat'' between wire $i-1$ and wire $i$ for each $i \in W$, no gate crosses any moat.

\subsection{Matrix Steinberg relations}

\newcommand{\TriMat}[2]{\mathrm{TriMat}^{\kappa}_{#1}(#2)}
\newcommand{\RectMat}[3]{\mathrm{RectMat}^{\kappa}_{#1}(#2 \times #3)}
\newcommand{\TriMatIE}[3]{\TriMat{#1}{\rangeEI{#2}{#3}}}
\newcommand{\RectMatIE}[5]{\RectMat{#1}{\rangeEI{#2}{#3}}{\rangeEI{#4}{#5}}}

\begin{definition}[``Triangular'' graded matrices]
Let $n, \kappa \in \N$, $\F$ be a field, and $\calI \subseteq \rangeII{0}{n}$ be an interval.
We define the set of matrices \[
\TriMat{\F}{\calI} \coloneqq \braces*{ A \in \F[X]^{\calI \times \calI} :\ 
\begin{aligned}
\forall i < j \in \calI,&\ \deg(A_{i,j}) \le \kappa(j-i), \\
\forall i \ge j \in \calI,&\ A_{i, j} = 0.
\end{aligned}
}. \qedhere \]
\end{definition}

Note that each matrix $A \in \TriMat{\F}{\calI}$ satisfies $\deg(A_{i,j}) \le \kappa(j-i)$ for every $i, j \in \calI$
under our standard convention $\deg(0) = -\infty$.

\begin{definition}[``Rectangular'' graded matrices]
Let $n, \kappa \in \N$, $\F$ be a field, and $\calI \prec \calJ \subset \rangeII{0}{n}$ be two intervals.
We define the set of matrices \[
\RectMat{\F}{\calI}{\calJ} \coloneqq \{ F \in \F[X]^{\calI \times \calJ} :
\forall i \in \calI, j \in \calJ,\ \deg(F_{i,j}) \le \kappa(j-i) \}. \qedhere \]
\end{definition}

The following proposition expresses closure conditions for these sets under matrix addition and multiplication:

\begin{fact}\label{fact:matrix:closed}
Let $n, \kappa \in \N$ and $\F$ be a field.
\begin{enumerate}
\item Let $\calI \subseteq \rangeII{0}{n}$ be an interval.
Then for every $A, A' \in \TriMat{\F}{\calI}$, we have $FF', F+F' \in \TriMat{\F}{\calI}$.
\item Let $\calI \prec \calJ \subset \rangeII{0}{n}$ be intervals.
Then for every $A \in \TriMat{\F}{\calI}$ and $F \in \RectMat{\F}{\calI}{\calJ}$,
we have $AF \in \RectMat{\F}{\calI}{\calJ}$.
Similarly, for every $B \in \TriMat{\F}{\calJ}$ and $F \in \RectMat{\F}{\calI}{\calJ}$,
we have $FB \in \RectMat{\F}{\calI}{\calJ}$.
\item Let $\calI \prec \calJ \prec \calK \subset \rangeII{0}{n}$.
Then for every $F \in \RectMat{\F}{\calI}{\calJ}$ and $G \in \RectMat{\F}{\calJ}{\calK}$,
we have $FG \in \RectMat{\F}{\calI}{\calK}$.
\end{enumerate}
\end{fact}

\begin{proof}
We prove each item separately, and only prove one representative inclusion for each item.
\begin{enumerate}
\item We expand $(AA')_{i,j} = \sum_{k=\min(\calI)}^{\max(\calI)} A_{i,k} A'_{k,j}$.
Note that if $i \ge j$, then for every $k$, either $i \ge k$ (so $A_{i,k} = 0$)
or $k \ge i$ (so $A'_{k,j} = 0$).
Otherwise, $\deg((AA)_{i,j}) \le \max_{k=\min(\calI)}^{\max(\calI)} (\deg(A_{i,k}) + \deg(A'_{k,j}))
	\le \kappa((k - i) + (j - k)) = \kappa(j-i)$, as desired.
Similarly, we expand $(A + A')_{i,j} = A_{i,j} + A'_{i,j}$, so if $i \ge j$ all terms vanish 
and otherwise $\deg((A+A')_{i,j}) \le \max\{\deg(A_{i,\bar{a'}}), \deg(A'_{i,\bar{a'}})\} \le \kappa(j-i)$.
\item We prove only the first, as the second is similar.
We expand $(AF)_{i,k} = \sum_{k=\min(\calJ)}^{\max(\calJ)} A_{i,j} F_{j,k}$,
so $\deg((AF)_{i,k}) \le \max_{k=\min(\calJ)}^{\max(\calJ)} (\deg(A_{i,j}) + \deg(F_{j,k}))
\le \kappa((j-i) + (k-j)) = \kappa(k - i)$.
\item Similar to the previous item. \qedhere
\end{enumerate}
\end{proof}

We also have the following fact regarding inverses of strictly upper-triangular matrices with added $1$'s on the diagonal:

\begin{fact}\label{fact:matrix:inv}
For every $n, \kappa \in \N$, field $\F$, interval $\calI \subseteq \rangeII{0}{n}$, and $A \in \TriMat{\F}{\calI}$,
there exists a (unique) $A^\dagger \in \diamA$ such that $AA^\dagger + A + A^\dagger = 0$ (equiv., $(A+\Id) (A^\dagger+\Id) = \Id$).
\end{fact}

\begin{definition}[Group elements corresponding to matrices]\label{def:matrix:elements}
    Let $n, \kappa \in \N$ and $\F$ be a field.
    For an interval $\calI \subseteq \rangeII{0}{n}$ and $A \in \TriMat{\F}{\calI}$, we define the group element: \[
    \El{A} \coloneqq \prod_{j=\max(\calI)}^{\min(\calI)} \parens*{
    \prod_{i=\max(\calI)}^{j-1} e_{i,j}(A_{i,j}) } \in \GrTri{n}{\F}{\calI}. \]
    For intervals $\calI \prec \calJ \subset \rangeII{0}{n}$ and $F \in \RectMat{\F}{\calI}{\calJ}$, we define the group element: \[
    \El{F} \coloneqq \prod_{i=\min(\calI)}^{\min(\calJ)} \prod_{j=\min(\calJ)}^{\max(\calJ)} e_{i,j}(F_{i,j}) \in \GrTriII{n}{\F}{\min(\calI)}{\max(\calJ)}. \]
    We always think of the matrices $A$ and $F$ as ``carrying'' the appropriate type information (triangle vs. rectangle) and the index sets $\calI$ (and $\calJ$ in the rectangular case),
    and therefore use the same notation $\El{\cdot}$ to unambiguously denote group elements corresponding to arbitrary matrices.
\end{definition}
Note that by the trivial Steinberg relations, all elements in the definition of $\El{F}$ commute, and therefore their order is not important.
(However, the ordering of elements in $\El{A}$ is important.)

We have the following matrix version of the Steinberg relations, applied to the group elements defined in the prior definition:

\begin{proposition}[Matrix Steinberg relations]\label{prop:matrix:steinberg}
Let $n, \kappa \in \N$; we have the following:
\begin{enumerate}
    \item \emph{Homogeneous triangle relations:} For every interval $\calI \subseteq \rangeII{0}{n}$, \[
    \forall A, A' \in \TriMat{\F}{\calI},\quad \El{A} \cdot \El{A'}
    = \El{AA' + A + A'}. \]
    \item \emph{Homogeneous rectangle relations:} For all intervals $\calI \prec \calJ \subset \rangeII{0}{n}$, \[
    \forall F, F' \in \RectMat{\F}{\calI}{\calJ},\quad \El{F} \cdot \El{F'}
    = \El{F + F'}. \]
    \item \emph{Heterogeneous triangle-triangle relations:} For all intervals $\calI \prec \calJ \subset \rangeII{0}{n}$ \[
    \forall A \in \TriMat{\F}{\calI}, B \in \TriMat{\F}{\calJ},\quad
    \Comm{\El{A}}{\El{B}} = \Id. \]
    \item \emph{Nontrivial heterogeneous rectangle-rectangle relations:} For all intervals $\calI \prec \calJ \prec \calK \subset \rangeII{0}{n}$, \[
    \forall F \in \RectMat{\F}{\calI}{\calJ}, G \in \RectMat{\F}{\calJ}{\calK},\quad
    \Comm{\El{F}}{\El{G}} = \El{FG}. \]
    \item \emph{Trivial heterogeneous rectangle-rectangle relations:} For all intervals $\calI \prec \calJ \subset \rangeII{0}{n}$ and $\calK \prec \calL \subset \rangeII{0}{n}$
    such that $\calI \cap \calL = \calJ \cap \calK = \emptyset$, \[
    \forall F \in \RectMat{\F}{\calI}{\calJ}, H \in \RectMat{\F}{\calK}{\calL},\quad
    \Comm{\El{F}}{\El{H}} = \Id. \]
    \item \emph{Nontrivial triangle-rectangle relations:} For all intervals $\calI \prec \calJ \subset \rangeII{0}{n}$,
    \begin{align*}
    \forall A \in \TriMat{\F}{\calI}, F \in \RectMat{\F}{\calI}{\calJ},&\quad
    \Comm{\El{A}}{\El{F}} = \El{AF}, \\
    \forall B \in \TriMat{\F}{\calJ}, F \in \RectMat{\F}{\calI}{\calJ},&\quad
    \Comm{\El{B}}{\El{F}} = \El{FB^\dagger}.
    \end{align*}
    \item \emph{Trivial triangle-rectangle relations}: For all intervals $\calI \prec \calJ \subset \rangeII{0}{n}$, $\calK \subset \rangeII{0}{n}$
    with $\calK \cap (\calI \cup \calJ) = \emptyset$, \[
    \forall C \in \TriMat{\F}{\calK}, F \in \RectMat{\F}{\calI}{\calJ},\quad
    \Comm{\El{C}}{\El{F}} = \Id. \]
\end{enumerate}
\end{proposition}

The proof is just linear algebra and we omit it.

\section{Proof of \Cref{thm:main-coboundary}}\label{sec:proof}

We now turn to establishing \Cref{thm:main-coboundary}, which is the main technical result of our paper.
We restate it in the following formal form:

\begin{theorem}
    For every $r \in \N$, there exists $n_0 = O(r^2) \in \N$ such that
    for every $n \ge n_0 \in \N$, there exists $p_0 = O(n^2 r^2) \in \N$
    such that for every prime $p \ge p_0 \in \N$ and $\kappa \in \N$,
    the local Kaufman--Oppenheim complex $\LinkCplxA{n}{\F_p}$ has $r$-triword expansion at least $2^{O(-r^{0.99})}$ over every group $\Gamma$.
\end{theorem}

Our proof combines the induction method of~\cite{BLM24} (\Cref{thm:linear condition}) with the following key lemmas:

\begin{lemma}[Separated diameter bound]\label{lemma:ko:diam}
    For every $\xi > 0$, there exist $n_0, C_0 \in \N$ such that the following holds.
    For every field $\F$, $\kappa \ge 1 \in \N$, $n \ge n_0 \in \N$, $W \subseteq \rangeOne{n}$ and $\ell_1 < \ell_2 \in \rangeOne{n}$ 
    such that $W \cap \rangeII{\ell_1}{\ell_2} = \emptyset$ and $\ell_2 - \ell_1 \ge \xi n$,
    for every $W$-word $w$,
    the complex $\resLink{(\LinkCplxA{n}{\F})}{\{\ell_1\},\{\ell_2\}}{w}$ has diameter at most $C_0$.
\end{lemma}

\begin{lemma}[Separated coboundary expansion bound]\label{lemma:ko:cbdy}
    For every $\xi > 0$, there exist $n_0, C_0 \in \N$ such that the following holds.
    For every field $\F$, $\kappa \ge 1 \in \N$, $n \ge n_0 \in \N$, $W \subseteq \rangeOne{n}$ and $\ell_1 < \ell_2 < \ell_3 \in \rangeOne{n}$ 
    such that $W \cap \rangeII{\ell_1}{\ell_3} = \emptyset$ and $\ell_3 - \ell_2, \ell_2 - \ell_1 \ge \xi n$,
    for every $W$-word $w$,
    the complex $\resLink{(\LinkCplxA{n}{\F})}{\{\ell_1\},\{\ell_2\},\{\ell_3\}}{w}$ has diameter at most $C_0$.
\end{lemma}
\Cref{lemma:ko:cbdy} is the formal version of \Cref{thm:main-separated}.

These lemmas will both follow from our key group-theoretic theorem:

\unip*

The idea behind the proofs of \Cref{lemma:ko:diam,lemma:ko:cbdy} is
is to combine \Cref{prop:dehn:diam,prop:dehn:cbdy} with \Cref{thm:main-unip}.
However, \Cref{thm:main-unip} is only stated for a local Kaufman--Oppenheim complex,
while $\link{(\LinkCplxA{n}{\F})}{w}$ is not a local Kaufman--Oppenheim complex (merely a link of one).
Thus, we first have to apply \Cref{cor:ko:self reducible} to pass to a (possibly) smaller local Kaufman--Oppenheim complex.

\begin{proof}[Proof of \Cref{lemma:ko:diam}]
    Let $n'_0, C_0$ denote the constants guaranteed by \Cref{thm:main-unip} for our value of $\xi$.
    Set $n_0 \coloneqq n'_0/\xi$.
    Let $i \coloneqq \max ((W \cap \rangeIE{1}{\ell_1}) \cup \{0\})$
    and $j \coloneqq \min ((W \cap \rangeEI{\ell_2}{n}) \cup \{n+1\})$.
    By construction, $i < j \in W \cup \{0,n+1\}$, $i < \ell_1$, and $j > \ell_2$.
    Further, \begin{equation}\label{eq:W disj}
    W \cap \rangeEE{i}{j} = \emptyset.
    \end{equation}
    (Indeed, suppose that $a \in W \cap \rangeEE{i}{j}$.
    Then $a \not\in \rangeII{\ell_1}{\ell_2}$ by assumption,
    hence either $a \in \rangeEE{i}{\ell_1} \cup \rangeEE{\ell_2}{j}$.
    We would then have $a \in W \cap \rangeIE{1}{\ell_1}$ (in which case $i \ge a$) or $a \in W \cap \rangeEI{\ell_2}{n}$ (in which case $j \le a$).)

    By definition of restrictions, \[
    \resLink{(\LinkCplxA{n}{\F})}{\{\ell_1\},\{\ell_2\}}{w}
    = \res{\big( \resLink{(\LinkCplxA{n}{\F})}{\rangeEE{i}{j}}{w} \big)}{\{\ell_1\},\{\ell_2\}} \]
    (since $\ell_1, \ell_2 \in \rangeEE{i}{j}$).
    Hence by \Cref{eq:W disj,cor:ko:self reducible}, \[
    \resLink{(\LinkCplxA{n}{\F})}{\{\ell_1\},\{\ell_2\}}{w} \cong \res{(\LinkCplxA{n'}{\F})}{\{\ell'_1\},\{\ell'_2\}} \]
    where $n' \coloneqq j-i-1$, $\ell'_2 \coloneqq \ell_2 - i$ and $\ell'_1 \coloneqq \ell_1 - i$.
    Now $\ell'_2 - \ell'_1 = (\ell_2 - i) - (\ell_1 - i) = \ell_2 - \ell_1 \ge \xi n$
    and $n' = j - i - 1 \ge \ell_2 - \ell_1 \ge \xi n \ge \xi n_0 \ge n'_0$.
    So \Cref{eq:diam} gives that \[
    \diam(\resLink{(\LinkCplxA{n'}{\F})}{\{\ell'_1\},\{\ell'_2\}}{w}) \le C_0. \]
    Finally, we apply \Cref{prop:dehn:diam}.
\end{proof}

Similarly:

\begin{proof}[Proof of \Cref{lemma:ko:cbdy}]
    Follows from \Cref{prop:dehn:cbdy,eq:cbdy} just as \Cref{lemma:ko:diam} did from \Cref{prop:dehn:diam,eq:diam}.
\end{proof}

Given these lemmas, we can now prove:

\begin{proof}[Proof of \Cref{thm:main-coboundary}]
    We pick the $n_0$ needed in \Cref{thm:main-unip} and $\epsilon_0$ needed in \Cref{thm:linear condition}.
    Given $r$ and $n$, we then set $p_0$ large enough such that $\frac1{\sqrt{p}-n+1} < \epsilon_0 r^{-2}$.
    Then, we apply \Cref{thm:linear condition}, and simply need to verify its hypotheses.
    \Cref{item:lin:spectral} follows from \Cref{cor:ko:product}.
    \Cref{item:lin:pol} (productization on the line) follows from \Cref{prop:ko:pol}.
    \Cref{item:lin:symmetry} (strong symmetry) follows from  \Cref{prop:coset complex:strongly symmetric}.
    \Cref{item:lin:diam} (diameter bound on separated singletons) is \Cref{lemma:ko:diam}
    and similarly \Cref{item:lin:cbdy} (coboundary bound on separated singletons) is \Cref{lemma:ko:cbdy}.
\end{proof}

\section{Miscellaneous technical tools}\label{sec:setup}

Before proving \Cref{thm:main-unip}, we develop a few more useful tools.

\subsection{Diameter and area under homomorphisms}

We develop some simple machinery to transfer bounds on diameter and area
from one indexed subgroup family to another using surjective group homomorphisms.

\begin{proposition}[Padding homomorphism]\label{lemma:pad}
    Let $n^{(1)}, n, n^{(2)}, \kappa \in \N$ and $\F$ be any field.
    Then there exists a surjective homomorphism \[
    \psi : \GrUnip{n^{(1)}+n+n^{(2)}}{\F} \twoheadrightarrow \GrUnip{n}{\F} \]
    such that for every $\ell \in \rangeOne{n}$, $\psi(\GrStair{n}{\F}{n^{(1)}+a}) = \GrStair{n}{\F}{a}$.
\end{proposition}

\begin{proof}
    We set $\psi \coloneqq \isomII{n^{(1)}+n+n^{(2)}}{\F}{n^{(1)}}{n^{(1)}+n} \circ \retractII{n^{(1)}+n+n^{(2)}}{\F}{n^{(1)}}{n^{(1)}+n}$
    (recalling \Cref{prop:unip:retract,prop:unip:iso}).
    By \Cref{cor:unip:isom stair}, $\psi(\GrStair{n}{\F}{n^{(1)}+a}) = \GrStair{n}{\F}{|\calI|-1}{a-\min(\calI)}$
    where $\calI \coloneqq \rangeII{n^{(1)}}{n^{(1)}+n}$.
    Hence, $|\calI| = n + 1$ and $\min(\calI) = n^{(1)}$.
\end{proof}

\begin{proposition}\label{prop:homo}
    Let $G, G'$ be groups and $\calH = (H_i < G)_{i \in \calI}, \calH' = (H'_i < G')_{i \in \calI}$ be $\calI$-indexed subgroup families.
    Let $\psi : G' \to G$ be a group homomorphism such that $\psi(H'_i) \le H_i$ for every $i \in \calI$.
    For every word $w' = \Embed{g'_1} \cdots \Embed{g'_T}$ over $\bigcup \calH'$,
    define the word \[ \Psi(w') \coloneqq \Embed{\psi(g'_1)} \cdots \Embed{\psi(g'_T)} \] over $\calH$.
    Then:
    \begin{enumerate}
        \item $\Psi(w'_1 w'_2) = \Psi(w'_1) \Psi(w'_2)$ and $\Psi((w')^{-1}) = (\Psi(w'))^{-1}$.\label{item:homo:homo}
        \item For every word $w'$ over $\calH'$, $|\Psi(w')| = |w'|$.\label{item:homo:length}
        \item For every word $w'$ over $\calH'$, $\eval(\Psi(w')) = \psi(\eval(w'))$. Hence, if $w'$ is a relator, then so is $\Psi(w')$.\label{item:homo:eval}
        \item For all words $x',y'$ over $\calH'$, if $x' \sim y'$, then $\Psi(x') \sim \Psi(y')$.\label{item:homo:sim}
        \item For every common-subgroup relator $r' \in \SubRels{3}{G'}{\calH'}$, $\Psi(r')$ is also a common-subgroup relator ($\Psi(r') \in \SubRels{3}{G}{\calH}$).\label{item:homo:subgroup}
        \item For all words $x',y',r'$ over $\calH'$, if $y'$ is derived from $x'$ via $r'$,
        then $\Psi(y')$ is derived from $\Psi(x')$ via $\Psi(r')$.\label{item:homo:derive}
    \end{enumerate}
\end{proposition}

\begin{proof}
    \Cref{item:homo:homo,item:homo:length,item:homo:sim} are immediate by definition.
    \Cref{item:homo:eval} is because $\eval(\Psi(w')) = \eval(\Embed{\psi(g'_1)} \cdots \Embed{\psi(g'_T)}) = \psi(g'_1) \cdots \psi(g'_T) = \psi(g'_1 \cdots g'_T) = \psi(\eval(w'))$
    (where the third equality uses that $\psi$ is a group homomorphism).
    \Cref{item:homo:subgroup} follows immediately from \Cref{item:homo:eval} and the definition.
    For \Cref{item:homo:derive}, if $y'$ is derived from $x'$ via $r'$, then there exist words $p',q',u',v'$ over $\calH'$
    such that $x' \sim p'u'q'$, $y' \sim p'v'q'$, and $r' \sim \Eq{u'}{v'}$.
    Hence $\Psi(x') \sim \Psi(p') \Psi(u') \Psi(q')$, $\Psi(y') \sim \Psi(p') \Psi(v') \Psi(q')$, and $\Psi(r') \sim \Eq{\Psi(u')}{\Psi(v')}$ by \Cref{item:homo:homo,item:homo:sim}.
    Since $\Psi(r')$ is a common-subgroup relator by \Cref{item:homo:subgroup}, we conclude that $\Psi(y')$ is derived from $\Psi(x')$ via $\Psi(r')$, as desired.
\end{proof}

\begin{proposition}\label{prop:diam:homo}
    Let $G, G'$ be groups and $\calH = (H_i < G)_{i \in \calI}, \calH' = (H'_i < G')_{i \in \calI}$ be $\calI$-indexed subgroup families.
    Let $\psi : G' \twoheadrightarrow G$ be a surjective group homomorphism mapping each $H'_i$ onto $H_i$, i.e., $\psi(H'_i) \le H_i$.
    Then \[
    \gdiam{G}{\calH} \le \gdiam{G'}{\calH'}. \]
\end{proposition}

\begin{proof}
    Let $C'_0 \coloneqq \gdiam{G'}{\calH'}$.
    Consider any $x \in G$.
    By surjectivity of $\psi$, there exists $x' \in G'$ with $\psi(x') = x$.
    By assumption, there is some subgroup word $w'$ over $\bigcup \calH'$ with $\eval(w') = x'$ and $|w'| \le C'_0$.
    As in \Cref{prop:homo}, the word $\Psi(w')$ over $\calH$ satisfies $\eval(\Psi(w')) = \psi(\eval(w')) = \psi(w') = w$ and $|\Psi(w')| = |w'| \le C'_0$.
\end{proof}

\begin{proposition}\label{prop:cbdy:homo}
    Let $G, G'$ be groups and $\calH = (H_i < G)_{i \in \calI}, \calH' = (H'_i < G')_{i \in \calI}$ be $\calI$-indexed subgroup families.
    Let $\psi : G' \twoheadrightarrow G$ be a surjective group homomorphism mapping each $H'_i$ \emph{surjectively} onto $H_i$, i.e., $\psi(H'_i) = H_i$.
    Then for every $C_0 \in \N$, \[
    \garea{C_0}{G}{\calH} \le \garea{C_0}{G'}{\calH'}. \]
\end{proposition}

\begin{proof}
    Consider any relator $w$ over $\calH$ of length $|w| \le C_0$;
    we want to show that $\area{w}{\calH} \le \garea{C_0}{G'}{\calH'}$.
    
    Hence, $w = \Embed{g_1} \cdots \Embed{g_T}$, each $g_t \in \bigcup \calH$, and $\eval(w) = g_1 \cdots g_T = \Id \in G$.
    By surjectivity, for every $t \in [T]$, there exists $g'_t \in \bigcup \calH'$ with $\psi(g'_t) = g_t$.
    We can therefore define a word $w' \coloneqq \Embed{g'_1} \cdots \Embed{g'_T}$ over $\calH'$.
    Note that (as in \Cref{prop:homo}), $\Psi(w') = w$ by definition, and $w'$ is a relator with $|w'| = |w| \le C_0$.

    We now claim that $\area{w}{\calH} \le R \coloneqq \area{w'}{\calH'}$;
    this is sufficient because $R \le \garea{C_0}{G'}{\calH'}$ since $|w'| = |w| \le C_0$.
    Indeed, by definition of area, there are words $(w')^{(0)} = w', \ldots, (w')^{(R)} = \zero$ over $\calH'$
    and relators $(r')^{(1)},\ldots,(r')^{(R)} \in \SubRels{3}{G'}{\calH'}$
    such that $(w')^{(i)}$ is derived from $(w')^{(i-1)}$ via $(r')^{(i)}$ for each $i \in [R]$.
    By \Cref{prop:homo}, the words $w^{(i)} \coloneqq \Psi((w')^{(i)})$ and $r^{(i)} \coloneqq \Psi((r')^{(i)})$ satisfy
    $w^{(0)} = w$, $w^{(R)} = 0$, $r^{(i)} \in \SubRels{3}{G}{\calH}$, and each $w^{(i)}$ is derived from $w^{(i-1)}$ via $r^{(i)}$ for each $i \in [R]$.
    Hence, $w$ can be derived from $\SubRels{3}{G}{\calH}$ in at most $R$ steps, as desired.
\end{proof}

\subsection{Basic matrices}

We next define a useful special set of matrices satisfying certain degree-bound conditions, which we call \emph{basic matrices}.
In the subsequent sections, we will use these basic matrices to decompose operators on large sets of wires into a bounded number of ``basic block'' operators.

\begin{definition}[Basic matrix]\label{def:diam:basic}
    Let $\F$ be a field, $X$ an indeterminate, $m, s, \kappa \in \N$.
    We define the set of \emph{$s$-bounded basic matrices}: \[
    \basic{s} \coloneqq \braces*{ M \in \F[X]^{\rangeIE{0}{m}\times\rangeIE{0}{m}} :
    	\forall i, j \in \rangeIE{0}{m}, \quad \deg(M_{i,j}) \le \kappa (s + j - i) }. \]
    (Here, we recall the convention that $\deg(0) = -\infty$.
    Hence if $s + j - i < 0$, the constraint is that $M_{i,j} = 0$.)
    In particular, every $M \in \basic{s}$ satisfies $\deg(M_{0,0}), \deg(M_{m-1,m-1}) \le \kappa s$;
    $\deg(M_{m-1,0}) \le \kappa (s + m-1)$, and $\deg(M_{0,m-1}) \le \kappa (s - m + 1)$.
\end{definition}

Let $I_m \in \F[X]^{\rangeIE{0}{m} \times \rangeIE{0}{m}}$ denote the identity matrix.
\begin{proposition}[Examples of basic matrices]\label{fact:basic:example}
    Let $\kappa, m \in \N$ and $\F$ be any field.
    \begin{enumerate}
        \item For every $s \in \N$ and $t \in \rangeII{0}{\kappa s}$, $X^t I_m \in \basic{s}$.
        (In particular, $I_m \in \basic{0}$.)
        \item For every $s \in \N$, if $M \in \F[X]^{\rangeIE{0}{m} \times \rangeIE{0}{m}}$
        satisfies $\deg_{i,j}(M) \le \kappa(s - (m-1))$ for every $i, j \in \rangeIE{0}{m}$,
        then $M \in \basic{s}$.\label{item:basic:uniform}
        \item For every $s \in \N$, if $M, N \in \basic{s}$,
        then $M + N \in \basic{s}$.
        \item For every $s, s', s \in \N$, if $M \in \basic{s}$, $N \in \basic{s'}$,
        then $MN \in \basic{s+s'}$.
    \end{enumerate}
\end{proposition}

\begin{proof}
    We calculate:
    \begin{enumerate}
        \item The only nonzero entries of $X^t I_m$ are the diagonal entries.
        These have degree $\deg(M_{i,i}) = t \le \kappa(s + i - i) = \kappa s$ as long as $t \le \kappa s$.
        \item Note that $s - (m-1) \le s + j - i$ since $j - i \ge 0 - (m - 1)$.
        \item For every $i,j \in \rangeIE{0}{m}$,
        $\deg((M+N)_{i,j}) \le \max\{\deg(M_{i,j}),\deg(N_{i,j})\} \le \kappa(s+j-i)$.
        \item For every $i, j \in \rangeIE{0}{m}$,
        \[ \deg((MN)_{i,j}) = \max_{k=0}^{m-1} (\deg(M_{i,k}) + \deg(M_{k,j})) \le \kappa \max_{k=0}^{m-1} (s + (k-i) + s' + (j-k)) = \kappa(s + s' + j - i). \]
    \end{enumerate}
\end{proof}

We now give a proposition, which we will use several times in the following sections,
allowing us to decompose a large-bounded basic matrix into a sum of small-bounded basic matrices.

\begin{proposition}[Decomposition into basic matrices]\label{prop:basic:decomp}
    Let $\F$ be a field, $X$ an indeterminate, $m, s, \Delta \in \N$, $s \ge m$.
    Let $M \in \basic{\Delta+2s}$ be a $(\Delta+2s)$-bounded basic matrix,
    and let $M_{i,j} = \sum_{t=0}^{\kappa(\Delta + 2s + j-i)} X^t \mu_{i,j}^{(t)}$ be the degree decomposition of $M_{i,j}$.
    Let $\tau \coloneqq \lfloor \Delta/s \rfloor$.
    Define the matrices $\Mbody^{(k)}$ ($k < \tau$), $\Mgap$,
    and $\Mtail$ in $\F[X]^{\rangeIE{0}{m} \times \rangeIE{0}{m}}$ via:
    \[
    (\Mbody^{(k)})_{i,j} \coloneqq \sum_{t=0}^{\kappa s - 1} X^t \mu^{(t+\kappa \cdot ks)}_{i,j},
    (\Mgap)_{i,j} \coloneqq \sum_{t=0}^{\kappa(\Delta-\tau s) - 1} X^t \mu^{(t+ \kappa \cdot \tau s)}_{i,j}, \\
    (\Mtail)_{i,j} \coloneqq \sum_{t=0}^{\kappa (2s + j - i)} X^t \mu^{(t+\kappa \cdot \Delta)}_{i,j}.
    \]
    Then
    \[
    M = \sum_{k=0}^{\tau - 1} X^{\kappa \cdot ks} \Mbody^{(k)} + X^{\kappa \cdot \tau s} \Mgap + X^{\kappa \Delta} \Mtail
    \]
    and $\Mbody^{(k)}, \Mgap, \Mtail \in \basic{2s}$.
\end{proposition}

\begin{proof}
    For the first equality, note that the decomposition partitions the terms in $M_{i,j}$ by degree:
    $X^{\kappa \cdot ks} (\Mbody^{(k)})_{i,j}$ takes the terms with degrees in $\rangeIE{\kappa \cdot ks}{\kappa \cdot (k+1) s - 1}$;
    $X^{\kappa \cdot \tau s} (\Mgap)_{i,j}$ takes the terms with degrees in $\rangeIE{\kappa \cdot \tau s}{\kappa \Delta}$;
    and
    $X^{\kappa \Delta} (\Mtail)_{i,j}$ takes the terms with degrees in $\rangeII{\kappa \Delta}{\kappa (\Delta + 2s + j - i)}$.
    
    For the second claim, note that all entries in $\Mbody^{(k)}$ and $\Mgap$
    have degree at most $\kappa s - 1 \le \kappa (s + m - (m - 1)) \le \kappa (2s - (m - 1))$ (using the assumption $s \ge m$);
    then we use \Cref{item:basic:uniform} in \Cref{fact:basic:example}.
    For $\Mtail$, the $(i,j)$ entry of $\Mtail$ has degree $\kappa (2s + j - i)$
    and so $\Mtail$ is $2s$-bounded by definition.
\end{proof}

\begin{remark}
    When we apply this proposition in the sequel, we will always want $s = \Omega_\rho(\Delta)$, where $\rho$ is the ``aspect ratio'' parameter,
    so that the decomposition has size $O_\rho(1)$.
\end{remark}

\subsection{Packing}

Let $x, y \in \N$ and $z \in \rangeII{0}{x + y}$,
we define
\begin{equation}\label{eq:pack}
\pack{z}{x}{y} \coloneqq \begin{cases} (z, 0) & z \le x, \\ (x, z-x) & z > x. \end{cases}
\end{equation}
This should be viewed as a packing of $z$ units of an item into two buckets of size $x$ and $y$,
where we pack greedily into the first bucket (of size $x$) as much as possible.

\subsection{Collecting commutators}

In this subsection, we write some simple inequalities bounding the area of a ``rearrangement'' from commutators of products to products of commutators.
We give slightly more refined bounds than we will actually need.

We first have the following trivial equivalence: For all words $u,v,w$ over a subgroup family $\calH$ of some group $G$,
\begin{equation}\label{eq:rearrange}
    \area{\Comm{u}{v} \equiv w}{\calH} = \area{uv = wvu}{\calH}.
\end{equation}

\begin{fact}\label{fact:comm:drag}
    Let $G$ be a group, $\calH$ an indexed subgroup family, and $u$, $v_1,\ldots,v_\ell$ words over $\calH$.
    Then: \[
    \area{u \parens*{ \prod_{j=1}^\ell v_j} \equiv \parens*{ \prod_{j=1}^\ell \Comm{u}{v_j} v_j } u}{\calH} \le \ell |u| + \sum_{j=1}^\ell |v_j|. \]
\end{fact}

\begin{proof}
    Induct on $\ell$.
    In the base case $\ell = 0$, $u=u$ is a tautology.
    For $\ell > 0$, we have $u \parens*{ \prod_{j=1}^\ell v_j } = u v_1 \parens*{ \prod_{j=2}^\ell v_j }$ and therefore
    $\area{u \parens*{ \prod_{j=1}^\ell v_j } \equiv uv_1 u^{-1} v_1^{-1} v_1 u \parens*{ \prod_{j=2}^\ell v_j }}{\calH} \le \area{u^{-1} v_1^{-1} v_1 u}{\calH} \le |u| + |v_1|$.
    Finally, $uv_1 u^{-1} v_1^{-1} v_1 u \parens*{ \prod_{j=2}^\ell v_j } = \Comm{u}{v_1} v_1 u \parens*{ \prod_{j=2}^\ell v_j}$ and we continue inductively.
\end{proof}

\begin{claim}
    Let $G$ be a group, $\calH$ an indexed subgroup family, and $u$, $v_1,\ldots,v_\ell$ words over $\calH$.
    Then: \[
    \area{\Comm{u}{ \prod_{j=1}^\ell v_j } \equiv
    \prod_{j=1}^\ell \Comm{u}{v_j} }{\calH}
    \le \ell |u| + \sum_{j=1}^\ell |v_j| + \sum_{j=1}^\ell \sum_{j'=j+1}^\ell \area{\Comm{v_j}{\Comm{u}{v_{j'}}}}{\calH}. \]
\end{claim}
Note that this lemma essentially presupposes that $\Comm{v_{j'}}{\Comm{u}{v_j}} \equiv \Id$
for every $j < j' \in \rangeOne{\ell}$, since otherwise the RHS is $\infty$.

\begin{proof}
    We apply \Cref{eq:rearrange} and again induct on $\ell$.
    In the base case $\ell = 0$, we again get the trivial equality $u = u$.
    For $\ell > 0$, we have $\prod_{j=1}^\ell \Comm{u}{v_j} v_j = \Comm{u}{v_1} v_1 \parens*{ \prod_{j=2}^\ell \Comm{u}{v_j} v_j}$.
    We then use:
    \begin{align*}
    \area{\prod_{j=1}^\ell \Comm{u}{v_j} v_j \equiv \Comm{u}{v_1} v_1 \parens*{\prod_{j=2}^\ell \Comm{u}{v_j}} \parens*{ \prod_{j=2}^\ell v_j}}{\calH}
    &\le \sum_{j=2}^\ell \sum_{j'=j+1}^\ell \area{\Comm{v_j}{\Comm{u}{v_{j'}}}}{\calH}, \\
    \area{\Comm{u}{v_1} v_1 \parens*{\prod_{j=2}^\ell \Comm{u}{v_j}} \parens*{ \prod_{j=2}^\ell v_j} \equiv \parens*{\prod_{j=1}^\ell \Comm{u}{v_j}} \parens*{ \prod_{j=1}^\ell v_j}}{\calH}
    &\le \sum_{j'=2}^\ell \area{\Comm{v_1}{\Comm{u}{v_{j'}}}}{\calH},
    \end{align*}
    where the first inequality is inductive and the second since $\area{\Comm{v_1}{\Comm{u}{v_1}}}{\calH} = \area{v_1 \Comm{u}{v_j} \equiv \Comm{u}{v_j} v_1}{\calH}$.
    Chaining these together with \Cref{fact:comm:drag} gives the inequality.
\end{proof}




\begin{claim}
    Let $G$ be a group, $\calH$ an indexed subgroup family, and $u_1,\ldots,u_k$, $v_1,\ldots,v_\ell$ words over $\calH$.
    Then:
    \begin{multline*}
    \area{\Comm{ \prod_{i=1}^k u_i }{ \prod_{j=1}^\ell v_j } \equiv
    \prod_{i=k}^1 \parens*{ \prod_{j=1}^\ell \Comm{u_i}{v_j}}}{\calH} \\
    \le \ell \sum_{i=1}^k |u_i| + k \sum_{j=1}^\ell |v_j| 
    + \sum_{i=1}^k \sum_{i'=i+1}^k \sum_{j=1}^\ell \area{\Comm{u_i}{\Comm{u_{i'}}{v_j}}}{\calH}
    + \sum_{i=1}^k \sum_{j=1}^\ell \sum_{j'=j+1}^\ell \area{\Comm{v_j}{\Comm{u_i}{v_{j'}}}}{\calH}.
    \end{multline*}
\end{claim}

This statement may appear slightly unwieldy, but all we really will need is that the RHS is a ``polynomial''; see \Cref{eq:collecting commutators} below.

\begin{proof}
    We apply \Cref{eq:rearrange} and again induct on $k$.
    The base case $k=0$ again gives the trivial equality $\prod_{j=1}^\ell v_j = \prod_{j=1}^\ell v_j$.
    For $k > 0$, we have $\parens*{ \prod_{i=1}^k u_i } \parens*{ \prod_{j=1}^\ell v_j} = u_1 \parens*{ \prod_{i=2}^k u_i } \parens*{ \prod_{j=1}^\ell v_j}$.
    Inductively, we have:
    \begin{multline*}
        \area{\parens*{ \prod_{i=2}^k u_i } \parens*{ \prod_{j=1}^\ell v_j} 
        \equiv \parens*{ \prod_{i=k}^2 \parens*{ \prod_{j=1}^\ell \Comm{u_i}{v_j}} } \parens*{ \prod_{j=1}^\ell v_j}  \parens*{ \prod_{i=2}^k u_i }}{\calH} \\
        \le \ell \sum_{i=2}^k |u_i| + (k-1) \sum_{j=1}^\ell |v_j| 
        + \sum_{i=2}^k \sum_{i'=i+1}^k \sum_{j=1}^\ell \area{\Comm{u_i}{\Comm{u_{i'}}{v_j}}}{\calH}
        + \sum_{i=2}^k \sum_{j=1}^\ell \sum_{j'=j+1}^\ell \area{\Comm{v_j}{\Comm{u_i}{v_{j'}}}}{\calH}.
    \end{multline*}
    At the same time, we have:
    \begin{align*}
        \area{u_1 \parens*{ \prod_{i=k}^2 \parens*{ \prod_{j=1}^\ell \Comm{u_i}{v_j}} } \equiv \parens*{ \prod_{i=k}^2 \parens*{ \prod_{j=1}^\ell \Comm{u_i}{v_j}} } u_1}{\calH}
        &\le \sum_{i=2}^k \sum_{j=1}^\ell \area{\Comm{u_1}{\Comm{u_i}{v_j}}}{\calH}, \\
        \area{u_1 \parens*{ \prod_{j=1}^\ell v_j} \equiv \parens*{ \prod_{j=1}^\ell \Comm{u_1}{v_j} } \parens*{ \prod_{j=1}^\ell v_j } u_1}{\calH} 
        &\le \ell |u_1| + \sum_{j=1}^\ell |v_j| + \sum_{j=1}^\ell \sum_{j'=j+1}^\ell \area{\Comm{v_j}{\Comm{u_1}{v_{j'}}}}{\calH},
    \end{align*}
    where the second inequality follows from the previous claim (and the first is trivial).
    Chaining these together gives the inequality.
\end{proof}

\begin{corollary}\label{eq:collecting commutators}
    Let $G$ be a group, $\calH$ an indexed subgroup family, and $u_1,\ldots,u_k, v_1,\ldots,v_\ell$ words over $\calH$.
    Then: \[
    \area{\Comm{ \prod_{i=1}^k u_i }{ \prod_{j=1}^\ell v_j } \equiv
    \prod_{i=k}^1 \parens*{ \prod_{j=1}^\ell \Comm{u_i}{v_j}}}{\calH} \le k\ell (U+V + k A + \ell B), \]
    where:
    \begin{align*}
        U &\coloneqq \max_i u_i,
        & V &\coloneqq \max_j v_j,
        & A &\coloneqq \max_{i,i',j} \area{\Comm{u_i}{\Comm{u_{i'}}{v_j}}}{\calH},
        & B &\coloneqq \max_{i,j,j'} \area{\Comm{v_j}{\Comm{u_i}{v_{j'}}}}{\calH}.
    \end{align*}
\end{corollary}

\section{Diameter of well-separated restrictions}\label{sec:diam}

In this section, we prove \Cref{eq:diam}.

\subsection{Setup}

We actually prove the following lemma with additional technical conditions:

\begin{lemma}\label{lemma:diam:padded}
For every field $\F$, $\kappa \ge 1 \in \N$, and $\nA,\nB,\nC,m \in \N$, with $m \ge 1$, $\nB \ge 4m$ and $\nA \equiv \nC \equiv 0 \pmod{m}$, we have \[
\gdiam{\GrUnip{n}{\F}}{(\GrStair{n}{\F}{\ell_i})_{i \in \{1,2\}}} \le O(\rho^3), \]
where $n \coloneqq \nA+\nB+\nC - 1$, $\ell_1 \coloneqq \nA$, $\ell_2 \coloneqq \nA+\nB$, and $\rho \coloneqq n/m$.
\end{lemma}

It is not difficult to see how \Cref{eq:diam} follows from \Cref{lemma:diam:padded}:

\begin{proof}[Proof of \Cref{eq:diam} modulo \Cref{lemma:diam:padded}]
    We set $n_0 \coloneqq \ceil*{ 4 \xi^{-1} }$ and $C_0 \coloneqq O((4\xi^{-1} + 2)^3)$ (with the same constants as in \Cref{lemma:diam:padded}).
    Now given an input, set $m \coloneqq \floor*{ \frac14 (\ell_2-\ell_1) }$.
    We therefore have $m \ge \frac14 (\ell_2 - \ell_1) \ge \frac14 \xi n \ge \frac14 \xi n_0 \ge 1$ and $\ell_2 - \ell_1 \ge 4m$.

    Next, let $\nA \coloneqq m \ceil*{ \ell_1 / m }$, $\nB \coloneqq \ell_2 - \ell_1$, and $\nC \coloneqq m \ceil*{ (n - \ell_2 + 1)/ m }$;
    these fulfill the hypotheses of \Cref{lemma:diam:padded} and therefore \[
    \gdiam{\GrUnip{n'}{\F}}{(\GrStair{n'}{\F}{\ell'_i})_{i \in \{1,2\}}} \le O(\rho^3), \]
    where $n' \coloneqq \nA + \nB + \nC - 1$, $\ell'_1 \coloneqq \nA$, $\ell'_2 \coloneqq \nA+\nB$, and $\rho \coloneqq n'/m$.
    Letting $n^{(1)} \coloneqq m \ceil*{ \ell_1 / m } - \ell_1$ and $n^{(2)} \coloneqq (m \ceil*{ (n - \ell_2 + 1)/ m } - (n - \ell_2 + 1)$,
    we have $n' - n = n^{(1)} + n^{(2)} \le 2m$ and therefore $\rho \le \frac{n}m + 2 \le 4\xi^{-1} + 2$.
    Finally, we apply \Cref{prop:diam:homo,lemma:pad} to give the desired bound.
\end{proof}

In the remainder of this section, we fix $\F,\kappa,n,m,\nA,\nB,\nC,\ell_1,\ell_2,\rho$ as in \Cref{lemma:diam:padded}, and prove \Cref{lemma:diam:padded}.
We define \[
\begin{aligned}
    \WiresA &\coloneqq \rangeIE{0}{\nA}, \\
    \WiresB &\coloneqq \rangeIE{\nA}{\nA+\nB}, \\
    \WiresC &\coloneqq \rangeIE{\nA+\nB}{\nA+\nB+\nC}, \\
    \Wires &\coloneqq \rangeII{0}{n}
\end{aligned}
\hspace{.5in}\text{so that}\hspace{.5in}
\begin{aligned}
    |\WiresA| &= \nA, \\
    |\WiresB| &= \nB, \\
    |\WiresC| &= \nC, \\
    |\Wires| &= n+1
\end{aligned} \]
and $\WiresA \sqcup \WiresB \sqcup \WiresC = \Wires$.
We use the convention that indices in $\WiresA$, $\WiresB$, and $\WiresC$ are written as $\ia$, $\ib$, and $\ic$, respectively.
We abbreviate $\GU \coloneqq \GrUnip{n}{\F}$, $\GUa \coloneqq \GrStair{n}{\F}{\ell_1} < \GU$,
$\GUb \coloneqq \GrStair{n}{\F}{\ell_2} < \GU$, and $\calH \coloneqq \{\GUa,\GUb\}$.

\subsection{Matrices and group elements}

We now define the sets of matrices:
\begin{alignat*}{3}
    \diamA &\coloneqq \TriMat{\F}{\WiresA} &&& \subseteq \F[X]^{\WiresA\times\WiresA}, \\
    \diamB &\coloneqq \TriMat{\F}{\WiresB} &&&\subseteq \F[X]^{\WiresB\times\WiresB}, \\
    \diamC &\coloneqq \TriMat{\F}{\WiresC} &&&\subseteq \F[X]^{\WiresC\times\WiresC}, \\
    \diamF &\coloneqq \RectMat{\F}{\WiresA}{\WiresA}\ &&&\subseteq \F[X]^{\WiresA\times\WiresB}, \\
    \diamG &\coloneqq \RectMat{\F}{\WiresB}{\WiresC}\ &&&\subseteq \F[X]^{\WiresB\times\WiresC}, \\
    \diamP &\coloneqq \RectMat{\F}{\WiresA}{\WiresC}\ &&&\subseteq \F[X]^{\WiresA\times\WiresC}.
\end{alignat*}

Correspondingly, \Cref{def:matrix:elements} gives us group elements corresponding to such matrices, satisfying:
\begin{align*}
A \in \diamA,&\quad\El{A} \in \GUa \cap \GUb, \\
B \in \diamB,&\quad\El{B} \in \GUa \cap \GUb, \\
C \in \diamC,&\quad\El{C} \in \GUa \cap \GUb, \\
F \in \diamF,&\quad\El{F} \in \GUb, \\
G \in \diamG,&\quad\El{G} \in \GUa, \\
P \in \diamP,&\quad\El{P} \in \GU. \qedhere
\end{align*}

\begin{figure}
    \centering
    \begin{tikzpicture}[framed]
    \addwiregroup{7}{0}{$\WiresA$};
    \addwiregroup{7}{1}{$\WiresB$};
    \addwiregroup{7}{2}{$\WiresC$};
    \addmoat{7}{0}{$\ell_1$};
    \addmoat{7}{1}{$\ell_2$};
    \addselfgate{1}{0}{$A$};
    \addselfgate{2}{1}{$B$};
    \addselfgate{3}{2}{$C$};
    \addgategroup{4}{0}{1}{$F$};
    \addgategroup{5}{1}{2}{$G$};
    \addgategroup{6}{0}{2}{$P$};
    \end{tikzpicture}
    \caption{
    The circuit view of the general form of an operator in $\GU$ vis-\`a-vis the subgroups $\GUa$ and $\GUb$ (cf. \Cref{prop:diam:expressing}):
    The wires $\Wires$ are bundled into sets $\WiresA$, $\WiresB$, and $\WiresC$ by the ``moats'' $\ell_1$ and $\ell_2$.
    We then divide up the operator's action into ``blocks'' which perform weighted additions
    within $\WiresA$ ($\El{A}$), within $\WiresB$ ($\El{B}$), within $\WiresC$ ($\El{C}$),
    from $\WiresA$ to $\WiresB$ ($\El{F}$), from $\WiresB$ to $\WiresC$ ($\El{G}$), or from $\WiresA$ to $\WiresC$ ($\El{P}$).}\label{fig:diam}
\end{figure}

Next, \Cref{prop:matrix:steinberg} gives:

\begin{proposition}[(A subset of the) Steinberg relations]\label{prop:diam:steinberg}
For every $F,F' \in \diamF$, $G,G' \in \diamG$, and $P,P' \in \diamP$, we have the following identities inside $\GU$:
\begin{itemize}
\item Homogeneous relations:
\begin{itemize}
	\item Rectangle:
    \begin{align*}
        \El{F} \El{F'} &= \El{F+F'},& \El{G} \El{G'} &= \El{G+G'},& \El{P} \El{P'} &= \El{P+P'}.
    \end{align*}
\end{itemize}
\item Heterogeneous relations:
\begin{itemize}
	\item Rectangle-rectangle:
    \begin{align*}
        \Comm{\El{F}}{\El{G}} &= \El{FG},& \Comm{\El{F}}{\El{P}} &= \Id,& \Comm{\El{G}}{\El{P}} &= \Id.
    \end{align*}
\end{itemize}
\end{itemize}
\end{proposition}

\begin{proposition}\label{prop:diam:expressing}
The following hold for elements of $\GU$ and its subgroups:
\begin{itemize}
\item Every element of $\GrUnip{n}{F}$ can be written uniquely as $\El{A} \El{B} \El{C} \El{F} \El{G} \El{P}$
for some $A \in \diamA, B \in \diamB, C \in \diamC, F \in \diamF, G \in \diamG, P \in \diamP$.
\item Every element of $\GUa$ can be written uniquely as $\El{A} \El{B} \El{C} \El{G}$
for some $A \in \diamA, B \in \diamB, C \in \diamC, G \in \diamG$.
\item Every element of $\GUb$ can be written uniquely as $\El{A} \El{B} \El{C} \El{F}$
for some $A \in \diamA, B \in \diamB, C \in \diamC, F \in \diamF$.
\end{itemize}
\end{proposition}

\begin{proof}
    Follows from \Cref{prop:expressing group elements}.
\end{proof}

See \Cref{fig:diam} for a graphical depiction of \Cref{prop:diam:expressing}.

In light of \Cref{prop:diam:expressing}, \Cref{lemma:diam:padded} follows from the following simpler lemma:

\begin{lemma}\label{lemma:diam:P bound}
For every $P \in \diamP$, there exists a word $w$ over $\calH$ of length $O(\rho^3)$ evaluating to $\El{P}$.
\end{lemma}

\begin{proof}[Proof of \Cref{lemma:diam:padded} modulo \Cref{lemma:diam:P bound}]
    We want to show that every element of $\GU$ can be expressed a word over $\calH$ of length at most $O(\rho^3)$.
    By \Cref{prop:diam:expressing}, every element of $\GU$ can be expressed as \[
    \underbrace{\El{A} \El{B} \El{C} \El{F}}_{\in \GUb} \underbrace{\El{G}}_{\in \GUa} \El{P}, \]
    and then we apply \Cref{lemma:diam:P bound} to $\El{P}$.
\end{proof}

\subsection{Decomposing wires into blocks}

We next turn to proving \Cref{lemma:diam:P bound}.
We essentially need to divide up the wire-bands $\WiresA$ and $\WiresC$ into \emph{blocks} of size $m$.
Towards this, we define:
\begin{alignat*}{3}
    \InitsA &\coloneqq \rangeII{0}{\nA-m}\ &&& \subset \WiresA, \\
    \InitsB &\coloneqq \rangeII{\nA}{\nA+\nB-m}\ &&& \subset \WiresB, \\
    \InitsC &\coloneqq \rangeII{\nA+\nB}{\nA+\nB+\nC-m}\ &&& \subset \WiresC.
\end{alignat*}
We typically use the variables $\initA$, $\initB$, and $\initC$ to refer to elements of these subsets, respectively.

For $\initA \in \InitsA$, we define $\IntA{\initA} \coloneqq \rangeIE{\initA}{\initA+m} \subset \WiresA$.
We think of this as a contiguous block of $m$ wires starting at wire $\initA$.
Similarly, we define $\I{\initB} \subset \WiresB$ for every $\initB \in \InitsB$,
and $\IntC{\initC} \subset \WiresC$ for every $\initB \in \InitsC$.

Next, we divide up $\WiresA$ and $\WiresC$ in a canonical way.
Recalling that $\nA$ and $\nC$ are multiples of $m$, we define:
\begin{equation}\label{eq:diam:A and C partition}
\begin{aligned}
&\partA{0} \coloneqq 0, \\
&\partA{1} \coloneqq m, \\
&\ldots, \\
&\partA{\nA/m - 1} \coloneqq (\nA/m -1) m = \nA-m
\end{aligned}
\in \InitsA \quad\text{and}\quad
\begin{aligned}
&\partC{0} \coloneqq \nA+\nB, \\
&\partC{1} \coloneqq \nA+\nB+m, \\
&\ldots, \\
&\partC{{\nC/m - 1}} \coloneqq \nA+\nB+ (\nC/m -1) m,
\end{aligned}
\in \InitsC.
\end{equation}
We view the former as inducing a partition $\IntPartA{0}, \IntPartA{1},\ldots,\IntPartA{\nA/m-1}$ of $\WiresA$,
and the latter as inducing a partition $\IntPartC{0}, \IntPartC{1},\ldots,\IntPartC{\nC/m-1}$ of $\WiresC$.
We will only use this decomposition briefly in the proof of \Cref{claim:diam:matrix:decompose P},
but it will serve as a useful point of comparison for decompositions we perform in the following section.

\subsection{Type-\texorpdfstring{$\typP$}{P} basic block matrices and decomposing type-\texorpdfstring{$\typP$}{P} matrices}

\begin{definition}[Type-$\typP$ basic block matrix]\label{def:diam:block:P}
For every $s \in \N$,
$\initA \in \InitsA$, 
$\initC \in \InitsC$,
$\powP \in \rangeII{0}{\initC - \initA - s}$,
and $s$-bounded basic matrix $M \in \basic{s}$,
we define the \emph{type-$\typP$ $s$-bounded basic block matrix}
$\blockP{\powP}{\initA}{\initC}{M} \in \F[X]^{\WiresA \times \WiresC}$ via
\[ (\blockP{\powP}{\initA}{\initC}{M})_{\ia,\ic} \coloneqq \begin{cases}
	X^{\kappa \powP} M_{\ia-\initA, \ic-\initC} &
		(\ia, \ic) \in \IntA{\initA} \times \IntC{\initC}, \\
	0 &
		\text{otherwise}.
\end{cases} \]
Further, $\deg((\blockP{\powP}{\initA}{\initC}{M})_{\ia,\ic}) \le \kappa \powP + \deg (M_{\ia-\initA, \ic-\initC}) \le \kappa(\powP + s + (\ic-\initC) - (\ia-\initA)) \le \kappa(\ic-\ia)$
and so
$\blockP{\powP}{\initA}{\initC}{M} \in \diamP$,
i.e., $\blockP{\powP}{\initA}{\initC}{M}$ is a type-$\typP$ matrix.
We call the corresponding group element $\El{\blockP{\powP}{\initA}{\initC}{M}}$ a \emph{type-$\typP$ $s$-bounded basic block element}.
\end{definition}

\begin{remark}
    The most important part of this definition is that $\powP$ is assumed to be in the interval $\rangeII{0}{\initC-\initA-s}$.
    (This assumption is necessary in order to guarantee that $\blockP{\powP}{\initA}{\initC}{M} \in \calP$.)
    In particular, we will use that at $s = 2m$, this interval is always nonempty because $\initC \ge \nA+\nB \ge \nA+m \ge \initA + 2m$.
\end{remark}

\begin{claim}[Decomposing type-$\typP$ block matrices]\label{claim:diam:matrix:block:decompose P}
    Let $\initA \in \InitsA$ and $\initC \in \InitsC$.
    Suppose $P \in \diamP$ is supported only on the $\IntA{\initA} \times \IntC{\initC}$ entries.
    Then $P$ is a sum of $\floor{ (\initC-\initA) / m } = O(\rho)$ $2m$-bounded basic block matrices
    (i.e., matrices of the form $\blockP{\powP}{\initA}{\initC}{M}$
    for some $\powP \in \rangeII{0}{\initC-\initA-2m}$, $M \in \basic{2m}$).
\end{claim}

\begin{proof}
    Let $M_{i,j} \coloneqq P_{\initA+i, \initC+j}$.
    Hence $\deg(M_{i,j}) = \deg(P_{\initA+i,\initC+j}) \le \kappa((\initC + j) - (\initA + i))$.
    Hence $M$ is $(\initC-\initA)$-bounded
    and $P = \blockP{0}{\initA}{\initC}M$.

    Now, letting $\Delta \coloneqq \initC-\initA - 2m$,
    $s \coloneqq m$, $\tau \coloneqq \lfloor \Delta / m \rfloor$,
    $M$ is $(\Delta+2s)$-bounded and \Cref{prop:basic:decomp} lets us decompose $M$ as
    \[
    M = \sum_{k=0}^{\tau - 1} X^{\kappa \cdot km} \Mbody^{(k)} + X^{\kappa \cdot \tau m} \Mgap + X^{\kappa \Delta} \Mtail
    \]
    for $\Mbody^{(k)}, \Mgap, \Mtail \in \basic{2m}$.
    Correspondingly,
    \begin{multline*}
    P = \blockP{0}{\initA}{\initC}M
    =  \sum_{k=0}^{\tau - 1} \blockP{0}{\initA}{\initC}{X^{\kappa \cdot km} \Mbody^{(k)}}
    + \blockP{0}{\initA}{\initC}{X^{\kappa \cdot \tau m} \Mgap}
    + \blockP{0}{\initA}{\initC}{X^{\kappa \Delta} \Mtail} \\
    = \sum_{k=0}^{\tau - 1} \blockP{km}{\initA}{\initC}{\Mbody^{(k)}}
    + \blockP{\tau m}{\initA}{\initC}\Mgap
    + \blockP{\Delta}{\initA}{\initC}{\Mtail}.
    \end{multline*}
    This decomposition uses $\tau + 2 = O(\rho)$ matrices, as desired.
\end{proof}

We combine this with the following:

\begin{claim}[Decomposing type-$\typP$ matrices]\label{claim:diam:matrix:decompose P}
    Every $P \in \diamP$ is the sum of $O(\rho^2)$ type-$\typP$ $2m$-bounded basic block matrices.
\end{claim}

\begin{proof}
	By partitioning $P$ into blocks according to the partition in \Cref{eq:diam:A and C partition},
    we can write $P = \sum_{\alpha=0}^{\nA/m-1} \sum_{\gamma=0}^{\nC/m-1} P^{(\alpha,\gamma)}$ where 
    $P^{(\alpha,\gamma)}$ is supported only on the $\IntPartA{\alpha} \times \IntPartC{\gamma}$ entries.
\end{proof}

By the commutator relation for type-$\typF$ and -$\typG$ elements (\Cref{prop:diam:steinberg}), we deduce:

\begin{corollary}\label{cor:diam:decompose P}
For every $P \in \calP$, $\El{P}$ can be expressed as a product of at most $O(\rho^3)$ type-$\typP$ $2m$-bounded basic block elements
(i.e., elements of the form $\El{\blockP{\powP}{\initA}{\initC}{M}}$ for some $\initA \in \InitsA$, $\initC \in \InitsC$, $\powP \in \rangeII{0}{\initC-\initA-2m}$, $M \in \basic{2m}$).
\end{corollary}

\subsection{Type-\texorpdfstring{$\typF$}{F} (and -\texorpdfstring{$\typG$}{G}) basic block matrices}

\begin{definition}[Type-$\typF$ basic block matrix]\label{def:diam:block:F}
    For every $s \in \N$,
    $\initA \in \InitsA$,
    $\initB \in \InitsB$,
    $\powF \in \rangeII{0}{\initB - \initA - s}$,
    and $M \in \basic{s}$,
    we define $\blockF{\powF}{\initA}{\initB}{M}$ via
    \[ (\blockF{\powF}{\initA}{\initB}{M})_{\ia,\ib} \coloneqq \begin{cases}
    	X^{\kappa f} M_{\ia-\initA, \ib-\initB}
    		& (\ia, \ib) \in \IntA{\initA} \times \I{\initB}, \\
    	0 &
    		\text{otherwise}. \end{cases} \]
    Further, $\deg((\blockF{\powF}{\initA}{\initB}{M})_{\ia,\ib})
    \le \kappa \powF + \deg(M_{\ia-\initA, \ib-\initB})
    \le \kappa(\powF + s + (\ib - \ia) - (\initB - \initA))
    \le \kappa(\ib-\ia)$
    and hence $\blockF{\powF}{\initA}{\initB}{M} \in \diamF$.
    We call the corresponding group element $\El{\blockF{\powF}{\initA}{\initB}{M}}$ a \emph{type-$\typF$ $s$-bounded basic block element}.
    For $\initB \in \InitsB$, $\initC \in \InitsC$,
    $\powG \in \rangeII{0}{\initC-\initB-s}$, and $M \in \basic{s}$,
    we similarly define $\blockG{\powG}{\initB}{\initC}{M}$ and $\El{\blockG{\powG}{\initB}{\initC}{M}}$.
\end{definition}

\begin{remark}
There is an important difference between this definition and the definition of type-$\typP$ bounded basic block elements:
E.g. for type-$\typF$ elements, at $s = 2m$, the interval $\rangeII{0}{\initB-\initA-s}$ might be empty, since we can have $\initB = \nA$ and $\initA = \nA-m$.
\end{remark}

\begin{fact}
For every $s,s' \in \N$,
$\initA \in \InitsA$, $\initB \in \InitsB$, $\initC \in \InitsC$,
$\powF \in \rangeII{0}{\initB-\initA-s}$, $\powG \in \rangeII{0}{\initC-\initB-s'}$,
$M \in \basic{s}$, and $N \in \basic{s'}$,
we have:
$\blockF{\powF}{\initA}{\initB}{M} \cdot \blockG{\powG}{\initB}{\initC}{n} = \blockP{\powF+\powG}{\initA}{\initC}{MN}$.
\end{fact}

\begin{proof}
    Direct calculation of a product of block matrices.
\end{proof}

\begin{lemma}\label{lemma:diam:P as F and G}
    For every
    $\initA \in \InitsA$ and $\initC \in \InitsC$,
    $\powP \in \rangeII{0}{\initC-\initA-2m}$,
    and $M \in \basic{2m}$,
    there exist $F \in \diamF$ and $G \in \diamG$ such that
    	$\Comm{\El{F}}{\El{G}} = \El{\blockP{\powP}{\initA}{\initC}{M}}$.
\end{lemma}
\begin{proof}
    Fix arbitrary $\initB' \in \rangeII{\nA+m}{\nA+\nB - m}$ with $\I{\initB} \cap \I{\initB'} = \emptyset$.
    (This is possible since $\nB \ge 4m$. For instance, case on whether $\initB \le \nA+2m$.
    If so, then $\initB=\nA+3m$ works, and if not, then $\initB' = \nA+2m$ works.)
    Also, fix arbitrary $\powF,\powG \in \N$ such that
    \begin{align*}
        \powF &\in \rangeII{0}{\initB'-\initA-2m}, \\
        \powG &\in \rangeII{0}{\initC-\initB'}, \\
        \powP &= \powF+\powG,
    \end{align*}
	which is possible since
	\begin{align*}
		(\initB'-\initA-2m) + (\initC-\initB') &=
			\initC-\initA-2m \ge \powP, \\
		\initB'-\initA-2m &\ge (\nA+m) - (\nA - m) - 2m = 0, \\
		\initC-\initB' &\ge (\nA+\nB + m) - (\nA+\nB - m) = 2m \ge 0.
	\end{align*}
	For instance, one can take $(\powF, \powG) = \pack{\powP}{\initB'-\initA-2m}{\initC-\initB'}$.
	
    Now, define
    \[ F \coloneqq \blockF{\powF}{\initA}{\initB'}{M} \text{ and } G \coloneqq \blockG{\powG}{\initB'}{\initC}{I_m} \]
    so that
    \[ FG = \blockP{\powF+\powG}{\initA}{\initC}{M} = \blockP{\powP}{\initA}{\initC}{M}, \]
    as desired.
\end{proof}

Finally, we can give:

\begin{proof}[Proof of \Cref{lemma:diam:P bound}]
    Let $P \in \diamP$.
    Let $R = O(\rho^3)$ denote the size of the product decomposition of $\El{P}$ into type-$\typP$ $2m$-bounded basic block elements in \Cref{cor:diam:decompose P}: \[
    \El{P} = \prod_{i=1}^R \El{\blockP{\powP^{(i)}}{\initA^{(i)}}{\initC^{(i)}}{M^{(i)}}}, \]
    where $\initA^{(i)} \in \InitsA$, $\initC^{(i)} \in \InitsC$, $\powP^{(i)} \in \rangeII{0}{\initC^{(i)}-\initA^{(i)}-2m}$, and $M^{(i)} \in \basic{2m}$.
    By \Cref{lemma:diam:P as F and G}, for each $i \in [R]$, there exists  $F^{(i)} \in \diamF, G^{(i)} \in \diamG$
    such that $\El{\blockP{\powP^{(i)}}{\initA^{(i)}}{\initC^{(i)}}{M^{(i)}}} = \Comm{\El{F^{(i)}}}{\El{G^{(i)}}}$.
    Hence \[
    \El{P} = \prod_{i=1}^R \Comm{\underbrace{\El{F^{(i)}}}_{\ni \GUb}}{\underbrace{\El{G^{(i)}}}_{\ni \GUa}}, \]
    which is a product of length $4R$.
\end{proof}

\section{Area of well-separated restrictions}\label{sec:cbdy}

In this section, we prove \Cref{eq:cbdy}.
The derivation follows the template, outlined in \Cref{sec:high-groups}, of implementing the Steinberg relations on ``blocks''.

\subsection{Setup}

Again, we prove the following ``padded'' form of \Cref{eq:cbdy}:

\begin{lemma}\label{lemma:cbdy:padded}
For every field $\F$, $\kappa \ge 1 \in \N$, and $\nA,\nB,\nC,\nD, m, C_0 \in \N$, with $m \ge 1$, $\nB, \nC \ge 4m$, and $\nA \equiv \nD \equiv 0 \pmod{m}$, \[
\garea{C_0}{\GrUnip{n}{\F}}{(\GrStair{n}{\F}{\ell_i})_{i \in \{1,2,3\}}} \le O(\rho^{17}) \]
where $n \coloneqq \nA+ \nB + \nC + \nD - 1$, $\ell_1 \coloneqq \nA$, $\ell_2 \coloneqq \nA+ \nB$, $\ell_3 \coloneqq \nA+ \nB + \nC$, and $\rho \coloneqq n/m$.
\end{lemma}

\begin{proof}[Proof of \Cref{eq:cbdy} modulo \Cref{lemma:cbdy:padded}]
    Same as the proof of \Cref{eq:diam} modulo \Cref{lemma:diam:padded}, now using \Cref{prop:cbdy:homo}.
\end{proof}

In the remainder of this section, we fix the parameters $\F, \kappa, \nA, \nB, \nC, \nD, m, C_0, n, \ell_1,\ell_2,\ell_3$ and prove \Cref{lemma:cbdy:padded}.

We define \[
\begin{aligned}
    \WiresA &\coloneqq \rangeIE{0}{\nA}, \\
    \WiresB &\coloneqq \rangeIE{\nA}{\nA+\nB}, \\
    \WiresC &\coloneqq \rangeIE{\nA+\nB}{\nA+\nB+\nC}, \\
    \WiresD &\coloneqq \rangeIE{\nA+\nB+\nC}{\nA+\nB+\nC+\nD}, \\
    \Wires &\coloneqq \rangeII{0}{n}
\end{aligned}
\hspace{.5in}\text{so that}\hspace{.5in}
\begin{aligned}
    |\WiresA| &= \nA, \\
    |\WiresB| &= \nB, \\
    |\WiresC| &= \nC, \\
    |\WiresD| &= \nD, \\
    |\Wires| &= n+1,
\end{aligned} \]
and $\WiresA \sqcup \WiresB \sqcup \WiresC \sqcup \WiresD = \Wires$.

We abbreviate $\GU \coloneqq \GrUnip{n}{\F}$, $\GUa \coloneqq \GrStair{n}{\F}{\ell_1} < \GU$, $\GUb \coloneqq \GrStair{n}{\F}{\ell_2} < \GU$,
$\GUc \coloneqq \GrStair{n}{\F}{\ell_3} < \GU$, and $\calH \coloneqq \{\GUa,\GUb,\GUc\}$.

\subsection{Matrices and group elements}

We again define sets of matrices:
\begin{alignat*}{3}
    \cbdyA &\coloneqq \TriMat{\F}{\WiresA} &&& \subseteq \F[X]^{\WiresA\times\WiresA}, \\
    \cbdyB &\coloneqq \TriMat{\F}{\WiresB} &&&\subseteq \F[X]^{\WiresB\times\WiresB}, \\
    \cbdyC &\coloneqq \TriMat{\F}{\WiresC} &&&\subseteq \F[X]^{\WiresC\times\WiresC}, \\
    \cbdyD &\coloneqq \TriMat{\F}{\WiresD} &&&\subseteq \F[X]^{\WiresD\times\WiresD}, \\
    \cbdyF &\coloneqq \RectMat{\F}{\WiresA}{\WiresA}\ &&&\subseteq \F[X]^{\WiresA\times\WiresB}, \\
    \cbdyG &\coloneqq \RectMat{\F}{\WiresB}{\WiresC}\ &&&\subseteq \F[X]^{\WiresB\times\WiresC}, \\
    \cbdyH &\coloneqq \RectMat{\F}{\WiresC}{\WiresD}\ &&&\subseteq \F[X]^{\WiresC\times\WiresD}, \\
    \cbdyP &\coloneqq \RectMat{\F}{\WiresA}{\WiresC}\ &&&\subseteq \F[X]^{\WiresA\times\WiresC}, \\
    \cbdyQ &\coloneqq \RectMat{\F}{\WiresB}{\WiresD}\ &&&\subseteq \F[X]^{\WiresB\times\WiresD}, \\
    \cbdyZ &\coloneqq \RectMat{\F}{\WiresA}{\WiresD}\ &&&\subseteq \F[X]^{\WiresA\times\WiresD}.
\end{alignat*}
(See also \Cref{tab:types}.)

\begin{figure}
    \centering
    \begin{tikzpicture}[framed]
        \addwiregroup{11}{0}{$\WiresA$};
        \addwiregroup{11}{1}{$\WiresB$};
        \addwiregroup{11}{2}{$\WiresC$};
        \addwiregroup{11}{3}{$\WiresD$};
        \addmoat{11}{0}{$\ell_1$};
        \addmoat{11}{1}{$\ell_2$};
        \addmoat{11}{2}{$\ell_3$};
        \addselfgate{1}{0}{$A$};
        \addselfgate{2}{1}{$B$};
        \addselfgate{3}{2}{$C$};
        \addselfgate{4}{3}{$D$};
        \addgategroup{5}{0}{1}{$F$};
        \addgategroup{6}{1}{2}{$G$};
        \addgategroup{7}{2}{3}{$H$};
        \addgategroup{8}{0}{2}{$P$};
        \addgategroup{9}{1}{3}{$Q$};
        \addgategroup{10}{0}{3}{$Z$};
    \end{tikzpicture}
    \caption{
    The circuit view of the general form of an operator in $\GU$ vis-\`a-vis the subgroups $\GUa$, $\GUb$, and $\GUc$ (cf. \Cref{prop:cbdy:expressing}):
    The wires $\Wires$ are bundled into sets $\WiresA$, $\WiresB$, $\WiresC$, and $\WiresD$ by the ``moats'' $\ell_1$, $\ell_2$, and $\ell_3$.
    We then divide up the operator's action into ``blocks'' which perform weighted additions
    within $\WiresA$ ($\El{A}$), within $\WiresB$ ($\El{B}$), within $\WiresC$ ($\El{C}$), within $\WiresD$ ($\El{D}$),
    from $\WiresA$ to $\WiresB$ ($\El{F}$), from $\WiresB$ to $\WiresC$ ($\El{G}$), from $\WiresC$ to $\WiresD$ ($\El{H}$),
    from $\WiresA$ to $\WiresC$ ($\El{P}$), from $\WiresB$ to $\WiresD$ ($\El{Q}$), and from $\WiresA$ to $\WiresD$ ($\El{Z}$).}\label{fig:cbdy}
\end{figure}

and we consider the corresponding group elements:
\begin{align*}
A \in \cbdyA,&\quad\El{A} \in \GUa \cap \GUb \cap \GUc, \\
B \in \cbdyB,&\quad\El{B} \in \GUa \cap \GUb \cap \GUc, \\
C \in \cbdyC,&\quad\El{C} \in \GUa \cap \GUb \cap \GUc, \\
D \in \cbdyD,&\quad\El{D} \in \GUa \cap \GUb \cap \GUc, \\
F \in \cbdyF,&\quad\El{F} \in \GUb \cap \GUc, \\
G \in \cbdyG,&\quad\El{G} \in \GUa \cap \GUc, \\
H \in \cbdyH,&\quad\El{H} \in \GUa \cap \GUb, \\
P \in \cbdyP,&\quad\El{P} \in \GUc, \\
Q \in \cbdyQ,&\quad\El{Q} \in \GUa, \\
Z \in \cbdyZ,&\quad\El{Z} \in \GU. \qedhere
\end{align*}

We have the following analogue of \Cref{prop:diam:expressing}:

\begin{proposition}\label{prop:cbdy:expressing}
The following hold for elements of $\GU$ and its subgroups:
\begin{itemize}
\item Every element of $\GU$ can be written uniquely as
$\El{A} \El{B} \El{C} \El{D} \El{F} \El{G} \El{H} \El{P} \El{Q} \El{Z}$
for some $A \in \cbdyA, B \in \cbdyB, C \in \cbdyC, D \in \cbdyD,
	F \in \cbdyF, G \in \cbdyG, H \in \cbdyH, P \in \cbdyP, Q \in \cbdyQ, Z \in \cbdyZ$.
\item Every element of $\GUa$ can be written uniquely as
$\El{A} \El{B} \El{C} \El{D} \El{G} \El{H} \El{Q}$
for some $A \in \cbdyA, B \in \cbdyB, C \in \cbdyC, D \in \cbdyD,
	G \in \cbdyG, H \in \cbdyH, Q \in \cbdyQ$;
\item Every element of $\GUb$ can be written uniquely as
$\El{A} \El{B} \El{C} \El{D} \El{F}  \El{H}$
for some $A \in \cbdyA, B \in \cbdyB, C \in \cbdyC, D \in \cbdyD,
	F \in \cbdyF, H \in \cbdyH$.
\item Every element of $\GUc$ can be written uniquely as
$\El{A} \El{B} \El{C} \El{D} \El{F} \El{G} \El{P}$
for some $A \in \cbdyA, B \in \cbdyB, C \in \cbdyC, D \in \cbdyD,
	F \in \cbdyF, G \in \cbdyG, P \in \cbdyP$.
\end{itemize}
\end{proposition}

See \Cref{fig:cbdy} for a graphical depiction of \Cref{prop:cbdy:expressing}.
We also have the following analogue of \Cref{prop:diam:steinberg}:

\begin{proposition}[Steinberg relations]\label{prop:cbdy:steinberg}
For every $A, A' \in \cbdyA$, $B,B' \in \cbdyB$, $C, C' \in \cbdyC$, $D, D' \in \cbdyD$,
	$F, F' \in \cbdyF$, $G, G' \in \cbdyG$, $H, H' \in \cbdyH$, $P, P' \in \cbdyP$, $Q, Q' \in \cbdyQ$, and $Z, Z' \in \cbdyZ$, we have:
\begin{itemize}
\item Homogeneous relations:
\begin{itemize}
	\item Rectangle:
    \begin{align*}
        \El{F} \El{F'} &= \El{F+F'}, & \El{G} \El{G'} &= \El{G+G'}, & \El{H} \El{H'} &= \El{H+H'} \\
        \El{P} \El{P'} &= \El{P+P'}, & \El{Q} \El{Q'} &= \El{Q+Q'}, & \El{Z} \El{Z'} &= \El{Z+Z'}.
    \end{align*}
	\item Triangle:
    \begin{align*}
        \El{A} \El{A'} &= \El{AA'+A+A'},& \El{B} \El{B'} &= \El{BB'+B+B'} \\
        \El{C} \El{C'} &= \El{CC'+C+C'},& \El{D} \El{D'} &= \El{DD'+D+D'}.
    \end{align*}
		
\end{itemize}
\item Heterogeneous relations:
\begin{itemize}
	\item Rectangle-rectangle: 
	\begin{align*}
		\Comm{\El{F}}{\El{G}} &= \El{FG},& \Comm{\El{F}}{\El{P}} &= \Id,& \Comm{\El{G}}{\El{P}} &= \Id, \\
		\Comm{\El{G}}{\El{H}} &= \El{GH},& \Comm{\El{G}}{\El{Q}} &= \Id,& \Comm{\El{H}}{\El{Q}} &= \Id, \\
		\Comm{\El{F}}{\El{H}} &= \Id,&&&& \\
		\Comm{\El{P}}{\El{Q}} &= \Id,& \Comm{\El{F}}{\El{Q}} &= \El{FQ}, &\Comm{\El{P}}{\El{H}} &= \El{PH}, \\
		\Comm{\El{F}}{\El{Z}} &= \Id,& \Comm{\El{G}}{\El{Z}} &= \Id,& \Comm{\El{H}}{\El{Z}}
			&= \Id,\\
        && \Comm{\El{P}}{\El{Z}} &= \Id,& \Comm{\El{Q}}{\El{Z}} &= \Id.
	\end{align*}
		
	\item Triangle-triangle:
	\begin{align*}
	    \Comm{\El{A}}{\El{B}} &= \Id,& \Comm{\El{B}}{\El{C}} &= \Id,& \Comm{\El{C}}{\El{D}} &= \Id, \\ 
		\Comm{\El{A}}{\El{C}} &= \Id,& \Comm{\El{B}}{\El{D}} &= \Id,& \Comm{\El{A}}{\El{D}} &= \Id.
	\end{align*}
	\item Rectangle-triangle:
	\begin{align*}
	\Comm{\El{A}}{\El{F}} &= \El{AF},& \Comm{\El{B}}{\El{F}} &= \El{FB^\dagger},& \Comm{\El{C}}{\El{F}} &= \Id,& \Comm{\El{D}}{\El{F}} &= \Id, \\
	\Comm{\El{B}}{\El{G}} &= \El{BG},& \Comm{\El{C}}{\El{G}} &= \El{GC^\dagger},& \Comm{\El{A}}{\El{G}} &= \Id,& \Comm{\El{D}}{\El{G}} &= \Id, \\
	\Comm{\El{C}}{\El{H}} &= \El{CH},& \Comm{\El{D}}{\El{H}} &= \El{HD^\dagger},& \Comm{\El{A}}{\El{H}} &= \Id,& \Comm{\El{B}}{\El{H}} &= \Id, \\
	\Comm{\El{A}}{\El{P}} &= \El{AP},& \Comm{\El{C}}{\El{P}} &= \El{PC^\dagger},& \Comm{\El{B}}{\El{P}} &= \Id,& \Comm{\El{D}}{\El{P}} &= \Id, \\
	\Comm{\El{B}}{\El{Q}} &= \El{BQ},& \Comm{\El{D}}{\El{Q}} &= \El{QD^\dagger},& \Comm{\El{A}}{\El{Q}} &= \Id,& \Comm{\El{C}}{\El{Q}} &= \Id, \\
	\Comm{\El{A}}{\El{Z}} &= \El{AZ},& \Comm{\El{D}}{\El{Z}} &= \El{ZD^\dagger},& \Comm{\El{B}}{\El{Z}} &= \Id,& \Comm{\El{C}}{\El{Z}} &= \Id.
	\end{align*}
\end{itemize}
\end{itemize}
\end{proposition}

\subsection{Reducing to Steinberg relations}

Let $\Types \coloneqq \{\typA, \typB, \typC, \typD, \typF, \typG, \typH, \typP, \typQ\}$.
(Note $\typZ$ is \emph{not} included.)
For each $\typL \in \Types$ and $L \in \Mats{\typL}$, we use $\Let{L}$ to denote the subgroup element $\El{L}$ regarded as a length-$1$ word over $\calH$.
(The advantage of this notation is that it lets us clearly distinguish between multiplication of words over $\calH$, i.e., concatenation, and multiplication of elements within $\GU$.)

In this section, we prove the following reduction lemma, which can be thought of as a ``higher-order'' analogue of the ``bubble sort'' argument of~\cite[Lemma~7.3]{KO21}.

\begin{lemma}\label{lemma:cbdy:reduce to steinberg}
    Let $R \in \N$.
    Suppose that for every $Z \in \cbdyZ$, there exists a subgroup word $\alias{Z}$ of length at most $R$ with $\eval(\alias{Z}) = \El{Z}$ such that the following holds.
    Let $\RSt$ denote the set of subgroup relators over $\calH$ formed by taking all equations in \Cref{prop:cbdy:steinberg} and replacing:
    \begin{itemize}
        \item Every equality in $\GU$ (the $=$ sign) with equality of evaluations of words (the $\equiv$ sign),
        \item Every instance of $\El{L}$ with $\Let{L}$ for $L \in \Mats{\typL}$, $\typL \in \Types$,
        \item Every instance of $\El{Z}$ with $\alias{Z}$ for $Z \in \cbdyZ$, and
        \item Every instance of $\Id$ with $\zero$.
    \end{itemize}
    If $\area{\RSt}{\calH} \le R$, then \[
    \garea{C_0}{G}{\calH} \le O(C_0 \cdot R), \]
    where $O(\cdot)$ means absolute universal constants.
\end{lemma}
We call $\alias{Z}$ the \emph{alias word} for the group element $\El{Z}$; when we are performing derivations via words over $\calH$, we use it to ``stand in'' for $\El{Z}$ since $\El{Z}$ itself is not in any of the subgroups in $\calH$.

For $S \subseteq \Types$, we say a word over $\calH$ is an \emph{$S$-word} if every element in it is of the form $\Let{L}$ where $L \in \Mats{\typL}$ for some $\typL \in S$.
In this section $O(\cdot)$ denotes some absolute linear function.
Our first lemma lets us convert arbitrary words over $\calH$ to $\Types$-words:

\begin{claim}\label{lemma:cbdy:reduce to alphabet}
    For every $g \in \GUa \cup \GUb \cup \GUc$,
    there exists an $\Types$-word $w$ such that $\area{\Embed{g} \equiv w}{\calH} \le 1$ and $|w| \le 7$.
\end{claim}

Note that the finiteness of the area of $\Embed{g} \equiv w$ in particular implies that $g = \eval(w)$.

\begin{proof}
    Suppose $g \in \GUa$.
    According to \Cref{prop:cbdy:expressing}, we have $g = \El{A} \El{B} \El{C} \El{D} \El{G} \El{H} \El{Q}$
    for some $A \in \cbdyA, B \in \cbdyB, C \in \cbdyC, D \in \cbdyD, G \in \cbdyG, H \in \cbdyH, Q \in \cbdyQ$.
    Hence $\Embed{g} \equiv \Let{A} \Let{B} \Let{C} \Let{D} \Let{G} \Let{H} \Let{Q}$.
    The proofs for the other cases are similar.
\end{proof}

\begin{claim}\label{claim:cbdy:zero}
    For every $\typL \in \Types$, $\area{\Let{0_\typL}}{\calH} \le O(R)$ (where $0_\typL \in \Mats{\typL}$ is the all-zeroes type-$\typL$ matrix)
    and further, $\area{\alias{0_\typZ}}{\calH} \le O(R)$.
\end{claim}

\begin{proof}
    For the first claim, we use that $0_\typL + 0_\typL = 0$ and therefore by the definition of $\RSt$, \[
    \area{\Let{0_\typL} \cdot \Let{0_\typL} \equiv \Let{0_\typL}}{\calH} \le R. \]
    We similarly have \[
    \area{\alias{0_\typZ} \cdot \alias{0_\typZ} \equiv \alias{0_\typZ}}{\calH} \le R. \]
    We can cancel quantities from both sides to achieve the desired forms in $O(R)$ steps (since $|\alias{0_\typZ}| \le R$).
\end{proof}

Next, we define a total ordering $\prec$ on $\Types$ as the alphabetical ordering $\typA \prec \typB \prec \typC \prec \typD \prec \typF \prec \cdots$.
For $\typL \in \Types$, we define $\Succ_\typL \coloneqq \{ \typL' \in \Types : \typL \prec \typL' \}$ as the set of letters which are lexicographically larger than $\typL$.
Note that a $\Succ_\typL$ word does \emph{not} contain any type-$\typL$ elements.

Our first claim lets us move a type-$\typL$ element from right to left across a $\Succ_\typL$-word $w$,
at the cost of doubling the length of $w$ and possibly creating a type-$\typZ$ alias:

\newcommand{\pre}{{\mathrm{pre}}}
\newcommand{\suff}{{\mathrm{suff}}}
\newcommand{\comm}{{\mathrm{comm}}}
\renewcommand{\Succ}{{\mathrm{Succ}}}
\renewcommand{\next}{{\mathrm{next}}}

\begin{claim}\label{claim:cbdy:move 1}
    Let $\typL \in \Types$ be any letter.
    For every matrix $L \in \Mats{\typL}$ and ``prefix'' $\Succ_{\typL}$-word $w_\pre$,
    there exists a matrix $L' \in \Mats{\typL}$, a $\Succ_{\typL}$-word $w'_\pre$, and a matrix $Z' \in \cbdyZ$ such that \[
    \area{w_\pre \cdot \Let{L} \equiv \Let{L'} \cdot w'_\pre \cdot \alias{Z'}}{\calH} \le O(|w_\pre|R), \] and $|w'_\pre| \le 2|w_\pre|$.
\end{claim}

\begin{proof}
    We induct on $|w_\pre|$.
    In the base case, $|w_\pre| = 0$;
    we set $L' \coloneqq L$, $w_\pre' \coloneqq \zero$, and $Z \coloneqq 0$, and we have $\area{\alias{0} \equiv \zero}{\calH} \le O(R)$ by \Cref{claim:cbdy:zero}.
    
    Otherwise, peel one element off $w_\pre$: Decompose $w_\pre = y_\pre \cdot \Let{L_\pre}$,
    where $L_\pre \in \Mats{\typL_\pre}$, $\typL_\pre \in \Succ_{\typL}$.
    (Thus, $|y_\pre| = |w_\pre|-1$.)
    Hence $w_\pre \cdot \Let{L_\pre} = y_\pre \cdot \Let{L_\pre} \cdot \Let{L}$.
    Now, by definition of $\RSt$ and inspection of \Cref{prop:cbdy:steinberg},
    there exists a $\Succ_{\typL}$-word $x$ which ``commutes'' $\Let{L}$ and $\Let{L_\pre}$, i.e., such that
    \begin{equation}\label{eq:x}
    \area{\Let{L_\pre} \cdot \Let{L} \equiv \Let{L} \cdot \Let{L_\pre} \cdot x}{\calH} \le R,
    \end{equation}
    and further, $x$ has one of the following three forms: (i) empty,
    (ii) $\Let{L_\comm}$ for some $L_\comm \in \Mats{\typL_\comm}$, $\typL_\comm \in \Succ_{\typL}$, or
    (iii) $\alias{Z}$ for some $Z \in \cbdyZ$.
    (E.g., case (iii) occurs iff $\{\typL,\typL_\pre\} \in \{\{\typF,\typQ\}, \{\typP,\typH\}\}$.)

    Next, apply the inductive hypothesis to $y_\pre \cdot \Let{L}$.
    Thus, there is some $L' \in \Mats{\typL}$, $\Succ_{\typL}$-word $y_\pre'$, and $Z \in \Mats{\typZ}$
    such that
    \begin{equation}\label{eq:y2}
        \area{y_\pre \cdot \Let{L} \equiv \Let{L'} \cdot y_\pre' \cdot \alias{Z}}{\calH} \le O(|y_\pre|R)
    \end{equation}
    and $|y_\pre'| \le 2|y_\pre|$.    
    Consequently, there is a $\Succ_\typL$-word $x'$ and $Z' \in \cbdyZ$ such that
    \begin{equation}\label{eq:x'}
    \area{x \cdot \alias{Z} \equiv x' \cdot \alias{Z'}}{\calH} \le R.
    \end{equation}
    (In cases (i) and (ii), $x' = x$ and $Z' = Z$.
    In case (iii), $x'$ is empty and $Z' = Z + Z$, and we again use the definition of $\RSt$ and inspection of \Cref{prop:cbdy:steinberg}.)
    We now have:
    \begin{align*}
    w \cdot \Let{L} &= y_\pre \cdot \Let{L_\pre} \cdot \Let{L} \\
    &\equiv y_\pre \cdot \Let{L} \cdot \Let{L_\pre} \cdot x \tag{\Cref{eq:x}} \\
    &\equiv \Let{L'} \cdot y_\pre' \cdot \alias{Z} \cdot \Let{L_\pre} \cdot x \tag{\Cref{eq:y2}} \\
    &\equiv \Let{L'} \cdot y_\pre' \cdot \Let{L_\pre} \cdot x \cdot \alias{Z} \tag{comm. of type-$\typZ$ elts. in $\RSt$} \\
    &\equiv \Let{L'} \cdot \underbrace{y_\pre' \cdot \Let{L_\pre} \cdot x'}_{\eqqcolon w_\pre'} \cdot \alias{Z'} \tag{\Cref{eq:x'}},
    \end{align*}
    Further, $|y_\pre'| \le 2|y_\pre|$ by inductive hypothesis, so $|w_\pre'| \le |y_\pre'| + 2 \le 2|y_\pre| + 2 = 2|w_\pre|$.
    Then, the four steps have costs $R$, $O((|w_\pre|-1)R)$, $2R$, and $R$, respectively, for a total cost of $O(|w_\pre|R)$.
\end{proof}

We next show how to move a type-$\typL$ element from right to left over a $(\Succ_\typL \cup \{\typL\})$-word:

\begin{claim}\label{claim:cbdy:move many}
    Let $\typL \in \Types$ be any letter.
    For every matrix $L \in \Mats{\typL}$ and ``prefix'' $(\Succ_{\typL} \cup \{\typL\})$-word $w_\pre$,
    there exists a matrix $L' \in \Mats{\typL}$, a $\Succ_\typL$-word $w_\pre'$, and a matrix $Z \in \cbdyZ$ such that \[
    \area{w_\pre \cdot \Let{L} \equiv \Let{L'} \cdot w_\pre' \cdot \alias{Z}}{\calH} \le O(|w_\pre|R) \]
    and $|w_\pre'| \le 2|w_\pre|$.
\end{claim}

(The difference with the preceding proposition is that we do \emph{not} require that $w$ does not contain any elements of type $\typL$.)

\begin{proof}
    Let $|w_\pre|_\typL$ denote the number of elements of the form $\Let{L'}$ ($L' \in \Mats{\typL}$) in $w_\pre$.
    We induct on $|w_\pre|_\typL$.
    If $|w_\pre|_\typL = 0$, i.e., $w_\pre$ is already a $\Succ_\typL$-word, then we just apply the previous \Cref{claim:cbdy:move 1}.
    
    Otherwise, we have $w_\pre = x_\pre \cdot \Let{L_\pre} \cdot y_\pre$ where $x_\pre$ is a $(\Succ_\typL \cup \{\typL\})$-word,
    $y_\pre$ is a $\Succ_\typL$-word (i.e., it does not contains elements of type $\typL$),
    and $|x_\pre|_\typL = |w_\pre|_\typL - 1$.
    Hence:
    \begin{align*}
    w \cdot \Let{L} &= x_\pre \cdot \Let{L_\pre} \cdot y_\pre \cdot \Let{L} \\
    &\equiv x_\pre \cdot \Let{L_\pre} \cdot \Let{L} \cdot y_\pre' \cdot \alias{Z_\suff} \tag{\Cref{claim:cbdy:move 1} on $y_\pre \cdot \Let{L}$} \\
    &\equiv x_\pre \cdot \Let{L_\pre + L} \cdot y_\pre' \cdot \alias{Z_\suff} \tag{lin. of type-$\typL$ elts. in $\RSt$} \\
    &\equiv \Let{L'} \cdot x_\pre' \cdot \alias{Z_\pre} \cdot y_\pre' \cdot \alias{Z_\suff} \tag{ind. hyp. on $x_\pre \cdot \Let{L_\pre + L}$} \\
    &\equiv \Let{L'} \cdot x_\pre' \cdot y_\pre' \cdot \alias{Z_\pre} \cdot \alias{Z_\suff} \tag{comm. of type-$\typZ$ elts. in $\RSt$} \\
    &\equiv \Let{L'} \cdot \underbrace{x_\pre' \cdot y_\pre'}_{\eqqcolon w_\pre'} \cdot \alias{\underbrace{Z_\pre+Z_\suff}_{\eqqcolon Z'}} \tag{lin. of type-$\typZ$ elts. in $\RSt$}.
    \end{align*}
    We have $|y_\pre'| \le 2|y_\pre|$ and $|x_\pre'| \le 2|x_\pre|$.
    The costs of the steps are $O(|y_\pre|R)$, $R$, $O(|x_\pre|R)$, $|y_\pre'| R \le 2|y_\pre|R$, and $R$, respectively,
    for a total cost of $O((|x_\pre|+|y_\pre|)R) \le O(|w_\pre|R)$, as desired.
\end{proof}

\begin{claim}\label{claim:cbdy:move many suff}
    For every $\typL \in \Types$ and $(\Succ_{\typL} \cup \{\typL\})$-word $w$,
    there exists $L \in \Mats{\typL}$, a $\Succ_\typL$-word $y$, and $Z \in \cbdyZ$ such that \[
    \area{w \equiv \Let{L} \cdot y \cdot \alias{Z}}{\calH} \le O(|w|R) \]
    and $|y| \le 2|w|$.
\end{claim}

\begin{proof}
    If $w$ is already a $\Succ_\typL$-word (i.e., it contains no element of the form $\Let{L}$ for $L \in \Mats{\typL}$),
    we set $L \coloneqq 0$, $y \coloneqq w$, and $Z \coloneqq 0$.
    By \Cref{claim:cbdy:zero}, we already have $\area{w \equiv \Let{0} \cdot y \cdot \alias{0}}{\calH} \le O(R)$.

    Otherwise, $w = w_\pre \cdot \Let{L} \cdot w_\suff$ for some $(\Succ_\typL \cup \{\typL\})$-word $w_\pre$ and $\Succ_\typL$-word $w_\suff$.
    We apply the prior claim (\Cref{claim:cbdy:move many}) to $w_\pre \cdot \Let{L}$, so that
    there exists $L' \in \Mats{\typL}$, $y_\pre$ a $\Succ_\typL$-word, and $Z \in \cbdyZ$ such that
    \begin{equation}\label{eq:split}
    \area{w_\pre \cdot \Let{L} \equiv \Let{L'} \cdot y_\pre \cdot \alias{Z}}{\calH} \le O(|w_\pre|R).
    \end{equation}

    Now we have the $G$-equalities
    \begin{align*}
        w_\pre \cdot \Let{L} \cdot w_\suff &\equiv \Let{L'} \cdot y_\pre \cdot \alias{Z} \cdot w_\suff \tag{\Cref{eq:split}} \\
        &\equiv \Let{L'} \cdot \underbrace{y_\pre \cdot w_\suff}_{\eqqcolon y} \cdot \alias{Z}. \tag{comm. of type-$\typZ$ elts. in $\RSt$}
    \end{align*}
    The cost of the former is $O(|w_\pre|R)$ and the latter is $|w_\suff|R$.
\end{proof}

With this setup, we can now prove \Cref{lemma:cbdy:reduce to steinberg}:

\begin{proof}[Proof of \Cref{lemma:cbdy:reduce to steinberg}]
    We are given a relator $w$ over $\calH$ and want to show that $\area{w}{\calH} \le O(|w|R)$.
    
    Our first step is to convert $w$ into an $\Types$-word.
    In particular, for $w = \Embed{g_1} \cdots \Embed{g_T}$ with each $g_t \in \bigcup \calH$,
    we apply \Cref{lemma:cbdy:reduce to alphabet} to each element $g_t$ to get $\Types$-words $w_1,\ldots,w_T$ such that
    for every $t \in [T]$, $\area{\Embed{g_t} \equiv w_t}{\calH} \le 1$ and $|w_t| \le 7$.
    Chaining these equalities together in any arbitrary order gives that, for the $\Types$-relator $w^\typA \coloneqq w_1 \cdots w_T$, we have $\area{w^\typA}{\calH} \le |w|$ and $|w^\typA| \le O(|w|)$.

    Next, we progressively eliminate letters from $w^\typA$ using \Cref{claim:cbdy:move many suff}.
    For each letter $\typL$ in $\Types$, we let $\next(\typL)$ denote the succeeding letter to $\typL$ in the natural ordering on $\Types \cup \{\typZ\}$.
    (If $\typL \prec \typQ$, then $\next(\typL) \in \Types$, whereas $\next(\typQ) = \typZ \not\in \Types$.)
    For every $\typL \in \Types$, note that $\Succ_{\next(\typL)} \cup \{\typL\} = \Succ_\typL$ (with the convention that $\Succ_\typZ \coloneqq \emptyset$).
    Consider the following iterative process.
    For $\typL$ initialized to $\typA$, $w^\typL$ is a $(\Succ_\typL \cup \{\typL\})$-word.
    Hence there exists $L \in \Mats{\typL}$, a $\Succ_\typL$-word $w^{\next(\typL)}$, and $Z^\typL \in \cbdyZ$ such that $w^{\next(\typL)} \le 2|w^\typL|$ and \[
    \area{w^\typL \equiv \Let{L} \cdot w^{\next(\typL)} \cdot \alias{Z^\typL}}{\calH} \le O(|w^\typL|R). \]
    We now replace $\typL$ with $\next(\typL)$ and iterate (until $\typL = \typZ$, when we halt).
    In the end, $w^\typZ$ is an $\emptyset$-word and therefore must be the empty word $\zero$.
    
    Finally, we define \[
    Z \coloneqq Z^\typA + Z^\typB + Z^\typC + Z^\typD + Z^\typF + Z^\typG + Z^\typH + Z^\typP + Z^\typQ, \]
    so that letting \[
    \tilde{Z} \coloneqq \alias{Z^\typA} \cdot \alias{Z^\typB} \cdot \alias{Z^\typC} \cdot \alias{Z^\typD} \cdot \alias{Z^\typF} \cdot \alias{Z^\typG} \cdot \alias{Z^\typH} \cdot \alias{Z^\typP} \cdot \alias{Z^\typQ},
    \] we have $\area{\alias{Z} \equiv \tilde{Z}}{\calH} \le 8R$.

    Substituting backward, we get that \[
    \area{w \equiv \Let{A} \Let{B} \Let{C} \Let{D} \Let{F} \Let{G} \Let{H} \Let{P} \Let{Q} \cdot \alias{Z}}{\calH} \le O(|w|R). \]
    
    Finally, we claim that we must have $A = B = C = D = F = G = H = P = Q = Z = 0$.
    Indeed, we assumed that $\eval(w) \equiv \Id$, and so \[
    \El{A} \El{B} \El{C} \El{D} \El{F} \El{G} \El{H} \El{P} \El{Q} \El{Z} = \Id \]
    (using that e.g. $\eval(\Let{A}) = \El{A}$ and $\eval(\alias{Z}) = \El{Z}$).
    At the same time, by \Cref{claim:cbdy:zero}, \[
    \El{0_\typA} \El{0_\typB} \El{0_\typC} \El{0_\typD} \El{0_\typF} \El{0_\typG} \El{0_\typH} \El{0_\typP} \El{0_\typQ} \El{0_\typZ} = \Id. \]
    Hence the equalities follow from the uniqueness statement in \Cref{prop:cbdy:expressing}.
    Hence by \Cref{claim:cbdy:zero}, $\area{w}{\calH} \le O(|w|R)$, as desired.
\end{proof}

\subsection{The ``missing'' Steinberg relations}

We now introduce a final notation: For an equation $\Eq{x}{y}$ over $\calH$ (or the corresponding relator $xy^{-1}$), and $i \in \N$, ``$\vdash_i x \equiv y$'' denotes that $\area{x\equiv y}{\calH} = O(\rho^i)$.
We now list the relations $w \in \RSt$ that already have the common subgroup property, and therefore have $\area{w}{\calH} = 0$:

\begin{lemma}[``Present'' Steinberg relations]\label{lemma:cbdy:present}
For every $A, A' \in \cbdyA$, $B,B' \in \cbdyB$, $C, C' \in \cbdyC$, $D, D' \in \cbdyD$,
$F, F' \in \cbdyF$, $G, G' \in \cbdyG$, $H, H' \in \cbdyH$, $P, P' \in \cbdyP$, $Q, Q' \in \cbdyQ$, we have:
\begin{itemize}
\item Homogeneous relations:
\begin{itemize}
	\item Rectangle: \[
    \vdash_0 \left\{
    \begin{aligned}
        \Let{F} \Let{F'} &\equiv \Let{F+F'}, & \Let{G} \Let{G'} &\equiv \Let{G+G'}, & \Let{H} \Let{H'} &\equiv \Let{H+H'} \\
        \Let{P} \Let{P'} &\equiv \Let{P+P'}, & \Let{Q} \Let{Q'} &\equiv \Let{Q+Q'}. &&
    \end{aligned}\right. \]
	\item Triangle: \[
    \vdash_0 \left\{
    \begin{aligned}
        \Let{A} \Let{A'} &\equiv \Let{AA'+A+A'},& \Let{B} \Let{B'} &\equiv \Let{BB'+B+B'} \\
        \Let{C} \Let{C'} &\equiv \Let{CC'+C+C'},& \Let{D} \Let{D'} &\equiv \Let{DD'+D+D'}.
    \end{aligned}\right. \]
		
\end{itemize}
\item Heterogeneous relations:
\begin{itemize}
	\item Rectangle-rectangle:  \[
    \vdash_0 \left\{
        \begin{aligned}
        \Comm{\Let{F}}{\Let{G}} &\equiv \Let{FG} \text{ and } \Comm{\Let{F}}{\Let{P}} \equiv \Comm{\Let{G}}{\Let{P}} \equiv \zero, \\
		\Comm{\Let{G}}{\Let{H}} &\equiv \Let{GH} \text{ and } \Comm{\Let{G}}{\Let{Q}} \equiv \Comm{\Let{H}}{\Let{Q}} \equiv \zero, \\
		\Comm{\Let{F}}{\Let{H}} &\equiv \zero.
        \end{aligned}\right.
	\]
		
	\item Triangle-triangle: \[
	\vdash_0 \left\{
    \begin{aligned}
	    \Comm{\Let{A}}{\Let{B}} &\equiv \zero,& \Comm{\Let{B}}{\Let{C}} &\equiv \zero,& \Comm{\Let{C}}{\Let{D}} &\equiv \zero, \\ 
		\Comm{\Let{A}}{\Let{C}} &\equiv \zero,& \Comm{\Let{B}}{\Let{D}} &\equiv \zero,& \Comm{\Let{A}}{\Let{D}} &\equiv \zero.
	\end{aligned}\right. \]
	\item Rectangle-triangle: \[
    \vdash_0 \left\{
    \begin{aligned}
	\Comm{\Let{A}}{\Let{F}} &\equiv \Let{AF},& \Comm{\Let{B}}{\Let{F}} &\equiv \Let{FB^\dagger},& \Comm{\Let{C}}{\Let{F}} &\equiv \zero,& \Comm{\Let{D}}{\Let{F}} &\equiv \zero, \\
	\Comm{\Let{B}}{\Let{G}} &\equiv \Let{BG},& \Comm{\Let{C}}{\Let{G}} &\equiv \Let{GC^\dagger},& \Comm{\Let{A}}{\Let{G}} &\equiv \zero,& \Comm{\Let{D}}{\Let{G}} &\equiv \zero, \\
	\Comm{\Let{C}}{\Let{H}} &\equiv \Let{CH},& \Comm{\Let{D}}{\Let{H}} &\equiv \Let{HD^\dagger},& \Comm{\Let{A}}{\Let{H}} &\equiv \zero,& \Comm{\Let{B}}{\Let{H}} &\equiv \zero, \\
	\Comm{\Let{A}}{\Let{P}} &\equiv \Let{AP},& \Comm{\Let{C}}{\Let{P}} &\equiv \Let{PC^\dagger},& \Comm{\Let{B}}{\Let{P}} &\equiv \zero,& \Comm{\Let{D}}{\Let{P}} &\equiv \zero, \\
	\Comm{\Let{B}}{\Let{Q}} &\equiv \Let{BQ},& \Comm{\Let{D}}{\Let{Q}} &\equiv \Let{QD^\dagger},& \Comm{\Let{A}}{\Let{Q}} &\equiv \zero,& \Comm{\Let{C}}{\Let{Q}} &\equiv \zero.
    \end{aligned}\right. \]
\end{itemize}
\end{itemize}
\end{lemma}

\begin{proof}
    Follows by inspection of the relators in \Cref{prop:cbdy:steinberg} and the subgroup memberships of the group elements.
\end{proof}

Because of \Cref{lemma:cbdy:reduce to steinberg,lemma:cbdy:present}, in order to prove \Cref{lemma:cbdy:padded}, we will just need to prove the following lemma:

\begin{lemma}[``Missing'' Steinberg relations]\label{lemma:cbdy:missing}
For every $P \in \cbdyP, Q \in \cbdyQ$, $\vdash_6 \Comm{\Let{P}}{\Let{Q}} \equiv \zero$.
Further, for every $Z \in \cbdyZ$, there exists a subgroup word $\alias{Z}$ of length $O(\rho^3)$ with $\eval(\alias{Z}) = \El{Z}$ such that the following holds:
For every $A \in \cbdyA$, $B \in \cbdyB$, $C \in \cbdyC$, $D \in \cbdyD$,
	$F \in \cbdyF$, $G \in \cbdyG$, $H \in \cbdyH$, $P \in \cbdyP$, $Q \in \cbdyQ$, and $Z, Z' \in \cbdyZ$,
\[
\vdash_{17} \left\{
\begin{aligned}
\alias{Z} \cdot \alias{Z'} &\equiv \alias{Z+Z'}, &&&& \\
\Comm{\Let{F}}{\Let{Q}} &\equiv \alias{FQ}, & \Comm{\Let{P}}{\Let{H}} &\equiv \alias{PH}, && \\
\Comm{\Let{F}}{\alias{Z}} &\equiv \zero, & \Comm{\Let{G}}{\alias{Z}} &\equiv \zero, & \Comm{\Let{H}}{\alias{Z}} &\equiv \zero, \\ \Comm{\Let{P}}{\alias{Z}} &\equiv \zero, & \Comm{\Let{Q}}{\alias{Z}} &\equiv \zero, && \\
\Comm{\Let{A}}{\alias{Z}} &\equiv \alias{AZ}, & \Comm{\Let{D}}{\alias{Z}} &\equiv \alias{ZD^\dagger},&& \\
\Comm{\Let{B}}{\alias{Z}} &\equiv \zero,& \Comm{\Let{C}}{\alias{Z}} &\equiv \zero,&&
\end{aligned}\right.
\]
\end{lemma}

\begin{proof}[Proof of \Cref{lemma:cbdy:padded} modulo \Cref{lemma:cbdy:missing}]
    Follows immediately since \Cref{lemma:cbdy:present,lemma:cbdy:missing} cover the conditions in \Cref{lemma:cbdy:padded}.
\end{proof}

\subsection{Lifted proofs}

Next, we give two statements which follow from what is essentially a ``lifting'' argument in \cite{OS25}.

\begin{lemma}[Lifted type-$\typP$ and -$\typQ$ elements commute]\label{lemma:cbdy:lifted P and Q commute}
    For every $F \in \cbdyF$, $G \in \cbdyG$, and $H \in \cbdyH$, \[
    \vdash_0 \Comm{\Let{FG}}{\Let{GH}} \equiv \zero. \]
\end{lemma}

 \begin{proof}
 	We begin with eight useful equations which are all constant-length in-subgroup relations (\Cref{lemma:cbdy:present}).
    The following four express type-$\typP$ and -$\typQ$ elements as commutators:
     \begin{align}
         \Let{FG}  &\equiv \Let{-G} \Let{F} \Let{G} \Let{-F}, \label{eq:a3:comm:alpha+beta:beta+gamma:1} \\
         \Let{GH} &\equiv \Let{-2H} \Let{\tfrac12 G} \Let{2H} \Let{-\tfrac12 G}, \label{eq:a3:comm:alpha+beta:beta+gamma:2} \\
         \Let{-FG} &\equiv \Let{\tfrac12 G} \Let{2F} \Let{-\tfrac12 G} \Let{-2F}, \label{eq:a3:comm:alpha+beta:beta+gamma:3} \\
         \Let{-GH} &\equiv \Let{-H} \Let{-G} \Let{H} \Let{G}. \label{eq:a3:comm:alpha+beta:beta+gamma:4} 
     \end{align}
     and the following four are commutator expressions:
     \begin{align}
         \Let{-F} \Let{\tfrac12 G} \Let{2F} &\equiv \Let{G} \Let{F} \Let{-\tfrac12 G}, \label{eq:a3:comm:alpha+beta:beta+gamma:5} \\
         \Let{2H} \Let{-\tfrac12 G} \Let{-H} &\equiv \Let{\tfrac12 G} \Let{H} \Let{-G}, \label{eq:a3:comm:alpha+beta:beta+gamma:6} \\
         \Let{G} \Let{-2H} \Let{G} &\equiv \Let{-H} \Let{2G} \Let{-H}, \label{eq:a3:comm:alpha+beta:beta+gamma:7} \\
         \Let{-G} \Let{-2F} \Let{-G} &\equiv \Let{-F} \Let{-2G} \Let{-F}. \label{eq:a3:comm:alpha+beta:beta+gamma:8}
     \end{align}
     The basic plan is now to write down $\Comm{\Let{FG}}{\Let{GH}}$ and expand using \Cref{eq:a3:comm:alpha+beta:beta+gamma:1,eq:a3:comm:alpha+beta:beta+gamma:2,eq:a3:comm:alpha+beta:beta+gamma:3,eq:a3:comm:alpha+beta:beta+gamma:4}.
     We proceed as follows using in-subgroup relations (\Cref{lemma:cbdy:present}) and these equations:
     \begin{align*}
         & \Comm{\Let{FG}}{\Let{GH}} \\
         &\equiv \Let{FG} \Let{GH} \Let{-FG} \Let{-GH} \\
         &\equiv \Let{FG} \Let{-2H} \Let{\tfrac12 G} \Let{2H} \Let{-\tfrac12 G} \Let{\tfrac12 G} \Let{2F} \Let{-\tfrac12 G} \Let{-2F}  \Let{-GH} \tag{\Cref{eq:a3:comm:alpha+beta:beta+gamma:2,eq:a3:comm:alpha+beta:beta+gamma:3}} \\
         &\equiv \Let{FG} \Let{-2H} \Let{\tfrac12 G} \Let{2H} \Let{2F} \Let{-\tfrac12 G} \Let{-2F} \Let{-GH}  \tag{canceling $\typG$'s} \\
         &\equiv \Let{-G} \Let{F} \Let{G} \Let{-F} \Let{-2H} \Let{\tfrac12 G} \Let{2H} \Let{2F} \Let{-\tfrac12 G} \Let{-2F}  \Let{-H} \Let{-G} \Let{H} \Let{G}  \tag{\Cref{eq:a3:comm:alpha+beta:beta+gamma:1,eq:a3:comm:alpha+beta:beta+gamma:4}}\\
         &\equiv \Let{-G} \Let{F} \Let{G} \Let{-2H} \Let{-F} \Let{\tfrac12 G} \Let{2F} \Let{2H} \Let{-\tfrac12 G} \Let{-H} \Let{-2F} \Let{-G} \Let{H}  \Let{G} \tag{commuting $\typF$'s and $\typH$'s} \\
         &\equiv \Let{-G} \Let{F} \Let{G} \Let{-2H} \Let{G} \Let{F} \Let{-\tfrac12 G} \Let{\tfrac12 G} \Let{H} \Let{-G} \Let{-2F} \Let{-G} \Let{H} \Let{G} \tag{\Cref{eq:a3:comm:alpha+beta:beta+gamma:5,eq:a3:comm:alpha+beta:beta+gamma:6}} \\
         &\equiv \Let{-G} \Let{F} \Let{G} \Let{-2H} \Let{G} \Let{F} \Let{H} \Let{-G} \Let{-2F} \Let{-G} \Let{H} \Let{G} \tag{canceling $\typG$'s} \\
         &\equiv \Let{-G} \Let{F} \Let{-H} \Let{2G} \Let{-H} \Let{F} \Let{H} \Let{-F} \Let{-2G} \Let{-F} \Let{H}  \Let{G} \tag{\Cref{eq:a3:comm:alpha+beta:beta+gamma:7,eq:a3:comm:alpha+beta:beta+gamma:8}} \\
         &\equiv \Let{-G} \Let{F} \Let{-H} \Let{2G}  \Let{-2G} \Let{-F} \Let{H}  \Let{G} \tag{commuting and cancelling $\typF$'s and $\typH$'s} \\
         &\equiv \Let{-G} \Let{F} \Let{-H}  \Let{-F} \Let{H}  \Let{G} \tag{cancelling $\typG$'s} \\
         &\equiv \Let{-G} \Let{G} \tag{commuting and cancelling $\typF$'s and $\typH$'s} \\
         &\equiv \zero, \tag{cancelling $\typG$'s}
     \end{align*}
     as desired.
 \end{proof}

\begin{lemma}[Interchange]\label{lemma:cbdy:interchange}
    For every $F \in \cbdyF$, $G \in \cbdyG$, and $H \in \cbdyH$, \[
    \vdash_0 \Comm{\Let{FG}}{\Let{H}} \equiv \Comm{\Let{F}}{\Let{GH}}. \]
\end{lemma}
\begin{proof}
    We calculate using in-subgroup relations (\Cref{lemma:cbdy:present}) and \Cref{lemma:cbdy:lifted P and Q commute}:
    \begin{align*}
         & \Let{F} \Let{GH} \\
         &\equiv \Let{F} \Let{G} \Let{H} \Let{-G} \Let{-H} \tag{commutator of $\typG$ and $\typH$} \\
         &\equiv \Let{FG} \Let{G} \Let{F} \Let{H} \Let{-G} \Let{-H} \tag{commutator of $\typF$ and $\typG$} \\
         &\equiv \Let{FG} \Let{G} \Let{H} \Let{F} \Let{-G} \Let{-H} \tag{commuting $\typF$ and $\typH$} \\
         &\equiv \Let{FG} \Let{G} \Let{H} \Let{-FG} \Let{-G} \Let{F} \Let{-H} \tag{commutator of $\typF$ and $\typG$} \\
         &\equiv \Let{FG} \Let{G} \Let{H} \Let{-FG} \Let{-G} \Let{-H} \Let{F} \tag{commuting $\typF$ and $\typH$} \\
         &\equiv \Let{FG} \Let{GH} \Let{H} \Let{G} \Let{-FG} \Let{-G} \Let{-H}\Let{F} \tag{commutator of $\typG$ and $\typH$} \\
         &\equiv \Let{FG} \Let{GH} \Let{H} \Let{-FG} \Let{-H} \Let{F} \tag{commuting $\typP$ and $\typG$ and cancelling $\typG$'s} \\
         &\equiv \Let{FG} \Let{H} \Let{-FG}  \Let{-H} \Let{GH} \Let{F} \tag{commuting $\typQ$ with $\typH$'s and \Cref{lemma:cbdy:lifted P and Q commute}} \\
         &\equiv \Comm{\Let{FG}}{\Let{H}} \Let{GH} \Let{F}
    \end{align*}
    as desired.
\end{proof}

\subsection{Decomposing wires into blocks}

Define
\begin{alignat*}{3}
    \InitsA &\coloneqq \rangeII{0}{\nA-m}\ &&& \subset \WiresA, \\
    \InitsB &\coloneqq \rangeII{\nA}{\nA+\nB-m}\ &&& \subset \WiresB, \\
    \InitsC &\coloneqq \rangeII{\nA+\nB}{\nA+\nB+\nC-m}\ &&& \subset \WiresC, \\
    \InitsD &\coloneqq \rangeII{\nA+\nB+\nC}{\nA+\nB+\nC+\nD-m}\ &&& \subset \WiresD.
\end{alignat*}
We typically use the variables $\initA$, $\initB$, $\initC$, and $\initD$ to refer to elements of these subsets, respectively.

Recall $a$ and $d$ are multiples of $m$.
Define
\begin{equation}\label{eq:cbdy:A and D partition}
\begin{aligned}
&\partA{0} \coloneqq 0, \\
&\partA{1} \coloneqq m, \\
&\ldots, \\
&\partA{\nA/m-1} \coloneqq (\nA/m -1) m
\end{aligned}
\in \InitsA \quad\text{and}\quad
\begin{aligned}
&\partD{0} \coloneqq \nA+\nB+\nC, \\
&\partD{1} \coloneqq \nA+\nB+\nC+m, \\
&\ldots, \\
&\partD{{\nD/m-1}} \coloneqq \nA+\nB+\nC+ (\nD/m -1) m.
\end{aligned}
\in \InitsD.
\end{equation}
Similarly, define
\begin{equation}\label{eq:cbdy:C and B partition}
\begin{aligned}
    \partC{0} &\coloneqq \nA+\nB, \\
    \partC{1} &\coloneqq \nA+\nB+(\nC-\floor{ \nC/m }  m), \\
    \partC{2} &\coloneqq \nA+\nB+(\nC-\floor{ \nC/m }  m)+m, \\
    &\ldots, \\
    \partC{\floor{ \nC/m }} &\coloneqq \nA+\nB+\nC-m.
\end{aligned}
\in \InitsC \quad\text{and}\quad
\begin{aligned}
    \partB{0} &\coloneqq \nA, \\
    \partB{1} &\coloneqq \nA+m, \\
    &{}\ldots, \\
    \partB{\floor{ \nB/m }-1} &\coloneqq \nA+(\floor{ \nB/m } -1) m, \\
    \partB{\floor{ \nB/m }} &\coloneqq \nA+\nB-m.
\end{aligned}
\in \InitsB.
\end{equation}
One unfortunate issue is that $\nB$ and $\nC$ need not be multiples of $m$,
and therefore it may be impossible to partition $\WiresB$ and $\WiresC$ into blocks of size $m$.
Therefore, we divide them into $\floor{\nC/m}+1$ and $\floor{\nB/m}+1$ blocks of size $m$.

\subsection{Types-\texorpdfstring{$\typH$}{H} and -\texorpdfstring{$\typQ$}{Q} basic block matrices}

Recall the definition of a basic matrix (\Cref{def:diam:basic})
and of $s$-bounded basic block matrices
$\blockF{\powF}{\initA}{\initB}{M}$ and $\blockG{\powG}{\initB}{\initC}{M}$ (\Cref{def:diam:block:F})
and $\blockP{\powP}{\initA}{\initC}{M}$ (\Cref{def:diam:block:P})
which are members of $\cbdyF$, $\cbdyG$, and $\cbdyP$, respectively.
We also define the $s$-bounded basic block matrices:
\begin{definition}
We define:
\begin{itemize}
    \item $\blockH{\powH}{\initC}{\initD}{M} \in \cbdyH$
    (for $\initC \in \InitsC$, $\initD \in \InitsD$,
    $\powH \in \rangeII{0}{\initD-\initC-s}$, and $M \in \basic{s}$),
    \item $\blockQ{\powQ}{\initB}{\initD}{M} \in \cbdyQ$
    (for $\initB \in \InitsB$, $\initD \in \InitsD$,
    $\powQ \in \rangeII{0}{\initD-\initB-s}$, and $M \in \basic{s}$),
    \item $\blockZ{\powZ}{\initA}{\initD}{M} \in \cbdyZ$
    (for $\initA \in \InitsA$, $\initD \in \InitsD$,
    $\powZ \in \rangeII{0}{\initD-\initA-s}$, and $M \in \basic{s}$)
\end{itemize}
analogously to \Cref{def:diam:block:F,def:diam:block:P}.
\end{definition}

The following useful corollary of \Cref{lemma:cbdy:interchange} allows us to pass between type-$\typP$/$\typH$ and type-$\typF$/$\typQ$ basic block element commutators.

\begin{lemma}[Block interchange]\label{lemma:cbdy:block:interchange}
    Let $s,s',s'' \in \N$,
    $\initA \in \InitsA$,
    $\initB \in \InitsB$,
    $\initC \in \InitsC$,
    $\initD \in \InitsD$,
    $\powF \in \rangeII{0}{\initB-\initA-s}$,
    $\powG \in \rangeII{0}{\initC-\initB-s'}$,
    $\powH \in \rangeII{0}{\initD-\initC-s''}$,
    $M \in \basic{s}$,
    $M' \in \basic{s'}$,
    $M'' \in \basic{s''}$.
    Then: \[
    \vdash_0 \Comm{\Let{\blockP{\powF+\powG}{\initA}{\initC}{MM'}}}{\Let{\blockH{\powH}{\initC}{\initD}{M''}}} \equiv \Comm{\Let{\blockF{\powF}{\initA}{\initB}{M}}}{\Let{\blockQ{\powG+\powH}{\initB}{\initD}{M'M''}}}. \]
\end{lemma}
The notation in the lemma is justified as $MM' \in \basic{s+s'}$, $\powF+\powH \le \initC-\initA-(s+s')$,
$M'M'' \in \basic{s'+s''}$, and $\powG+\powH \le \initD-\initB-(s'+s'')$.
\begin{proof}
Follows from the matrix calculations
\begin{align*}
\blockP{\powF+\powG}{\initA}{\initC}{MM'} &= \blockF{\powF}{\initA}{\initB}{M} \cdot \blockG{\powG}{\initB}{\initC}{M'}, \\
\blockQ{\powG+\powH}{\initB}{\initD}{M'M''} &= \blockG{\powG}{\initB}{\initC}{M'} \cdot \blockH{\powH}{\initC}{\initD}{M''}
\end{align*}
and \Cref{lemma:cbdy:interchange}.
\end{proof}

\subsection{Commuting type-\texorpdfstring{$\typP$}{P} and -\texorpdfstring{$\typQ$}{Q} matrices}

We next show that type-$\typP$ and -$\typQ$ elements commute,
beginning with basic block elements and then using decompositions to lift to all elements.

\begin{lemma}[Type-$\typP$ and -$\typQ$ basic block elements commute]\label{lemma:cbdy:block:P and Q commute}
    Let $\initA \in \InitsA$,
    $\initB \in \InitsB$,
    $\initC \in \InitsC$,
    $\initD \in \InitsD$,
    $\powP \in \rangeII{0}{\initC-\initA-2m}$,
    $\powQ \in \rangeII{0}{\initD-\initB-2m}$,
    and $M,N \in \basic{2m}$.
    Then \[
    \vdash_0 \Comm{\Let{\blockP{\powP}{\initA}{\initC}{M}}}{\Let{\blockQ{\powQ}{\initB}{\initD}{N}}} \equiv \zero. \]
\end{lemma}
\begin{proof}
	As in the proof of \Cref{lemma:diam:P as F and G}, we can fix
    arbitrary $\initB' \in \rangeII{\nA+m}{\nA+\nB-m}$ with $\IntB{\initB'} \cap \IntB{\initB'} = \emptyset$
    and $\initC' \in \rangeII{\nA+\nB+m}{\nA+\nB+\nC-m}$ with $\IntC{\initC'} \cap \IntC{\initC} = \emptyset$,
    and fix arbitrary $\powF,\powG \in \N$ such that
    \begin{align*}
        \powF &\in \rangeII{0}{\initB'-\initA-2m}, \\
        \powG &\in \rangeII{0}{\initC-\initB'}, \\
        \powP &= \powF+\powG,
    \end{align*}
    and $\powH,\powG' \in \N$ such that
    \begin{align*}
    	\powH &\in \rangeII{0}{\initD-\initC'-2m}, \\
        \powG' &\in \rangeII{0}{\initC'-\initB}, \\
        \powQ &= \powG'+\powH.
    \end{align*}

    Now, define \[
    	F \coloneqq \blockF{\powF}{\initA}{\initB'}{M},
		G \coloneqq \blockG{\powG}{\initB'}{\initC}{I_m}+\blockG{\powG'}{\initB}{\initC'}{I_m},
    	\text{ and }
		H \coloneqq \blockH{\powH}{\initC'}{\initD}{N}, \]
	so that \[
		FG = \blockP{\powF+\powG}{\initA}{\initC}{M} = \blockP{\powP}{\initA}{\initC}{M} \text{ and }
        GH = \blockQ{\powG'+\powH}{\initB}{\initD}{N} = \blockQ{\powQ}{\initB}{\initD}{N}. \]
    and then we can apply \Cref{lemma:cbdy:lifted P and Q commute}.
\end{proof}

We have the following analogue of \Cref{claim:diam:matrix:block:decompose P}:

\begin{claim}[Decomposing type-$\typP$ and -$\typQ$ block matrices]\label{claim:cbdy:matrix:block:decompose P and Q}
    Let $\initA \in \InitsA$ and $\initC \in \InitsC$.
    Suppose $P \in \cbdyP$ is supported only on the $\IntA{\initA} \times \IntC{\initC}$ entries.
    Then $P$ is a sum of $\floor{ (\initA-\initC) / m } = O(\rho)$ $2m$-bounded basic block matrices 
    (i.e., matrices of the form $\blockP{\powP}{\initA}{\initC}{M}$
    for some $\powP \in \rangeII{0}{\initC-\initA-2m}$, $M \in \basic{2m}$).
    Similarly, for $\initB \in \InitsB$ and $\initD \in \InitsD$
    and $Q \in \cbdyQ$ supported only on the $\IntB{\initB} \times \IntD{\initD}$ entries,
    $Q$ is a sum of $\floor{ (\initD-\initB) / m } = O(\rho)$ $2m$-bounded basic block matrices.
\end{claim}

In turn, we get:

\begin{claim}[Decomposing type-$\typP$ and -$\typQ$ matrices]\label{claim:cbdy:matrix:decompose P and Q}
    Every $P \in \cbdyP$ can be decomposed as \[
    P = \sum_{\alpha=0}^{\nA/m-1} \sum_{\gamma=0}^{\floor{ \nC/m} } P^{(\alpha,\gamma)} \]
    where $P^{(\alpha,\gamma)} \in \cbdyP$ is supported on
    $\IntPartA{\alpha} \times \IntPartC{\gamma}$.
    Similarly, every $Q \in \cbdyQ$ can be decomposed as $Q = \sum_{\beta=0}^{\floor{ \nB/m } -1} \sum_{\delta=0}^{\nD/m-1} Q^{(\beta,\delta)}$
    where $Q^{(\beta,\delta)} \in \cbdyQ$ is supported on
    $\IntPartB{\beta} \times \IntPartD{\delta}$.
\end{claim}

This is identical the same as \Cref{cor:diam:decompose P} except we have to account for
the fact that the size of the ``inner'' wireset ($\nB$ for type-$\typQ$, and $\nC$ for type-$\typP$)
may not being a multiple of $m$.

\begin{lemma}[Type-$\typP$ and -$\typQ$ elements commute]\label{lemma:cbdy:P and Q commute}
    For every $P \in \cbdyP$ and $Q \in \cbdyQ$, \[
    \vdash_6 \Comm{\Let{P}}{\Let{Q}} \equiv \zero. \]
\end{lemma}

\begin{proof}
    By \Cref{claim:cbdy:matrix:decompose P and Q,claim:cbdy:matrix:block:decompose P and Q}, $P$ is a sum of at $O(\rho^3)$ $2m$-bounded basic block matrices
    and $Q$ is a sum of at most $O(\rho^3)$ $2m$-bounded basic block matrices.
    Hence by in-subgroup relations, $\Let{P}$ is a product of $O(\rho^3)$ $2m$-bounded basic block elements
    and $\Let{Q}$ is a product of at most $O(\rho^3)$ $2m$-bounded basic block elements.
    Finally, by \Cref{lemma:cbdy:block:P and Q commute}, each pair of $2m$-bounded basic block elements commutes with cost $O(1)$.
\end{proof}

\subsection{Defining type-\texorpdfstring{$\typZ$}{Z} basic block aliases}

Next, we define alias elements corresponding to type-$\typZ$ basic blocks
and prove that these elements can be constructed as commutators of appropriate pairs
of type-$\typP$/$\typH$ and type-$\typF$/$\typQ$ basic blocks.

Recall the definition of packing (\Cref{eq:pack}).

\begin{definition}[Type-$\typZ$ basic block aliases]\label{def:cbdy:block:Z}
    Fix $c^* \coloneqq \nA+\nB+\nC+m \in \InitsC$.
    For $\initA \in \InitsA$,
    $\initD \in \InitsD$,
    and $\powZ \in \rangeII{0}{\initD-\initA-4m}$,
    let $(\powP,\powH) \coloneqq \pack{\powZ}{c^*-\initA-4m}{\initD-c^*}$.
    Then for $M \in \basic{4m}$,
    define \[
    \ElBlockZ{\powZ}{\initA}{\initD}{M} \coloneqq \Comm{\Let{\blockP{\powP}{\initA}{c^*}{M}}}{\Let{\blockH{\powH}{c^*}{\initD}{I_m}}}. \qedhere \]
\end{definition}

We have the following normalization lemma for packing:

\begin{proposition}\label{prop:cbdy:normalization via packing}
    For every $\powP, \powH, \mF, \mG, \mH \in \N$, $\powP \le \mF+\mG$, and $\powH \le \mH$, the following holds.
    Let $\powP^{(0)} \coloneqq \powP$, $\powH^{(0)} \coloneqq \powH$,
    $\powZ \coloneqq \powP+\powH$, $(\powP^*,\powH^*) \coloneqq \pack{\powZ}{\mF+\mG}{\mH}$, and $K \coloneqq \ceil{ \mG/\mF }$.
    For $k = 1,\ldots,K$, iterate the following:
    \begin{align}
        (\powF^{(k)}, \powG^{(k)}) &\coloneqq \pack{\powP^{(k-1)}}{\mF}{\mG}, \label{eq:norm-z:fg} \\
        ((\powG')^{(k)}, \powH^{(k)}) &\coloneqq \pack{\powG^{(k)}+\powH^{(k-1)}}{\mG}{\mH}, \label{eq:norm-z:gh} \\
        \powP^{(k)} &\coloneqq \powF^{(k)}+(\powG')^{(k)} \label{eq:norm-z:p}.
    \end{align}
    Then $(\powP^{(K)}, \powH^{(K)}) = (\powP^*, \powH^*)$.
\end{proposition}

\begin{proof}
    We have the invariants $\powP^{(k)}+\powH^{(k)} = \powZ$
    (because $\powP^{(k-1)}+\powH^{(k-1)} = \powF^{(k)}+\powG^{(k)}+\powH^{(k-1)} = \powF^{(k)}+(\powG')^{(k)}+\powH^{(k)} = \powP^{(k)}+\powH^{(k)} $
    by \Cref{eq:norm-z:p,eq:norm-z:gh,eq:norm-z:fg}, respectively),
    $\powP^{(k)} \le \mF+\mG$ (since $\powF^{(k)} \le \mF$ and $(\powG')^{(k)} \le \mG$ by \Cref{eq:norm-z:fg,eq:norm-z:gh}, respectively),
    and $\powH^{(k)} \le \mH$.
    By definition of $\pack{\powZ}{\mF+\mG}{\mH}$,
    these inequalities guarantee that $\powP^{(k)} \le \powP^*$ and $\powH^{(k)} \ge \powH^*$.

    Now, we claim that if $\powP^{(k-1)} \ge \mF$,
    then $\powP^{(k)} = \powP^*$.
    This is because $\powF^{(k)} = \mF$,
    $\powG^{(k)} = \powP^{(k-1)}-\mF$,
    $\powG^{(k)}+\powH^{(k-1)} = \powP^{(k-1)}-\mF+\powH^{(k-1)} = \powZ-\mF$,
    and so $(\powG')^{(k)} = \min\{\powZ-\mF, \mG\}$.
    In the former case, $\powP^{(k)} = \powF^{(k)}+(\powG')^{(k)} = \powZ$,
    while in the latter case, $\powP^{(k)} = \powF^{(k)}+(\powG')^{(k)} = \mF+\mG$.
    In either case, since $\powP^{(k)} \le \powP^*$, we conclude $\powP^{(k)} = \powP^*$ by definition of $\powP^*$.

    So, it suffices to check that at every step $k$,
    either $\powP^{(k)} = \powP^*$,
    $\powP^{(k)} \ge \mF$,
    or $\powP^{(k)}-\powP^{(k-1)} = \mG$.
    Indeed, observe that \[
    \powP^{(k)}-\powP^{(k-1)} = \powH^{(k-1)}-\powH^{(k)} = \powH^{(k-1)}-((\powG^{(k)}+\powH^{(k-1)})-(\powG')^{(k)}) = (\powG')^{(k)}-\powG^{(k)}. \]
    Since $\powP^{(k-1)} < \mF$ by assumption,
    then $\powF^{(k)} = \powP^{(k-1)}$
    and $\powG^{(k)} = 0$,
    and so $(\powG')^{(k)} = \min\{\powH^{(k-1)},\mG\}$.
    Hence either $(\powG')^{(k)} = \powH^{(k-1)}$,
    in which case $\powP^{(k)} = \powF^{(k)}+(\powG')^{(k)} = \powP^{(k-1)}+\powH^{(k-1)} = \powZ = \powP^*$,
    or $(\powG')^{(k)} = \mG$,
    in which case $\powP^{(k)}-\powP^{(k-1)} = \mG$.
\end{proof}

\begin{lemma}[Normalizing the type-$\typZ$ exponent]\label{lemma:cbdy:block:normalize z exponent}
    Let $c^*$ be as in \Cref{def:cbdy:block:Z}.
    For every
    $\initA \in \InitsA$,
    $\initD \in \InitsD$,
    $\powP \in \rangeII{0}{c^*-\initA-4m}$,
    $\powH \in \rangeII{0}{\initD-c^*}$,
    and $M \in \basic{4m}$, \[
    \vdash_1 \ElBlockZ{\powP+\powH}{\initA}{\initD}{M} \equiv \Comm{\Let{\blockP{\powP}{\initA}{c^*}{M}}}{\Let{\blockH{\powH}{c^*}{\initD}{I_m}}}. \]
\end{lemma}

\begin{proof}
    Pick $\initB \coloneqq \nA+3m$ and run the process described in \Cref{prop:cbdy:normalization via packing}
    with $\mF = \initB-\initA-3m$, $\mG = c^*-\initB-m$, and $\mH = \initD-c^*$.
    We have iteratively:
    \begin{align*}
        \Comm{\Let{\blockP{\powP^{(k-1)}}{\initA}{c^*}{M}}}{\Let{\blockH{\powH^{(k-1)}}{c^*}{\initD}{I}}}
        &= \Comm{\Let{\blockP{\powF^{(k)}+\powG^{(k)}}{\initA}{c^*}{M \cdot I_m}}}{\Let{\blockH{\powH^{(k-1)}}{c^*}{\initD}{I_m}}} \tag{\Cref{eq:norm-z:fg}} \\
        &\equiv \Comm{\Let{\blockF{\powF^{(k)}}{\initA}{\initB}{M}}}{\Let{\blockQ{\powG^{(k)}+\powH^{(k-1)}}{\initB}{\initD}{I_m \cdot I_m}}} \tag{\Cref{lemma:cbdy:block:interchange}} \\
        &= \Comm{\Let{\blockF{\powF^{(k)}}{\initA}{\initB}{M}}}{\Let{\blockQ{(\powG')^{(k)}+\powH^{(k)}}{\initB}{\initD}{I_m \cdot I_m}}} \tag{\Cref{eq:norm-z:gh}} \\
        &\equiv \Comm{\Let{\blockP{\powF^{(k)}+(\powG')^{(k)}}{\initA}{c^*}{M \cdot I_m}}}{\Let{\blockH{\powH^{(k)}}{c^*}{\initD}{M}}} \tag{\Cref{lemma:cbdy:block:interchange}} \\
        &= \Comm{\Let{\blockP{\powP^{(k)}}{\initA}{c^*}{M}}}{\Let{\blockH{\powH^{(k)}}{c^*}{\initD}{I_m}}}, \tag{\Cref{eq:norm-z:p}}
    \end{align*}
    and by the conclusion of \Cref{prop:cbdy:normalization via packing}, \[
    \Comm{\Let{\blockP{\powP^{(K)}}{\initA}{c^*}{M}}}{\Let{\blockH{\powH^{(K)}}{c^*}{\initD}{I_m}}}
    = \Comm{\Let{\blockP{\powP^*}{\initA}{c^*}{M}}}{\Let{\blockH{\powH^*}{c^*}{\initD}{I_m}}}
    = \ElBlockZ{\powP^*+\powH^*}{\initA}{\initC}{M}, \]
    as desired.
    Also, $\mG \le c^* \le n$ while $\mF \ge (\nA+3m) - (\nA-m) - 3m = m$ and so $K \le \ceil{ n/m } = \rho$.
\end{proof}

\begin{lemma}[Commutators of compatible $\typP$-$\typH$ and $\typF$-$\typQ$ pairs are type-$\typZ$ basic block aliases]\label{lemma:cbdy:block:establish Z}
    For every $s, s' \in \N$ with $s+s' = 4m$,
    $\initA \in \InitsA$,
    $\initD \in \InitsD$,
    $M \in \basic{s}$ and $N \in \basic{s'}$,
    \begin{itemize}
        \item for every $\initC \in \InitsC$,
        $\powP \in \rangeII{0}{\initC-\initA-s}$,
        and $\powH \in \rangeII{0}{\initD-\initC-s'}$, \[
        \vdash_1 \Comm{\Let{\blockP{\powP}{\initA}{\initC}{M}}}{\Let{\blockH{\powH}{\initC}{\initD}{N}}} \equiv \ElBlockZ{\powP+\powH}{\initA}{\initD}{MN}. \]
        
        \item for every $\initB \in \InitsB$,
        $\powF \in \rangeII{0}{\initB-\initA-s}$,
        and $\powQ \in \rangeII{0}{\initD-\initB-s'}$, \[
        \vdash_1 \Comm{\Let{\blockF{\powF}{\initA}{\initB}{M}}}{\Let{\blockQ{\powQ}{\initB}{\initD}{N}}} \equiv \ElBlockZ{\powF+\powQ}{\initA}{\initD}{MN}. \]
    \end{itemize}
\end{lemma}
Note that e.g. in the first bullet, if $\initC = \nA+\nB+\nC-m$, $\initD = \nA+\nB+\nC$, and $s' > m$, the interval $\rangeII{0}{\initD-\initC-s'}$ that contains $\powH$ is empty.
See the next statement for a nonvacuous statement which we need in this ``corner case''.

\begin{proof}
    Let $c^*$ be as in \Cref{def:cbdy:block:Z}.
    We prove the items in sequence.
    \begin{itemize}
    \item
    Fix $\initB \coloneqq \nA+\nB+m$.
    Let $(\powF,\powG) \coloneqq \pack{\powP}{\initB-\initA-s}{\initC-\initB}$,
    and let $(\powG',\powH') \coloneqq \pack{\powG+\powH}{c^*-\initB-s'}{\initD-c^*}$,
    $\powP' \coloneqq \powF+\powG'$.
    Applying \Cref{lemma:cbdy:block:interchange} twice:
    \begin{multline*}
        \Comm{\Let{\blockP{\powP}{\initA}{\initC}{M}}}{\Let{\blockH{\powH}{\initC}{\initD}{N}}} = 
        \Comm{\Let{\blockP{\powF+\powG}{\initA}{\initC}{M \cdot I_m}}}{\Let{\blockH{\powH}{\initC}{\initD}{N}}} 
        \equiv \Comm{\Let{\blockF{\powF}{\initA}{\initB}{M}}}{\Let{\blockQ{\powG+\powH}{\initB}{\initD}{I_m \cdot N}}} \\
        = \Comm{\Let{\blockF{\powF}{\initA}{\initB}{M}}}{\Let{\blockQ{\powG'+\powH'}{\initB}{\initD}{N \cdot I_m}}}
        \equiv \Comm{\Let{\blockP{\powF+\powG'}{\initA}{c^*}{MN}}}{\Let{\blockH{\powH'}{c^*}{\initD}{I_m}}}
        = \Comm{\Let{\blockP{\powP'}{\initA}{c^*}{MN}}}{\Let{\blockH{\powH'}{c^*}{\initD}{I_m}}}.
    \end{multline*}
    Finally, we apply the previous lemma (\Cref{lemma:cbdy:block:normalize z exponent}).
    
    \item Let $(g,h) \coloneqq \pack{\powQ}{c^*-\initB-2m}{\initD-c^*}$.
    Hence by \Cref{lemma:cbdy:block:interchange} again:
    \begin{multline*}
    \Comm{\Let{\blockF{\powF}{\initA}{\initB}{M}}}{\Let{\blockQ{\powQ}{\initB}{\initD}{N}}}
    = \Comm{\Let{\blockF{\powF}{\initA}{\initB}{M}}}{\Let{\blockQ{\powG+\powH}{\initB}{\initD}{N}}} \\
    \equiv \Comm{\Let{\blockP{\powF+\powG}{\initA}{c^*}{MN}}}{\Let{\blockH{\powH}{c^*}{\initD}{I_m}}}
    = \ElBlockZ{\powF+\powQ}{\initA}{\initD}{MN}
    \end{multline*}
    by the first item. \qedhere
    \end{itemize}
\end{proof}

\begin{corollary}[Type-$\typZ$ basic block aliases as commutators of type-$\typF$ and -$\typQ$ elements]\label{lemma:cbdy:block:Z as F and Q}
    For every $\initA \in \InitsA$,
    $\initD \in \InitsD$,
    $\powZ \in \rangeII{0}{\initD-\initA-4m}$,
    and $M \in \basic{4m}$,
    there exist $F \in \cbdyF$ and $Q \in \cbdyQ$ such that: \[
    \vdash_1 \ElBlockZ{\powZ}{\initA}{\initD}{M} \equiv \Comm{\Let{F}}{\Let{Q}}. \]
\end{corollary}

\begin{proof}
    Fix $\initB \coloneqq \nA+m$ and $(\powF,\powQ) \coloneqq \pack{\powZ}{\initB-\initA}{\initD-\initB-4m}$
    and use the previous lemma (\Cref{lemma:cbdy:block:establish Z}) to get \[
    \Comm{\Let{\blockF{\powF}{\initA}{\initB}{I_m}}}{\Let{\blockQ{\powQ}{\initB}{\initD}{M}}} \equiv \ElBlockZ{\powZ}{\initA}{\initD}{M}. \]
    (Note that this is well-defined because $d - b - 4m \ge (\nA + \nB + \nC) - (\nA + m) - 4m = \nB + \nC - 5m \ge 3m \ge 0$.)
\end{proof}

\begin{lemma}[Corner case]\label{lemma:cbdy:block:corner}
    We have:
    \begin{itemize}
        \item For every $\initA \in \InitsA$,
        $M \in \basic{2m}$, $N \in \basic{m}$,
        $\initC \coloneqq \nA+\nB+\nC-m$, $\initD \coloneqq \nA+\nB+\nC$,
        $\powP \in \rangeII{0}{\initC-\initA-2m}$, \[
        \vdash_0 \Comm{\Let{\blockP{\powP}{\initA}{\initC}{M}}}{\Let{\blockH{0}{\initC}{\initD}{N}}} \equiv \ElBlockZ{\powP}{\initA}{\initD}{MN} . \]
        
        \item For every $\initD \in \InitsD$,
        $M \in \basic{m}$, $N \in \basic{2m}$,
        $\initA \coloneqq \nA-m$, $\initB \coloneqq \nA$,
        $\powQ \in \rangeII{0}{\initD-\initB-2m}$, \[
        \vdash_0 \Comm{\Let{\blockF{0}{\initA}{\initB}{M}}}{\Let{\blockQ{\powQ}{\initB}{\initD}{N}}} \equiv \ElBlockZ{\powQ}{\initA}{\initD}{MN}. \]
    \end{itemize}
\end{lemma}

\begin{proof}
    We prove the first equality as the second is dual.
    Let $\initB \coloneqq \nA+m$.
    Let $(\powF,\powG) \coloneqq \pack{\powP}{\initB-\initA-2m}{\initC-\initB}$.
    By \Cref{lemma:cbdy:block:interchange,lemma:cbdy:block:establish Z}:
    \begin{multline*}
    \Comm{\Let{\blockP{\powP}{\initA}{\initC}{M}}}{\Let{\blockH{0}{\initC}{\initD}{N}}}
    = \Comm{\Let{\blockP{\powF+\powG}{\initA}{\initC}{M \cdot I_m}}}{\Let{\blockH{0}{\initC}{\initD}{N}}} \\
    \equiv \Comm{\Let{\blockF{\powF}{\initA}{\initB}{M}}}{\Let{\blockQ{\powG}{\initB}{\initD}{I_m \cdot N}}}
    = \ElBlockZ{\powF+\powG}{\initA}{\initD}{MN}. \qedhere
    \end{multline*}
\end{proof}

\begin{lemma}[Incompatible $\typP$-$\typH$ and $\typF$-$\typQ$ pairs commute]\label{lemma:cbdy:block:incompatible}
For every $\initA \in \InitsA$ and $\initD \in \InitsD$,
$M \in \basic{s}$ and $N \in \basic{s'}$:
    \begin{itemize}
        \item for every $\initC, \initC' \in \InitsC$
        such that $\IntC{\initC} \cap \IntC{\initC'} = \emptyset$,
        $\powP \in \rangeII{0}{\initC-\initA-s}$,
        and $\powH \in \rangeII{0}{\initD-\initC'-s'}$, \[
        \vdash_0 \Comm{\Let{\blockP{\powP}{\initA}{\initC}{M}}}{\Let{\blockH{\powH}{\initC'}{\initD}{N}}} \equiv \zero. \]
        
        \item for every $\initB, \initB' \in \InitsB$
        such that $\IntB{\initB'} \cap \IntB{\initB'} = \emptyset$,
        $\powF \in \rangeII{0}{\initB-\initA-s}$,
        and $\powQ \in \rangeII{0}{\initD-\initB'-s'}$, \[
        \vdash_0 \Comm{\Let{\blockF{\powF}{\initA}{\initB}{M}}}{\Let{\blockQ{\powQ}{\initB'}{\initD}{N}}} \equiv \zero. \]
    \end{itemize}
\end{lemma}

\begin{proof}
    For the first item, pick $\initB \coloneqq \initA+m$
    and let $(\powF,\powG) \coloneqq \pack{\powP}{\initB-\initA-s}{\initC-\initB}$.
    Then
    $\blockP{\powP}{\initA}{\initC}{M} = \blockF{\powF}{\initA}{\initB}{M} \cdot \blockG{\powG}{\initB}{\initC}{I_m}$
    and so \[
    \Let{\blockP{\powP}{\initA}{\initC}{M}} \equiv \Comm{\Let{\blockF{\powF}{\initA}{\initB}{M}}}{\Let{\blockG{\powG}{\initB}{\initC}{I_m}}}. \]
    But $\Let{\blockH{\powH}{\initC'}{\initD}{N}}$
    commutes with both these elements by constant-length in-subgroup relations (\Cref{lemma:cbdy:present}):
    The first since all type-$\typF$ and -$\typH$ elements commute,
    and the second because \[
    \Comm{\Let{\blockG{\powG}{\initB}{\initC}{I_m}}}{\Let{\blockH{\powH}{\initC'}{\initD}{N}}} \equiv \Let{\blockG{\powG}{\initB}{\initC}{I_m} \cdot \blockH{\powH}{\initC'}{\initD}{N}} \equiv \Let{0} \equiv \zero \]
    by the disjointness assumption.
    The second item is similar.
\end{proof}

\subsection{Remaining relations for type-\texorpdfstring{$\typZ$}{Z} basic block aliases}

We now prove all remaining relations which show that type-$\typZ$ basic block aliases commute with other types of elements.
These proofs essentially follow from \cite{OS25} but we provide them for completeness.

\begin{lemma}[Type-$\typB$ elements commute with type-$\typZ$ basic block aliases]\label{lemma:cbdy:block:B and Z commute}
	For every $\initA \in \InitsA$,
    $\initD \in \InitsD$,
	$\powZ \in \rangeII{0}{\initD-\initA-4m}$,
    $M \in \basic{4m}$,
    and $B \in \cbdyB$,
     \[ \vdash_0 \Comm{\Let{B}}{\ElBlockZ{\powZ}{\initA}{\initD}{M}} \equiv \zero. \]
\end{lemma}

\begin{proof}
    By definition (\Cref{def:cbdy:block:Z}), $\ElBlockZ{\powZ}{\initA}{\initD}{M} = \Comm{\Let{P}}{\Let{H}}$
    for some $P \in \cbdyP$ and $H \in \cbdyH$.
    Both parts commute with $\Let{B}$ by constant-length in-subgroup relations (\Cref{lemma:cbdy:present}).
\end{proof}

\begin{lemma}[Type-$\typC$ elements commute with type-$\typZ$ basic block aliases]\label{lemma:cbdy:block:C and Z commute}
	For every $\initA \in \InitsA$,
    $\initD \in \InitsD$,
	$\powZ \in \rangeII{0}{\initD-\initA-4m}$,
    $M \in \basic{4m}$,
    and $C \in \cbdyC$, \[ \vdash_1 \Comm{\Let{C}}{\ElBlockZ{\powZ}{\initA}{\initD}{M}} \equiv \zero. \]
\end{lemma}

\begin{proof}
    By \Cref{lemma:cbdy:block:Z as F and Q}, $\vdash_1 \ElBlockZ{\powZ}{\initA}{\initD}{M} \equiv \Comm{\Let{F}}{\Let{Q}}$
    for some $F \in \cbdyF$ and $Q \in \cbdyQ$.
    Both parts commute with $\Let{C}$ by constant-length in-subgroup relations (\Cref{lemma:cbdy:present}).
\end{proof}

\begin{lemma}[Type-$\typF$ elements commute with type-$\typZ$ basic block aliases]\label{lemma:cbdy:block:F and Z commute}
	For every $\initA \in \InitsA$,
    $\initD \in \InitsD$,
	$\powZ \in \rangeII{0}{\initD-\initA-4m}$,
    $M \in \basic{4m}$,
    and $F \in \cbdyF$,
     \[ \vdash_0 \Comm{\Let{F}}{\ElBlockZ{\powZ}{\initA}{\initD}{M}} \equiv \zero. \]
\end{lemma}

\begin{proof}
    By definition (\Cref{def:cbdy:block:Z}), $\ElBlockZ{\powZ}{\initA}{\initD}{M} = \Comm{\Let{P}}{\Let{H}}$
    for some $P \in \cbdyP$ and $H \in \cbdyH$.
    Both parts commute with $\Let{F}$ by constant-length in-subgroup relations (\Cref{lemma:cbdy:present}).
\end{proof}

\begin{lemma}[Type-$\typH$ elements commute with type-$\typZ$ basic block aliases]\label{lemma:cbdy:block:H and Z commute}
	For every $\initA \in \InitsA$,
    $\initD \in \InitsD$,
	$\powZ \in \rangeII{0}{\initD-\initA-4m}$,
    $M \in \basic{4m}$,
    and $H \in \cbdyH$, \[
    \vdash_1 \Comm{\Let{H}}{\ElBlockZ{\powZ}{\initA}{\initD}{M}} \equiv \zero. \]
\end{lemma}

\begin{proof}
    By \Cref{lemma:cbdy:block:Z as F and Q}, $\vdash_1 \ElBlockZ{\powZ}{\initA}{\initD}{M} \equiv \Comm{\Let{F}}{\Let{Q}}$
    for some $F \in \cbdyF$ and $Q \in \cbdyQ$.
    Both parts commute with $\Let{H}$ by constant-length in-subgroup relations (\Cref{lemma:cbdy:present}).
\end{proof}

\begin{lemma}[Type-$\typG$ elements commute with type-$\typZ$ basic block aliases]\label{lemma:cbdy:block:G and Z commute}
	For every $\initA \in \InitsA$,
    $\initD \in \InitsD$,
	$\powZ \in \rangeII{0}{\initD-\initA-4m}$,
    $M \in \basic{4m}$,
    and $G \in \cbdyG$,
	\[ \vdash_6 \Comm{\Let{G}}{\ElBlockZ{\powZ}{\initA}{\initD}{M}} \equiv \zero. \]
\end{lemma}

\begin{proof}
    By definition (\Cref{def:cbdy:block:Z}), $\ElBlockZ{\powZ}{\initA}{\initD}{M} = \Comm{\Let{P}}{\Let{H}}$
    for some $P \in \cbdyP$ and $H \in \cbdyH$.
    Hence using constant-length in-subgroup relations (\Cref{lemma:cbdy:present}):
    \begin{align*}
        & \Let{G} \ElBlockZ{\powZ}{\initA}{\initD}{M} \\
        &= \Let{G} \Let{P} \Let{H} \Let{-P} \Let{-H} \tag{def. of $\typZ$ basic block alias, \Cref{def:cbdy:block:Z}} \\
        &\equiv \Let{P} \Let{G} \Let{H} \Let{-P} \Let{-H} \tag{$\typG$ and $\typP$ commute} \\
        &\equiv \Let{P} \Let{H} \Let{GH} \Let{G} \Let{-P} \Let{-H} \tag{commutator of $\typG$ and $\typH$} \\
        &\equiv \Let{P} \Let{H} \Let{GH} \Let{-P} \Let{G} \Let{-H} \tag{$\typG$ and $\typP$ commute} \\
        &\equiv \Let{P} \Let{H} \Let{GH} \Let{-P} \Let{-GH} \Let{-H} \Let{G} \tag{commutator of $\typG$ and $\typH$} \\
        &\equiv \Let{P} \Let{H} \Let{-P} \Let{-H} \Let{G} \tag{$\typP$ and $\typQ$ commute (\Cref{lemma:cbdy:P and Q commute}) and linearity of $\typQ$} \\
        &= \ElBlockZ{\powZ}{\initA}{\initD}{M} \Let{G}, \tag{def. of $\typZ$ basic block alias, \Cref{def:cbdy:block:Z}}
    \end{align*}
    as desired.
\end{proof}
 
\begin{lemma}[Type-$\typP$ elements commute with type-$\typZ$ basic block aliases]\label{lemma:cbdy:block:P and Z commute}
	For every $\initA \in \InitsA$,
    $\initD \in \InitsD$,
	$\powZ \in \rangeII{0}{\initD-\initA-4m}$,
    $M \in \basic{4m}$,
    and $P \in \cbdyP$, \[
    \vdash_1 \Comm{\Let{P}}{\ElBlockZ{\powZ}{\initA}{\initD}{M}} \equiv \zero. \]
\end{lemma}

\begin{proof}
    By \Cref{lemma:cbdy:block:Z as F and Q}, $\vdash_1 \ElBlockZ{\powZ}{\initA}{\initD}{M} \equiv \Comm{\Let{F}}{\Let{Q}}$
    for some $F \in \cbdyF$ and $Q \in \cbdyQ$.
    Both parts commute with $\Let{P}$, the first by in-subgroup relations (\Cref{lemma:cbdy:present})
    and the second by \Cref{lemma:cbdy:block:P and Q commute}.
\end{proof}

\begin{lemma}[Type-$\typQ$ elements commute with type-$\typZ$ basic block aliases]\label{lemma:cbdy:block:Q and Z commute}
	For every $\initA \in \InitsA$,
    $\initD \in \InitsD$,
	$\powZ \in \rangeII{0}{\initD-\initA-4m}$,
    $M \in \basic{4m}$,
    and $Q \in \cbdyQ$, \[ \vdash_0 \Comm{\Let{Q}}{\ElBlockZ{\powZ}{\initA}{\initD}{M}} \equiv \zero. \]
\end{lemma}

\begin{proof}
    By definition (\Cref{def:cbdy:block:Z}), $\ElBlockZ{\powZ}{\initA}{\initD}{M} = \Comm{\Let{P}}{\Let{H}}$
    for some $P \in \cbdyP$ and $H \in \cbdyH$.
    Both parts commute with $\Let{Q}$, the first by \Cref{lemma:cbdy:block:P and Q commute} and the second by in-subgroup relations (\Cref{lemma:cbdy:present}).
\end{proof}

\begin{lemma}[Type-$\typZ$ basic block aliases commute]\label{lemma:cbdy:block:Z and Z commute}
	For every $\initA, \initA' \in \InitsA$,
    $\initD,\initD' \in \InitsD$,
	$\powZ \in \rangeII{0}{\initD-\initA-4m}$,
    $\powZ' \in \rangeII{0}{\initD'-\initA'-2m}$, and
	$M, M' \in \basic{4m}$, \[
    \vdash_1 \Comm{\ElBlockZ{\powZ}{\initA}{\initD}{M}}{\ElBlockZ{\powZ'}{\initA'}{\initD'}{M'}} \equiv \zero. \]
\end{lemma}

\begin{proof}
    By definition (\Cref{def:cbdy:block:Z}), $\ElBlockZ{\powZ}{\initA}{\initD}{M} \equiv \Comm{\Let{P}}{\Let{H}}$
    for some $P \in \cbdyP$ and $H \in \cbdyH$.
    Both parts commute with $\ElBlockZ{\powZ'}{\initA'}{\initD'}{M'}$ by \Cref{lemma:cbdy:block:P and Z commute,lemma:cbdy:block:H and Z commute}.
\end{proof}

\begin{lemma}[Linearity for type-$\typZ$ basic block aliases]\label{lemma:cbdy:block:Z linearity}
	For every $\initA \in \InitsA$, $\initD \in \InitsD$,
	$\powZ \in \rangeII{0}{\initD-\initA-4m}$, and $M,M' \in \basic{4m}$: \[
    \vdash_1 \ElBlockZ{\powZ}{\initA}{\initD}{M} \cdot \ElBlockZ{\powZ}{\initA}{\initD}{M'} \equiv \ElBlockZ{\powZ}{\initA}{\initD}{M+M'}. \]
\end{lemma}

\begin{proof}
    Using (\Cref{def:cbdy:block:Z}), we have:
    \begin{align*}
        & \ElBlockZ{\powZ}{\initA}{\initD}{M} \cdot \ElBlockZ{\powZ}{\initA}{\initD}{M'} \\
        &= \Let{\blockP{\powP}{\initA}{\initC}{M}} \Let{\blockH{\powH}{\initC}{\initD}{I_m}}
        	\Let{\blockP{\powP}{\initA}{\initC}{-M}} \Let{\blockH{\powH}{\initC}{\initD}{-I_m}} \ElBlockZ{\powZ}{\initA}{\initD}{M'}
            \tag{def. of type-$\typZ$ basic block elt., \Cref{def:cbdy:block:Z}} \\
        &\equiv \Let{\blockP{\powP}{\initA}{\initC}{M}} \cdot \ElBlockZ{\powZ}{\initA}{\initD}{M'} \cdot \Let{\blockH{\powH}{\initC}{\initD}{I_m}}
        	\Let{\blockP{\powP}{\initA}{\initC}{-M}} \Let{\blockH{\powH}{\initC}{\initD}{-I_m}}
            \tag{type-$\typZ$ basic block elts. commute with type-$\typP$ and -$\typH$ elts.,
            \Cref{lemma:cbdy:block:P and Z commute,lemma:cbdy:block:H and Z commute}} \\
        &\equiv \Let{\blockP{\powP}{\initA}{\initC}{M}} \Let{\blockP{\powP}{\initA}{\initC}{M'}} \Let{\blockH{\powH}{\initC}{\initD}{I_m}}
        	\Let{\blockP{\powP}{\initA}{\initC}{-M'}} \Let{\blockH{\powH}{\initC}{\initD}{-I_m}} \Let{\blockH{\powH}{\initC}{\initD}{I_m}}
        	\Let{\blockP{\powP}{\initA}{\initC}{-M}} \Let{\blockH{\powH}{\initC}{\initD}{-I_m}} 
            \tag{def. of type-$\typZ$ basic block elt., \Cref{def:cbdy:block:Z}} \\
		&\equiv \Let{\blockP{\powP}{\initA}{\initC}{M+M'}} \Let{\blockH{\powH}{\initC}{\initD}{I_m}}
        	\Let{\blockP{\powP}{\initA}{\initC}{-(M+M')}} \Let{\blockH{\powH}{\initC}{\initD}{-I_m}} 
            \tag{lin. of type-$\typP$ basic block elts., \Cref{prop:cbdy:steinberg}} \\
		&= \ElBlockZ{\powZ}{\initA}{\initD}{M+M'}
            \tag{def. of type-$\typZ$ basic block elt., \Cref{def:cbdy:block:Z}}
    \end{align*}
    as desired.
\end{proof}

\subsection{Defining type-\texorpdfstring{$\typZ$}{Z} aliases}

Having defined type-$\typZ$ basic block aliases and shown they commute with other relevant elements,
we next turn to defining aliases corresponding to general type-$\typZ$ matrices.
We emphasize that starting from this subsection, we will begin considering words of non-constant (i.e., $\rho$-dependent) length.

\begin{definition}[Decomposition of type-$\typZ$ matrices]\label{def:cbdy:Z decomposition}
    For $Z \in \cbdyZ$, for each $\alpha \in \rangeIE{0}{\nA/m}$, $\delta \in \rangeIE{0}{\nD/m}$, let $Z^{(\alpha,\delta)}$
    denote $Z$ restricted to the entries $\IntPartA{\alpha} \times \IntPartD{\delta}$,
    so that $Z = \sum_{\alpha=0}^{\nA/m-1} \sum_{\delta=0}^{\nD/m-1} Z^{(\alpha,\delta)}$.
    
    For each $\alpha,\delta$,
    let $M^{(\alpha,\delta)}_{i,j} = Z_{\partA{\alpha}+i,\partD{\delta}+j}$
    so that $Z^{(\alpha,\delta)} = \blockPartZ{0}{\alpha}{\delta}{M^{(\alpha,\delta)}}$.
    Letting
    $\Delta^{(\alpha,\delta)} \coloneqq \partD{\delta}-\partA{\alpha}-4m$, $s = 2m$,
    and $\tau^{(\alpha,\delta)} \coloneqq \floor{ \Delta^{(\alpha,\delta)}/(2m) } $,
    \Cref{prop:basic:decomp} lets us decompose $M^{(\alpha,\delta)}$ as \[
    M^{(\alpha,\delta)} = \sum_{k=0}^{\tau^{(\alpha,\delta)}-1} X^{\kappa \cdot 2km} \Mbody^{(\alpha,\delta,k)}+X^{\kappa \cdot 2\tau^{(\alpha,\delta)} m} \Mgap^{(\alpha,\delta)}+X^{\kappa \Delta} \Mtail^{(\alpha,\delta)} \]
    for $\Mbody^{(\alpha,\delta,k)}, \Mgap^{(\alpha,\delta)}, \Mtail^{(\alpha,\delta)} \in \basic{4m}$.
    Correspondingly, \[ Z^{(\alpha,\delta)} = \blockPartZ{0}{\partA{\alpha}}{\partD{\delta}}{M^{(\alpha,\delta)}} = \sum_{k=0}^{\tau^{(\alpha,\delta)}-1} \blockZ{k\cdot 2m}{\partA{\alpha}}{\partD{\delta}}{\Mbody^{(\alpha,\delta,k)}}+\blockZ{2\tau^{(\alpha,\delta)}m}{\partA{\alpha}}{\partD{\delta}}{\Mgap^{(\alpha,\delta)}}+\blockZ{\Delta^{(\alpha,\delta)}}{\alpha}{\delta}{\Mtail^{(\alpha,\delta)}}. \qedhere \]
\end{definition}

\begin{definition}[``Alias'' word corresponding to type-$\typZ$ matrix]\label{def:cbdy:Z}
    For every $Z \in \cbdyZ$, we define: \[
    \alias{Z} \coloneqq \prod_{\alpha=0}^{\nA/m-1} \prod_{\delta = 0}^{\nD/m-1} \parens*{ \parens*{
    \prod_{k=0}^{\tau^{(\alpha,\delta)}-1}
    \ElBlockZ{km}{\initA}{\initD}{\Mbody^{(\alpha,\delta,k)}}} \cdot
    \ElBlockZ{\tau^{(\alpha,\delta)}-1}{\initA}{\initD}{\Mgap^{(\alpha,\delta)}} 
    \cdot \ElBlockZ{\Delta^{(\alpha,\delta)}}{\initA}{\initD}{\Mtail^{(\alpha,\delta)}}} \]
    where $\Delta^{(\alpha,\delta)},\tau^{(\alpha,\delta)}, \Mbody^{(\alpha,\delta,k)}, \Mgap^{(\alpha,\delta)}, \Mtail^{(\alpha,\delta)}$ are as in \Cref{def:cbdy:Z decomposition}.
    This is a subgroup word over $\calH$ of length $(\nA/m-1)(\nD/m-1) (\tau^{(\alpha,\delta)} - 1 + 2) \le O(\rho^3)$,
    since $a, d, \Delta^{(\alpha,\delta)} \le n$.
\end{definition}

\begin{lemma}[Linearity for type-$\typZ$ aliases]\label{lemma:cbdy:Z:linearity}
    For every $Z, Z' \in \cbdyZ$, \[
    \vdash_7 \alias{Z} \cdot \alias{Z'} \equiv \alias{Z+Z'}. \]
\end{lemma}

\begin{proof}
    Abbreviate $\Delta \coloneqq \Delta^{(\alpha,\delta)}$ and $\tau \coloneqq \tau^{(\alpha,\delta)}$.
    By construction, the decomposition in \Cref{def:cbdy:Z decomposition} is linear, in the sense that:
    \begin{align*}
        \Mbody^{(\alpha,\delta,k)}+\Nbody^{(\alpha,\delta,k)} &= (M+N)_{\mathrm{body}}^{(\alpha,\delta,k)}, \\
        \Mgap^{(\alpha,\delta)}+\Ngap^{(\alpha,\delta)} &= (M+N)_{\mathrm{gap}}^{(\alpha,\delta)}, \\
        \Mtail^{(\alpha,\delta)}+\Ntail^{(\alpha,\delta)} &= (M+N)_{\mathrm{tail}}^{(\alpha,\delta)}
    \end{align*}
    (the first for every $k \in \rangeIE{0}{\tau}$).

    By the definition of type-$\typZ$ aliases (\Cref{def:cbdy:Z}),
    type-$\typZ$ basic block alias linearity (\Cref{lemma:cbdy:block:Z linearity}) and autocommutation (\Cref{lemma:cbdy:block:Z and Z commute}),
    we have
    \begin{align*}
    &\alias{Z} \cdot \alias{Z'} \\
    &= \prod_{\alpha=0}^{\nA/m-1} \prod_{\delta = 0}^{\nD/m-1} \parens*{ \parens*{
    \prod_{k=0}^{\tau -1}
    \ElBlockZ{km}{\initA}{\initD}{\Mbody^{(\alpha,\delta,k)}}}
    \cdot \ElBlockZ{\tau}{\initA}{\initD}{\Mgap^{(\alpha,\delta)}} 
    \cdot \ElBlockZ{\Delta}{\initA}{\initD}{\Mtail^{(\alpha,\delta)}}} \\
    &\hspace{1in} \cdot \prod_{\alpha=0}^{\nA/m-1} \prod_{\delta = 0}^{\nD/m-1} \parens*{ \parens*{
    \prod_{k=0}^{\tau -1}
    \ElBlockZ{km}{\initA}{\initD}{\Nbody^{(\alpha,\delta,k)}}}
    \cdot \ElBlockZ{\tau}{\initA}{\initD}{\Ngap^{(\alpha,\delta)}} 
    \cdot \ElBlockZ{\Delta}{\initA}{\initD}{\Ntail^{(\alpha,\delta)}} } \\
    &\equiv \prod_{\alpha=0}^{\nA/m-1} \prod_{\delta = 0}^{\nD/m-1} \parens*{ \parens*{
    \prod_{k=0}^{\tau -1}
    \ElBlockZ{km}{\initA}{\initD}{\Mbody^{(\alpha,\delta,k)}+\Nbody^{(\alpha,\delta,k)}}}
    \cdot \ElBlockZ{\tau}{\initA}{\initD}{\Mgap^{(\alpha,\delta)}+\Ngap^{(\alpha,\delta)}} 
    \cdot \ElBlockZ{\Delta}{\initA}{\initD}{\Mtail^{(\alpha,\delta)}+\Ntail^{(\alpha,\delta)}} } \\
    &\equiv \parens*{ \parens*{
    \prod_{k=0}^{\tau -1}
    \ElBlockZ{km}{\initA}{\initD}{(M+N)_{\mathrm{body}}^{(\alpha,\delta,k)}}}
    \cdot \ElBlockZ{\tau}{\initA}{\initD}{(M+N)_{\mathrm{gap}}^{(\alpha,\delta)}} 
    \cdot \ElBlockZ{\Delta}{\initA}{\initD}{(M+N)_{\mathrm{tail}}^{(\alpha,\delta)}} } \\
    &= \alias{Z + Z'}.
    \end{align*}
    The number of steps is $O(\rho^7)$ since the subgroup words $\alias{Z}$ and $\alias{Z'}$ each have length $O(\rho^3)$
    and each application of \Cref{lemma:cbdy:block:Z linearity,lemma:cbdy:block:Z and Z commute} uses $O(\rho)$ derivation steps.
\end{proof}

\begin{lemma}\label{lemma:cbdy:Z:commute}
    For every $Z \in \cbdyZ$, $\typL \in \{\typB,\typC,\typF,\typG,\typH,\typP,\typQ\}$, and $L \in \Mats{\typL}$, \[
    \vdash_9 \Comm{\Let{L}}{\alias{Z}} \equiv \zero. \]
\end{lemma}

\begin{proof}
    By \Cref{def:cbdy:Z}, $\alias{Z}$ is a product of $O(\rho^3)$ type-$\typZ$ basic block aliases.
    Each basic block alias commutes with type-$\typB$, -$\typC$, -$\typF$, -$\typG$, -$\typH$, -$\typP$, and -$\typQ$ elements by
    \Cref{lemma:cbdy:block:B and Z commute,lemma:cbdy:block:C and Z commute,lemma:cbdy:block:F and Z commute,lemma:cbdy:block:G and Z commute,lemma:cbdy:block:H and Z commute,lemma:cbdy:block:P and Z commute,lemma:cbdy:block:Q and Z commute},
    respectively.
    $\alias{Z}$ has length $O(\rho^3)$ and each application of these lemmas requires $O(\rho^6)$ derivation steps.
\end{proof}

\subsection{Normalizing type-\texorpdfstring{$\typZ$}{Z} basic block aliases}

In the previous section, we defined words corresponding to general type-$\typZ$ matrices.
We now show how to relate the type-$\typZ$ basic block aliases we defined in \Cref{def:cbdy:Z} to this new definition.

\begin{lemma}[Shift]\label{lemma:cbdy:shift}
    For every $s \in \rangeII{0}{4m}$,
    $\initA \in \InitsA$,
    $\initD \in \InitsD$,
    $\powZ \in \rangeII{0}{\initD-\initA-4m-s}$,
    if $M \in \basic{4m-s}$ (equiv., $X^{\kappa s} M \in \basic{4m}$), then \[
    \vdash_1 \ElBlockZ{\powZ+s}{\initA}{\initD}{M} \equiv \ElBlockZ{\powZ}{\initA}{\initD}{X^{\kappa s} M}. \]
\end{lemma}

\begin{proof}
    By definition, for some fixed $c^*$ and $(\powP,\powH) \coloneqq \pack{\powZ}{c^*-\initA-4m}{\initD-c^*}$, \[
    \ElBlockZ{\powZ}{\initA}{\initD}{X^{\kappa s} M} = \Comm{\Let{\blockP{\powP}{\initA}{c^*}{X^{\kappa s} M}}}{\Let{\blockH{\powH}{c^*}{\initD}{I_m}}}. \]
    Note that we actually have $\powH \in \rangeII{0}{\initD-c^*-s}$ by assumption on $\powZ$.
    Now note that $\Let{\blockP{\powP}{\initA}{c^*}{X^{\kappa s} M}} = \Let{\blockP{\powP+s}{\initA}{c^*}{M}}$
    and $M \in \basic{4m-s}$ while $I_m \in \basic{s}$.
    Hence by \Cref{lemma:cbdy:block:establish Z}, we get: \[
    \Comm{\Let{\blockP{\powP+s}{\initA}{c^*}{M}}}{\Let{\blockH{\powH}{c^*}{\initD}{I_m}}} \equiv \ElBlockZ{\powZ+s}{\initA}{\initD}{M}, \]
    as desired.
\end{proof}

\begin{lemma}[Normalization]\label{lemma:cbdy:norm}
    Let $\alpha \in \rangeIE{0}{\nA/m}$,
    $\delta \in \rangeIE{0}{\nD/m}$,
    $\powZ \in \rangeII{0}{\partD{\delta}-\partA{\alpha}-4m}$,
    and $M \in \basic{4m}$.
    Then \[
    \vdash_2 \ElBlockZ{\powZ}{\partA{\alpha}}{\partD{\delta}}M \equiv \alias{\blockPartZ{\powZ}{\alpha}{\delta}{M}}. \]
\end{lemma}

\begin{proof}
    We abbreviate $\Delta \coloneqq \Delta^{(\alpha,\delta)} = \partD{\delta}-\partA{\alpha}-4m$ 
    and $\tau \coloneqq \tau^{(\alpha,\delta)} = \floor{ \Delta^{(\alpha,\delta)}/(2m) }$ as in \Cref{def:cbdy:Z decomposition}.

    For $(i,j) \in \rangeIE{0}{m} \times \rangeIE{0}{m}$, let $M_{i,j} = \sum_{t=0}^{\kappa(4m+j-i)} X^t \mu_{i,j}^{(t)}$ be the degree decomposition of the $(i,j)$ entry of $M$,
    and $\mu_{i,j}^{(t)} = 0$ for every $t < 0$ or $t > \kappa(4m+j-i)$.
    Correspondingly, define $\Nbody^{(k)}, \Ngap, \Ntail \in \F[X]^{\rangeIE{0}{m}\times\rangeIE{0}{m}}$ via:
    \begin{align*}
    (\Nbody^{(k)})_{i,j} &\coloneqq \sum_{t=0}^{\kappa \cdot 2m-1} X^{t+\kappa (k \cdot 2m-\powZ)} \mu^{(t+\kappa (k \cdot 2m-\powZ))}_{i,j}, \\
    (\Ngap)_{i,j} &\coloneqq \sum_{t=0}^{\kappa (\Delta-\tau \cdot 2m)-1} X^{t+\kappa (\tau \cdot 2m-\powZ)} \mu^{(t+\kappa (\tau \cdot 2m-\powZ))}_{i,j}, \\
    (\Ntail)_{i,j} &\coloneqq \sum_{t=0}^{\kappa (4m+j-i)} X^{t+\kappa (\Delta-\powZ)} \mu^{(t+\kappa (\Delta-\powZ))}_{i,j}.
    \end{align*}
    Note that there may be negative powers of $X$ here, but by definition of $\mu$,
    only terms with degrees in $\rangeII{0}{4m+j-i}$ have nonzero $\mu$ values.
    Indeed, this means that $\Nbody^{(k)}, \Ngap, \Ntail$ are all in $\basic{4m}$.
    
    Since $M = \sum_{k=0}^{\tau -1} \Nbody^{(k)}+\Ngap+\Ntail$
    (because the terms in the expansion exactly partition the terms in $M_{i,j}$),
    and since $\Delta \le n$, by block linearity (\Cref{lemma:cbdy:block:Z linearity}), \[
    \vdash_2 \ElBlockZ{\powZ}{\initA}{\initD}{M} \equiv
    \prod_{k=0}^{\tau -1} \ElBlockZ{\powZ}{\initA}{\initD}{\Nbody^{(k)}}
    \cdot \ElBlockZ{\powZ}{\initA}{\initD}{\Ngap} \cdot \ElBlockZ{\powZ}{\initA}{\initD}{\Ntail}. \]
    Recalling the definition of type-$\typZ$ alias (\Cref{def:cbdy:Z}),
    it suffices to verify that \[
    \vdash_1 \left\{
    \begin{aligned}
        \ElBlockZ{\powZ}{\initA}{\initD}{\Nbody^{(k)}} &\equiv \ElBlockZ{k \cdot 2m}{\initA}{\initD}{\Mbody^{(k)}}, \\ 
        \ElBlockZ{\powZ}{\initA}{\initD}{\Ngap} &\equiv \ElBlockZ{\tau \cdot 2m}{\initA}{\initD}{\Mgap}, \\
        \ElBlockZ{\powZ}{\initA}{\initD}{\Ntail} &\equiv \ElBlockZ{\Delta}{\initA}{\initD}{\Mtail}.
    \end{aligned} \right. \]
    Each of these follows by \Cref{lemma:cbdy:shift}.
    For instance, for the first equality, we observe that if $k \cdot 2m-\powZ \ge 0$, then $\Nbody^{(k)} = X^{\kappa(k \cdot 2m-\powZ)} \Mbody^{(k)}$
    and otherwise $\Mbody^{(k)} = X^{\kappa(\powZ- k \cdot 2m)} \Nbody^{(k)}$.
    In the former case, we use $s = k \cdot 2m-\powZ$ and have
    $\ElBlockZ{\powZ+s}{\initA}{\initD}{\Mbody^{(k)}} \equiv \ElBlockZ{\powZ}{\initA}{\initD}{\Nbody^{(k)}}$
    (using that $\powZ \le \Delta-s$ because $\powZ+s = k \cdot 2m \le \Delta$).
    In the latter case, we use $s = \powZ-k \cdot 2m$ and have
    $\ElBlockZ{k \cdot 2m+s}{\initA}{\initD}{\Nbody^{(k)}} \equiv \ElBlockZ{k \cdot 2m}{\initA}{\initD}{\Mbody^{(k)}}$
    (using that $k \cdot 2m \le \Delta-s$ because $k \cdot 2m+s = \powZ \le \Delta$).
    The proofs for the other two equalities are the same, with either $\tau \cdot 2m$ or $\Delta$ itself playing the role of $k \cdot 2m$ here.
\end{proof}

\subsection{Establishing type-\texorpdfstring{$\typZ$}{Z} elements}

We now turn to proving the following lemma:

\begin{lemma}[$\typF$-$\typQ$ and $\typP$-$\typH$ commutators]\label{lemma:cbdy:establish Z}
    For every $P \in \cbdyP, H \in \cbdyH$, $\vdash_{15} \Comm{\Let{P}}{\Let{H}} \equiv \alias{PH}$.
    Similarly, for every $F \in \cbdyF, Q \in \cbdyQ$, $\vdash_{15} \Comm{\Let{F}}{\Let{Q}} \equiv \alias{FQ}$.
\end{lemma}

We focus on proving the first equality, since the latter is dual.
Specifically, all proofs in this subsection will work if we replace $P$ with $Q$, $F$ with $H$, $\gamma$ with $\floor{\nB/m}-\beta$, $\alpha$ with $\delta$, and $\delta$ with $\nA/m - \alpha - 1$.

\begin{claim}[Decomposing type-$\typH$ block matrices]\label{claim:cbdy:matrix:block:decompose H}
    For $\gamma \in \rangeII{0}{\floor{\nC/m}}$ and $\delta \in \rangeIE{0}{\nD/m}$,
    and $H \in \cbdyH$ supported only on the $\IntPartC{\gamma} \times \IntPartD{\delta}$ entries:
    \begin{itemize}
    \item If $\gamma = \floor{\nC/m}$ and $\delta = 0$, then
    $H$ is an $m$-bounded basic block matrix,
    i.e., a matrix of the form $\blockPartH{\powH}{\gamma}{\delta}{M}$
    for some $\powH \in \rangeII{0}{\partD{\delta}-\partC{\gamma}-m}$, $M \in \basic{m}$.
    \item Otherwise, $H$ is a sum of $\floor{ (\partD{\delta}-\partC{\gamma}) / m } = O(\rho)$ $2m$-bounded basic block matrices
    (i.e., matrices of the form $\blockPartH{\powH}{\gamma}{\delta}{M}$
    for some $\powH \in \rangeII{0}{\partD{\delta}-\partC{\gamma}-2m}$, $M \in \basic{2m}$).
    \end{itemize}
\end{claim}

\begin{proof}
    Define $M \in \F[X]^{\rangeIE{0}{m}\times\rangeIE{0}{m}}$ via
    $M_{i,j} \coloneqq H_{\partC{\gamma}+i,\partD{\delta}+j}$.
    Then $\deg(M_{i,j}) = \deg(H_{i,j}) \le \kappa((\partD{\delta}+j)-(\partC{\gamma}+i))$,
    and so $M$ is $(\partD{\delta}-\partC{\gamma})$-bounded.
    Therefore $H = \blockPartH{0}{\gamma}{\delta}{M}$.

    Now we consider two cases.
    Firstly, if $\partD{\delta}-\partC{\gamma} \le 2m$,
    which (by \Cref{eq:cbdy:A and D partition,eq:cbdy:C and B partition}) happens iff $\gamma = \floor{\nC/m}$ and $\delta = 0$,
    i.e., $\partD{\delta}-\partC{\gamma} = m$,
    then $H$ is already an $m$-bounded basic block matrix.
    Otherwise, let $\Delta \coloneqq \partD{\delta}-\partC{\gamma}-2m \ge 0$,
    $s \coloneqq m$, $\tau \coloneqq \floor{ \Delta / m } $,
    so that $M$ is $(\Delta+2s)$-bounded and by \Cref{prop:basic:decomp}, \[
    M = \sum_{k=0}^{\tau-1} X^{\kappa \cdot km} \Mbody^{(k)}+X^{\kappa \cdot \tau m} \Mgap+X^{\kappa \Delta} \Mtail \]
    for $\Mbody^{(k)}, \Mgap, \Mtail \in \basic{2m}$.
    Correspondingly,
    \begin{multline*}
    H = \blockPartH{0}{\gamma}{\delta}M = 
    \sum_{k=0}^{\tau-1} \blockPartH{0}{\gamma}{\delta}{X^{\kappa \cdot km} \Mbody^{(k)}}+\blockPartH{0}{\gamma}{\delta}{X^{\kappa \cdot \tau m} \Mgap}+\blockPartH{0}{\gamma}{\delta}{X^{\kappa \Delta} \Mtail} \\
    = \sum_{k=0}^{\tau-1} \blockPartH{km}{\gamma}{\delta}{\Mbody^{(k)}}+\blockPartH{\tau m}{\gamma}{\delta}\Mgap+\blockPartH{\Delta}{\gamma}{\delta}{\Mtail},
    \end{multline*}
    as desired.
\end{proof}

\begin{claim}\label{claim:cbdy:P and H form}
    Let $\alpha \in \rangeIE{0}{\nA/m}$, $\gamma, \gamma' \in \rangeII{0}{\floor{\nC/m}}$, and $\delta \in \rangeIE{0}{\nD/m}$.
    Suppose $P \in \cbdyP$ is supported only on
    $\IntPartA{\alpha} \times \IntPartC{\gamma}$
    and $H \in \cbdyH$ is supported only on
    $\IntPartC{\gamma'} \times \IntPartD{\delta}$.
   
    Let $\rhoP \coloneqq \floor{ (\partC{\gamma}-\partA{\alpha})/m } + 2$ and
    $\rhoH \coloneqq \floor{ (\partD{\delta}-\partC{\gamma'}) / m } + 2 \le O(\rho)$.
    Then $P$ is a sum of $\rhoP$ $2m$-bounded basic block matrices,
    $P = \sum_{k = 1}^{\rhoP} \blockPartP{\powP^{(k)}}{\alpha}{\gamma}{M^{(k)}}$
    where $M^{(k)} \in \basic{2m}$ and $\powP^{(k)} \in \rangeII{0}{\partC{\gamma}-\partA{\alpha}-2m}$.
    Correspondingly
    \begin{equation}\label{eq:P form}
    \vdash_1 \Let{P} = \Let{\sum_{k = 1}^{\rhoP} \blockPartP{\powP^{(k)}}{\alpha}{\gamma}{M^{(k)}}}
    \equiv \prod_{k = 1}^{\rhoP} \Let{\blockPartP{\powP^{(k)}}{\alpha}{\gamma}{M^{(k)}}}.
    \end{equation}
    $H$ is either an $m$-bounded basic block matrix $\blockPartH{0}{\floor{ \nC/m }}{0}{N}$ (for some $N \in \basic{m}$) if $\gamma' = \floor{ \nC/m } $ and $\delta = 0$,
    and otherwise a sum of $\rhoH$ $2m$-bounded basic block matrices,
    $H = \sum_{\ell = 1}^{\rhoH} \blockP{\powH^{(\ell)}}{\partC{\gamma'}}{\partD{\delta}}{N^{(\ell)}}$
    where $N^{(\ell)} \in \basic{2m}$ and $\powH^{(\ell)} \in \rangeII{0}{\partD{\delta}-\partC{\gamma'}-2m}$.
    Correspondingly, either
    \begin{subequations}
    \begin{align}
        \Let{H} &=
        \Let{\blockPartH{0}{\floor{ \nC/m }}{0}{N}} &&& \text{if }\gamma' = \floor{ \nC/m } \wedge \delta = 0, \label{eq:H form single} \\
        \vdash_1 \Let{H} &= \Let{\sum_{\ell = 1}^{\rhoH} \blockPartH{\powH^{(\ell)}}{\gamma'}{\delta}{N^{(\ell)}}}
        \equiv \prod_{\ell = 1}^{\rhoH} \Let{\blockPartH{\powH^{(\ell)}}{\gamma'}{\delta}{N^{(\ell)}}} &&& \text{otherwise}. \label{eq:H form multi}
    \end{align}
    \end{subequations}
\end{claim}

\begin{proof}
    Follows immediately from \Cref{claim:cbdy:matrix:block:decompose P and Q,claim:cbdy:matrix:block:decompose H,lemma:cbdy:present}.
\end{proof}

\begin{claim}\label{claim:cbdy:establish Z:block:disj}
    Let $\alpha \in \rangeIE{0}{\nA/m}$, $\gamma, \gamma' \in \rangeII{0}{\floor{\nC/m}}$, and $\delta \in \rangeIE{0}{\nD/m}$.
    Suppose $P \in \cbdyP$ is supported only on
    $\IntPartA{\alpha} \times \IntPartC{\gamma}$
    and $H \in \cbdyH$ is supported only on
    $\IntPartC{\gamma'} \times \IntPartD{\delta}$.
    Further, suppose that $\IntPartC{\gamma} \cap \IntPartC{\gamma'} = \emptyset$.
    Then \[
    \vdash_2 \Comm{\Let{P}}{\Let{H}} \equiv \zero. \]
\end{claim}

\begin{proof}
    We apply \Cref{claim:cbdy:P and H form}.
    By \Cref{lemma:cbdy:block:incompatible}, each element in the RHS of \Cref{eq:P form} commutes with each element the RHS of \Cref{eq:H form single} or \Cref{eq:H form multi}.
    Each invocation of \Cref{lemma:cbdy:block:incompatible} requires $O(1)$ derivation steps and we do so $\rhoP \cdot \rhoH \le O(\rho^2)$ times.
\end{proof}

\begin{claim}\label{claim:cbdy:establish Z:block:single}
    Let $\alpha \in \rangeIE{0}{\nA/m}$, $\gamma \coloneqq \floor{\nC/m} \in \rangeII{0}{\floor{\nC/m}}$, and $\delta \in \rangeIE{0}{\nD/m}$.
    Suppose $P \in \cbdyP$ is supported only on
    $\IntPartA{\alpha} \times \IntPartC{\gamma}$
    and $H \in \cbdyH$ is supported only on
    $\IntPartC{\gamma} \times \IntPartD{\delta}$.
    Then \[
    \vdash_8 \Comm{\Let{P}}{\Let{H}} \equiv \alias{PH}. \]
\end{claim}

\begin{proof}
    We have, by \Cref{eq:P form,eq:H form single}: \[
    \Comm{\Let{P}}{\Let{H}} = \Comm{\prod_{k = 1}^{\rhoP} \underbrace{\Let{\blockPartP{\powP^{(k)}}{\alpha}{\floor{ \nC/m }}{M^{(k)}}}}_{\eqqcolon u_k}}
    {\underbrace{\Let{\blockPartH{\powH}{\floor{ \nC/m }}{0}{N}}}_{\eqqcolon v} }. \]
    We apply \Cref{eq:collecting commutators}.
    The pairwise commutators are \[
    \Comm{u_k}{v} = \Comm{\Let{\blockPartP{\powP^{(k)}}{\alpha}{\floor{ \nC/m }}{M^{(k)}}}}
    {\Let{\blockPartH{\powH}{\floor{ \nC/m }}{0}{N}} }
    \equiv \ElBlockZ{\powP^{(k)} + \powH}{\partA{\alpha}}{\partD{0}}{M^{(k)} N} \]
    via \Cref{lemma:cbdy:block:corner} (each application costs $O(1)$).
    The triple commutators are
    \begin{multline*}
    \Comm{u_{k'}}{\Comm{u_k}{v}}
    = \Comm{\Let{\blockPartP{\powP^{(k)}}{\alpha}{\floor{ \nC/m }}{M^{(k)}}}}
    {\Comm{\Let{\blockPartP{\powP^{(k)}}{\alpha}{\floor{ \nC/m }}{M^{(k)}}}}
    {\Let{\blockPartH{\powH}{\floor{ \nC/m }}{0}{N}}} } \\
    \equiv \Comm{\Let{\blockPartP{\powP^{(k)}}{\alpha}{\floor{ \nC/m }}{M^{(k)}}}}
    {\ElBlockZ{\powP^{(k)} + \powH}{\partA{\alpha}}{\partD{0}}{M^{(k)} N}} \equiv \zero
    \end{multline*}
    by \Cref{lemma:cbdy:block:P and Z commute} (each application costs $O(1)$).
    All factors have size $O(1)$.
    Thus,
    \begin{multline*}
    \Comm{\Let{P}}{\Let{H}}
    \equiv \prod_{k=\rhoP}^{1} \ElBlockZ{\powP^{(k)} + \powH}{\partA{\alpha}}{\partD{0}}{M^{(k)} N}
    \equiv \prod_{k=\rhoP}^{1} \alias{\blockPartZ{\powP^{(k)} + \powH}{\alpha}{0}{M^{(k)} N}} \\
    \equiv \alias{\sum_{k=\rhoP}^{1} \blockPartZ{\powP^{(k)} + \powH}{\alpha}{0}{M^{(k)} N}}
    = \alias{Z},
    \end{multline*}
    where the equivalences use \Cref{eq:collecting commutators,lemma:cbdy:norm,lemma:cbdy:Z:linearity}.
    The last equation follows because \[
    \blockPartZ{\powP^{(k)} + \powH}{\alpha}{0}{M^{(k)} N}
    = \blockPartP{\powP^{(k)}}{\alpha}{\floor{ \nC/m }}{M^{(k)}}
    \cdot \blockPartH{\powH}{\floor{ \nC/m }}{0}{N} \]
    and therefore \[
    \sum_{k=1}^{\rhoP} \blockPartZ{\powP^{(k)} + \powH}{\alpha}{0}{M^{(k)} N}
    = \parens*{ \sum_{k=1}^{\rhoP} \blockPartP{\powP^{(k)}}{\alpha}{\floor{ \nC/m }}{M^{(k)}} }
    \cdot \blockPartH{\powH}{\floor{ \nC/m }}{0}{N}
    = PH. \]
    Each invocation of \Cref{lemma:cbdy:block:corner,lemma:cbdy:block:P and Z commute,lemma:cbdy:block:H and Z commute} takes $O(1)$ derivation steps.
    \Cref{lemma:cbdy:norm} takes $O(\rho^2)$,
    and \Cref{lemma:cbdy:Z:linearity} takes $O(\rho^{7})$.
    The invocation of \Cref{lemma:cbdy:Z:linearity} $O(\rho)$ times dominates.
\end{proof}

\begin{claim}\label{claim:cbdy:establish Z:block:multi}
    Let $\alpha \in \rangeIE{0}{\nA/m}$, $\gamma \in \rangeIE{0}{\floor{\nC/m}}$, and $\delta \in \rangeIE{0}{\nD/m}$.
    Suppose $P \in \cbdyP$ is supported only on
    $\IntPartA{\alpha} \times \IntPartC{\gamma}$
    and $H \in \cbdyH$ is supported only on
    $\IntPartC{\gamma} \times \IntPartD{\delta}$.
    Then \[
    \vdash_9 \Comm{\Let{P}}{\Let{H}} \equiv \alias{PH}. \]
\end{claim}

\begin{proof}
    By \Cref{eq:P form,eq:H form multi}: \[
    \Comm{\Let{P}}{\Let{H}} = \Comm{\prod_{k = 1}^{\rhoP} \underbrace{\Let{\blockPartP{\powP^{(k)}}{\alpha}{\gamma}{M^{(k)}}}}_{\eqqcolon u_k} }
    {\prod_{\ell = 1}^{\rhoH} \underbrace{\Let{\blockPartH{\powH^{(\ell)}}{\gamma}{\delta}{N^{(\ell)}}}}_{\eqqcolon v_\ell} }. \]
    We again apply \Cref{eq:collecting commutators}.
    The pairwise commutators are \[
    \Comm{u_k}{v_\ell} = \Comm{\Let{\blockPartP{\powP^{(k)}}{\alpha}{\gamma}{M^{(k)}}}}
    {\Let{\blockPartH{\powH^{(\ell)}}{\gamma}{\delta}{N^{(\ell)}}} }
    \equiv \ElBlockZ{\powP^{(k)} + \powH^{(\ell)}}{\partA{\alpha}}{\partD{\delta}}{M^{(k)} N^{(\ell)}} \]
    via \Cref{lemma:cbdy:block:corner} (cost is $O(1)$).
    The triple commutators are
    \begin{multline*}
    \Comm{u_{k'}}{\Comm{u_k}{v_\ell}}
    = \Comm{\Let{\blockPartP{\powP^{(k')}}{\alpha}{\gamma}{M^{(k')}}}}
    {\Comm{\Let{\blockPartP{\powP^{(k)}}{\alpha}{\gamma}{M^{(k)}}}}
    {\Let{\blockPartH{\powH^{(\ell)}}{\gamma}{\delta}{N^{(\ell)}}}} } \\
    \equiv \Comm{\Let{\blockPartP{\powP^{(k')}}{\alpha}{\gamma}{M^{(k')}}}}
    {\ElBlockZ{\powP^{(k)} + \powH^{(\ell)}}{\partA{\alpha}}{\partD{\delta}}{M^{(k)} N^{(\ell)}}}
    \equiv \zero
    \end{multline*}
    by \Cref{lemma:cbdy:block:P and Z commute}, and similarly for $\Comm{v_{\ell'}}{\Comm{u_k}{v_\ell}}$ by \Cref{lemma:cbdy:block:H and Z commute} (cost is $O(1)$).
    Thus,
    \begin{multline*}
    \Comm{\Let{P}}{\Let{H}}
    \equiv \prod_{k=\rhoP}^1 \prod_{\ell=1}^{\rhoH} \ElBlockZ{\powP^{(k)} + \powH^{(\ell)}}{\partA{\alpha}}{\partD{\delta}}{M^{(k)} N^{(\ell)}}
    \equiv \prod_{k=\rhoP}^1 \prod_{\ell=1}^{\rhoH} \alias{\blockPartZ{\powP^{(k)} + \powH^{(\ell)}}{\alpha}{\delta}{M^{(k)} N^{(\ell)}}} \\
    \equiv \alias{\sum_{k=1}^{\rhoP} \sum_{\ell=1}^{\rhoH} \blockPartZ{\powP^{(k)} + \powH^{(\ell)}}{\alpha}{\delta}{M^{(k)} N^{(\ell)}}}
    = \alias{Z},
    \end{multline*}
    where the equivalences use \Cref{eq:collecting commutators,lemma:cbdy:norm,lemma:cbdy:Z:linearity}.
    The last equation follows because \[
    \blockPartZ{\powP^{(k)} + \powH^{(\ell)}}{\alpha}{\delta}{M^{(k)} N^{(\ell)}}
    = \blockPartP{\powP^{(k)}}{\alpha}{\gamma}{M^{(k)}}
    \cdot \blockPartH{\powH^{(\ell)}}{\gamma}{\delta}{N^{(\ell)}} \]
    and therefore \[
    \sum_{k=1}^{\rhoP} \sum_{\ell=1}^{\rhoH} \blockPartZ{\powP^{(k)} + \powH^{(\ell)}}{\alpha}{\delta}{M^{(k)} N^{(\ell)}}
    = \parens*{ \sum_{k=1}^{\rhoP} \blockPartP{\powP^{(k)}}{\alpha}{\gamma}{M^{(k)}} }
    \cdot \parens*{ \sum_{\ell=1}^{\rhoH} \blockPartH{\powH^{(\ell)}}{\gamma}{\delta}{N^{(\ell)}} }
    = PH. \]
    Each invocation of \Cref{lemma:cbdy:block:corner,lemma:cbdy:block:P and Z commute,lemma:cbdy:block:H and Z commute} takes $O(1)$ derivation steps, \Cref{lemma:cbdy:norm} takes $O(\rho^2)$,
    and \Cref{lemma:cbdy:Z:linearity} takes $O(\rho^{7})$.
    The invocation of \Cref{lemma:cbdy:Z:linearity} $O(\rho^2)$ times dominates.
\end{proof}

\begin{claim}\label{claim:cbdy:establish Z:block}
    Let $\alpha \in \rangeIE{0}{\nA/m}$, $\gamma, \gamma' \floor{\nC/m} \in \rangeII{0}{\floor{\nC/m}}$, and $\delta \in \rangeIE{0}{\nD/m}$.
    Suppose $P \in \cbdyP$ is supported only on
    $\IntPartA{\alpha} \times \IntPartC{\gamma}$
    and $H \in \cbdyH$ is supported only on
    $\IntPartC{\gamma'} \times \IntPartD{\delta}$.
    Then \[
    \vdash_9 \Comm{\Let{P}}{\Let{H}} \equiv \alias{PH}. \]
\end{claim}

\begin{proof}
    We consider cases based on $\alpha,\gamma,\gamma',\delta$.

    \paragraph{Case 1: $\IntPartC{\gamma} \cap \IntPartC{\gamma'} = \emptyset$.}
    We immediately apply \Cref{claim:cbdy:establish Z:block:disj}.

    \paragraph{Case 2: $\gamma = \gamma'$.}
    We split into two further subcases.

    \paragraph{Case 2a: $\gamma = \floor{\nC/m}$ and $\delta = 0$.}
    We apply \Cref{claim:cbdy:establish Z:block:single}.

    \paragraph{Case 2b: $\gamma < \floor{\nC/m}$ or $\delta > 0$.}
    We apply \Cref{claim:cbdy:establish Z:block:multi}.
    
    \paragraph{Case 3 (leftover).}
    Since we fell through the first two cases, $\IntPartC{\gamma}$ and $\IntPartC{\gamma'}$ intersect but $\gamma \ne \gamma'$.
    The only way that this can happen, by construction of the intervals $\IntPartC{\cdot}$ (\Cref{eq:cbdy:C and B partition}),
    is that $\gamma = 0$ and $\gamma' = 1$ or $\gamma = 1$ and $\gamma' = 0$.
    Hence, $P$ is supported on $\IntPartA{\alpha} \times (\IntPartC{0} \cup \IntPartC{1})$ and $H$ on $(\IntPartC{0} \cup \IntPartC{1}) \times \IntPartD{\delta}$.
    Unfortunately, the intervals $\IntPartC{0}$ and $\IntPartC{1}$ overlap:
    Since $\partC{0} = \nA + \nB$ and $\partC{1} = \nA + \nB + (\nC - (\floor{\nC/m})m)$,
    \begin{align*}
    \IntPartC{0} &= \rangeIE{\nA + \nB}{\nA + \nB + m}, \\
    \IntPartC{1} &= \rangeIE{\nA + \nB + (\nC - \floor{\nC/m}m)}{\nA + \nB + (\nC - \floor{\nC/m}m) + m}, \\
    \IntPartC{0} \cap \IntPartC{1} &= \rangeIE{\nA + \nB + (\nC - \floor{\nC/m}m)}{\nA + \nB + m}.
    \end{align*}
    (Note that $0 \le \nC - \floor{\nC/m}m < m$.)
    Define $\initC' \coloneqq \nA + \nB + m$, so that \[
    \IntC{\initC'} = \rangeIE{\nA + \nB + m}{\nA + \nB + 2m},
    \quad \IntPartC{0} \cap \IntC{\initC'} = \emptyset,\quad \text{and}\quad 
    \IntPartC{1} \subset \IntPartC{0} \cup \IntC{\initC'}. \]
    
    In this case, decompose $P$ as a sum $P_0 + P'$, where $P_0$ is supported on $\IntPartA{\alpha} \times \IntPartC{0}$
    and $P'$ is supported on $\IntPartA{\alpha} \times (\IntC{\initC'} \cap \IntPartC{1})$.
    Similarly, write $H$ as a sum $H_0 + H'$, where $H_0$ is supported on $\IntPartC{0} \times \IntPartD{\delta}$
    and $H'$ is supported on $(\IntC{\initC'} \cap \IntPartC{1}) \times \IntPartD{\delta}$.
    By disjointness, we have the matrix calculation $PH = P_0H_0 + P'H'$.

    With an eye towards applying \Cref{eq:collecting commutators}, \Cref{claim:cbdy:establish Z:block:disj} gives,
    using that $H'$ is supported on $\IntPartC{1} \times \IntPartD{\delta}$,
    and $P'$ on $\IntPartA{\alpha} \times \IntPartC{1}$: \[
    \Comm{\Let{P_0}}{\Let{H'}} \equiv \Comm{\Let{P'}}{\Let{H_0}} \equiv \Id. \]
    Similarly, \Cref{claim:cbdy:establish Z:block:multi} gives, using that $H'$ is supported on $\IntC{\initC'} \times \IntPartD{\delta}$,
    and $P'$ on $\IntPartA{\alpha} \times \IntC{\initC'}$, we calculate the double commutators:
    \begin{equation}\label{eq:overlap}
    \Comm{\Let{P_0}}{\Let{H_0}} \equiv \alias{P_0 H_0}\quad \text{and} \quad
    \Comm{\Let{P'}}{\Let{H'}} \equiv \alias{P'H'}.
    \end{equation}
    All triple commutators vanish, since the double commutators are either $\Id$ or type-$\typZ$ elements, and we can use \Cref{lemma:cbdy:Z:commute}.
    We get:
    \begin{align*}
    \Comm{\Let{P}}{\Let{H}} &\equiv \Comm{\Let{P_0}\Let{P'}}{\Let{H_0}\Let{H'}} \tag{\Cref{lemma:cbdy:present}} \\
    &\equiv \alias{P_0H_0} \cdot \alias{P'H'} \tag{\Cref{eq:collecting commutators}} \\
    &\equiv \alias{P_0H_0 + P'H'}
    = \alias{PH} \tag{\Cref{lemma:cbdy:Z:linearity}},
    \end{align*}
    as desired.
\end{proof}

Finally, we prove:

\begin{proof}[Proof of \Cref{lemma:cbdy:establish Z}]
    We prove the first equality, as the second is dual.
    \Cref{claim:cbdy:matrix:decompose P and Q} gives us a decomposition
    $P = \sum_{\alpha=0}^{\nA/m-1} \sum_{\gamma=0}^{\floor{ \nC/m } } P^{(\alpha,\gamma)}$
    where $P^{(\alpha,\gamma)}$ is supported only on $\IntPartA{\alpha} \times \IntPartC{\gamma}$.
    The same argument gives an analogous decomposition
    $H = \sum_{\gamma'=0}^{\floor{ \nC/m } } \sum_{\delta=0}^{\nD/m-1} H^{(\gamma',\delta)}$,
    where $H^{(\gamma',\delta)}$ is supported only on $\IntPartC{\gamma'} \times \IntPartD{\delta}$.
    Hence
    \begin{multline*}
    \Comm{\Let{P}}{\Let{H}}
    = \Comm{ \Let{\sum_{\alpha=0}^{\nA/m-1} \sum_{\gamma=0}^{\floor{ \nC/m } } P^{(\alpha,\gamma)}}}{\Let{\sum_{\gamma'=0}^{\floor{ \nC/m } } \sum_{\delta=0}^{\nD/m-1} H^{(\gamma',\delta)}}} \\
    \equiv \Comm{\prod_{\alpha=0}^{\nA/m-1} \prod_{\gamma=0}^{\floor{ \nC/m } } \underbrace{\Let{P^{(\alpha,\gamma)}}}_{\eqqcolon u_{\alpha,\gamma}}}{\prod_{\gamma'=0}^{\floor{ \nC/m } } \prod_{\delta=0}^{\nD/m-1} \underbrace{\Let{H^{(\gamma',\delta)}}}_{\eqqcolon v_{\gamma',\delta}}}.
    \end{multline*}
    We again apply \Cref{eq:collecting commutators}.
    By \Cref{claim:cbdy:establish Z:block}, the pairwise commutators are \[
    \Comm{u_{\alpha,\gamma}}{v_{\gamma',\delta}}
    = \Comm{\Let{P^{(\alpha,\gamma)}}}{\Let{H^{(\gamma',\delta)}}}
    \equiv \alias{P^{(\alpha,\gamma)}H^{(\gamma',\delta)}} \]
    and by \Cref{claim:cbdy:establish Z:block,lemma:cbdy:Z:commute}, the triple commutators are \[
    \Comm{u_{\alpha',\gamma''}}{\Comm{u_{\alpha,\gamma}}{v_{\gamma',\delta}}}
    = \Comm{\Let{P^{(\alpha',\gamma'')}}}{\Comm{\Let{P^{(\alpha,\gamma)}}}{\Let{H^{(\gamma',\delta)}}}}
    \equiv \Comm{\Let{P^{(\alpha',\gamma'')}}}{\alias{P^{(\alpha,\gamma)}H^{(\gamma',\delta)}}}
    \equiv \zero \]
    and similarly for $\Comm{v_{\gamma'',\delta'}}{\Comm{u_{\alpha,\gamma}}{v_{\gamma',\delta}}}$.
    Hence by \Cref{eq:collecting commutators},
    \begin{align*}
    \Comm{\Let{P}}{\Let{H}}
    &\equiv \prod_{\alpha=\nA/m-1}^{0} \prod_{\gamma=\floor{ \nC/m }}^{0} \prod_{\gamma'=0}^{\floor{ \nC/m } } \prod_{\delta=0}^{\nD/m-1} \alias{P^{(\alpha,\gamma)} \cdot H^{(\gamma',\delta)}} \\
    &\equiv \alias{\sum_{\alpha=0}^{\nA/m-1} \sum_{\gamma=0}^{\floor{ \nC/m } } \sum_{\gamma'=0}^{\floor{ \nC/m } } \sum_{\delta=0}^{\nD/m-1}  P^{(\alpha,\gamma)} \cdot H^{(\gamma',\delta)}} \\
    &= \alias{\parens*{ \sum_{\alpha=0}^{\nA/m-1} \sum_{\gamma=0}^{\floor{ \nC/m } }  P^{(\alpha,\gamma)} }
    \cdot \parens*{ \sum_{\gamma'=0}^{\floor{ \nC/m } } \sum_{\delta=0}^{\nD/m-1}  H^{(\gamma',\delta)}}}
    = \alias{PH},
    \end{align*}
    as desired.
    Each invocation of \Cref{claim:cbdy:establish Z:block} takes $O(\rho^9)$ derivation steps,
    \Cref{lemma:cbdy:Z:commute} $O(\rho^9)$ steps, and \Cref{lemma:cbdy:Z:linearity} $O(\rho^7)$ steps.
    The dominant cost is invoking \Cref{claim:cbdy:establish Z:block} in the triple commutator calculation (which happens $O(\rho^6)$ times).
\end{proof}

\subsection{Remaining relations}

Finally, we verify the remaining Steinberg relations.

\begin{claim}\label{claim:cbdy:A commutator with PH}
    For every $A \in \cbdyA$, $P \in \cbdyP$, and $H \in \cbdyH$, \[
    \vdash_{15} \Comm{\Let{A}}{\Comm{\Let{P}}{\Let{H}}} \equiv \alias{APH}. \]
    Dually, for every $D \in \cbdyD$, $F \in \cbdyF$, and $Q \in \cbdyQ$, \[
    \vdash_{15} \Comm{\Let{D}}{\Comm{\Let{F}}{\Let{Q}}} \equiv \alias{D^\dagger FQ}. \]
\end{claim}

\begin{proof}
    We prove the first equality, as the second is dual.
    We have:
    \begin{align*}
    \Let{A} \Comm{\Let{P}}{\Let{H}}
    &\equiv \Let{A} \Let{P} \Let{H} \Let{-P} \Let{-H}
    \tag{inv. of $\typP$'s and $\typH$'s, \Cref{lemma:cbdy:present}} \\
    &\equiv \Let{AP} \Let{P} \Let{A} \Let{H} \Let{-P} \Let{-H}
    \tag{\Cref{eq:reorder} and comm. of $\typA$'s and $\typP$'s, \Cref{lemma:cbdy:present}} \\
    &\equiv \Let{AP} \Let{P} \Let{H} \Let{A} \Let{-P} \Let{-H}
    \tag{comm. of $\typA$'s and $\typH$'s, \Cref{lemma:cbdy:present}} \\
    & \equiv \Let{AP} \Let{P} \Let{H} \Let{-AP} \Let{-P} \Let{A} \Let{-H} 
    \tag{\Cref{eq:reorder} and comm. of $\typA$'s and $\typP$'s, \Cref{lemma:cbdy:present}} \\
    &\equiv \Let{AP} \Let{P} \Let{H} \Let{-AP} \Let{-P} \Let{-H} \Let{A} 
    \tag{comm. of $\typA$'s and $\typH$'s, \Cref{lemma:cbdy:present}} \\
    &\equiv \Let{P} \Comm{\Let{AP}}{\Let{H}} \Let{H} \Let{AP}  \Let{-AP} \Let{-P} \Let{-H} \Let{A}
    \tag{\Cref{eq:reorder}} \\
    &\equiv \Let{P} \cdot \alias{APH} \cdot \Let{H} \Let{AP}  \Let{-AP} \Let{-P} \Let{-H} \Let{A} 
    \tag{\Cref{lemma:cbdy:establish Z}} \\
    &\equiv \alias{APH} \cdot \Comm{\Let{P}}{\Let{H}} \Let{A}
    \tag{\Cref{lemma:cbdy:block:P and Z commute}}.
    \end{align*}
    The dominant cost is applying \Cref{lemma:cbdy:establish Z}.
\end{proof}

\begin{lemma}[Remaining Steinberg relations]\label{lemma:cbdy:A and D commututators with Z}
    For every $Z \in \cbdyZ$, $A \in \cbdyA$, and $D \in \cbdyD$: \[
    \vdash_{17}
    \left\{
    \begin{aligned}
        \Comm{\Let{A}}{\alias{Z}} &\equiv \alias{AZ}, \\
        \Comm{\Let{D}}{\alias{Z}} &\equiv \alias{ZD^\dagger}.
    \end{aligned}
    \right. \]
\end{lemma}

\newcommand{\rhoZ}{\rho_{\mathrm{Z}}}

\begin{proof}
    We prove the first equation, as the second is dual.
    By \Cref{def:cbdy:Z}, $Z$ is a sum of $\rhoZ = O(\rho^3)$ type-$\typZ$ basic block matrices: \[
    Z = \sum_{i=1}^{\rhoZ} \blockPartZ{\powZ^{(i)}}{\alpha^{(i)}}{\delta^{(i)}}{M^{(i)}}
    = \sum_{i=1}^{\rhoZ} P^{(i)} \cdot
    H^{(i)}. \]
    Therefore by \Cref{lemma:cbdy:Z:linearity}, \[
    \alias{Z} = \prod_{i=1}^{\rhoZ} \Comm{P^{(i)}}{H^{(i)}}. \]
    We now apply \Cref{eq:collecting commutators} a final time to: \[
    \Comm{\Let{A}}{\alias{Z}} = \Comm{\underbrace{\Let{A}}_{\eqqcolon g}}{\prod_{i=1}^{\rhoZ} 
    \underbrace{\Comm{\Let{P^{(i)}}}{\Let{H^{(i)}}}}_{\eqqcolon v_i}}. \]
    By \Cref{claim:cbdy:A commutator with PH}, the pairwise commutators are: \[
    \Comm{g}{v_i}
    = \Comm{\Let{A}}{\Comm{\Let{P^{(i)}}}{\Let{H^{(i)}}}}
    \equiv \alias{AP^{(i)}H^{(i)}}. \]
    and by \Cref{claim:cbdy:A commutator with PH,lemma:cbdy:establish Z}, the triple commutators are: \[
    \Comm{v_{i'}}{\Comm{g}{v_i}} = \Comm{\Comm{\Let{P^{(i')}}}{\Let{H^{(i')}}}}{\Comm{\Let{A}}{\Comm{\Let{P^{(i)}}}{\Let{H^{(i)}}}}}
    \equiv \Comm{\alias{P^{(i')}H^{(i')}}}{\alias{AP^{(i)}H^{(i)}}}
    \equiv \zero. \]
    Hence \[
    \Comm{\Let{A}}{\alias{Z}} = \Comm{g}{\prod_{i=1}^{\rhoZ} v_i} \equiv \prod_{i=1}^{\rhoZ} \alias{AP^{(i)}H^{(i)}}
    \equiv \alias{A \sum_{i=1}^{\rhoZ} P^{(i)}H^{(i)}} \equiv \alias{AZ}. \]
    Each invocation of \Cref{lemma:cbdy:establish Z} takes $O(\rho^{15})$ derivation steps, \Cref{claim:cbdy:A commutator with PH} $O(\rho^{15})$, and \Cref{lemma:cbdy:Z:linearity} $O(\rho^9)$.
    The dominant cost is invoking \Cref{claim:cbdy:A commutator with PH}, which happens $O(\rho^2)$ times.
\end{proof}

\subsection{Proof of \Cref{lemma:cbdy:missing}}

Given all this setup, we can now prove \Cref{lemma:cbdy:missing}, thereby establishing \Cref{lemma:cbdy:padded}.

\begin{proof}[Proof of \Cref{lemma:cbdy:missing}]
    \Cref{lemma:cbdy:P and Q commute} gives that $\Let{P}$ and $\Let{Q}$ commute.
    Type-$\typZ$ aliases (i.e., $\alias{Z}$) is defined in \Cref{def:cbdy:Z}.
    Linearity for type-$\typZ$ aliases is \Cref{lemma:cbdy:Z:linearity}.
    The commutation relation of type-$\typZ$ aliases with all $\typL$ elements for $\typL \in \{\typB,\typC,\typF,\typG,\typH,\typP,\typQ\}$ is in \Cref{lemma:cbdy:Z:commute}.
    The commutators of type-$\typZ$ aliases with type-$\typA$ and -$\typD$ elements are given in \Cref{lemma:cbdy:A and D commututators with Z}.
\end{proof}
\section{The Kaufman--Oppenheim complex: all the theorems} \label{sec:megatheorem}

In the following theorem we collect many properties of the KO complexes, including that they yield sparse low-soundness agreement testers and can be used in the \cite{BMVY25} PCP system.  Many of the properties were previously known or essentially known; it is the last ones, \Cref{itm:mega-diam,itm:mega-cob,itm:mega-triword,itm:mega-agree,itm:mega-pcp}, that are new to this paper.
\begin{theorem} \label{thm:megatheorem}
    For sufficiently large $n$ and $p \geq \poly(n)$ prime, the~\cite{KO23,KOW25} complexes $\CplxA{n}{\sorl}{\F_p}$ ($\kappa = 3$) for $\sorl > 3n$ satisfy the following properties:
    \begin{enumerate}
        \item \label{itm:mega-clique} $\CplxA{n}{\sorl}{\F_p}$ is an $(n+1)$-partite, pure rank-$(n+1)$ clique coset complex.
        \item \label{itm:mega-conn} All links of $\CplxA{n}{\sorl}{\F_p}$ are connected.
        \item \label{itm:mega-size} \label{KO:N} Defining $Q = p^s$, the complex $\CplxA{n}{\sorl}{\F_p}$ has $N$ vertices, for 
        \begin{equation}
            N = \frac{n+1}{p^{2n(n+1)}}\cdot |\mathrm{SL}_{n+1}(\F_Q)|, \quad |\mathrm{SL}_{n+1}(\F_Q)| = \tfrac{1}{Q-1} \prod_{i=0}^{n} (Q^{n+1}-Q^i).
        \end{equation}
        \item \label{itm:mega-target-size} For a fixed $n$, and  given targets $N_0$ and $p_0 \leq \polylog N_0$, one can in deterministic $\polylog N_0$ time find values $p_0 \leq p \leq O(p_0)$ and $s$ and hence~$N$ such that $N_0 \leq N \leq N_0^{1+O(1/p_0^{.4})}$, the latter bound being $O(N_0)$ if $p_0 \geq \log^{2.5} N$.
        \item \label{itm:mega-explicit}  For fixed $n$, and assuming $p \leq \polylog N$, the complexes $\CplxA{n}{\sorl}{\F_p}$ are strongly explicit, in the sense that the following things can be done in $\polylog N$ time: Computing $N$, picking a uniformly random vertex, computing the parts and neighborhoods of a vertex.
        \item \label{itm:mega-degree} $\CplxA{n}{\sorl}{\F_p}$ is regular: every vertex participates in $\frac{p^{2n(n+1)}}{n+1}$ maximal cliques.
        \item \label{itm:mega-links} Every vertex-link is isomorphic to $\LinkCplxA{n}{\F_p}$, an $n$-partite, pure rank-$n$ clique coset complex.
        \item \label{itm:mega-symmetric} The automorphism group of $\CplxA{n}{\sorl}{\F_p}$ acts transitively on the maximal cliques, and the same is true of $\LinkCplxA{n}{\F_p}$.
        \item \label{itm:mega-one-sided} Every link $\CplxA{n}{\sorl}{\F_p}_w$ with $|w| \leq n-1$ has normalized $2$nd eigenvalue at most $\frac{1}{\sqrt{p} - n}$.
        \item \label{itm:mega-product} Every pinned bipartite restriction $\CplxA{n}{\sorl}{\F_p}_w[\{i\},\{j\}]$ has normalized $2$nd eigenvalue at most $\frac{2}{\sqrt{p}}$.
        \item \label{itm:mega-two-sided} The graph underlying $\CplxA{n}{\sorl}{\F_p}$ has all nontrivial normalized eigenvalues at most $\frac{2}{d}$ in absolute value.
        \item \label{itm:mega-diam} Within $\LinkCplxA{n}{\F_p}$, the bipartite graph between parts $a$ and $b$ has diameter $\poly(1/\xi)$ provided $|b-a| \geq \xi n$.
        \item \label{itm:mega-cob} Within any pinning (link) $\LinkCplxA{n}{\F_p}_w$, provided $a < b < c$ are inside a contiguous sequence of unpinned parts, the triangle complex on parts $a, b, c$ has $1$-coboundary expansion at least $\poly(\xi)$ over any group, provided $b - a, c -b \geq \xi n$.
        \item \label{itm:mega-triword} Provided $r \ll \sqrt{n}$, 
        $\LinkCplxA{n}{\F_p}$ has 
        $r$-triword expansion at least $\exp(-O(r^{.67}))$ over any group (i.e., its $r$-faces complex has $1$-coboundary expansion at least $\exp(-O(r^{.67}))$).
        \item \label{itm:mega-agree} For $\delta > 0$, provided $k \geq \exp(\poly(1/\delta))$ and $n \geq \exp(\poly(k))$, the V-test performed with the size-$k$ sub-edges of $\CplxA{n}{\sorl}{\F_p}$ is a $\delta$-sound agreement tester in the sense of \Cref{def:agreement-tester}.
        \item \label{itm:mega-pcp} $\CplxA{n}{\sorl}{\F_p}$ may be used for the direct-product testing component of the Bafna--Minzer--Vyas PCP system (specifically, in \cite[Lemma 5.3]{BMVY25}).  This achieves their arbitrary constant $\delta >0$ soundness with $N \mapsto N \polylog N$ proof length using only elementary ingredients.
    \end{enumerate}
\end{theorem}
\begin{proof}
    \Cref{itm:mega-clique} is mostly by definition; the fact that $\CplxA{n}{\sorl}{\F_p}$ is a clique complex is \Cref{prop:KO:clique}.  
    \Cref{itm:mega-conn} is \Cref{prop:ko:connected}.  
    \Cref{itm:mega-size} arises from the definition of $\CplxA{n}{\sorl}{\F_p}$ and is discussed in \Cref{sec:high-dp-concrete}.
    \Cref{itm:mega-target-size} takes a little work.  Treating $n$ as constant, note that $N = p^{\Theta(\sorl)}$.  Thus given $p_0$ and $N_0$, the desired $\sorl$ will satisfy $\sorl = \Theta(\frac{\log N_0}{\log p_0})$.  Adjusting $s$ by $\pm 1$ changes $N$ by a factor of $p^{\Theta(1)}$; so assuming $p_0 \leq \polylog N$, an algorithm can (in $\polylog N$ time) determine an~$s$ that gets $N$ into the range $\rangeII{\frac{N_0}{\polylog N_0}}{N_0}$.  Next, the algorithm can consider adjusting the prime $p_0$ upward.  It is known~\cite{BHP01} that for any prime $p$, there is another that is at most $p + O(p^{.6})$.  Thus (in $\polylog N$ time) an algorithm can increase $p$ by a factor of $1 + O(1/p^{.4}) = \exp(O(1/p^{.4}))$, which increases $N$ by a factor of $\exp(O(s/p^{.4})) \leq N^{O(1/p^{.4})}$.  At most $O(p_0^{.4})$ such repeated adjustments get $N$ and $p$ into the claimed ranges.
    
    The preceding analysis covers the ``computing~$N$'' portion of \Cref{itm:mega-explicit}, and the fact that computing the incidence structure is also computable in $\polylog N$ time is clear from the description in \Cref{sec:high-dp-concrete}.  (The equivalence classes are of size $\poly(p)$ and can be maintained explicitly; the computations over $\F_Q$ can also be done in $\polylog N$ time.) 
    The regularity and degree facts of \Cref{itm:mega-degree} are in \Cref{sec:high-dp-concrete}.
    \Cref{itm:mega-links} is \Cref{prop:ko:vertex links}, and \Cref{itm:mega-symmetric} is \Cref{prop:coset complex:strongly symmetric}.

    Moving on to eigenvalue bounds, \Cref{itm:mega-one-sided} is from \cite[Theorem 2.8]{KOW25}.  \Cref{itm:mega-product} is \Cref{cor:ko:product}. \Cref{itm:mega-two-sided} follows from the one-sided spectral bound in \Cref{itm:mega-one-sided}, together with $(n+1)$-partiteness, trickling down~\cite{Opp18}, and $p \gg n^2$; see the discussion of trickling down in \Cref{sec:sufficient conditions for direct-product testing}.

    We come to the new results in our work.
    \Cref{itm:mega-diam} is from \Cref{lemma:ko:diam},  \Cref{itm:mega-cob} is \Cref{thm:main-separated}, and   \Cref{itm:mega-triword} is \Cref{thm:main-coboundary}.
    With these and preceding results in hand, we can now finally prove that the KO complex is a $\delta$-sound agreement tester (i.e., establish \Cref{thm:main-dp} and \Cref{itm:mega-agree}) by appealing to Bafna--Minzer/Dikstein--Dinur's \Cref{thm:sufficient conditions}; the only piece missing for that theorem is its trivial global cohomology condition, which is provided by Kaufman--Oppenheim--Weinberger's \Cref{thm:KOWabunga}.
    At last, to verify the PCP result \Cref{itm:mega-pcp}, we observe \cite[Lemma 5.3]{BMVY25} requires a complex that satisfies the long list of properties in \cite[Theorem 2.13]{BMVY25}; this theorem shows that all those properties are satisfied except one.  The lacuna is that  \cite[Theorem 2.13]{BMVY25} formally requires the diameter bound from \Cref{itm:mega-diam} to be $O(1/\xi)$, whereas we have only $\poly(1/\xi)$.  However it is not hard to check that having this slightly worse bound makes no material difference in the one place it is used, \cite[Section 4.2]{BMVY25}.
\end{proof}

\printbibliography

\appendix

\section{Coboundary expansion of ``linearly ordered'' complexes}\label{sec:triword expansion in linear complexes}

\newcommand{\vol}{\mathsf{vol}}
\newcommand{\incons}{\mathsf{incons}}
\newcommand{\viol}{\mathsf{viol}}
\newcommand{\Bad}{\mathsf{mismatch}}

In this appendix, we prove \Cref{thm:linear condition}, which we restate here for convenience:

\line*

We emphasize that the proof of \Cref{thm:linear condition} is extracted from the analysis of the triword expansion of the~\cite{CL25-applications} complex in the work of~\textcite{BLM24}.
Our contribution in this appendix is to restate their method in a generic way which can be applied to other complexes (and may be useful to future work)
and achieves a slightly better final expansion bound (though this does not matter for applications).
(We do also make a slight technical adjustment to their lemma, since we prove triword expansion without additive error,
though this is also essentially already done in the work of~\cite{DD24-swap}.)

\subsection{Coboundary expansion from pinned triword expansion}

We now state a lemma due to~\textcite{BLM24} which can be used to prove triword expansion of an induced triangle complex via triword expansion of ``smaller'' complexes.
The lemma is not stated explicitly in the work of~\cite{BLM24}, but it is present as equation (8) within the proof of Lemma 3.6 in that paper.

The general setup is as follows: Suppose we have $\X$ an $\calI$-indexed simplicial complex,
subsets $R_1,R_2,R_3 \sqsubset_{\le r} \calI$,
and subsets $S_1,S_2,S_3 \sqsubset R_1 \cup R_2 \cup R_3$.
We relate the triword expansion of the ``global'' triangle complex $\Tri \coloneqq \res{\X}{R_1,R_2,R_3}$
to the triword expansion of the ``central'' triangle complex $\Tri^* \coloneqq \res{\X}{S_1,S_2,S_3}$
as well as the ``peripheral'' complexes $\Tri_{s_j} \coloneqq \resLink{\X}{R_1\setminus S_j, R_2\setminus S_j, R_3\setminus S_j}{s_j}$,
where $j \in \{1,2,3\}$ and $s_j \in \supp{\res{\X}{S_j}}$ is a valid $S_j$-word.
The peripheral complex $\Tri_{s_j}$ can be viewed as an ``induced subcomplex'' of $\Tri$ only on words consistent with $s_j$ on $S_j$.
In particular, every word/biword/triword in the peripheral complex $\Tri_{s_j}$ can also be viewed as a word/biword/triword in $\Tri$.
($\Tri^*$ can't be viewed this way.)
See \Cref{fig:pinning} below for a depiction of this situation.
Note that, although $S_1$, $S_2$, and $S_3$ need to be pairwise disjoint, 
they need not interact nicely with $R_1$, $R_2$, and $R_3$'s; each $R_i$ could intersect each $S_j$.
They also need not have the same size.

\begin{lemma}[{from \cite[\S3]{BLM24}}]\label{lemma:pinning}
    There exist absolute constants $\Kpin > 1, \epin > 0$ such that the following holds.
    Let $\X$ be an $\calI$-indexed simplicial complex, $\Gamma$ be any coefficient group,
    $r \in \N$, $R_1, R_2, R_3 \sqsubset_{\le r} \calI$,
    and $S_1,S_2,S_3 \sqsubset R_1\cup R_2\cup R_3$.
    Let $0 < \beta, \beta^* < 1$, and suppose that $\X$ satisfies the following conditions:
    \begin{conditions}
        \conditem{Spectral expansion}
        $\X$ is $\epin r^{-2}$-product.
        \conditem{``Central'' triword expansion}
        $\Tri^* \coloneqq \res{\X}{S_1,S_2,S_3}$ has $1$-triword expansion at least $\beta^*$ over $\Gamma$.
        \label{item:pinning:central}
        \conditem{``Peripheral'' triword expansion}
        For each $j \in \{1,2,3\}$ and $s_j \in \supp(\res{\X}{S_j})$,
        $\Tri_{s_j} \coloneqq \resLink{\X}{R_1\setminus S_j,R_2\setminus S_j,R_3\setminus S_j}{s_j}$
        has $1$-triword expansion at least $\beta$ over $\Gamma$.\footnote{
            The subscript-$j$ is included to emphasize which coordinates are pinned.}
        \label{item:pinning:peripheral}
    \end{conditions}
    Then $\Tri \coloneqq \res{\X}{R_1,R_2,R_3}$ has $1$-triword expansion at least $\Kpin \beta^* \beta$ over $\Gamma$.
\end{lemma}

The constant $\Kpin$ is reasonable (at least $0.01$).

\begin{remark}
    The proof of this lemma due to \textcite{BLM24} also allows additive loss in the triword expansion, but we will not need this.
\end{remark}

\begin{figure}
\center
\begin{tikzpicture}[scale=1.75]

\def\rectx{0.2}
\def\rectwidth{3}
\def\rectheight{0.5}

\def\boxA{red!20!white}
\def\boxB{blue!20!white}
\def\boxC{green!20!white}



\newcommand{\drawBlock}[5][]{%
  \begin{scope}[shift={#2}, rotate=#3]
    \pgfmathsetmacro{\N}{#4+1} 
    \pgfmathsetmacro{\step}{(\rectwidth-2*\rectx)/\N}
    
    \foreach \kv in {#1} {
      \foreach \i/\c in \kv {
        \pgfmathsetmacro{\xleft}{-\rectwidth/2+\rectx+\i*\step}
        \pgfmathsetmacro{\xright}{\xleft+\step}
        \fill[\c] (\xleft,0) rectangle (\xright,-\rectheight);
      }
    }

    \draw (-\rectwidth/2+\rectx,0) rectangle (\rectwidth/2-\rectx,-\rectheight);
    \foreach \i in {1,...,#4} {
      \draw ({-\rectwidth/2+\rectx+\i*\step},0) -- ({-\rectwidth/2+\rectx+\i*\step},-\rectheight);
    }

    \foreach \i in {0,...,#4} {
      \node at ({-\rectwidth/2+\rectx+(\i+0.5)*\step},-\rectheight/3) {$0$};
      \node at ({-\rectwidth/2+\rectx+(\i+0.5)*\step},-.7*\rectheight) {$1$};
    }
    \node at (0, -0.8) {#5};
  \end{scope}
}

\drawBlock[{2/\boxA,4/\boxA,5/\boxC}]{(0,{-sqrt(3)/6*\rectwidth})}{0}{5}{$R_1$}
\drawBlock[{0/\boxB,1/\boxA,3/\boxB}]{({-1/4*\rectwidth},{(sqrt(3)/4-sqrt(3)/6)*\rectwidth})}{-120}{7}{$R_2$}
\drawBlock[{1/\boxC,3/\boxB,4/\boxB}]{({1/4*\rectwidth},{(sqrt(3)/4-sqrt(3)/6)*\rectwidth})}{120}{4}{$R_3$}

\end{tikzpicture}
\caption{The situation posited in \Cref{lemma:pinning}, in the special case where all the underlying vertex-sets $\X[i]$ are $\{0,1\}$.
That is, let $\X$ be a distribution on $\{0,1\}^n$.
Let $R_1,R_2,R_3\sqsubset \calI$ and $S_1,S_2,S_3\sqsubset R_1\cup R_2\cup R_3$.
The overall triangle complex $\Tri$ is the distribution on ``triwords''
$\Tri = \res{\X}{R_1,R_2,R_3}$,
i.e., on triples in $\{0,1\}^{R_1} \times \{0,1\}^{R_2} \times \{0,1\}^{R_3}$.
The sets $S_1,S_2,S_3$ are depicted in red, green, and blue, respectively.
The central complex $\Tri^* = \res{\X}{S_1,S_2,S_3}$ is just a distribution on triples in $\{0,1\}^{S_1} \times \{0,1\}^{S_2} \times \{0,1\}^{S_3}$.
Each peripheral complex $\Tri_{s_j} = \resLink{\X}{R_1\setminus S_j,R_2\setminus S_j,R_3\setminus S_j}{s_j}$
is a distribution on $\{0,1\}^{R_1\setminus S_j} \times \{0,1\}^{R_2\setminus S_j} \times \{0,1\}^{R_3\setminus S_j}$
resulting from conditioning on the $S_j$-coordinates equaling $s_j$.}
\label{fig:pinning}
\end{figure}
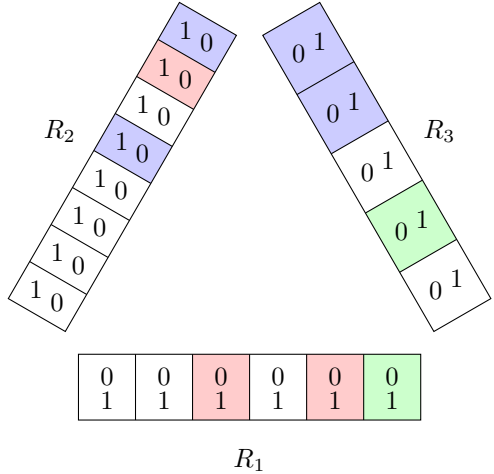

\subsection{More conditions for triword expansion of induced complexes}

Next, we state another lemma due to \textcite{BLM24} which repeatedly applies their pinning lemma (\Cref{lemma:pinning} in this paper)
to prove coboundary expansion of induced complexes when the corresponding restrictions satisfy certain ``spreadness'' conditions.
In particular, we make the following definition, which is a generic condition of \cite[Item 3 in Assumption 1]{BLM24}.
(Similar notions of spreadness are considered in \cite[Def. 8.3]{DD24-swap} and \cite[Def. 6.12]{DDL24}.)

\begin{definition}[Pseudorandom spreadness]\label{def:pseudorandom spreadness}
    Let $r \le n \in \N$, $R \subset_r \rangeOne{n}$, and $T, \eta \in \N$.
    $R$ is \emph{$(T,\eta)$-pseudorandomly spread} if the following holds:
    Define $\calI_t \coloneqq \rangeEI{(t-1) \frac{n}T}{t \frac{n}T} $ for $t \in \rangeOne{T}$,
    so that each $\calI_t$ has size $\frac{n}T$ and $\calI_1,\ldots,\calI_T$ partition $\rangeOne{n}$.
    Then for every $t \in \rangeOne{T}$, \[
    \frac{r}T - \eta \le \abs* { R_i \cap \calI_t } \le \frac{r}T + \eta. \qedhere \]
\end{definition}
Note that each $R_i$ has size $r$, so by linearity of expectation,
$\Exp_{\bt \sim \DUnif{\rangeOne{T}}} \bracks*{ |R_i \cap \calI_\bt| } = r \cdot \frac{1}T = \frac{r}T$.

\begin{lemma}[{``Lopsided induction'', \cite[Lemma~3.5]{BLM24}}]\label{lemma:lopsided induction}
    Let $\Kpin > 1$, $\epin > 0$ be the absolute constants from \Cref{lemma:pinning}.
    Let $\Gamma$ be any group.
    Let $r \le n \in \N$,
    let $\X$ be an $\calI$-indexed simplicial complex.
    Let $\beta^* \le 1$ and suppose that $\X$ satisfies the following conditions:
    \begin{conditions}
        \conditem{Spectral expansion}
        $\X$ is an $\epin r^{-2}$-local spectral expander.
        \conditem{Singleton coboundary expansion}
        For every distinct $i_1,i_2,i_3 \in \calI$, the triangle complex
        $\res{\X}{\{i_1\},\{i_2\},\{i_3\}}$
        has $1$-triword expansion at least $\beta^*$ over $\Gamma$.
    \end{conditions}
    Then for every $k_1,k_2,k_3 \in \N$ with $k_1 + k_2 + k_3 \le r$,
    $R_1,R_2,R_3 \sqsubset \calI$ with $|R_1| = k_1$, $|R_2| = k_2$, $|R_3| = k_3$, the triangle complex $\Tri \coloneqq \res{\X}{R_1,R_2,R_3}$
    is has $1$-triword expansion at least $(\Kpin \beta^*)^{k_1+k_2+k_3}$ over $\Gamma$.
\end{lemma}

The following lemma reduces proving triword expansion of a pseudorandomly spread restriction to triword expansion of ordered restrictions:

\begin{lemma}[{``Non-lopsided induction'', \cite[Lemma~3.6]{BLM24}}]\label{lemma:balanced induction}
    Let $\Kpin > 1$, $\epin > 0$ be the absolute constants from \Cref{lemma:pinning}.
    Let $\X$ be an $\rangeOne{n}$-indexed simplicial complex and $\Gamma$ be a coefficient group.
    Let $r \le n \in \N$, $R_1,R_2,R_3 \sqsubset_r \rangeOne{n}$, and $k, T, \eta \in \N$.
    Let $\beta^* < 1$.
    Suppose that $\X$ satisfies the following conditions:
    \begin{conditions}
    \conditem{Spectral expansion}
    $\X$ is $\epsilon_0 r^{-2}$-product.
    \conditem{Coboundary expansion of ordered restrictions}
    For every $W \subset_{\le r-3k} R_1 \cup R_2 \cup R_3$,
    restriction $w \in \supp(\res{\X}{W})$,
    $S_1 \prec S_2 \prec S_3 \sqsubset_k R_1 \cup R_2 \cup R_3 \setminus W$,
    the complex $\resLink{\X}{S_1,S_2,S_3}{w}$ has $1$-triword expansion at least $\beta^*$ over $\Gamma$.\label{item:balanced induction:ordered expansion}
    \conditem{Pseudorandom spreadness}
    For each $i \in \{1,2,3\}$, $R_i$ is $(T,\eta)$-pseudorandomly spread.
    \label{item:balanced induction:pseudorandom spread}
    \end{conditions}
    Then the triangle complex $\res{\X}{R_1,R_2,R_3}$ has $1$-triword expansion at least $\beta'$ over $\Gamma$, where \[
    \beta' \coloneqq (\Kpin \beta^*)^{3(T(k+2\eta)+2(r/T))+r/k}. \]
\end{lemma}

\begin{remark}\label{rmk:induction}
The proof of this lemma in~\cite{BLM24} sets $k \coloneqq r^{0.01}$, $T \coloneqq r/k^{10} = r^{0.9}$, $\eta \coloneqq k^6 = r^{0.06}$, and $\tau \coloneqq k^9 = r^{0.09}$,
giving an exponent $O(r^{0.9}(r^{0.09}+r^{0.06})+r^{0.1}+r^{0.99}) = O(r^{0.99}$).
Instead, we will set $k \coloneqq r^{\tfrac13+\delta}$, $T \coloneqq r^{\tfrac13+\delta}$, and $\eta \coloneqq r^{\tfrac13+\delta}$.
This gives an exponent of $O(r^{\tfrac23+\delta} + r^{\tfrac23}) = O(r^{\tfrac23+\delta})$.
We therefore give a reproof of the lemma so that we can account for the exponent with generic parameters, which we can then optimize (though it does not matter for the final testing result).
We stress that our proof exactly follows theirs.
\end{remark}

\begin{proof}
    Set $\tau \coloneqq k+\eta$, so that $\tau - \eta = k$.
    By assumption, for every $i \in \{1,2,3\}, t\in \rangeOne{T}$, we have $|R_i \cap \calI_t| = \frac{r}{T} + c_{i,t}$ for some $|c_{i,t}| \le \eta$,
    and for each $i \in \{1,2,3\}$, $\sum_{t=1}^T c_{i,t} = 0$.
    The argument is recursive.

    \paragraph{The recursion.}
    Every node of the recursion tree is of the form $N = (R'_1,R'_2,R'_3;W)$ maintaining the following invariants:
    \begin{conditions}
        \item $W \subseteq \rangeOne{n}$ and for every $i \in \{1,2,3\}$, $R_i' \subseteq R_i \setminus W$.
        (Hence, $R'_1$, $R'_2$, and $R'_3$ are pairwise disjoint.)
        \item $|R_1'| = |R_2'| = |R_3'| = r - kD$, where $D$ is the depth of the node $N$.
        \item For every $t \in \rangeOne{T}$, there exists $0 \le r_t \le \frac{r}T$ such that
        for every $i \in \{1,2,3\}$,
        $|R'_i \cap \calI_t| = r_t + c_{i,t}$.
        \item $N$ is an \emph{internal node} if and only if there exist $t_1 < t_2 < t_3 \in \rangeOne{T}$ with $r_{t_j} \ge \tau$ for each $j \in \{1,2,3\}$.
    \end{conditions}    
    The root of the tree is $N^0 \coloneqq (R_1,R_2,R_3;\emptyset)$.
    Given an internal node $N = (R'_1,R'_2,R'_3;W)$, we recurse to three nodes $N_j$, one for each $j \in \{1,2,3\}$.
    Fix $j \in \{1,2,3\}$.
    Since $R_1',R_2',R_3'$ are pairwise disjoint,
    and $|R_i' \cap \calI_{t_j}| = r_{t_j} + c_{i,t_j} \ge \tau - \eta = k$ for each $i \in \{1,2,3\}$,
    we can pick $S_j \subseteq (R'_1 \cup R'_2 \cup R'_3) \cap \calI_{t_j}$ such $|R_i' \cap S_j|=k$ for each $i \in \{1,2,3\}$,
    and then $N_j \coloneqq (R'_1\setminus S_j, R'_2\setminus S_j, R'_3\setminus S_j;W \cup S_j)$.

    In any leaf node $N = (R_1,R_2,R_3;W)$, each $r_t < \tau$ for all but (at most) two intervals $t \in \rangeOne{T}$,
    and therefore for every $i \in \{1,2,3\}$,
    \begin{equation}\label{eq:ind:rho}
    |R'_i| = \sum_{t=1}^T |R'_i \cap \calI_t| = \sum_{t=1}^T (r_t + c_{i,t}) \le T\eta + \sum_{t=1}^T r_t \le (T-2) \tau + 2 \parens*{ \frac{r}T } = T (\tau + \eta) + 2 \parens*{ \frac{r}T-\tau } \eqqcolon \rho.
    \end{equation}

    \paragraph{Inductive argument.}
    We claim that inductively, every node $N = (R'_1,R'_2,R'_3;W)$ of height $h$ satisfies the following:

    \begin{quote}
        For every $w \in \supp(\res{\X}{W})$, the complex $\resLink{\X}{R'_1,R'_2,R'_3}{w}$
        has $1$-triword expansion at least $(\Kpin \beta^*)^{3\rho+h}$ over $\Gamma$.
    \end{quote}
    Noting that the root height is at most $r/k$ (since at a leaf, the sets must be nonempty, they start at size $r$, and decrease by $k$ at each iteration),
    this gives the desired bound.
    
    Consider a node $N = (R'_1,R'_2,R'_3;W)$.
    If $N$ is a leaf node, then we apply the lopsided induction lemma (\Cref{lemma:lopsided induction}).
    Indeed, for every $i_1 \in R'_1, i_2 \in R'_2, i_3 \in R'_3$, letting $i'_1,i'_2,i'_3$ be the sorted order of $i_1,i_2,i_3$
    (so that $\{i_1,i_2,i_3\} = \{i'_1,i'_2,i'_3\}$ and $i'_1 < i'_2 < i'_3$),
    and every $w \in \supp(\res{\X}{W})$, $\{i'_1\} \prec \{i'_2\} \prec \{i'_3\}$,
    so the assumption gives that $\resLink{\X}{\{i_1\},\{i_2\},\{i_3\}}{w} = \resLink{\X}{\{i'_1\},\{i'_2\},\{i'_3\}}{w}$ has $1$-triword expansion at least $\beta^*$.
    Hence by \Cref{lemma:lopsided induction},
    for every $w \in \supp(\link{\X}{w})$, $\resLink{\X}{R'_1,R'_2,R'_3}{w}$
    has $1$-triword expansion at least $(\Kpin \beta^*)^{3\rho}$
    (using that $|R'_1| + |R'_2| + |R'_3| \le 3\rho$ by \Cref{eq:ind:rho}).

    Otherwise, $N$ is an internal node, say of height $h$.
    Its three children $N_1,N_2,N_3$ are all nodes of height (at most) $h-1$.
    Recall that each $N_j = (R'_1 \setminus S_j, R'_2 \setminus S_j, R'_3 \setminus S_j; W \cup S_j)$.
    By the inductive hypothesis, for each $j \in \{1,2,3\}$, $(w,s_j) \in \supp(\res{\X}{W,S_j})$,
    the peripheral complex \[ \Tri_{s_j} \coloneqq \resLink{\X}{R'_1\setminus S_j, R'_2\setminus S_j, R'_3\setminus S_j}{w, s_j} \]
    has $1$-triword expansion at least $(\Kpin \beta^*)^{3\rho+(h-1)}$.
    Further, $S_1 \prec S_2 \prec S_3$ by construction (since each $S_j \subseteq \calI_{t_j}$);
    hence, by \Cref{item:balanced induction:ordered expansion}, the central complex $\resLink{\X}{S_1,S_2,S_3}{w}$ has $1$-triword expansion at least $\beta^*$.
    Together, \Cref{lemma:pinning} therefore gives that $\resLink{\X}{R'_1,R'_2,R'_3}{w}$ has $1$-triword expansion at least $\beta$ coboundary expander for \[
        \beta \coloneqq \Kpin  \beta^* \cdot (\Kpin \beta^*)^{3\rho+(h-1)} = (\Kpin \beta^*)^{3\rho+h}, \]
    as desired.
\end{proof}

\subsection{Diameter bound from pinned diameter bound}

We now prove the following analogue of \Cref{lemma:pinning} for the case of ``rank-two'' induced complexes (i.e., bipartite graphs).
This lemma seems to be used implicitly in the proof of \cite[Lemma 4.19]{BLM24}; we make it explicit here.

\begin{lemma}\label{lemma:pinning diameter}
    There exist absolute constants $K > 1, \epsilon_0 > 0$ such that the following holds.
    Let $r \in \N$ and $\X$ be an $\calI$-indexed simplicial complex.
    Let $R_1, R_2 \sqsubset \calI$ and $S_1,S_2 \sqsubset R_1\cup R_2$.
    Suppose that $\X$ satisfies the following conditions:
    \begin{conditions}
        \conditem{``Central'' diameter bound}
        $\Avg^* \coloneqq \res{\X}{S_1,S_2}$ has $\diam(\Avg^*) \le C_0^*$.
        \label{item:pinning diameter:central}
        \conditem{``Peripheral'' diameter bound}
        For each $j \in [2]$ and $s_j \in \supp(\res{\X}{S_j})$,
        $\Avg_{s_j} \coloneqq \resLink{\X}{R_1\setminus S_j,R_2\setminus S_j}{s_j}$
        has $\diam(\Avg_{s_j}) \le C_0$.
        \label{item:pinning diameter:peripheral}
    \end{conditions}
    Then $\Avg \coloneqq \res{\X}{R_1,R_2}$ has $\diam(\Avg) \le 6C_0^* C_0$.
\end{lemma}

\begin{proof}
    Fix any $r_1 \in \supp(\res{\X}{R_1})$ and $r_2 \in \supp(\res{\X}{R_2})$;
    we show that their distance in $\Avg$ is at most $3C_0^*C_0$.
    This gives the desired bound by the triangle inequality.
    
    Pick an arbitrary $S_1$-word $s_1 \in \supp(\resLink{\X}{S_1}{r_1})$ and $S_2$-word $s_2 \in \supp(\resLink{\X}{S_2}{r_2})$.
    Since $s_1$ is an $S_1$-word and $s_2$ is an $S_2$-word, there exists a path between $s_1$ and $s_2$ them of alternating $S_2$- and $S_1$-words of length $k \le C^*_0$.
    That is, there exists, for each $i \in [k]$, an $S_1$-word $s_1^{(i)} \in \supp(\res{\X}{S_1})$
    and $S_2$-word $s_2^{(i)} \in \supp(\res{\X}{S_2})$,
    such that: $s_1^{(1)} = s_1$, $s_2^{(k)} = s_2$,
    $(s_1^{(i)},s_2^{(i)}) \in \supp(\res{\X}{S_1,S_2})$ is an $(S_1,S_2)$-biword (for each $i \in [k]$),
    and $(s_1^{(i+1)},s_2^{(i)}) \in \supp(\res{\X}{S_1,S_2})$ is an $(S_1,S_2)$-biword (for each $i \in [k-1]$).
    
    Next, we ``lift'' this path of $S_1$- and $S_2$-words to a path of $R_1$- and $R_2$-words in $\Avg$.
    For each $i \in [k]$, pick an $(R_1,R_2)$-biword $(u_1^{(i)},u_2^{(i)}) \in \supp(\resLink{\X}{R_1,R_2}{s_1^{(i)},s_2^{(i)}})$ consistent with $s_1^{(i)}$ and $s_2^{(i)}$.
    Similarly, for $i \in [k-1]$, pick an $(R_1,R_2)$-biword $(v_1^{(i+1)},v_2^{(i)}) \in \supp(\resLink{\X}{R_1,R_2}{s_1^{(i+1)},s_2^{(i)}})$  consistent with $s_1^{(i+1)}$ and $s_2^{(i)}$.
    We also define $v_1^{(1)} \coloneqq r_1$ and $v_2^{(k)} \coloneqq r_2$.

    Finally, we claim that: For every $j \in [2]$ and $i \in [k]$, 
    there is a path of length $C_0$ in $\Avg$ between $u_j^{(i)}$ and $v_j^{(i)}$.
    Note that these are both $R_j$-words consistent with the $S_j$-word $s_j^{(i)}$.
    (I.e., both are in $\supp(\resLink{\X}{R_j}{s_j^{(i)}})$.)
    Incorporating the $(R_1,R_2)$-biwords $(u_1^{(i)},u_2^{(i)})$ and $(v_1^{(i+1)},v_2^{(i)})$ gives an overall path \[
    r_1 = v_1^{(1)} \rightsquigarrow u_1^{(1)} \to u_2^{(1)} \rightsquigarrow v_2^{(1)} \to v_1^{(2)}
    \rightsquigarrow u_1^{(2)} \to v_1^{(2)} \to \cdots \to u_2^{(k)} \rightsquigarrow v_2^{(k)} = r_2. \]
    The total length of this path is \[
    C_0(2k+1) + (2k-1) \le C_0(2C_0^*+1) + (2C_0^*-1) \le 3C_0^*C_0. \]

    Now, we prove the claim for the $R_j$-words $u_j^{(i)}$ and $v_j^{(i)}$.
    Recalling that $u_j^{(i)},v_j^{(i)} \in \supp(\resLink{\X}{R_j}{s_j^{(i)}})$,
    we recall that there is a path of length at most $C_0$ between the $(R_j \setminus S_j)$-words $u_j^{(i)}|_{R_j \setminus S_j}$
    and $v_j^{(i)}|_{S_j \setminus S_j}$ in the link $\resLink{\X}{R_1\setminus S_1,R_2 \setminus S_2}{s_1}$.
    This path lifts to a path of length $C_0$ between $u_j^{(i)}$ and $v_j^{(i)}$ by appending $s_j$ to every word.
\end{proof}

\subsection{Coboundary expansion of ordered restrictions}

The following lemma due to \textcite{BLM24} allows us to prove triword expansion in symmetric complexes which productize:

\begin{lemma}[{\cite[Lem.~4.18]{BLM24}}]\label{thm:triword expansion given productization}
    There exists an absolute constant $K > 0$ such that the following holds.
    Let $n \in \N$ and $\X$ an $\calI$-indexed simplicial complex.
    If $\X$ is a strongly symmetric clique complex,
    then for every $S_1,S_2,S_3 \sqsubset \calI$ such that $\ptiz{\X}{S_1}{S_j}$ for both $j \in \{2,3\}$,
    $\res{\X}{S_1,S_2,S_3}$ has $1$-triword expansion at least $K / \diam(\res{\X}{S_2,S_3})$ over every coefficient group $\Gamma$.
\end{lemma}

(This lemma is very similar to \cite[Claim~2.11]{DDL24}.
The hypothesis that $\X$ is a clique complex appears to be necessary,
but is omitted from both \cite[Lem.~4.18]{BLM24} and \cite[Claim~2.11]{DDL24}.)

\begin{proposition}\label{prop:diam from prod}
    Let $\X$ be an $\calI$-indexed simplicial complex.
    If $S_1,S_2 \sqsubset \calI$ such that $\ptiz{\X}{S_1}{S_2}$, then $\diam(\res{\X}{S_1,S_2}) \le 2$.
\end{proposition}

\begin{proof}
    We claim that for every $S_1$-word $s_1 \in \supp(\res{\X}{S_1})$ and $S_2$-word $s_2 \in \supp(\res{\X}{S_2})$,
    $(s_1,s_2)$ is a valid $(S_1,S_2)$-biword, i.e., $(s_1,s_2) \in \supp(\res{\X}{S_1,S_2})$.
    Thus, there is a path of length $1$ between $s_1$ and $s_2$,
    and therefore a path of length $2$ between every pair of words in $\res{\X}{S_1,S_2}$.
    The claim follows immediately from productization: We have \[
    \Pr_{(\bs_1,\bs_2) \sim \res{\X}{S_1,S_2}}[(\bs_1,\bs_2) = (s_1,s_2)]
    = \Pr_{\bs_1 \sim \res{\X}{S_1}}[\bs_1 = s_1] \cdot \Pr_{\bs_2 \sim \res{\X}{S_2}}[\bs_2 = s_2]
    > 0 \cdot 0 = 0, \]
    where the inequality uses that both $s_1$ and $s_2$ are valid words.
\end{proof}

We also have the following trivial upper bound on the diameter in productizing complexes:

\begin{proposition}[Diameter bound for separated-and-ordered restrictions]\label{lemma:diam:product on line}
    There exists a universal constant $K > 0$ such that the following holds.
    Let $n \in \N$ and $\X$ be an $\rangeOne{n}$-indexed simplicial complex.
    Let $\sigma$ be a parameter.
    Suppose that $\X$ satisfies the following conditions:
    \begin{conditions}
        \conditem{Productization}
        $\X$ is productizing on the line (in the sense of \Cref{def:pol}).
        
        \conditem{Diameter bound for separated singleton restrictions}
        For every $W \subset \rangeOne{n}$ and $\ell_1 < \ell_2 \in \rangeOne{n}$ such that
        $\{\ell_1,\ell_2\}$ is $\sigma$-separated and $W \cap \rangeII{\ell_1}{\ell_2} = \emptyset$,
        for every $w \in \supp(\res{\X}{W})$,
        $\diam(\resLink{\X}{\{\ell_1\},\{\ell_2\}}{w}) \le C_0$.\label{item:diam:product on line:radius}
    \end{conditions}
    Then for every $W \subset \rangeOne{n}$, $S_1 \prec S_2 \sqsubset \rangeOne{n}\setminus W$
    such that $S_1 \cup S_2$ is $\sigma$-separated, and $w \in \supp(\res{\X}{W})$,
    $\diam(\resLink{\X}{S_1,S_2}{w}) \le KC_0$.
\end{proposition}

Using these, we can prove the following:

\begin{proof}
    Let $\ell_1 \coloneqq \max S_1$ and $\ell_2 \coloneqq \min S_2$.
    We apply \Cref{lemma:pinning diameter} with $R_1 \coloneqq S_1$, $R_2 \coloneqq S_2$, $S_1 \coloneqq \{\ell_1\}$, and $S_2 \coloneqq \{\ell_2\}$.
    Hence, we need to bound the ``central'' diameter $\diam(\resLink{\X}{\{\ell_1\},\{\ell_2\}}{w})$
    and the ``peripheral'' diameters 
    $\diam(\resLink{\X}{S_1\setminus\{\ell_1\},S_2}{w,s_1})$ 
    and $\diam(\resLink{\X}{S_1,S_2\setminus\{\ell_2\}}{w,s_2})$
    (where $s_1 \in \resLink{\X}{\ell_1}{w}$ and $s_2 \in \resLink{\X}{\ell_2}{w}$).
    The former is bounded by \Cref{item:diam:product on line:radius}
    and the latter by \Cref{prop:diam from prod} and the PoL assumption.
\end{proof}

Given this, we now state a lemma which is extracted from the proofs of the ``extended base cases'' in 
\cite[\S4.3]{BLM24}.

\begin{lemma}[Coboundary expansion for separated-and-ordered restrictions, due to {\cite{BLM24}}]\label{thm:product on line}
    There exist universal constants $\epsilon_0, K > 0$ such that the following holds.
    Let $r \le n \in \N$.
    Let $\X$ be an $\rangeOne{n}$-indexed simplicial complex and $\Gamma$ a coefficient group.
    Let $\sigma$ be a parameter.
    Let $\beta^* < 1$ and $C_0 \in \N$.
    Suppose that $\X$ satisfies the following conditions:
    \begin{conditions}
        \conditem{Spectral expansion}
        $\X$ is $\epsilon_0 r^{-2}$-product.
        \label{item:product on line:spectral}
        \conditem{Strong symmetry}
        $\X$ is strongly symmetric.
        \conditem{Productization}
        $\X$ is productizing on the line (in the sense of \Cref{def:pol}).\label{item:product on line:prod}
        \conditem{Coboundary expansion of separated singleton restrictions}
        For every $W \subset \rangeOne{n}$ and $\ell_1 < \ell_2 < \ell_3 \in \rangeOne{n}$
        such that $\{\ell_1,\ell_2,\ell_3\}$ is $\sigma$-separated
        and $W \cap \rangeII{\ell_1}{\ell_3} = \emptyset$,
        for every $w \in \supp(\res{\X}{W})$,
        the complex $\resLink{\X}{\{\ell_1\},\{\ell_2\},\{\ell_3\}}{w}$ has $1$-triword expansion at least $\beta^*$ over $\Gamma$.\label{item:product on line:disj triword expansion}
        \conditem{Diameter bound for separated singleton restrictions}
        For every $W \subset \rangeOne{n}$ and $\ell_1 < \ell_2 \in \rangeOne{n}$ such that
        $\{\ell_1,\ell_2\}$ is $\sigma$-separated and $W \cap \rangeII{\ell_1}{\ell_2} = \emptyset$,
        for every $w \in \supp(\res{\X}{W})$,
        $\diam(\resLink{\X}{\{\ell_1\},\{\ell_2\}}{w}) \le C_0$.\label{item:product on line:radius}
    \end{conditions}
    Then for every $W \subset \rangeOne{n}$, $S_1 \prec S_2 \prec S_3 \sqsubset_r \rangeOne{n}\setminus W$
    such that $S_1 \cup S_2 \cup S_3$ is $\sigma$-separated, and $w \in \supp(\res{\X}{W})$,
    the complex $\resLink{\X}{S_1,S_2,S_3}{w}$ has $1$-triword expansion at least $K\beta^* / C_0$ over $\Gamma$.
\end{lemma}

\begin{proof}
    We consider cases based on $r = |S_1| = |S_2| = |S_3|$.
    Fix $W$, $w$, $S_1$, $S_2$, and $S_3$; we bound the triword expansion of $\Tri \coloneqq \resLink{\X}{S_1,S_2,S_3}{w}$ over a fixed $\Gamma$.

    \paragraph{Case 1: $r = 1$.}
    In this case, $S_1=\{\ell_1\}$, $S_2=\{\ell_2\}$, and $S_3=\{\ell_3\}$, with $\ell_1 < \ell_2 < \ell_3$,
    and $\{\ell_1,\ell_2,\ell_3\}$ is $\sigma$-separated by construction.
    
    \paragraph{Subcase 1a: $W \cap \rangeII{\ell_1}{\ell_3}=\emptyset$.}
    In this subcase, we directly use the assumed \Cref{item:product on line:disj triword expansion} on $\Tri$,
    which has $1$-triword expansion at least $\beta^*$ by assumption.

    \paragraph{Subcase 1b: $W \cap \rangeII{\ell_1}{\ell_3}\ne\emptyset$.}
    Let $\ell \in W \cap \rangeII{\ell_1}{\ell_3}$.
    WLOG, $\ell < \ell_2$.
    By the PoL assumption (\Cref{item:product on line:prod}), $\ptiz{\link{\X}{w}}{\{\ell_1\}}{\{\ell_j\}}$ for both $j \in \{2,3\}$.
    So by \Cref{thm:triword expansion given productization}, $\Tri$ has $1$-triword expansion at least $K \cdot \diam(\res{\X}{\{\ell_2\},\{\ell_3\}})$.
    By \Cref{item:product on line:radius}, and since $\{\ell_2\} \prec \{\ell_3\}$ and $\{\ell_2,\ell_3\}$ is $\sigma$-separated, $\diam(\res{\X}{\{\ell_2\},\{\ell_3\}}) \le C_0$.
    Hence $\Tri$ has $1$-triword expansion at least $K/C_0$.

    \paragraph{Case 2: $r > 2$.}
    Let $\ell_1 \coloneqq \max S_1, \ell_2 \coloneqq \min S_2, \ell_3 \coloneqq \min S_3$.
    We apply \Cref{lemma:pinning} to $\Tri$ with $R_1 = S_1, R_2 = S_2, R_3 = S_3, S_1 = \{\ell_1\}, S_2 = \{\ell_2\}, S_3 = \{\ell_3\}$.
    By the first case ($r=1$), the central triangle complex $\Tri^* \coloneqq \resLink{\X}{\{\ell_1\},\{\ell_2\},\{\ell_3\}}{w}$
    has $1$-triword expansion at least $\beta^*$.
    On the other hand, consider the peripheral complexes
    \begin{align*}
        \Tri_{s_1} &\coloneqq \resLink{\X}{S_1\setminus\{\ell_1\}, S_2, S_3}{w,s_1}, \\
        \Tri_{s_2} &\coloneqq \resLink{\X}{S_1, S_2\setminus\{\ell_2\}, S_3}{w,s_2}, \\
        \Tri_{s_3} &\coloneqq \resLink{\X}{S_1, S_2, S_3\setminus\{\ell_3\}}{w,s_3}
    \end{align*}
    (for $s_j \in \supp(\resLink{\X}{\{\ell_j\}}{w})$).
    We have by the PoL assumption (\Cref{item:product on line:prod}):
    \begin{itemize}
        \item For every $j \in \{2,3\}$, $S_1 \setminus \{\ell_1\} \prec S_j$.
        Further, since $\min S_1 \setminus \{\ell_1\} = \min S_1$, we have
        $\ell_1 \in \rangeII{\min S_1 \setminus \{\ell_1\}}{\max S_2} \subseteq \rangeII{\min S_1 \setminus \{\ell_1\}}{\max S_3}$.
        Hence for every $s_1 \in \supp(\resLink{\X}{S_1}{w})$ and $j \in \{2,3\}$, $\ptiz{\Tri_{s_1}}{S_1\setminus\{\ell_1\}}{S_j}$.
        \item Dually, for every $j \in \{1,2\}$, $S_j \prec S_3 \setminus \{\ell_3\}$.
        Further, since $\max S_3 \setminus \{\ell_3\} = \max S_3$, we have
        $\ell_3 \in \rangeII{\min S_2}{\max S_3 \setminus \{\ell_3\}} \subseteq \rangeII{\min S_1}{\max S_3 \setminus \{\ell_3\}}$.
        Hence for every $s_3 \in \supp(\resLink{\X}{S_3}{w})$ and $j \in \{1,2\}$,
        $\ptiz{\Tri_{s_3}}{S_j}{S_3\setminus\{\ell_3\}}$.
        \item $S_1 \prec S_2 \setminus \{\ell_2\} \prec S_3$.
        Further, since $\max S_2 \setminus \{\ell_2\} = \max S_2$, we have
        $\ell_2 \in \rangeII{\min S_1}{\max S_2 \setminus \{\ell_2\}} \subseteq \rangeII{\min S_1}{\max S_3}$.
        Hence for every $s_2 \in \supp(\resLink{\X}{S_2}{w})$,
        $\ptiz{\Tri_{s_2}}{S_1}{S_3}$ and $\ptiz{\Tri_{s_2}}{S_1}{S_2 \setminus \{\ell_2\}}$.
    \end{itemize}
    Note that the $S_2$ case had to be handled slightly differently because of the asymmetry in the choice of $\ell_2$.\footnote{
        When $t \ge 3$, this issue can be resolved more symmetrically: We pick $t \in S_2 \setminus \{\min S_2, \max S_2\}$.
        Then $\ell_2 \in \rangeII{\min S_1}{\max S_2 \setminus \{\ell_2\}} \cap \rangeII{\min S_2 \setminus \{\ell_2\}}{\max S_3}$,
        so for every $j \in \{1,3\}$ and $s_2 \in \supp(\resLink{\X}{S_2}{w})$, $\ptiz{\Tri_{s_2}}{S_j}{S_2}$.
    }
    In all cases, the relevant restrictions are ordered and separated, so we can apply \Cref{lemma:diam:product on line}
    (along with our \Cref{item:product on line:radius,item:product on line:prod})
    and deduce that all the peripheral complexes have $1$-triword expansion at least $K/C_0$.
\end{proof}

\subsection{Coboundary expansion of random restrictions}

\begin{lemma}\label{lemma:pseudorandom spread}
    Let $r \le n \in \N$, $T, \eta \ge 1 \in \N$, $W \subset_{\le r} \rangeOne{n}$.
    Suppose that $r^2/\eta \le \tfrac14 n$ and $Tr \le \tfrac12 n$.
    Sample $\bR \subset_r \rangeOne{n} \setminus W$ uniformly.
    Then: \[
    \Pr[\bR\text{ is }(T,\eta)\text{-pseudorandomly spread}] \ge 1- 2T\exp\parens*{ \frac{ -\eta^2T }{24 r} }. \]
\end{lemma}

\begin{proof}
    Generate $\bR$ via the following process: Sample $\bx_1,\ldots,\bx_r$ uniformly from $\rangeOne{n} \setminus W$ without replacement and set $\bR \coloneqq \{\bx_1,\ldots,\bx_r\}$.
    For $\br_t \coloneqq |\bR \cap \calI_t|$, we therefore have $\br_t = \sum_{i=1}^r 1_{\bx_i \in \calI_t}$.
    Marginally, \[
    p \coloneqq \Pr[\bx_i \in \calI_t] = \frac{|\calI_t \setminus W|}{n - |W|}. \]
    Using the trivial bounds $\frac{n}T - |W| = |\calI_t| - |W| \le |\calI_t \setminus W| \le |\calI_t| = \frac{n}{T}$, we get: \[
    \frac1{T} - \frac{2r}{n} \le \frac{\frac{n}T - r}{n-r} \le \frac{\frac{n}T - |W|}{n-|W|}
    \le p \le \frac{\frac{n}T}{n-r} \le \frac1{T} + \frac{2r}{n} \]
    and consequently $|p - \frac1T| \le \frac{2r}{n}$ and so $|rp - \frac{r}T| \le \frac{2r^2}{n} \le \tfrac12 \eta$.
    We also use $p \le \frac1T + \frac{2r}n \le \frac1{T} + \frac1{T} = \frac{2}{T}$.
    
    The Chernoff bound for sampling without replacement~\cite{Chv79} gives \[
    \Pr[|\br_t-rp| \ge \tfrac12 \eta] \le 2\exp \parens*{ \frac{ -(\tfrac12 \eta)^2 }{3r p} }
    = 2\exp \parens*{ \frac{ -(\tfrac12 \eta)^2 }{3r p} }
    \le 2\exp \parens*{ \frac{ -\eta^2 T }{24 r} }. \]
    Finally, take a union bound over all $T$ blocks.
\end{proof}

\begin{lemma}\label{lemma:separated}
    Let $r \le \tfrac12 n \in \N$ and $W \subset_{\le r} \rangeOne{n}$,
    Sample $\bR \subset_r \rangeOne{n} \setminus W$ uniformly.
    For $\sigma \in \N$, we have: \[
    \Pr[\bR\text{ is }\sigma\text{-separated}] \ge 1-\sigma \cdot \frac{2r^2}{n}. \]
\end{lemma}

\begin{proof}
    We first claim the following:
    \begin{equation}\label{eq:dist of pair}
    \Pr_{\bx_1,\bx_2 \sim \DUnif{\rangeOne{n} \setminus W}} [|\bx_1 - \bx_2| \le \sigma \mid \bx_1 \ne \bx_2]
    \le \frac{2\sigma}{\tfrac12n-1}.
    \end{equation}
    This implies the lemma by a union bound if we generate $\bR$ by sampling $\bx_1,\ldots,\bx_r$ uniformly from $\rangeOne{n}$ without replacement.

    It remains to prove \Cref{eq:dist of pair}.
    Let $m \coloneqq |\rangeOne{n}\setminus W|= n-|W| \ge n-r \ge \tfrac12n$.
    First, we argue that the probability is maximized when $W = \rangeEI{n-m}{n}$;
    that is, \[
    \Pr_{\bx_1,\bx_2 \sim \DUnif{\rangeOne{n} \setminus W}} [|\bx_1 - \bx_2| \le \sigma \mid \bx_1 \ne \bx_2]
    \le \Pr_{\bi,\bj \sim \DUnif{[m]}} [|\bi - \bj| \le \sigma \mid \bi \ne \bj]. \]
    Indeed, enumerate the elements of $\rangeOne{n} \setminus W$ as $\ell_1 < \cdots < \ell_m$.
    For $i, j \in [m]$, we have $|\ell_j - \ell_i| \ge |j - i|$ (since if e.g. $j \ge i$ then $\ell_j > \cdots > \ell_i$).
    This implies this inequality by coupling, since the distribution of $(\ell_\bi,\ell_\bj)$ for $\bi, \bj \sim \DUnif{[m]} \mid \bi \ne \bj$
    is the same as the distribution of $(\bx_1,\bx_2)$ for $\bx_1, \bx_2 \sim \DUnif{[m]} \mid \bx_1 \ne \bx_2$.
    
    Finally, we show that \[
    \Pr_{\bi,\bj \sim \DUnif{[m]}} [|\bi - \bj| \le \sigma \mid \bi \ne \bj] \le \frac{2\sigma}{m-1}. \]
    Indeed, for fixed $i \in [m]$, we have $\Pr_{\bj \sim \DUnif{[m]\setminus\{i\}}} \bracks*{ |i - \bj| \le \sigma} \le \frac{2\sigma}{m-1}$.
    (The factor of $2$ is necessary when $\sigma \le i \le m-\sigma$.)
\end{proof}

This lets us give:

\begin{lemma}[Coboundary expansion of typical restrictions]\label{lemma:reduction:typical}
    There exist absolute constants $K > 1$ and $\epsilon_0 > 0$ such that the following holds.
    Let $r \le n \in \N$, $\X$ be a $\rangeOne{n}$-indexed simplicial complex and $\Gamma$ a coefficient group.
    Let $\sigma, \eta, k, T \in \N$ be parameters.
    Suppose that $\X$ satisfies the following conditions:
    \begin{enumerate}
    \conditem{Spectral expansion}
    $\X$ is $\epsilon_0 r^{-2}$-product.
    \conditem{Productization}
    $\X$ productizes on the line, in the sense of \Cref{def:pol}.
    \conditem{Symmetry}
    $\X$ is strongly symmetric, in the sense of \Cref{def:strong symm}.
    \conditem{Coboundary expansion of separated singleton restrictions}
    For every $W \subset \rangeOne{n}$ and $\ell_1 < \ell_2 < \ell_3 \in \rangeOne{n}$
    such that $W \cap \rangeII{\ell_1}{\ell_3} = \emptyset$ and $\{\ell_1,\ell_2,\ell_3\}$ is $\sigma$-separated,
    for every $w \in \supp(\res{\X}{W})$,
    the complex $\resLink{\X}{\{\ell_1\},\{\ell_2\},\{\ell_3\}}{w}$ has $1$-triword expansion at least $\beta^*$ over $\Gamma$.
    \conditem{Diameter of separated singleton restrictions}
    For every $W \subset \rangeOne{n}$ and $\ell_1 < \ell_2 \in \rangeOne{n}$
    such that $W \cap \rangeII{\ell_1}{\ell_2} = \emptyset$ and $\{\ell_1,\ell_2\}$ is $\sigma$-separated,
    for every $w \in \supp(\res{\X}{W})$,
    the complex $\resLink{\X}{\{\ell_1\},\{\ell_2\}}{w}$ has diameter at most $C_0$.
    \end{enumerate}
    Then for every $W \subset_{\le r} \rangeOne{n}$ and $w \in \supp(\res{\X}{W})$,
    \begin{multline*}
    \Pr_{\bR_1,\bR_2,\bR_3 \sqsubset_r \rangeOne{n} \setminus W} [\resLink{\X}{\bR_1,\bR_2,\bR_3}{w}\text{ has }1\text{-triword expansion at least }K\beta'\text{ over }\Gamma] \\
    \ge 1 - 3 \parens*{ 2T\exp\parens*{ \frac{ -\eta^2T }{24r} } + \sigma \cdot \frac{2r^2}n },
    \end{multline*}
    where $\beta' \coloneqq (K \beta^* /C_0)^{3(T(k+2\eta)+2(r/T))+r/k}$.
\end{lemma}

\begin{proof}
    Let $\calS \coloneqq \{(R_1,R_2,R_3) : R_1,R_2,R_3 \sqsubset_r \rangeOne{n} \setminus W\}$.
    We claim that for fixed $(R_1,R_2,R_3) \in \calS$ such that (i) $R_i$ is $(T,\eta)$-pseudorandomly spread for each $i \in \{1,2,3\}$ and
    (ii) $R_1\cup R_2\cup R_3$ is $\xi n$-separated,
    $\res{\X}{R_1,R_2,R_3}$ is a triword expander with the desired parameters.
    This implies the lemma when we sample $(\bR_1,\bR_2,\bR_3) \sim \DUnif{\calS}$,
    since marginally, each $\bR_i$ is a uniformly random $r$-set in $\rangeOne{n} \setminus W$;
    we can apply \Cref{lemma:pseudorandom spread,lemma:separated} and take a union bound.
    
    To prove the claim, we apply \Cref{lemma:balanced induction}; we need to check its three conditions.
    Spectral expansion holds by our assumption, and $R_1,R_2,R_3$ are pseudorandomly spread by construction.
    Thus, it suffices to therefore check \Cref{item:balanced induction:ordered expansion} of \Cref{lemma:balanced induction} (``ordered triword expansion''):
    That is, for every $W \subseteq R_1 \cup R_2 \cup R_3$, $w \in \supp(\res{\X}{W})$,
    and $S_1 \prec S_2 \prec S_3 \sqsubset_k R_1 \cup R_2 \cup R_3 \setminus W$,
    $\resLink{\X}{S_1,S_2,S_3}{w}$ is a triword expander.
    For this, we apply \Cref{thm:product on line}; we have all the necessary conditions, and get triword expansion at least $K\beta^* / C_0$.
\end{proof}

\subsection{$r$-triword expansion from typical triword expansion}

\newcommand{\bcalF}{{\boldsymbol{\calF}}}
\newcommand{\bcalR}{{\boldsymbol{\calR}}}

Now, we state and prove the following theorem, which is based on techniques from~\cite[\S8]{DD24-swap}:

\begin{theorem}\label{thm:ddl}
    There exists absolute constants $K, \epsilon_0 > 0$ such that the following holds.
    Let $r, n \in \N$ with $r < \frac1{10} \sqrt{n} \in \N$.
    Let $\calI$ be a set of size $n$,
    $\X$ be an $\calI$-indexed simplicial complex which is $\epsilon_0 r^{-2}$-product, and $p, \beta < 1$.
    Let $\Gamma$ be any coefficient group and suppose that \[
    \Pr_{\bR_1,\bR_2,\bR_3 \sqsubset_r \calI} \bracks*{
    \res{\X}{\bR_1,\bR_2,\bR_3} \text{ has }1\text{-triword expansion at least }\beta\text{ over }\Gamma} \ge p. \]
    Then $\X$ has $r$-triword expansion at least $K\beta p^2$ over $\Gamma$.
\end{theorem}

In order to prove this, we use the following lemma:

\begin{lemma}\label{thm:dd}
    There exists an absolute constant $\epsilon_0 > 0$ such that the following holds.
    Let $r, \ell, n \in \N$ with $\ell \ge 3r$ and $n \ge (\ell+2)k$.
    Let $\calJ$ be a set of size $n$,
    let $\frY$ be a $\calJ$-indexed simplicial complex which is $\epsilon_0 r^{-2}$-product.
    Let $\bcalR$ be a set of $\ell$ uniformly random pairwise disjoint $r$-sets in $\calJ$.
    Let $\Gamma$ any coefficient group and suppose that: \[
    \Pr_\bcalR \bracks*{
    \res{\frY}{\bcalR} \text{ has }1\text{-triword expansion at least } \beta\text{ over }\Gamma
    } \ge p. \]
    Then $\frY$ has $1$-triword expansion $\Omega(\beta p)$ over $\Gamma$.
\end{lemma}

This lemma is a variant of lemmas in~\cite[\S4]{DD24-local} and~\cite[\S5]{BLM24}.
As we have not been able to locate a self-contained proof, we give a proof in \Cref{sec:local-to-global} below.

The $\ell=1,r=3$ case of this lemma is the following, which can also be viewed as a direct consequence of~\cite[Lemma 4.1]{DD24-local}:
\begin{corollary}\label{cor:dd}
    There exists an absolute constant $\epsilon_0 > 0$ such that the following holds.
    Let $\frY$ be a $\calJ$-indexed simplicial complex which is $\epsilon_0$-product.
    Let $\bj_1,\bj_2,\bj_3$ be three uniformly random distinct indices in $\calJ$.
    Let $\Gamma$ be any coefficient group and suppose that: \[
    \Pr_{\bj_1,\bj_2,\bj_3} [\res{\frY}{\bj_1,\bj_2,\bj_3} \text{ has }1\text{-triword expansion at least } \beta\text{ over }\Gamma] \ge p. \]
    Then $\frY$ has $1$-triword expansion $\Omega(\beta p)$ over $\Gamma$.
\end{corollary}

We also use:

\begin{proposition}[{\cite[Lemma 3.5]{GLL22}}]\label{thm:gll}
    Let $\X$ be an $\calI$-indexed simplicial complex which is $\epsilon$-product.
    Then for every $R_1,R_2 \sqsubset \calI$, the edge complex $\res{\X}{R_1,R_2}$ is an $\epsilon |R_1| |R_2|$-bipartite expander.
\end{proposition}

\begin{corollary}\label{prop:gll:restrict}
    Let $\X$ be an $\calI$-indexed simplicial complex which is $\epsilon$-product.
    Let $\calR$ be a set of pairwise disjoint $k$-sets in $\calI$.
    Then the $\calR$-indexed simplicial complex $\res{\X}{\calR}$ is $\epsilon r^2$-product.
\end{corollary}

\begin{proof}
    We need to show that for every $R_1 \ne R_2 \in \calR$, $W \subseteq \calR \setminus \{R_1,R_2\}$, and $W$-word $w$,
    $\res{(\resLink{\X}{\calR}{w})}{R_1,R_2}$ is an $\epsilon r^2$-bipartite expander.
    But this graph is identical to the graph $\resLink{\X}{R_1,R_2}{w}$ (now viewing $w$ as a word on $\bigcup \calR$).
    Since $\link{\X}{w}$ is $\epsilon$-product (by heritability on links), we can therefore apply \Cref{thm:gll}.
\end{proof}

\begin{claim}\label{claim:ddl:core}
    Suppose $\X$ is an $\calI$-indexed simplicial complex satisfying the hypotheses of \Cref{thm:ddl}.
    Let $\bcalR$ be a set of $3r$ random pairwise disjoint $r$-sets in $\calI$.
    Then \[
    \Pr_\bcalR \bracks*{ \res{\X}{\bcalR}\text{ has }1\text{-triword expansion at least }\tfrac12 \beta p } \ge \tfrac12 p. \]
\end{claim}
\begin{proof}
    First, fix $\calR$, a collection of $3r$ disjoint $r$-sets, and define \[
    \gamma(\calR) \coloneqq \Pr_{\bR_1,\bR_2,\bR_3\text{ distinct} \in\calR} \bracks*{ \res{\X}{\bR_1,\bR_2,\bR_3}\text{ has }1\text{-triword expansion at least }\beta\text{ over }\Gamma}. \]
    $\gamma(\calR)$ is a deterministic function of $\calR$ and satisfies $0 \le \gamma(\calR) \le 1$.

    By \Cref{prop:gll:restrict}, $\res{\X}{\calR}$ is $\epsilon_0$-product.
    We can therefore apply \Cref{cor:dd} (with $\frY \coloneqq \X$, $r \coloneqq 1$, $\ell \coloneqq 3$) to deduce that $\res{\X}{\calR}$ has $1$-triword expansion at least $\beta \gamma(\calR)$ over $\Gamma$.
    
    Next, for random $\bcalR$, $\Exp_\bcalR \gamma(\bcalR) \ge p$ by assumption.
    (This is because the marginal distribution of three random distinct $r$-sets inside of $\bcalR$ is the same as the marginal distribution of three random pairwise disjoint $r$-sets.)
    Hence by Markov's inequality, \[
    \Pr_\bcalR[\gamma(\bcalR) \le \tfrac12p ] = \Pr_\bcalR[1-\gamma(\bcalR) \ge 1-\tfrac12p] \le \frac{1-p}{1-\tfrac12 p} \le 1-\tfrac12 p \]
    (just using that $p < 2$),
    and therefore $\Pr_\bcalR[\gamma(\bcalR) \ge \tfrac12 p] \ge \tfrac12 p$.
\end{proof}

\begin{proof}[Proof of \Cref{thm:ddl}]
    We use \Cref{thm:dd} again and set $\ell \coloneqq 3r$.
    Applying \Cref{claim:ddl:core}, the probability in the hypothesis of \Cref{thm:dd} is at least $\tfrac14 p$ (using $\beta' \coloneqq \tfrac12 \beta p$).
    This gives a final coboundary expansion bound of $\tfrac18 \beta p^2$, as desired.
\end{proof}

\Cref{thm:linear condition} now follows immediately from \Cref{thm:ddl,lemma:reduction:typical,rmk:induction} (setting $\xi \coloneqq 100/r^2$).

\section{Sufficient conditions for direct-product testing}\label{sec:sufficient conditions for direct-product testing}

In this section, we discuss the proof of \Cref{thm:sufficient conditions}, restated here:

\sufficient*

In particular, \Cref{thm:sufficient conditions} follows from \cite[Theorem 3.1]{DD24-covers} (henceforth ``3.1''), for the following reasons:

\begin{itemize}
    \item We are specializing 3.1 to the ``V-test'' (this is the ``$\calD$'' of 3.1; it is acceptable per \cite[Example~2.11]{DD24-covers}).
    \item Given $\X$ as in our \Cref{thm:sufficient conditions}, we do not apply 3.1 to it directly. Rather, we apply 3.1 to the $d$-skeleton $\X'$ of $\X$, for a certain $d \geq \exp(\poly(r))$ (and then we require $n \geq 2d^2$, say). This small trick is for achieving two-sided spectral expansion; see below. Note that $\X'$ is a rank-$d$ clique complex, and it has \Cref{itm:2,item:sufficient conditions:typical triword expander,itm:4} if (and only if) $\X$ has (since these properties only depend on the $4r$-skeleton).
    \item We should verify that \Cref{itm:1,itm:2,item:sufficient conditions:typical triword expander} imply that $\X'$ is ``$(d,k,\alpha)$-suitable'' in the sense of~3.1:
    \begin{itemize}
        \item First, recalling that 3.1's parameter ``$\alpha$'' is $\poly(1/\delta)$, by requiring $r \geq \poly(k)$ and then $d \geq \exp(\poly(r))$ large enough, we can let the ``suitable'' $d_1$ parameter equal~$r$, while also arranging for the expression $\alpha\frac{d_1}{k}$ to be at least $r^{.991}$.
        \item Next we will verify that $\X'$ is a $\frac{1}{d^2}$-two-sided local spectral expander.  This uses the two-sided ``trickling down'' results of~\cite{Opp18}.
        By \cite[Corollary~2.4]{DDL24} (or see, e.g., \cite[\S A]{HS24} for a fuller treatment, where one uses the trivial bound $\eta = -1$ to  control negative eigenvalues),
        we get that $\X'$ is a $\max\{\frac{\gamma}{1-(n-1)\gamma}, \frac{1}{n-d+1}\}$-two-sided local expander.
        Using $\gamma < 1/n^2$ and recalling $n \geq 2d^2$, the expansion parameter is at most $\frac{1}{d^2}$ as needed.
        \item $\X'$ is required to be ``$r$-well-connected'', which is precisely the conjunction of \Cref{itm:2} and ``simple-connectedness'' of the $r$-triword complex $\X'_w$ for every vertex~$w$, which we discuss next.
        \item This simple-connectedness is equivalent to \emph{nonzero} $r$-triword expansion of $\X'_w$ over every group, which is not quite implied by \Cref{item:sufficient conditions:typical triword expander}.\footnote{In fact, for the complex we study in this paper, $\X'_w$ \emph{will} have nonzero $r$-triword expansion over every group, rendering this discussion moot.}
        However, in the one place where simple-connectedness is actually used (end of first paragraph of proof of ``Claim~5.1'' in \cite{DD24-covers}) it is used to show that an $\ell$-cover of $\X'_w$, with $\ell \leq \poly(1/\delta)$ must in fact be $\ell$ disjoint copies of $\X'_w$.  But \Cref{item:sufficient conditions:typical triword expander}, in its weaker nonzero expansion form, is precisely equivalent to saying that the only $m$-covers of $\X'_w$ are trivial ones, given by $m$ disjoint copies of $\X'_w$.  (This will be discussed again in the final bullet-point.)  Thus \Cref{item:sufficient conditions:typical triword expander} suffices, provided $m \geq \poly(1/\delta)$ is taken larger than~$\ell$.
        \item Finally, we show that \Cref{item:sufficient conditions:typical triword expander} implies the last needed condition for ``suitability'', namely that  $\X'$ has $r$-triword ``cosystolic'' expansion $\exp(-r^{.991})$ over $\Sym{m_0}$ for all $1 < m_0 \leq m$.  (In fact, 3.1 formally requires the cosystolic expansion over every group~$\Gamma$.  However, it was written this way just for brevity; when the condition is used in the proof --- going from ``Lemma~3.6'' to ``Corollary~3.7'' --- only $\Gamma = \Sym{m}$ is needed.  Cf. ~\cite[Theorem 4.1]{BM24}.)
        We apply the local-to-global method of~\cite{KKL16,EK23,DD24-local}, specifically \cite[Theorem~1.2]{DD24-local} (with its $k = 1$ and $d = 3$): by combining the $r$-triword (``coboundary'') expansion of $\exp(-r^{.99})$ over $\Sym{m_0}$ for each $\X'_w$, with the one-sided spectral expansion of $\X'$ with parameter $\frac{1}{d^2} \ll \exp(-r^{.99})$, we conclude that $\X'$ has $r$-triword cosystolic expansion $\Omega(\exp(-r^{.99})) \geq \exp(-r^{.991})$, as needed.
    \end{itemize}
    \item The last distinction is that 3.1 only concludes the agreement tester achieves ``$\delta$-cover soundness'', rather than $\delta$-soundness; i.e., rather than $F(\bS)$ being often close to $f^{\wedge k}(\bS)$ for some $f : \calU \to \Sigma$, there is merely an explanatory $f$ on the vertices of an \emph{$m$-cover} $\mathfrak{Y}$ of $\X'$
    (See \cite[Definition~2.26]{DD24-covers} for the precise definition.)
    However, similar to the prior discussion of simple-connectedness, we have that \Cref{itm:4} is precisely equivalent to saying that the only $m$-covers of $\X'$ are trivial ones (given by $m$ disjoint copies of~$\X'$). In this case, it is easy to see (and the argument written in~\cite[Part~2 of the proof of Theorem~4.11]{DDL24}) that the explanatory~$f$ on $\mathfrak{Y}$ must yield an equally explanatory $f$ on $\X'$, as desired for $\delta$-soundness.
\end{itemize}

\section{The global Kaufman--Oppenheim complex}\label{sec:kaufman-oppenheim}

In this subsection, we describe a general way to construct Kaufman--Oppenheim-type complexes and some useful properties of them.
This construction generalizes both the original complex defined in \cite{KO18,KO23},
and also its modification in~\cite[\S2.5]{KOW25}.

\subsection{Defining the group}

For a field $\F$, indeterminate $X$, and $s \in \N$, define the quotient ring $\GrF{s}{\F} \coloneqq \F[X]/(X^s)$.
We use $\deg(f)$ to refer to the degree of the (unique) minimal-degree representative of an element of $\GrF{s}{\F}$.

\begin{remark}
For concreteness, here is a direct definition of this ring, without using the quotient ring construction:
$\GrF{s}{\F}$ consists of the set of elements $\{f \in \F[X] : \deg(f) < s\}$ equipped with the standard addition operation and the ``truncated'' multiplication operation.
That is, for $f = \sum_{i=0}^{s-1} f_i X^i$ and $g = \sum_{i=0}^{s-1} g_i X^i$, $fg \coloneqq \sum_{i=0}^{s-1} \sum_{j=0}^i f_j g_{i-j} X^i$.
Hence, $\deg(fg) = \min\{\deg(f)+\deg(g), s-1\}$.
\end{remark}

\begin{definition}[Special linear group]
    Let $n, s \in \N$ and $\F$ be a field.
    $\GrSL{n+1}{s}{\F}$ is the group of $(n+1) \times (n+1)$ matrices (indexed by $\rangeII{0}{n}$) with entries in the ring $\GrF{s}{\F}$ and determinant $1$.
    For $i \ne j \in \rangeII{0}{n}$ and $f \in \GrF{s}{\F}$, we let $e_{i,j}(f) \in \GrSL{n+1}{s}{\F}$ denote the matrix
    with $1$'s on the diagonal and $f$ in the $(i,j)$ entry.
    These matrices are known variously as \emph{elementary matrices}, \emph{transvections}, or \emph{shear matrices}.
\end{definition}

\begin{remark}
    One can more generally define $\SL{n+1}{R}$ when $R$ is any (commutative, unital) ring.
    For general $R$, there is a technical distinction between the group $\SL{n+1}{R}$ and its subgroup $\EL{n+1}{R}$,
    which is the subgroup generated by the elementary matrices $e_{i,j}(r)$.
    Some sources, in particular~\cite{KO23}, deal with the groups $\EL{n+1}{R}$ instead of $\SL{n+1}{R}$.
    However, $\EL{n+1}{R}$ and $\SL{n+1}{R}$ coincide (i.e., the elementary matrices generate $\SL{n+1}{R}$) whenever $\F$ is a field and $R = \GrF{s}{\F}$ (our case);
    indeed, they coincide under the (much) more general condition that $R$ is a \emph{semilocal ring} (see e.g.~\cite[Theorem 4.3.9]{HO89}).
    Hence, we are justified in using these notations interchangeably in this appendix.
\end{remark}

\begin{remark}
    This is an instantiation of a more general construction in \cite[\S3.1]{KO23};
    specifically, we apply that construction with the ring $R = \F$, finitely generated algebra $\calR = \F[X]$, and generating set $T = \{1,\ldots,X^\kappa\}$.
    We can therefore invoke some known properties of the constructions in \cite{KO23} in the sequel.
    In particular, \Cref{prop:KO:clique} follows immediately from~\cite[Theorem 3.5]{KO23}.
\end{remark}

\newcommand{\residue}[1]{\mu_{n+1}(#1)}

We use $\residue{\cdot} : \Z \to \rangeII{0}{n}$ to denote the function reducing an integer to its residue in $\rangeII{0}{n}$ modulo $n+1$.

\begin{definition}[Unipotent subgroups of special linear group]
    For $n, s, \kappa \in \N$ and $\ell \in \rangeII{0}{n}$, the subgroup $\CGrUnip{n}{\ell}{s}{\F} < \GrSL{n}{s}{\F}$ is
    the subgroup generated by elements of the form $e_{j,\residue{j+1}}(f)$ for $j \in \rangeII{0}{n} \setminus \{\ell\}$ and $f \in \GrF{s}{\F}$ with $\deg(f) \le \kappa$.
\end{definition}

\begin{remark}
    Note that if $s \le \kappa$ then the condition $\deg(f) \le \kappa$ is superfluous.
\end{remark}

These subgroups form an $\rangeII{0}{n}$-indexed subgroup family
$(\CGrUnip{n}{\ell}{s}{\F} < \GrSL{n}{s}{\F})_{\ell \in \rangeII{0}{n}}$.
This family has many nice properties; for instance, it has ``dihedral symmetry'' because there are automorphisms which cyclically shift and/or reverse the order of the subgroups (see \cite[\S3.2]{KO23}).

Hence, as in \Cref{def:indexed subgroup family}, for $W \subseteq \rangeII{0}{n}$ we can consider the intersection subgroup \[
\CGrUnip{n}{W}{s}{\F} \coloneqq \bigcap_{\ell \in W} \CGrUnip{n}{\ell}{s}{\F}. \]
(Again, $\CGrUnip{n}{\emptyset}{s}{\F} = \GrSL{n}{s}{\F}$.)

We have the following, which essentially follows from \cite[Theorem 2.20]{KOW25}:

\begin{proposition}\label{prop:ko:local isomorphism}
    Let $n, s, \kappa \in \N$ and $\F$ be a field.
    Suppose that $s \ge \kappa n$.
    For every $\ell \in \rangeII{0}{n}$, let $K_\ell \coloneqq \CGrUnip{n}{\ell}{s}{\F}$.
    Then for every $\ell \in \rangeII{0}{n}$, the subgroup family $(K_\ell \cap K_a \le K_\ell)_{a \in \rangeII{0}{n}\setminus\{\ell\}}$
    is isomorphic to the subgroup family $(\GrStair{n}{\F}{b} \le \GrUnip{n}{\F})_{b \in \rangeOne{n}}$ defined in \Cref{def:graded unipotent group}.
\end{proposition}

\subsection{Defining the complex}

Given the definitions of groups in the previous subsection, it is now possible to define the \emph{global} Kaufman--Oppenheim complex: 

\begin{definition}[The global Kaufman--Oppenheim complex]
    Let $n, s, \kappa \in \N$ and $\F$ be a field.
    The (global) Kaufman--Oppenheim complex, denoted $\CplxA{n}{s}{\F}$, is the coset complex defined by
    the group $\GrSL{n}{s}{\F}$ and its subgroup family $(\CGrUnip{n}{\ell}{s}{\F})_{\ell\in\rangeII{0}{n}}$.
\end{definition}

Then, we can easily prove \Cref{prop:ko:vertex links}, which states that
vertex-links in the global Kaufman--Oppenheim complex are isomorphic to the local Kaufman--Oppenheim complex (as long as $s \ge \kappa n$):

\begin{proof}[Proof of \Cref{prop:ko:vertex links}]
    First, we define shorthands for our groups:
    Let $G \coloneqq \GrSL{n}{s}{\F}$.
    For $\ell \in \rangeII{0}{n}$, let $K_\ell \coloneqq \CGrUnip{n}{\ell}{s}{\F} \le G$.
    Let $\calK \coloneqq (K_\ell \le G)_{\ell \in \rangeII{0}{n}}$ denote the corresponding indexed subgroup family of $G$,
    and, for each $\ell \in \rangeII{0}{n}$,
    let $\calH^\ell \coloneqq (K_\ell \cap K_a \le K_\ell)_{a \in \rangeII{0}{n}\setminus \{\ell\}}$ denote the corresponding indexed subgroup family of $K_\ell$.

    By definition, the global Kaufman--Oppenheim complex is $\CplxA{n}{s}{\F} = \CoCo{G}{\calK}$.
    By \Cref{prop:coset complex:link}, its vertex-links are:
    $\link{(\CoCo{G}{\calK})}{v} \cong \CoCo{K_\ell}{(K_\ell \cap K_a))_{a \in \rangeII{0}{n}\setminus \ell}} = \CoCo{K_\ell}{\calH^\ell}$.
    By \Cref{prop:ko:local isomorphism}, $\calH^\ell$ is isomorphic to the subgroup family $(\GrStair{n}{\F}{b} \le \GrUnip{n}{\F})_{b \in \rangeOne{n}}$ defined in \Cref{def:graded unipotent group}.
    But this is precisely the local Kaufman--Oppenheim complex, defined in \Cref{def:ko:local complex}.
\end{proof}

\section{Additional proofs regarding  groups and coset complexes}\label{sec:coset complex}

\subsection{Proof of \Cref{prop:coset complex:pass to quotient}}

\begin{proposition}\label{prop:prelim:lagrange}
For $H \le G$ groups and $\bg \sim \DUnif{G}$, the coset-valued random variable $\bg H$ is uniformly distributed on $G/H$.
\end{proposition}

\begin{proof}
    We have $\Pr [ \bg H = C ] = \frac1{|G|} |\{ g \in G : g H = C \}|$.
    Then, $\{ g \in G : g H = C \} = C$, and $|C| = |H|$, so that $\Pr [\bg H = C] = \frac{|H|}{|G|} = \frac1{[G:H]} = \frac1{|G/H|}$ by Lagrange's theorem.
    (Here, $[G:H]$ is the \emph{index} of $H$ in $G$, a.k.a., the number of cosets in $G/H$.)
\end{proof}

\begin{proof}[Proof of \Cref{prop:coset complex:pass to quotient}]
    We define the mapping
    \begin{align*}
        \phi_i :\ &G/H_i \to L/(\pi(H_i)) \\
        &g H_i \mapsto \pi(g) (\pi(H_i)).
    \end{align*}
    This is well-defined because if $g H_i = g' H_i$, then $g^{-1} g' \in H_i$,
    so $\pi(g)^{-1} \pi(g') = \pi(g^{-1} g') \in \pi(H_i)$,
    so $\pi(g) (\pi(H_i)) = \pi(g') (\pi(H_i))$.

    $\phi_i$ is obviously surjective: For any $\ell \in L$, by surjectivity, there is some $g \in G$ with $\pi(g) = \ell$;
    then $\phi(g H_i) = \pi(g) (\pi(H_i)) = \ell (\pi(H_i))$.

    Conversely, $\phi_i$ is also injective:
    If $\phi_i(g H_i) = \phi_i(g' H_i)$,
    then $\pi(g) (\pi(H_i)) = \pi(g') (\pi(H_i))$,
    so $\pi(g^{-1} g') = \pi(g)^{-1} \pi(g') \in \pi(H_i)$,
    hence there exists $h \in H_i$ with $\pi(g^{-1} g') = \pi(h)$,
    hence $\pi(g^{-1} g' h^{-1}) = \Id$,
    hence $g^{-1} g' h^{-1} \in \ker(\pi) \le H_i$,
    so finally $g^{-1} g' \in H_i$ and so $gH_i = g'H_i$.

    Finally, we show that $(\phi_i)_{i \in \calI}$ is an isomorphism of coset complexes.
    The coset complex $\CoCo{G}{K\calH}$ samples 
    $\bg$ and outputs the cosets $( \bg H_i)_{i \in \calI}$.
    The image under $(\phi_i)_{i \in \calI}$ of this distribution
    samples $\bg \sim G$ uniformly
    and outputs $(\pi(\bg) \pi(H_i))_{i \in \calI}$.
    The marginal distribution of $\pi(\bg)$ is uniform\footnote{
        This is because the fibers $\pi^{-1}(\ell)$ for $\ell \in L$ are in bijection with the cosets $G/\ker(\pi)$, which partition $G$.}
    on $\pi(G) = L$ and therefore that this is the same as sampling $\bl \sim L$
    and outputting $(\bl \pi(H_i))_{i\in \calI}$,
    which is the definition of $\CoCo{L}{(\pi(H_i))_{i \in \calI}}$.
\end{proof}

\subsection{Proof of \Cref{prop:coset complex:productization}}

We now prove the following strong form of \Cref{prop:coset complex:productization}:

\begin{proposition}
    Let $G$ be a group and $H_1, H_2 \le G$ two subgroups.
    The following are equivalent:
    \begin{enumerate}
        \item $H_1 H_2 = G$, i.e., eveAry $g \in G$ can be written $g = h_1 h_2$ for $h_1 \in H_1$, $h_2 \in H_2$.\label{item:H1H2 = G}
        \item Every $C_1 \in G/H_1$ and $C_2 \in G/H_2$ have nonempty intersection.\label{item:cosets intersect}
        \item For $\bg \sim G$ uniform, the coset-valued random variables $\bg H_1$ and $\bg H_2$ are independent.\label{item:independent}
        \item For $\bh_1 \sim H_1$ and $\bh_2 \sim H_2$ uniform and independent,
        $\bh_1 \bh_2$ is distributed uniformly on $G$.\label{item:uniform}
    \end{enumerate}
\end{proposition}

\begin{proof}
(\Cref{item:H1H2 = G} $\implies$ \Cref{item:cosets intersect}.)
Pick $a \in C_1$ and $b \in H_2$.
Then $a^{-1}b\in H_1H_2$ so $a^{-1}b=h_1h_2$ for some $h_1 \in H_1$ and $h_2 \in H_2$,
hence $b h_2^{-1}=a h_1\in aH_1\cap bH_2\neq \emptyset$.

(\Cref{item:cosets intersect} $\implies$ \Cref{item:independent}.)
Note that $gH_i = C_i \iff g \in C_i$.
Hence $gH_1 = C_1 \wedge gH_2 = C_2 \iff g \in C_1 \cap C_2$.
By \Cref{prop:prelim:coset intersection},
any nonempty intersection of two cosets is a coset of $H_1\cap H_2$,
so for every $C_1,C_2$, \[
\Pr[\bg H_1=C_1,\ \bg H_2=C_2]=\frac{|H_1\cap H_2|}{|G|}. \]
Hence the joint distribution of $(\bg H_1, \bg H_2)$ is uniform.

(\Cref{item:independent} $\implies$ \Cref{item:uniform}.)
Independence of $\bg H_1$ and $\bg H_2$ gives that for every coset pair $C_1\in G/H_1,\ C_2\in G/H_2$, \[
\Pr[\bg H_1=C_1,\ \bg H_2=C_2]
=\Pr[\bg H_1=C_1]\Pr[\bg H_2=C_2]
=\frac{|H_1||H_2|}{|G|^2}. \]
But as in the last implication, this probability (being nonzero) must equal $|H_1\cap H_2|/|G|$, and so \[
|H_1||H_2|=|G|\,|H_1\cap H_2|. \]
Hence by \Cref{prop:prelim:product formula}, $|H_1H_2|=|G|$ and thus $H_1H_2=G$.

(\Cref{item:H1H2 = G} $\implies$ \Cref{item:uniform}.)
For fixed $g\in G$, the pairs $(h_1,h_2) \
\in H_1 \times H_2$ such that $g=h_1h_2$ correspond bijectively to the set
\[
\{h_1\in H_1:\ h_1^{-1}g\in H_2\}=H_1\cap gH_2. \]
This set is nonempty (since $H_1 H_2 = G$) and so by \Cref{item:cosets intersect} it is a coset of $H_1 \cap H_2$ and has size $|H_1\cap H_2|$.
Therefore, \[
\Pr[\bh_1\bh_2=g]=\frac{|H_1\cap H_2|}{|H_1||H_2|}=\frac{1}{|G|}, \]
so $\bh_1\bh_2\sim \DUnif{G}$.

(\Cref{item:uniform} $\implies$ \Cref{item:H1H2 = G}.)
If $\bh_1\bh_2$ is uniform on $G$, then every $g\in G$ occurs with positive probability as a product $h_1h_2$, so $g\in H_1H_2$.
Hence $H_1H_2=G$.
\end{proof}

\subsection{Proof of \Cref{prop:group:intersection of products}}

\begin{proof}[Proof of \Cref{prop:group:intersection of products}]
    First, we claim that all three factors are contained in the intersection $H \cap K$.
    For the middle factor, we have $H_1 \subseteq H_1 H_2 = H$ and $K_2 \subseteq K_1 K_2 = K$.
    For e.g. the left factor, we have $K_1 \le K_1 \cap H_1 \le K \cap H$ and similarly for the right factor.
    
    Next, we verify that the three factors are pairwise commuting and trivially intersecting.
    Indeed, $K_1$ and $K_2$ commute and intersect trivially by assumption, and therefore so do $K_1 $ and $H_1 \cap K_2$.
    The proof for $H_2$ and $H_1 \cap K_2$ is similar.
    Finally, $K_1$ and $H_2$ commute and intersect trivially because $K_1 \le H_1$ and $K_1$ and $K_2$ commute and intersect trivially.
    (Alternatively, we could use that $H_2 \le K_2$ and that $H_1$ and $H_2$ commute and intersect trivially.)

    So, it remains to verify that the three factors together generate $H \cap K$.
    Let $s \in H \cap K$.
    Hence $s = h_1 h_2 = k_1 k_2$ for some $h_i \in H_i, k_i \in K_i$.
    In particular, we can define $t \coloneqq h_2 k_2^{-1} = h_1^{-1} k_1$.
    We have $t = h_2 k_2^{-1} \in H_2 K_2 = K_2$ (since $H_2 \le K_2$),
    while also $t = h_1^{-1} k_1 \in H_1 K_1 = H_1$ (since $K_1 \le H_1$).
    Hence $t \in H_1 \cap K_2$.
    Finally, we can write $s = h_1 h_2 = k_1 t^{-1} h_2$, which has the desired form.\footnote{
        This statement can also been seen via the so-called ``Dedekind law'' or ``modular law'',
        which states that if $A \le B$ then $A(B \cap C) = B \cap AC$ ({\cite[Exercise~2.49]{Rot95}}).
        In particular, our claim follows from applying this statement twice, using the two assumed containments.}
\end{proof}

\subsection{Proof of \Cref{prop:unip:stair form}}

\begin{proof}[Proof of \Cref{prop:unip:stair form}]
    Induct on $t$.
    In the base case, $t= 0$, we have that $\GrStair{n}{\F}{\emptyset} = \GrTriII{n}{\F}{0}{n} = \GrUnip{n}{\F}$, so the claim holds trivially.
    Otherwise, let $W' = \{\ell_1,\ldots,\ell_{t-1}\}$ so that $W = W' \cup \{\ell\}$.
    Inductively, \[
    \GrStair{n}{\F}{W'} = \underbrace{\parens*{ \GrTriIE{n}{\F}{0}{\ell_1}
    \times \parens*{ \prod_{i=1}^{t-2} \GrTriIE{n}{\F}{\ell_i}{\ell_{i+1}}} }}_{\eqqcolon K_1}
    \times \underbrace{\GrTriII{n}{\F}{\ell_{t-1}}{n}}_{\eqqcolon K_2}. \]
    while by \Cref{prop:unip:one step stair}, \[
    \GrStair{n}{\F}{\ell_t} = \underbrace{\GrTriII{n}{\F}{0}{\ell_t-1}}_{\eqqcolon H_1}
    \times \underbrace{\GrTriII{n}{\F}{\ell_t}{n}}_{\eqqcolon H_2}. \]
    Note that by \Cref{prop:unip:triangle contain}, $K_1 \le H_1$ while $H_2 \le K_2$.
    Also, by \Cref{prop:unip:overlapping triangles}, \[
    H_1 \cap K_2 = \GrTriII{n}{\F}{\ell_{t-1}}{\ell_t-1}. \]
    Hence by the prior \Cref{prop:group:intersection of products},
    \begin{multline*}
    \GrStair{n}{\F}{W} = \GrStair{n}{\F}{W'} \cap \GrStair{n}{\F}{\ell_t} \\
    = \GrTriIE{n}{\F}{0}{\ell_1}
    \times \parens*{ \prod_{i=1}^{t-2} \GrTriIE{n}{\F}{\ell_i}{\ell_{i+1}}}
    \times \GrTriIE{n}{\F}{\ell_{t-1}}{\ell_t}
    \times \GrTriII{n}{\F}{\ell_t}{n}.
    \end{multline*}
    This is precisely the desired result.
\end{proof}

\section{Local-to-global theorem for triword expansion}\label{sec:local-to-global}

In this section, we prove \Cref{thm:dd}.
To simplify notations, we replace $\frY$ with $\X$, $\calJ$ with $\calI$, and $r$ with $k$.
We remark that when reading the proof, it may be helpful to first consider the case where $k = 1$.
In this case, for $S_1, S_2 \sqsubset_k \calI$, $S_1 \cap S_2 = \emptyset \iff S_1 \ne S_2$,
and our argument recovers the proof in~\cite[\S4]{DD24-local}.
For general $k$, we have to incorporate a few additional inequalities.

\subsection{Spectral expansion}

The following simple proposition expresses edge expansion of induced complexes assuming productness:

\begin{proposition}\label{claim:dd:edge expansion}
    Let $\X$ be an $\calI$-indexed simplicial complex which is $\epsilon k^{-2}$-product
    and suppose that $\calR$ is a set of pairwise disjoint $k$-sets in $\calI$
    and $S$ is a $k$-set in $\calI$ disjoint from all sets in $\calR$.
    Let $s \in \X[S]$ be an $S$-word, $\Gamma$ be a finite set, and $h : V(\X_s[\calR]) \to \Gamma$ a word-labeling by $\Gamma$.
    Then there exists $x \in \Gamma$ such that \[
    \Pr_{\bR \in \calR, \br \sim \X_s[\bR]}[h(\br) \ne x] \le \frac1{1-\epsilon} \Pr_{\bR_1 \ne \bR_2 \in \calR, (\br_1,\br_2) \sim \X_s[\bR_1,\bR_2]}[h(\br_1) \ne h(\br_2)]. \]
\end{proposition}

\begin{proof}
    Let $G$ denote the $1$-skeleton of the induced complex $\X_s[\calR]$.
    By \Cref{prop:gll:restrict}, $\X_s[\calR]$ is $\epsilon$-product.
    Hence the induced subgraph on each pair of parts in $G$
    (a.k.a. the bipartite graph $\resLink{\X}{R_1,R_2}{s}$ for $R_1,R_2 \in \calR$)
    is an $\epsilon$-bipartite expander and therefore a $1/(1-\epsilon)$-edge expander.
    Hence $G$ itself is also an $1/(1-\epsilon)$-edge expander by an averaging argument
    (and since each pairwise induced subgraph has the same total weight), as desired.
\end{proof}

In the sequel, we set $\epsilon_0 \coloneqq \frac12$ so that for $\epsilon \le \epsilon_0$, $\frac1{1-\epsilon} \le 2$.

\subsection{Choosing a good center $\calR$}

We first define a parameter $\gamma$, depending only on $f$ and $\X$, which we call the \emph{triword inconsistency} of $f$.
Consider the following experiment: Sample three uniformly random pairwise disjoint $k$-sets $\bS_1,\bS_2,\bS_3\sqsubset_k \calI$,
sample an $(\bS_1,\bS_2,\bS_3)$-triword $(\bs_1,\bs_2,\bs_3) \sim \res{\X}{\bS_1,\bS_2,\bS_3}$,
and define $\gamma \coloneqq \Pr[f(\bs_1,\bs_2) f(\bs_2,\bs_3) f(\bs_3,\bs_1) \ne \Id]$.

Further, in this setup, for any fixed set $\calR$ of $\ell$ distinct $k$-sets in $\calI$, and $j \in \{0,1,2,3\}$,
we let $\gamma^{(\calR,j)} \coloneqq \Pr[f(\bs_1,\bs_2) f(\bs_2,\bs_3) f(\bs_3,\bs_1) \ne \Id \mid \bS_{1:j} \in \calR, \bS_{j+1:3} \not\in \calR]$.
(Here, e.g. $\bS_{1:j} \in \calR$ denotes that $\bS_i \in \calR$ for every $i \in \rangeII{1}{j}$.)
Note that for any fixed $j$ and uniformly random $\bcalR$, the marginal distribution of 
$(\bS_1,\bS_2,\bS_3)$ conditioned on $\bS_{1:j} \in \bcalR, \bS_{j+1:3} \not\in \bcalR$ is uniformly random,
and so $\Exp_\bcalR[\gamma^{(\bcalR,j)}]=\gamma$.

By the law of total probability, we can therefore fix $\calR$ such that $\X[\calR]$ has triword expansion at least $\beta$ over $\Gamma$, and $\max\{\gamma^{(\calR,1)},\gamma^{(\calR,2)},\gamma^{(\calR,3)}\} \le 3\gamma/p$.
For the remainder of this subsection, we treat $\calR$ as fixed.
We write $\gamma^{(j)} \coloneqq \gamma^{(\calR,j)}$.

\subsection{The construction}

Now we define a word-labeling $g : V_k(\X) \to \Gamma$ in two steps:
\begin{enumerate}
\item Consider the restriction $f^* : \vec{E}(\X[\calR]) \to \Gamma$, which labels the $(R_1,R_2)$-biwords for $R_1 \ne R_2 \in \calR$.
The triword inconsistency of $f^*$ is $\gamma^{(\calR,3)} \le p \gamma$.

By the assumed triword expansion of $\X[\calR]$, there is word-labeling $g^* : V(\X[\calR]) \to \Gamma$, which labels $k$-words for $R \in \calR$, in the following sense:
Suppose that we sample $\bR_1\ne \bR_2 \in \calR$ uniformly, sample $(\br_1,\br_2) \sim \X[\bR_1,\bR_2]$,
and set $\eta^* \coloneqq \Pr[g^*(\br_1) g^*(\br_2)^{-1} \ne f(\br_1,\br_2)]$.
Then we have $\eta^* \le \gamma^{(3)}/\beta \le 3\gamma/(\beta p)$.

\item Now, for every $S \subset_k \calI$ with $S \not\in \calR$, for every $S$-word $s \in \X[S]$,
we define a word-labeling $h_s : V(\X_s[\{R \in \calR : R \cap S = \emptyset\}]) \to \Gamma$,
by setting $h_s(r) \coloneqq f(s,r) \cdot g^*(r)$.

Finally, we do the plurality decoding: $g(s) \coloneqq \argmax_{\gamma \in \Gamma} \Pr[h_s(r) = \gamma]$ (ties broken arbitrarily).
(We also define $g(r) \coloneqq g^*(r)$ for $r \in \X[R]$, $R \in \calR$).
\end{enumerate}

We now consider the biword inconsistency $\eta$ of $g$.
Specifically, consider sampling $\bS_1,\bS_2 \sqsubset_k \calI$ uniformly at random
and $(\bs_1,\bs_2) \sim \X[\bS_1,\bS_2]$ and set $\eta \coloneqq \Pr[g(\bs_1) g(\bs_2)^{-1} \ne f(\bs_1,\bs_2)]$.
We want to upper-bound $\eta$ multiplicatively relative to $\gamma$.

To do so, we consider the following quantity $\eta^{(1)}$: Sample $\bR, \bS \sqsubset_k \calI$ uniformly, sample $(\br,\bs) \sim \X[\bR,\bS]$,
and set $\eta^{(1)} \coloneqq \Pr[g(\br) g(\bs)^{-1} \ne f(\br,\bs) \mid \bR \in \calR, \bS \not\in \calR]$.

Similarly, define $\eta^{(2)}$ via: Sample $\bS_1, \bS_2 \sqsubset_k \calI$ uniformly conditioned on $\bS_1, \bS_2 \not\in \calR$, sample $(\bs_1,\bs_2) \sim \X[\bS_1,\bS_2]$,
and then set $\eta^{(2)} \coloneqq \Pr[g(\bs_1) g(\bs_2)^{-1} \ne f(\bs_1,\bs_2)]$.

Now we claim that $\eta$ is a convex combination of $\eta^*$, $\eta^{(1)}$, and $\eta^{(2)}$;
indeed, we can condition on the overlap size $|\{\bS_1,\bS_2\} \cap \calR|$,
and the probabilities conditioned on overlap sizes $0$, $1$, and $2$ are precisely $\eta^{(2)}$, $\eta^{(1)}$, and $\eta^*$, respectively.
We have already bounded $\eta^*$ relative to $\gamma$.
In the sequel, we bound $\eta^{(1)}$ and $\eta^{(2)}$.

\subsection{Qualitative analysis}

We first establish some simple implications involving the labelings $h_s$ and $g$.

\begin{claim}\label{claim:dd:h consistent}
Let $R_1, R_2, S \sqsubset_k \calI$ with $R_1,R_2 \in \calR$, $S \not\in \calR$,
and suppose that $(r_1,r_2,s) \in \X[R_1,R_2,S]$ is a valid $(R_1,R_2,S)$-triword.
Then: \[
f(r_1,r_2) f(r_2,s) f(s,r_1) = \Id 
\wedge g^*(r_1) g^*(r_2)^{-1} = f(r_1,r_2)
\quad\implies\quad h_s(r_1) = h_s(r_2). \]
\end{claim}
\begin{proof}
    By definition, $h_s(r_1) = h_s(r_2) \iff f(s,r_1) \cdot g^*(r_1) = f(s,r_2) \cdot g^*(r_2) \iff f(r_2,s) \cdot f(s,r_1) \cdot g^*(r_1) \cdot g^*(r_2)^{-1} = \Id$.
    The two hypotheses finish the proof.
\end{proof}

\begin{claim}\label{claim:dd:eta 1}
Let $R, S \sqsubset_k \calI$ with $R \in \calR$ and $S \not\in \calR$,
and suppose that $(r,s) \in \X[R,S]$ is a valid $(R,S)$-biword.
Then \[ g(s) = h_s(r) \quad\implies\quad g(r) g(s)^{-1} = f(r,s). \]
\end{claim}
\begin{proof}
    By definition, $g(r) = g^*(r)$ and $h_s(r) = f(s,r) \cdot g^*(r)$.
    Hence by hypothesis, $g(r) g(s)^{-1} = g(r) h_s(r)^{-1} = f(r,s)$, as desired.
\end{proof}

\begin{claim}\label{claim:dd:eta 2}
Suppose $R,S_1,S_2 \sqsubset_k \calI$ with $R \in \calR$ and $S_1,S_2 \not\in \calR$,
and suppose that $(r,s_1,s_2) \in \X[R,S_1,S_2]$ is a valid $(R,S_1,S_2)$-triword.
Then: \[
g(s_1) = h_{s_1}(r) \wedge g(s_2) = h_{s_2}(r) \wedge f(s_1,s_2) f(s_2,r) f(r,s_1) = \Id
\quad\implies\quad g(s_1) g(s_2)^{-1}  = f(s_1,s_2). \]
\end{claim}

\begin{proof}
    By the first hypothesis, $g(s_1)g(s_2)^{-1} = h_{s_1}(r) h_{s_2}(r)^{-1}$.
    By definition, we have $h_{s_1}(r) h_{s_2}(r)^{-1} = f(s_1,r) g^*(r) g^*(r)^{-1} f(r,s_2) = f(s_1,r) f(r,s_2)$.
    The second hypothesis completes the proof.
\end{proof}

\subsection{Sampling sets}

Now, we turn to proving some simple lemmas about sampling $k$-sets conditioned on disjointness hypotheses.
We stress that this section is unnecessary in the $k=1$ case since disjointness and distinctness are equivalent.
\newcommand{\bomega}{\boldsymbol \omega}

Given two distributions $X$ and $Y$ on a finite set $\Omega$ and $\lambda \ge 1$, we say $X$ and $Y$ are
\emph{$\lambda$-multiplicatively close} if for every $\omega \in \Omega$, $\lambda^{-1} \Pr_{\bomega\sim Y}[\bomega=\omega] \le \Pr_{\bomega\sim X}[\bomega=\omega] \le \lambda \Pr_{\bomega\sim Y}[\bomega=\omega]$.
Note that this is a symmetric relationship, and that it implies that for every nonnegative function $f : \Omega \to \R_{\ge 0}$, we have
$\lambda^{-1} \Exp_{\bomega\sim Y}[f(\bomega)] \le \Exp_{\bomega\sim X}[f(\bomega)] \le \lambda\Exp_{\bomega\sim Y}[f(\bomega)]$.

\begin{claim}
    Let $G = (L,R,E)$ be an (unweighted) bipartite graph.
    Consider the following two distributions over $E$:
    \begin{enumerate}
        \item Sample $(\bu,\bv) \in E$ uniformly.
        \item Sample $\bv \in L$ uniformly, then sample $\bu \in L$ uniformly conditioned on $(\bu,\bv) \in E$.
    \end{enumerate}
    For $v \in R$, let $\deg(v) \coloneqq |\{ u \in L: (u,v) \in E \}|$.
    The two above distributions are $(\max_v \deg(v))/(\min_v \deg(v))$-multiplicatively close.
\end{claim}

\begin{proof}
    The probability of an edge $(u,v)$ in the first distribution is $1/|E|$
    and the probability in the second distribution is $1/(|R| \cdot \deg(v))$.
    Note also that $|E| = \sum_{v \in R} \deg(v)$.
    Hence the ratio of probabilities is $|L| \cdot \deg(v) / (\sum_{v' \in R} \deg(v'))$,
    which is e.g. uniformly upper-bounded by $|R| \cdot (\max_v \deg(v)) / (\sum_{v \in R} \min_{v'} \deg(v')) = (\max_v \deg(v))/(\min_v \deg(v))$.
\end{proof}

\begin{claim}\label{claim:dd:sample RS}
Consider the following distributions on pairs $(R,S)$ with $R, S \sqsubset_k \calI$, $R \in \calR$, and $S \not\in \calR$:
\begin{enumerate}
\item Sample $\bR, \bS \subset_k \calI$ uniformly and independently, conditioned on $\bR \in \calR$, $\bS \not\in \calR$, and $\bR \cap \bS = \emptyset$.
\item First, sample $\bS \subset_k \calI$ uniformly conditioned on $\bS \not\in \calR$.
Then, sample $\bR \in \calR$ uniformly conditioned on $\bR \cap \bS = \emptyset$.
\end{enumerate}
Then these two distributions are $3/2$-multiplicatively close.
\end{claim}

\begin{proof}
    Form the bipartite graph where the left-vertices are $\calR$,
    the right-vertices are $S \subset_k \calI$ with $S \not\in \calR$,
    and there is an edge between $R$ and $S$ if $R \cap S = \emptyset$.
    The maximum right-degree is at most $\ell$ (the total number of left-vertices)
    and the minimum right-degree is at least $\ell-k$ since any $k$-set can intersect at most $k$ sets in $\calR$ (since the sets in $\calR$ are pairwise disjoint).
    Finally, $\ell/(\ell-k) \le 3/2$ by the assumption that $\ell \ge 3k$.
\end{proof}

\begin{claim}\label{claim:dd:sample RRS}
Consider the following distributions on pairs $(R_1,R_2,S)$ with $R_1, R_2, S \sqsubset_k \calI$, $R_1, R_2 \in \calR$, $S \not\in \calR$:
\begin{enumerate}
\item Sample $\bR_1, \bR_2, \bS \subset_k \calI$ uniformly and independently, conditioned on $\bR_1, \bR_2 \in \calR$, $\bS \not\in \calR$, and $\bR_1,\bR_2,\bS$ are pairwise disjoint.
\item 
First, sample $\bS \subset_k \calI$ uniformly conditioned on $\bS \not\in \calR$.
Then, sample $\bR_1, \bR_2 \in \calR$ uniformly conditioned on $\bR_1 \ne \bR_2$, $\bR_1 \cap \bS = \emptyset$, and $\bR_2 \cap \bS = \emptyset$.
\end{enumerate}
Then these two distributions are $3$-multiplicatively close.
\end{claim}

\begin{proof}
    Form the bipartite graph where left-vertices are ordered pairs $(R_1,R_2)$ with $R_1 \ne R_2 \in \calR$,
    right-vertices are $S \subset_k \calI$ with $S \not\in \calR$,
    and there is an edge between $(R_1,R_2)$ and $S$ if $(R_1 \cup R_2) \cap S = \emptyset$.
    The maximum right-degree is at most $\ell(\ell-1)$ while the minimum right-degree $(\ell-k)(\ell-k-1)$,
    and the ratio is $(\ell(\ell-1))/((\ell-k)(\ell-k-1)) \le (3k(3k-1))/(2k(2k-1)) \le 3$.
\end{proof}


\begin{claim}\label{claim:dd:sample RSS}
Consider the following distributions on triples $(R, S_1,S_2)$ with $R, S_1, S_2 \sqsubset_k \calI$, $R \in \calR$, and $S_1,S_2 \not\in \calR$:
\begin{enumerate}
\item Sample $\bR, \bS_1, \bS_2 \sqsubset_k \calI$ uniformly
conditioned on $\bR \in \calR, \bS_1, \bS_2 \not\in \calR$.
\item Sample $\bS_1, \bS_2 \sqsubset_k \calI$ uniformly conditioned on $\bS_1, \bS_2 \not\in \calR$.
Then, sample $\bR \in \calR$ uniformly conditioned on $\bR \cap (\bS_1 \cup \bS_2) = \emptyset$.
\end{enumerate}
Then these two distributions are $3$-multiplicatively close.
\end{claim}

\begin{proof}
    Form the bipartite graph where the left-vertices are $\calR$,
    the left-vertices are $(S_1,S_2)$ with $S_1, S_2 \sqsubset_k \calI$ and $S_1, S_2 \not\in \calR$,
    and there is an edge between $R$ and $(S_1,S_2)$ if $R \cap (S_1 \cup S_2) = \emptyset$.
    The maximum right-degree is at most $\ell$ and the minimum right-degree is at least $\ell-2k$, which yields a ratio of $\ell/(\ell-2k) \le 3$.
\end{proof}

\subsection{Quantitative analysis: $\eta^{(1)}$}

To analyze $\eta^{(1)}$, we define a few more quantities.
For fixed $S \subset_k \calI$ with $S \not\in \calR$, we consider:
\begin{itemize}
\item Sample $\bR \in \calR$ uniformly conditioned on $\bR \cap S = \emptyset$,
sample $(\br,\bs) \sim \X[\bR,S]$,
and let $\eta_S \coloneqq \Pr[g(\br) g(\bs)^{-1} \ne f(\br,\bs)]$
and $\psi_S \coloneqq \Pr[g(\bs) \ne h_\bs(\br)]$.

\item Sample $\bR_1 \ne \bR_2 \in \calR$ uniformly conditioned on $(\bR_1 \cup \bR_2) \cap S = \emptyset$,
sample $(\br_1,\br_2,\bs) \sim \X[\bR_1,\bR_2,S]$,
and let $\theta_S \coloneqq \Pr[h_\bs(\br_1) \ne h_\bs(\br_2)]$ and $\gamma_S \coloneqq \Pr[f(\br_1,\br_2) f(\br_2,\bs) f(\bs,\br_1) \ne \Id]$.

\item Sample $\bR_1 \ne \bR_2 \in \calR$ uniformly conditioned on $(\bR_1 \cup \bR_2) \cap S = \emptyset$,
sample $(\br_1,\br_2) \sim \X[\bR_1,\bR_2]$,
and let $\eta^*_S \coloneqq \Pr[g(\br_1) g(\br_2)^{-1} \ne f(\br_1,\br_2)]$.
\end{itemize}

By \Cref{claim:dd:eta 1}, $\eta_S \le \psi_S$.
By \Cref{claim:dd:edge expansion}, $\psi_S \le 2 \theta_S$.
By \Cref{claim:dd:h consistent},
$\theta_S \le \gamma_S + \eta^*_S$.
Hence $\eta_S \le 2(\gamma_S + \eta^*_S)$.

Now, consider sampling $\bS \subset_k \calI$ uniformly conditioned on $\bS \not\in \calR$.
By \Cref{claim:dd:sample RS}, $\eta^{(1)} \le \frac32 \Exp \eta_\bS$.
By \Cref{claim:dd:sample RRS}, $\Exp \gamma_\bS \le 3\gamma^{(2)}$ and $\Exp \eta^*_\bS \le 3\eta^*$.
These combine to give $\Exp \eta_\bS \le 6 (\gamma^{(2)} + \eta^*)$
and therefore $\eta^{(1)} \le 9 (\gamma^{(2)} + \eta^*)$.

\subsection{Quantitative analysis: $\eta^{(2)}$}

Again, we define some quantities.
For fixed $S_1, S_2 \sqsubset_k \calI$ with $S_1,S_2 \not\in \calR$:
\begin{itemize}
\item Sample $(\bs_1,\bs_2) \sim \X[S_1,S_2]$,
and let $\eta_{S_1,S_2} \coloneqq \Pr[g(\bs_1) g(\bs_2)^{-1} \ne f(\bs_1,\bs_2)]$.

\item Sample $\bR \sim \calR$ conditioned on $\bR \cap (S_1 \cup S_2) = \emptyset$, sample $(\bs_1,\bs_2,\br) \sim \X[S_1,S_2,\bR]$,
and let $\gamma_{S_1,S_2} \coloneqq \Pr[f(\bs_1,\bs_2) f(\bs_2,\br) f(\br,\bs_1) \ne \Id]$.
\end{itemize}

By \Cref{claim:dd:eta 2}, $\eta_{S_1,S_2} \le \eta_{S_1} + \eta_{S_2} + \gamma_{S_1,S_2}$. 

Now, consider sampling $\bS_1,\bS_2 \sqsubset_k \calI$ conditioned on $\bS_1,\bS_2\not\in \calR$.
By definition, $\eta^{(2)} = \Exp \eta_{\bS_1,\bS_2}$.
Hence $\eta^{(2)} \le \Exp \eta_{\bS_1} + \Exp \eta_{\bS_2} + \Exp \gamma_{\bS_1,\bS_2}$.
By \Cref{claim:dd:sample RRS}, $\Exp \gamma_{\bS_1,\bS_2} \le 3\gamma^{(1)}$.
Further, noting that the marginal distribution of each $\bS_i$ is the same as sampling $\bS_i \subset_k \calI$ conditioned on $\bS_i \not\in \calR$,
we get $\Exp \eta_{\bS_i} \le 6(\gamma^{(2)} + \eta^*)$.
Altogether, we get $\eta^{(2)} \le 12(\gamma^{(2)} + \eta^*) + 3\gamma^{(1)}$.

\subsection{Finishing the proof}

Finally, we give:

\begin{proof}[Proof of \Cref{thm:dd}]
    By convexity, $\eta \le \max\{\eta^*, \eta^{(1)}, \eta^{(2)}\}$.
    We upper-bounded the second term by $9(\gamma^{(2)} + \eta^*)$ and the third term by $12(\gamma^{(2)} + \eta^*) + 3\gamma^{(1)}$.
    Using that $\eta^* \le \gamma^{(3)}/\beta$ and $\max\{\gamma^{(1)},\gamma^{(2)},\gamma^{(3)}\} \le 3\gamma/p$,
    we get a final bound of $9\gamma(5+4/\beta)/p \le 117\gamma/(\beta p)$.
    This gives a final expansion bound of $\beta p/117$, as desired.
\end{proof}

\end{document}